\documentclass[11pt]{article}

\def\final{0}  %

\usepackage{fullpage}
\usepackage[utf8]{inputenc}

\usepackage{graphicx}
\usepackage{amsmath,amsthm,amssymb}
\usepackage{algorithm}

\usepackage{algpseudocode}
\usepackage{multirow}
\usepackage{makecell}

\usepackage{thm-restate,color,xspace}
\usepackage{comment}
\usepackage{thmtools}
\usepackage{xcolor}
\usepackage{nameref}
\usepackage{array}
\usepackage{dsfont}
\usepackage{enumitem}
\usepackage{tikz}
\usepackage{mathrsfs}
\usetikzlibrary{shapes.geometric, arrows, positioning, calc}

\ifnum\final=0  %
\newcommand{\todo}[1]{{\color{red}[{\tiny TODO: \bf #1}]\marginpar{\color{red}*}}}
\newcommand{\yonggang}[1]{{\color{blue}[{\tiny Yonggang: \bf #1}]\marginpar{\color{blue}*}}}
\newcommand{\thatchaphol}[1]{\textcolor{purple}{thatchaphol: #1}}
\newcommand{\shengzhe}[1]{\textcolor{brown}{shengzhe: #1}}
\newcommand{\yaowei}[1]{\textcolor{orange}{yaowei: #1}}
\else %
\newcommand{\yonggang}[1]{}
\newcommand{\todo}[1]{}
\newcommand{\thatchaphol}[1]{}
\newcommand{\shengzhe}[1]{}
\newcommand{\yaowei}[1]{}
\fi

\definecolor{ForestGreen}{rgb}{0.1333,0.5451,0.1333}
\definecolor{DarkRed}{rgb}{0.65,0,0}
\definecolor{Red}{rgb}{1,0,0}
\usepackage[linktocpage=true,
pagebackref=true,colorlinks,
linkcolor=DarkRed,citecolor=ForestGreen,
bookmarks,bookmarksopen,bookmarksnumbered]
{hyperref}
\usepackage{cleveref}
\usepackage[T1]{fontenc}

\declaretheorem[numberwithin=section]{theorem}

\declaretheorem[numberlike=theorem]{lemma}

\declaretheorem[numberlike=theorem]{fact}

\declaretheorem[numberlike=theorem]{corollary}

\declaretheorem[numberlike=theorem]{claim}

\declaretheorem[numberlike=theorem]{observation}

\declaretheorem[numberlike=theorem,style=definition]{definition}
\declaretheorem[numberlike=theorem,style=definition]{remark}

\declaretheorem[numberlike=theorem,name=Open Problem]{question}

\global\long\def\leng{\mathrm{leng}}
\global\long\def\cong{\mathrm{cong}}
\global\long\def\step{\mathrm{step}}
\global\long\def\diam{\mathrm{diam}}
\global\long\def\Emulator{\textsc{Shortcut}}

\global\long\def\cov{\mathrm{cov}}
\global\long\def\ssmall{\mathrm{small}}
\global\long\def\big{\mathrm{large}}

\global\long\def\path{\mathrm{path}}
\global\long\def\wtilde{\widetilde}

\global\long\def\Dem{\mathrm{Dem}}
\global\long\def\jump{\mathrm{jump}}

\global\long\def\vvalue{\mathrm{value}}

\global\long\def\shortcut{\mathrm{shortcut}}
\global\long\def\HvLg{\mathrm{HvLg}}
\global\long\def\light{\mathrm{light}}
\global\long\def\pre{\mathrm{pre}}
\global\long\def\suf{\mathrm{suf}}

\global\long\def\sc{\mathrm{sc}}

\global\long\def\poly{\mathrm{poly}}

\global\long\def\sshort{\mathrm{short}}
\global\long\def\best{\mathrm{best}}

\global\long\def\totlen{\mathrm{totlen}}
\global\long\def\supp{\mathrm{supp}}
\global\long\def\dir{\mathrm{dir}}
\global\long\def\iin{\mathrm{in}}
\global\long\def\out{\mathrm{out}}
\global\long\def\sc{\mathrm{sc}}
\global\long\def\l{\ell}%
\global\long\def\ed{\mathrm{ed}}
\global\long\def\tepsk{\tilde{\epsilon_{\kappa}}}
\global\long\def\epsh{\epsilon_{h}}
\global\long\def\sg{\mathrm{sg}}
\global\long\def\mt{\mathrm{mt}}
\global\long\def\true{\mathrm{true}}
\global\long\def\old{\mathrm{old}}
\global\long\def\new{\mathrm{new}}
\global\long\def\output{\mathrm{output}}
\global\long\def\polylog{\mathrm{polylog}}

\newcommand{\dist}{\mathrm{dist}}

\newcommand{\ADLSC}{\textsc{ApxDLSC}}

\newcommand{\spa}{\operatorname{spars}}
\newcommand{\sep}{\operatorname{sep}}

\newcommand{\cond}{\operatorname{cond}}

\newcommand{\conge}{\operatorname{cong}}

\renewcommand{\l}{\ell}
\newcommand{\val}{\operatorname{val}}

\newcommand{\qLC}{\text{LC}}

\newcommand{\qLEC}{\text{LEC}}

\newcommand{\qLWSC}{\text{LWSC}}
\newcommand{\qLDSC}{\text{DLSC}}
\newcommand{\qLDSCS}{\text{DLSCS}}

\newcommand{\epo}{\mathsf{epoch}}

\newcommand{\LDSC}{\operatorname{DLSC}}
\newcommand{\LDSCS}{\operatorname{DLSCS}}

\newcommand{\eps}{\epsilon}
\newcommand{\Start}{\mathrm{start}}
\newcommand{\End}{\mathrm{end}}
\newcommand{\Succ}{\mathrm{succ}}
\newcommand{\Prec}{\mathrm{prec}}
\newcommand{\copies}{\mathrm{Copies}}
\newcommand{\pg}{\textnormal{\textsf{pg}}\xspace}

\newcommand{\epsl}{\epsilon_\ell}
\newcommand{\epsk}{\epsilon_\kappa}

\newcommand{\cK}{{\mathcal K}}

\newcommand{\cN}{{\mathcal N}}
\newcommand{\cO}{{\mathcal O}}
\newcommand{\cP}{{\mathcal P}}

\newcommand{\cS}{{\mathcal S}}

\newcommand{\mcD}{\mathcal{D}}

\newcommand{\mcN}{\mathcal{N}}

\newcommand{\bbN}{{\mathbb N}}

\newcommand{\bbR}{{\mathbb R}}

\newcommand{\tO}[1]{\widetilde{O}\left(#1\right)}

\newcommand{\Ohat}{\hat{O}}

\newcommand{\tA}{\tilde{A}}
\newcommand{\tB}{\tilde{B}}

\newcommand{\hA}{\hat{A}}
\newcommand{\hB}{\hat{B}}

\newcommand{\hP}{\hat{P}}

\title{Parallel $(1+\epsilon)$-Approximate Multi-Commodity Mincost Flow in Almost Optimal Depth and Work}
\author{Bernhard Haeupler\thanks{
        INSAIT, Sofia University ``St.~Kliment Ohridski'' and ETH Zürich,
        \texttt{bernhard.haeupler@insait.ai}.
        Partially funded by the Ministry of Education and Science of Bulgaria's support for INSAIT as part of the Bulgarian National Roadmap for Research Infrastructure and through the European Research Council (ERC) under the European Union's Horizon 2020 research and innovation program (ERC grant agreement 949272).}
       \and 
       Yonggang Jiang\thanks{MPI-INF and Saarland University, \texttt{yjiang@mpi-inf.mpg.de}.} 
       \and 
       Yaowei Long\thanks{University of Michigan, \texttt{yaoweil@umich.edu}. Part of this work was done while at INSAIT, Sofia University "St. Kliment Ohridski", Bulgaria. Partially funded by the Ministry of Education and Science of Bulgaria's support for INSAIT, Sofia University ``St.~Kliment Ohridski'' as part of the Bulgarian National Roadmap for Research Infrastructure.} 
       \and 
       Thatchaphol Saranurak\thanks{
        University of Michigan,
        \texttt{thsa@umich.edu}.
        Supported by NSF Grant CCF-2238138. Partially funded by the Ministry of Education and Science of Bulgaria's support for INSAIT, Sofia University ``St.~Kliment Ohridski'' as part of the Bulgarian National Roadmap for Research Infrastructure.
    }  
    \and 
    Shengzhe Wang\thanks{ETH Zürich, \texttt{shengzhe.wang@inf.ethz.ch}.}
    }
\date{}

\begin{document}

\maketitle

\pagenumbering{gobble}

\begin{abstract}

We present a parallel algorithm for computing $(1+\epsilon)$-approximate mincost flow on an undirected graph with $m$ edges, where capacities and costs are assigned to both edges and vertices. Our algorithm achieves $\Ohat(m)$ work and $\Ohat(1)$ depth when $\epsilon > 1/\polylog(m)$, making both the work and depth almost optimal, up to a subpolynomial factor.

Previous algorithms with $\Ohat(m)$ work required $\Omega(m)$ depth, even for special cases of mincost flow with only edge capacities or max flow with vertex capacities. Our result generalizes prior almost-optimal parallel $(1+\epsilon)$-approximation algorithms for these special cases, including shortest paths \cite{DBLP:conf/stoc/Li20,DBLP:conf/stoc/AndoniSZ20,DBLP:conf/stoc/RozhonGHZL22} and max flow with only edge capacities \cite{li2023near,DBLP:conf/soda/AgarwalKLPWWZ24}.

Our key technical contribution is the first construction of \emph{length-constrained flow shortcuts} with $(1+\epsilon)$ length slack, $\Ohat(1)$ congestion slack, and $\Ohat(1)$ step bound. This provides a strict generalization of the influential concept of $(\Ohat(1),\epsilon)$-hopsets \cite{DBLP:journals/jacm/Cohen00}, allowing for additional control over congestion. Previous length-constrained flow shortcuts \cite{DBLP:conf/stoc/HaeuplerH0RS24} incur a large constant in the length slack, which would lead to a large approximation factor.
To enable our flow algorithms to work under vertex capacities, we also develop a close-to-linear time algorithm for computing length-constrained \emph{vertex} expander decomposition.

Building on Cohen's idea of \emph{path-count flows} \cite{DBLP:journals/siamcomp/Cohen95}, we further extend our algorithm to solve $(1+\epsilon)$-approximate $k$-commodity mincost flow problems with almost-optimal $\Ohat(mk)$ work and $\Ohat(1)$ depth, independent of the number of commodities $k$.

\end{abstract}

\newpage
\setcounter{tocdepth}{2}
\tableofcontents

\newpage

\pagenumbering{arabic}

\section{Introduction}
\label{sect:Introduction}

The \emph{mincost flow} problem is one of the most fundamental combinatorial optimization problems and has been extensively studied since the 60s \cite{fulkerson1961out,DBLP:journals/jacm/EdmondsK72,DBLP:journals/combinatorica/Tardos85,orlin1984genuinely,DBLP:journals/jacm/GalilT88,DBLP:journals/mor/GoldbergT90,DBLP:journals/mp/AhujaGOT92,DBLP:journals/ior/Orlin93,DBLP:journals/mp/Orlin97,DBLP:conf/stoc/DaitchS08,DBLP:conf/stoc/BrandLLSS0W21,DBLP:conf/focs/AxiotisMV21,DBLP:conf/stoc/BrandGJLLPS22}, which recently culminates in the breakthrough almost-linear-time algorithm \cite{DBLP:conf/focs/ChenKLPGS22}. However, the aforementioned algorithms were mainly designed in the \emph{sequential} setting, and generally cannot be parallelized in an efficient way. 

This motivates the next natural goal of obtaining \emph{parallel} algorithms with almost-optimal\footnote{By convention, the term \emph{almost} means a subpolynomial overhead, e.g. almost-linear in $n$ means $n^{1+o(1)}$. Moreover, \emph{nearly} means a polylogarithmic overhead, and \emph{close-to} means an $n^{\epsilon}$ overhead. We use $\tilde{O}(\cdot)$ and $\hat{O}(\cdot)$ to hide polylogarithmic and subpolynomial factors.} work and depth. Previously, there have been exciting achievements for two special cases of the mincost flow problem, the \emph{$s$-$t$ shortest path} and \emph{maximum flow} problems, in undirected graphs. For shortest paths, the 25-year-old breakthrough by \cite{DBLP:journals/jacm/Cohen00} showed the first $(1+\epsilon)$-approximate algorithm with almost-linear work and subpolynomial depth, and recently, $(1+\epsilon)$-approximate algorithms with nearly-linear work and polylogarithmic depth were developed \cite{DBLP:conf/stoc/Li20,DBLP:conf/stoc/AndoniSZ20,DBLP:conf/stoc/RozhonGHZL22}. 
For maximum flows, the state-of-the-art $(1+\epsilon)$-approximate algorithm also attained nearly-linear work and polylogarithmic depth \cite{DBLP:conf/soda/AgarwalKLPWWZ24}. 
However, there is still no comparable result for either
\begin{enumerate}
\item\label{intro:problem1} the more generalized \emph{mincost flow} problem in undirected graphs (with edge capacities), or 
\item\label{intro:problem2} still the maximum flow problem but in undirected graphs with \emph{vertex capacities}.
\end{enumerate}

Solving either of the above problems will be challenging, and in particular, the techniques in the aforementioned works cannot be directly transferred to either of these two problems. The mincost flow problem requires a joint $\ell_{1}$ (lengths/costs) and $\ell_{\infty}$ (congestion) optimization, while previous techniques including hopsets \cite{DBLP:journals/jacm/Cohen00}, $\ell_{1}$-oblivious routings \cite{DBLP:conf/stoc/Li20,DBLP:conf/stoc/AndoniSZ20,DBLP:conf/stoc/RozhonGHZL22} and congestion approximators based on $\ell_{\infty}$-oblivious routings \cite{li2023near,DBLP:conf/soda/AgarwalKLPWWZ24} only capture either $\ell_{1}$ or $\ell_{\infty}$ information. Also, in vertex-capacitated undirected graphs, it is known that good $\ell_{\infty}$-oblivious routing for vertex-capacitated graphs does \emph{not} exist \cite{DBLP:journals/talg/HajiaghayiKRL07}. This gives a fundamental barrier to previous oblivious-routing-based techniques \cite{DBLP:conf/soda/AgarwalKLPWWZ24}. To summarize, the following natural question is still open, and answering it may require new approaches.

\begin{center}
\textit{Are there parallel $(1+\epsilon)$-approximate algorithms for either problems \ref{intro:problem1} or \ref{intro:problem2}\\ with almost-linear work and subpolynomial depth?}
\end{center}

\subsection{Our Results} 
We answer the above open question affirmatively. In particular, we present the first parallel $(1+\epsilon)$-approximate mincost flow algorithm with almost-linear work and subpolynomial depth for undirected graphs with vertex capacities and costs, which solves problems \ref{intro:problem1} and \ref{intro:problem2} simultaneously.

Formally, consider an $n$-vertex $m$-edge graph $G = (V,E)$ with non-negative polynomially-bounded integral capacities and costs on both edges \emph{and vertices}. Let $B$ be a given number representing the cost budget and $s,t\in V$ be two given vertices. The $(1+\epsilon)$-approximate mincost flow problem asks for a $(1+\epsilon)$-approximate maximum flow among all flows that have total cost at most $B$ and respect the capacity on each edge and vertex. Our result is formally stated in \Cref{thm:IntroMainResult1}.

\begin{theorem}
\label{thm:IntroMainResult1}
Given a graph $G$, cost budget $B$, two vertices $s,t\in V$, and a precision parameter $1/\polylog(n)< \epsilon<1$, there exists a parallel $(1+\epsilon)$-approximate mincost flow algorithm with $\hat{O}(m)$ work and $\hat{O}(1)$ depth.
\end{theorem}

\paragraph{Barriers towards Exact Algorithms or Directed Graphs.} 
It is noteworthy that simultaneously working in undirected graphs (with vertex capacities and costs) and allowing $(1+\epsilon)$-approximation is the strongest setting in which efficient parallel mincost flow algorithms may be achievable based on current knowledge. Considering either \emph{directed} graphs or \emph{exact} solutions will run up against the following long-standing barriers.
\begin{itemize}
\item The parallel (approximate) mincost flow problem in \emph{directed} graphs subsumes the basic \emph{parallel 
$s$-$t$ reachability} problem. Parallel $s$-$t$ reachability is a big open problem studied over decades and we are still far from knowing whether algorithms with almost-linear work and subpolynomial depth exist. The state-of-the-art almost-linear-work algorithm has $O(n^{0.5+o(1)})$ depth \cite{DBLP:conf/focs/LiuJS19}.
\item An efficient parallel \emph{exact} mincost flow in undirected graphs faces the barrier even for two special cases: the \emph{maximum flow} and \emph{$s$-$t$ shortest path} problems. Computing an exact maximum flow in undirected graphs is as hard as in directed graphs \cite{mkadry2011graphs}, so it again suffers from the aforementioned parallel reachability barrier. For the $s$-$t$ shortest path problem, even in the simplest undirected, unweighted graphs, the state-of-the-art almost-linear-work algorithms still have $O(n^{0.5+o(1)})$ depth \cite{DBLP:conf/stoc/RozhonHMGZ23,DBLP:conf/soda/CaoF23}.

\end{itemize}
Indeed, there have been attempts to design parallel exact mincost flow algorithms in directed graphs, but they either suffer from a large polynomial work \cite{DBLP:journals/orl/OrlinS93} or get stuck in roughly $\sqrt{n}$ depth \cite{DBLP:conf/focs/LeeS14,brand2025parallel}.

\paragraph{Multi-Commodity Mincost Flows.} 
There has been a long line of work \cite{shahrokhi1990maximum,leighton1991fast,radzik1996fast,grigoriadis1996approximate,nesterov2009fast,madry2010faster,kelner2012faster,sherman2017area,cohen2021solving,DBLP:conf/soda/ChuzhoyS21,chen2024high,chen2025accelerated} devoted to fast algorithms for computing \emph{multi-commodity flow}. See \Cref{sect:MincostFlowByBoosting} for the formal problem statements.

Quite recently, algorithms for multi-commodity flow with near-optimal work have been discovered. In the discussion below, we assume that $\epsilon > 1/\polylog(m)$. The first algorithm, presented by Sherman \cite{sherman2017area}, is a $(1+\epsilon)$-approximation algorithm for computing $k$-commodity concurrent flow, achieving near-optimal $\tilde{O}(mk)$ work. An alternative approach is to apply known reductions from \cite{radzik1996fast,grigoriadis1996approximate}, which reduce the problem of $(1+\epsilon)$-approximate $k$-commodity concurrent (mincost) flow to $\tilde{O}(k/\epsilon^{2})$ calls of \emph{single-commodity} $(1+\epsilon)$-approximate mincost flow. Plugging in the recent exact $\Ohat(m)$-time mincost flow algorithm \cite{DBLP:conf/focs/ChenKLPGS22}, this also yields an algorithm with $\Ohat(mk)$ work. Unfortunately, none of these approaches provide parallel algorithms with non-trivial sublinear depth.

Notably, by applying \Cref{thm:IntroMainResult1} to the above reduction, we immediately obtain a $(1+\epsilon)$-approximate $k$-commodity concurrent (mincost) flow algorithm with $\Ohat(mk)$ work and $\Ohat(k)$ depth. To our knowledge, this already gives the first parallel multi-commodity flow algorithm with almost-optimal work and non-trivial depth, motivating us to further investigate whether $\Ohat(k)$ depth is inherent.

\begin{center}

\emph{Is there a parallel $(1+\epsilon)$-approximate multi-commodity flow algorithm with almost-optimal work and depth that is sublinear in the number of commodities $k$?}

\par\end{center}

Conceptually, this question asks how much different commodities in a near-optimal flow must ``depend'' on each other. To our knowledge, the only related previous work is \cite{DBLP:conf/stoc/HaeuplerH0RS24}, which shows a parallel $2^{\poly(1/\epsilon)}$-approximation algorithm with $O((m+k)^{1+\epsilon})$ work\footnote{In fact, when the algorithm is required to output a flow solution explicitly, there is an $\Omega(mk)$ work lower bound due to the output size. Their algorithm bypasses this by outputting a flow solution \emph{implicitly} via a flow oracle.} and $O((m+k)^{\epsilon})$ depth for $k$ commodities and arbitrary constant $\epsilon>0$. The significant drawback of \cite{DBLP:conf/stoc/HaeuplerH0RS24} is that its approximation factor is a larger constant, significantly above $1+\epsilon$, and thus unsatisfactory in our context.

We provide a strong affirmative answer to the above question by showing an almost optimal algorithm with depth $\Ohat(1)$:

\begin{theorem}
\label{thm:IntroResult2}
Let $G$ be an undirected graph with non-negative polynomially-bounded integral costs and capacities. Given a precision parameter $1/\poly\log n< \epsilon< 1$, there are $(1+\epsilon)$-approximate algorithms for both the $k$-commodity concurrent and non-concurrent mincost flow problems with work $\hat{O}(mk)$ and depth $\hat{O}(1)$. The flow solution is outputted explicitly in its edge representation.
\end{theorem}

\Cref{thm:IntroResult2} gives a parallel $(1+\epsilon)$-approximation multi-commodity flow algorithm with almost optimal depth and work.

We point out that our algorithms are randomized, but they have good potential to be deterministic since the only randomized ingredients are parallel algorithms related to sparse neighborhood covers. We refer to \Cref{remark:deterministic} for a discussion.

\subsection{Our Techniques}

We now highlight our technical components that may be of independent interest. An overview of how everything fits together will be given in \Cref{sec:overview}.

\paragraph{Length-Constrained Flow Shortcuts with $(1+\epsilon)$ Length Slack.} The notion of \emph{length-constrained flow shortcuts}\footnote{They are called \emph{low-step flow emulator} in \cite{DBLP:conf/stoc/HaeuplerH0RS24}.} (LC-flow shortcuts) are introduced recently by \cite{DBLP:conf/stoc/HaeuplerH0RS24}, which is a new kind of graph augmentation capturing both length and congestion information. Roughly speaking, an LC-flow shortcut augments a graph with additional edges and vertices, enabling all multi-commodity flow demands to be routed with approximately the same length and congestion as before, but only along flow paths each of which has a small number (called \emph{step}) of edges.

\cite{DBLP:conf/stoc/HaeuplerH0RS24} showed the existence of an LC-flow shortcut of size $O(n^{1+O(\epsilon)})$ with length slack $O(1/\epsilon^{3})$, congestion slack $n^{O(\epsilon)}$ and $O(1/\epsilon)$ step\footnote{In fact, the shortcut in \cite{DBLP:conf/stoc/HaeuplerH0RS24} has length slack $O(1/\epsilon^{4})$ and step $O(1/\epsilon^{2})$, but they ensure that all shortcut edges have endpoints inside $V$. When allowing endpoints outside $V$, one can shave a factor of $1/\epsilon$ from both length slack and step by replacing their routers with star graphs.}, but for undirected graphs with only edge lengths and capacities. They also presented shortcuts that can be constructed in parallel and fast but with worse quality such as $2^{\poly(1/\epsilon)}$ length slack. For undirected graphs with vertex lengths and capacities, very recently, \cite{HaeuplerLSW25} showed that there exists an LC-flow shortcut of size $O(n^{1+O(\epsilon)})$ with length slack $O(1/\epsilon^{3})$, congestion slack $n^{O(\epsilon)}$ and step $2^{O(1/\epsilon)}$, but without an efficient construction.

The approaches behind both the aforementioned LC-flow shortcuts are based on tools related to \emph{length-constrained (vertex) expanders} \cite{HaeuplerRG22,HaeuplerHT24,HaeuplerLSW25}. However, these approaches fail to go beyond the inherent large length slack of these tools, and therefore do not result in LC-flow shortcuts with $1+\epsilon$ length slack.

We give the first LC-flow shortcut with $1+\epsilon$ length slack. Precisely, our LC-flow shortcut works for undirected graphs with vertex lengths and capacities, and it has size $\hat{O}(n)$, length slack $1+\epsilon$, congestion slack $\hat{O}(1)$ and step $\hat{O}(1)$ for any $1/\polylog(n)<\epsilon<1$. Furthermore, our shortcut can be constructed in $\hat{O}(m)$ work and $\hat{O}(1)$ depth. We refer to \Cref{thm:emulator} for a detailed statement of our result (or \Cref{thm:OverviewShortcut} in the overview for an informal version).

Our new LC-flow shortcut, thus, improves upon previous length slack from a large constant to $(1+\epsilon)$, at an expense of increasing the step from a large constant to $\hat{O}(1)$. However, the step will only affect the work and depth of the final algorithms, not the approximation. The construction of our LC-flow shortcuts can be viewed as a generalization of Cohen's hopset construction \cite{DBLP:journals/jacm/Cohen00} which  provides shortcuts with $(1+\epsilon)$ length slack but no guarantee on congestion slack.

\paragraph{Close-to-Optimal LC-Vertex Expander Decomposition Algorithm.} Our second technical contribution is the first parallel length-constrained vertex expander decomposition algorithm in close-to-linear work and close-to-constant depth. We refer to \Cref{subsec:preliminaryexpander} for related preliminaries and \Cref{thm:vertexLC-ED} for the formal statement of our results.

A fast LC-vertex-expander decomposition algorithm is required when we want our LC-flow shortcuts to work on undirected graphs with vertex lengths and capacities. However, previously \cite{HaeuplerLSW25} only showed the existence of good LC-vertex-expander decompositions without an efficient algorithm for them. 

Our approach mainly follows the fast LC-(edge)-expander decomposition algorithm, developed recently in \cite{HaeuplerHT24}. Given that algorithmic LC-(edge)-expander decompositions have already led to several applications \cite{DBLP:conf/stoc/HaeuplerH0RS24,DBLP:conf/focs/HaeuplerLS24}, our fast algorithm for LC-vertex-expander decomposition will likely have more future applications.

\paragraph{Length-Constrained Multi-Commodity Maxflows in Low Depths.} Our last intermediate result which may be of independent interest is a \emph{$(1+\epsilon)$-approximate length-constrained multi-commodity maxflow algorithm} with $\tilde{O}(mk)\cdot \poly(h,\epsilon^{-1})$ work and $\tilde{O}(1)\cdot \poly(h,\epsilon^{-1})$ depth, for $k$ commodities and a given $h$ representing the length constraint. We refer to \Cref{Sect:ApproxLCMCFlow} for the formal definition and \Cref{thm:ApproxLCMCMF} for the formal result statement.

Previously, \cite{HaeuplerHS23} studied this problem in the single-commodity setting, and showed a single-commodity version with $\tilde{O}(m)\cdot \poly(h,\epsilon^{-1})$ work and $\tilde{O}(1)\cdot \poly(h,\epsilon^{-1})$. They also generalized their result to the multi-commodity setting in an (almost) black-box way, but with an additional $k$ factor in depth compared to ours. 

Our approach towards improvement (i.e. the depth is now independent from $k$) is a white-box generalization of \cite{HaeuplerHS23} combining Cohen's idea of \emph{path-count flows} \cite{DBLP:journals/siamcomp/Cohen95}.

\subsection{Organization}

In \Cref{sec:overview}, we give an overview of our techniques. \Cref{sect:Preliminaries} introduces basic concepts and definitions we will use in the main body. In \Cref{sect:LCVED}, we formally state the result on our fast parallel length-constrained vertex expander decomposition algorithm (with proof deferred to the appendix). In \Cref{sect:Shortcut}, we discuss our LC-flow shortcuts with $(1+\epsilon)$ length slack, which will use the decomposition subroutine in \Cref{sect:LCVED}. In \Cref{Sect:ApproxLCMCFlow} we talk about our 
$(1+\epsilon)$-approximate length-constrained multi-commodity maxflow algorithm. Lastly, in \Cref{sect:MincostFlowByBoosting} we complete the proof of our parallel $(1+\epsilon)$-approximate multi-commodity mincost flow result, using our new tools in \Cref{sect:Shortcut,Sect:ApproxLCMCFlow}.
\section{Overview}
\label{sec:overview}

In this section, we give an overview of our parallel $(1+\epsilon)$-approximate multi-commodity mincost flow algorithm. 

At a high-level, there are three steps towards our final result. The first and most important step is to build \emph{length-constrained flow (LC-flow) shortcuts} with $1+\epsilon$ length slack, subpolynomial congestion slack and subpolynomial step (\Cref{sect:Shortcut}) and towards this we need a fast parallel length-constrained vertex expander decomposition algorithm (\Cref{sect:LCVED}). The second step is to develop a low-step multi-commodity mincost flow algorithm, but we allow a bicriteria approximation with $1+\epsilon$ cost slack and subpolynomial congestion slack (\Cref{sect:LowStepFlow}), which requires an approximate length-constrained multi-commodity maxflow algorithm (\Cref{Sect:ApproxLCMCFlow}). Lastly in the third step, combining the shortcuts and the low-step flow, we can compute a multi-commodity mincost flow in the original graph, also with a bicriteria $1+\epsilon$ cost slack and subpolynomial congestion slack (\Cref{sect:ApproxMTLFlow}). Further using the flow boosting template \cite{DBLP:journals/siamcomp/GargK07,DBLP:conf/stoc/HaeuplerH0RS24} based on multiplicative weight updates, the subpolynomial congestion slack can be boosted down to $1+\epsilon$, at the cost of an additional subpolynomial factor in depth, and this gives our final result (\Cref{sect:MainResultProof}). 

A large fraction of this overview will focus on the first step, i.e. LC-flow shortcuts with $(1+\epsilon)$ length slack. Although the formal argument is somewhat complicated, the key idea is intuitive. Therefore, in \Cref{sect:OverviewPart1}, we will present a relatively succinct argument, which is informal but (almost) precise and self-contained, by simplifying the setting and omitting the proofs of some less important properties. For example, we will mainly focus on undirected graphs with only edge capacities and costs/lengths, and we will only argue the \emph{existence} of such shortcuts without providing a fast construction algorithm (only in \Cref{Overview:VertexShortcut}, we will talk about undirected graphs with vertex costs and capacities, and briefly discuss the fast construction).

In \Cref{Overview:LowStepFlow}, we will briefly discuss the idea behind the second step without going into the technical details. We will not talk about the third step in this overview, because \cite{DBLP:conf/stoc/HaeuplerH0RS24} already presented a formal template of the flow boosting and we only need to plug in our new algorithm as a subroutine.

\subsection{LC-Flow Shortcuts with $(1+\epsilon)$ Length Slack}
\label{sect:OverviewPart1}

Before introducing our LC-flow shortcuts, we need to set up some basic notions on flows. Consider an undirected graph $G = (V,E)$ with edge lengths $\ell$ and capacities $U$. A \emph{(multi-commodity) flow} $F$ is a collection of weighted flow paths $P$. The weight of each path $P$ is called its value, denoted by $F(P)$, and the value of the flow $F$ is $\vvalue(F) = \sum_{P\in F}F(P)$. A \emph{demand} $D$ is a function that assigns each (ordered) vertex pair $(u,v)$ a non-negative number $D(u,v)$. We further have the following concepts about a flow $F$.
The \emph{demand routed by $F$}, denoted by $D_{F}$, has $D_{F}(u,v)$ equals to the total value of $u$-$v$ flow paths for any vertex pair $(u,v)$.
The \emph{length} of $F$, denoted by $\leng(F)$, is the maximum length of flow paths.
The \emph{congestion} of $F$, denoted by $\cong(F)$, is the maximum $F(e)/U(e)$ over $e\in E$, where $F(e)$ is the total value of flow paths going through $e$.
The \emph{step} of $F$, denoted by $\step(F)$, is the maximum number of edges on a flow path.
Our result about low-step LC-flow shortcuts is given in \Cref{thm:OverviewShortcut}. 

\begin{theorem}
\label{thm:OverviewShortcut}
Given an undirected graph $G$ and a parameter $1/\polylog(n)<\epsilon<1$, we can compute an additional edge set $H$ (possibly with endpoints outside $G$) satisfying the following. 
\begin{itemize}
\item (Forward Mapping) For any flow $F^{G}$ in $G$ with congestion $1$, there is a flow $F^{H}$ in $G\cup H$ routing the same demand with length $(1+\epsilon)\leng(F^{G})$, congestion $1$ and step $\hat{O}(1)$.
\item (Backward Mapping) For any flow $F^{H}$ in $G\cup H$ routing a demand between $V(G)$ with congestion $1$, there is a flow $F^{G}$ in $G$ routing the same demand with length $\leng(F^{H})$ and congestion $\hat{O}(1)$.
\end{itemize}
The size of $H$ is $\hat{O}(n)$, and $H$ can be computed in $\hat{O}(m)$ work and $\hat{O}(1)$ depth.
\end{theorem}

The additional edge set $H$ in \Cref{thm:OverviewShortcut} is what we called an LC-flow shortcut. The length slack and step of the LC-flow shortcut are those in the forward mapping, while the congestion slack is from the backward mapping. In other words, \Cref{thm:OverviewShortcut} shows an LC-flow shortcut with length slack $1+\epsilon$, congestion slack $\hat{O}(1)$ and step $\hat{O}(1)$.

To achieve LC-flow shortcuts with $1+\epsilon$ length slack, we show a new construction inspired by Cohen's \emph{hopsets} \cite{DBLP:journals/jacm/Cohen00}. Roughly speaking, a hopset is an LC-flow shortcut without guarantees on the congestion slack, i.e. it 
$(1+\epsilon)$-approximately preserves pairwise distances. Cohen showed that one can construct hopsets with $1+\epsilon$ length slack using \emph{neighborhood covers} that have a relatively large length slack (i.e., the ratio between the diameter and the covering radius). To further encode congestion information, we hope to integrate the recently developed \emph{LC-expander decompositions} into Cohen's approach, as LC-expander decompositions strengthen neighborhood covers in the sense that they capture both length and congestion information. Moreover, this time we can hope to overcome the large-length-slack barrier in \cite{DBLP:conf/stoc/HaeuplerH0RS24} since Cohen's hopset construction shows the possibility of obtaining $1+\epsilon$ length slack using tools with much larger length slack.

We now start discussing the existence of such LC-flow shortcuts in \Cref{thm:OverviewShortcut}. Our full argument is a \emph{two-dimensional} induction/recursion. In the first dimension, we perform induction on the graph size (this corresponds to Cohen's hopset construction), and the second-dimensional induction is on the levels of the \emph{LC-expander hierarchy}. For better understanding, in \Cref{Overview:ExpanderShortcuts}, we first consider a warm-up case in which the input graph is already an LC-expander, in which the argument only involves the first-dimensional induction. Next, in \Cref{Overview:GeneralShortcuts} we turn to general graphs, introduce the LC-expander hierarchy, and discuss the two-dimensional induction (based on a two-level hierarchy or more).

Before we move on, we emphasize that we will actually aim at shortcuts with $(1+o(\log n)\epsilon)$ length slack instead. Scaling down $\epsilon$ by $\log n$ gives a $1+\epsilon$ (for the original $\epsilon$) length slack. Furthermore, we point out that it suffices to show the existence of a weaker version of LC-flow shortcuts, called \emph{$h$-LC-flow shortcuts}. That is, the forward mapping only needs to work for $h$-length flows $F^{G}$ in $G$ (i.e. $\leng(F^{G})\leq h$). To obtain a (general) LC-flow shortcut, just take the union of $h$-LC-flow shortcuts for $h$ being different powers of $1+\epsilon$. This only multiplicatively increases the length slack by another factor of $1+\epsilon$, and the congestion slack by a factor of $O(\log n/\epsilon)$, assuming polynomially-bounded lengths.

\subsubsection{Warm Up: $h$-LC-Flow Shortcuts for LC-Expanders}
\label{Overview:ExpanderShortcuts}

First we set up the scenario by introducing LC-expanders. A \emph{node-weighting} $A$ is a function that assigns each vertex $v\in V$ a non-negative number $A(v)$. In particular, the \emph{degree node-weighting} of $G$, denoted by $A_{\deg}$, has $A_{\deg}(v) = \sum_{e\in E\text{ 
incident to }v}U(e)$ for each vertex $v$. A demand $D$ is \emph{$A$-respecting} if for each vertex $v$, the total demand units involving $v$, i.e. $D(v):=\sum_{v'}D(v,v') + D(v',v)$, is at most $A(v)$, and $D$ is \emph{$h$-length} if any pair of vertices $(u,v)$ s.t. $D(u,v)>0$ has $\dist_{G}(u,v)\leq h$. Finally, for a length parameter $h_{\ed}$, a length slack $s$ and a congestion parameter $\phi$, a node-weighting $A$ is \emph{$(h_{\ed},s)$-length $\phi$-expanding} in $G$ if any $h_{\ed}$-length $A$-respecting demand $D$ can be routed (by a flow) with length $h_{\ed}s$ and congestion $O(\log n/\phi)$. Particularly, we say $G$ is an $(h_{\ed},s)$-length $\phi$-expander if its degree node-weighting $A_{\deg}$ is expanding (for $h_{\ed},s,\phi$). 

In this warm-up scenario, we assume the input graph $G$ is an $(h_{\ed},s)$-length $\phi$-expander for \[
s = \log n, h_{\ed} = \epsilon h/s,\text{ 
and }\phi = 1/2^{\log^{0.75} n}.
\]
We choose such $s$ and $\phi$ because we want them to fit our argument in \Cref{Overview:GeneralShortcuts} later.

Second, we will need the aforementioned tool of neighborhood covers. A neighborhood cover ${\cal N}$ is a collection of vertex subsets $S$ called \emph{clusters}, which can be partitioned into $\omega$ (called the \emph{width}) \emph{clusterings}, where each clustering is a collection of disjoint clusters. The \emph{diameter} $h_{\diam}$ of ${\cal N}$ is the maximum diameter of a cluster, and ${\cal N}$ has \emph{covering radius} $h_{\cov}$ if for each vertex $v$, there is a cluster $S$ containing all vertices $v'$ within distance $h_{\cov}$ from $v$. 

\paragraph{Shortcut Construction.} Now we are ready to construct the $h$-LC-flow shortcut.

\medskip

\noindent{\underline{Step 1.}} First, we take a neighborhood cover ${\cal N}$ of $G$ with %
\[
h_{\diam} = h_{\ed}, h_{\cov} = h_{\diam}/\log n\text{ and }\omega = O(\log n).
\]
It is known that such ${\cal N}$ exists. 
We further partition clusters in ${\cal N}$ into \emph{small clusters} ${\cal N}_{\ssmall}$ and \emph{large clusters} ${\cal N}_{\big}$, such that ${\cal N}_{\ssmall}$ collects all clusters of size at most $n/\sigma$ and ${\cal N}_{\big}$ contains the rest.
Here $\sigma$ is some subpolynomial parameter, and here we set 
\[
\sigma = 2^{\log^{0.75} n}.
\]

\medskip

\noindent{\underline{Step 2.}} For each small cluster $S\in{\cal N}_{\ssmall}$, we \emph{inductively} assume there exists an $h_{\cov}$-LC-flow shortcut $H_{S}$ for the graph $G[S]$ (the subgraph induced by $S$) with length slack $\hat{\lambda}$, congestion slack $\hat{\kappa}$ and step $\hat{t}$ (think of $\hat{\lambda}$ as roughly $1+\epsilon$).

\medskip

\noindent{\underline{Step 3.}} For each large cluster $S\in{\cal N}_{\big}$, we add a \emph{star graph} $H_{S}$ in which a newly created center vertex $x_{S}$ is connected to each original vertex $v\in S$ by an edge with length $h_{\diam}s$ and capacity $A_{\deg}(v)$. Intuitively, the star edges have such lengths and capacities because we use the star $H_{S}$ as a low-step simulation of the expander routing. Recall that we choose $h_{\diam} = h_{\ed}$, so any $A_{\deg}$-respecting demand on $S$ (which is definitely $h_{\ed}$-length) can be routed in the LC-expander $G$ with length $h_{\diam}s$ and low congestion. 

\medskip

\noindent{\underline{Step 4.}} For each pair of large clusters $S_{1},S_{2}\in{\cal N}_{\big}$, and each $h'\in[1,h]$ rounded up to the nearest integral power of $(1+\epsilon)$, we add an edge, called an \emph{inter-cluster edge}, connecting the star centers $x_{S_{1}}$ and $x_{S_{2}}$ with length $h' + 2h_{\diam}s$. Moreover, the capacity is the value of the maximum $(h' + 2h_{\diam}s)$-length single-commodity flow from $x_{S_{1}}$ to $x_{S_{2}}$ in $G\cup H_{S_{1}}\cup H_{S_{2}}$. Intuitively, these parallel edges between $x_{S_{1}}$ and $x_{S_{2}}$ provide $1$-step shortcuts for any feasible $h$-length flow $F$ routing an $A_{\deg}$-respecting demand $D_{F}$ from $S_{1}$ to $S_{2}$.

\medskip

\noindent{\underline{Step 5.}} The final shortcut $H$ is the union of $H_{S}$ from steps 2 and 3, and the inter-cluster edges from step 4.

\paragraph{Forward and Backward Mappings.} We are not going to give a rigorous analysis here but only some intuitions about the forward and backward mappings.

For the forward mapping, we can hope for low length slack because we pay roughly $1+\epsilon$ multiplicative length slack and additionally $O(h_{\diam}s) = O(\epsilon h)$ \emph{additive length slack} in this inductive step. To see this, consider an $h$-length flow path $P$ in the given $h$-length flow $F^{G}$. Let us partition $P$ into roughly $O(h/h_{\cov}) = \tilde{O}(1/\epsilon)$ many subpaths $\hat{P}$, each of which has length at most $h_{\cov}$. Thus each subpath $\hat{P}$ is totally contained by some cluster $\hat{S}$ in the neighborhood cover ${\cal N}$. 
\begin{figure}[htbp]
    \centering
    \includegraphics[]{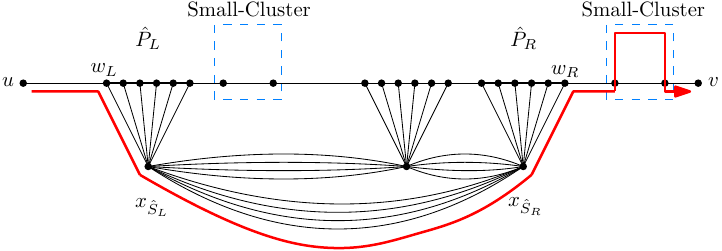}
    \caption{An example for the forward mapping in LC-expanders. Given flow path $P = P_{u,v}$, the red path is the output of the forward mapping. 
    We only draw few subpaths on $P$ and omit others for simplicity.
    We note that the routing given in the small-cluster part is only for illustration, which represents the inductive forward mapping results.}
    \label{fig:forward_mapping_expander}
\end{figure}
\medskip

\noindent{\underline{Small-Cluster Subpaths.}} For each subpath $\hat{P}$ whose $\hat{S}$ is a small cluster, we are good because the shortcut $H_{\hat{S}}$ from Step 2 can reroute $\hat{P}$ with multiplicative length slack $\hat{\lambda}$ (roughly $1+\epsilon$) and step $\hat{t}$. 

\medskip

\noindent{\underline{Large-Cluster Subpaths.}} For those $\hat{P}$ in large clusters $\hat{S}$, instead of rerouting them separately using star graphs on large clusters, we will do something different using the step-4 shortcut edges so that we only pay $O(h_{\diam}s)$ additive length slack. As illustrated in \Cref{fig:forward_mapping_expander}, let $\hat{P}_{L}$ and $\hat{P}_{R}$ be the leftmost and rightmost large-cluster subpaths, and let $w_{L}$ and $w_{R}$ be the left endpoint of $\hat{P}_{L}$ and the right endpoint of $\hat{P}_{R}$. Then we reroute the whole part of $P$ between $w_{L}$ and $w_{R}$ (denoted by $P[w_{L},w_{R}]$ and called the \emph{jumping subpath}) using shortcut edges
\[
(w_{L},x_{\hat{S}_{L}})\circ(x_{\hat{S}_{L}},x_{\hat{S}_{R}})\circ(x_{\hat{S}_{R}},w_{R}).
\]
Here $(w_{L},x_{\hat{S}_{L}})$ and $(x_{\hat{S}_{R}},w_{R})$ are from the star graphs $H_{\hat{S}_{L}}$ and $H_{\hat{S}_{R}}$, so they both have lengths $h_{\diam}s$. In the middle, $(x_{\hat{S}_{L}},x_{\hat{S}_{R}})$ is the inter-cluster edge between the star centers with appropriate length $(1+\epsilon)\cdot\leng(P[w_{L},w_{R}]) + 2h_{\diam}s$. In other words, rerouting the jumping subpath $P[w_{L},w_{R}]$ incurs $1+\epsilon$ multiplicative length slack and $4h_{\diam}s$ additive slack.

\medskip

Therefore, for this inductive step, the multiplicative length slack is $\max\{\hat{\lambda},1+\epsilon\}$ and additive length slack is $4h_{\diam}s$ as desired. Moreover, after rerouting, the step bound is roughly $\tilde{O}(1/\epsilon)\cdot \hat{t}$, because there are $\tilde{O}(1/\epsilon)$ subpaths and the bottleneck is the small-cluster subpaths. In summary, the key intuition behind this forward mapping argument is that we can get small additive length slack ($O(\epsilon h)$ instead of $\tilde{O}(h)$) by paying a $\tilde{O}(1/\epsilon)$ multiplicative factor on the step bound.

Regarding the backward mapping, we briefly discuss its congestion slack below. In Step 2, all shortcuts $H_{S}$ of small clusters together cause congestion slack $\hat{\kappa}\cdot \omega = O(\hat{\kappa}\log n)$ because ${\cal N}$ has width $\omega$. In Step 3, all star graphs of large clusters together cause congestion slack $O(\log n/\phi)\cdot \omega$ by the expander routing and the width of ${\cal N}$. In Step 4, each inter-cluster edge $(x_{S_{1}},x_{S_{2}})$ causes congestion slack $1 + O(\log n/\phi)$, because it corresponds to an underlying feasible flow in $G\cup H_{S_{1}}\cup H_{S_{2}}$ and further projecting this flow to $G$ requires expander routing inside $S_{1}$ and $S_{2}$. The number of inter-cluster edges is at most $(\omega\sigma)^{2}\cdot O(\log h/\epsilon)$ (note that there are at most $\omega\sigma$ large clusters and each pair has $O(\log h/\epsilon)$ parallel edges with different lengths), so step 4 causes congestion slack $O((\omega\sigma)^{2}\cdot(\log n/\phi)\cdot O(\log h/\epsilon))$. Adding up all above, the total congestion slack is $\tilde{O}(\hat{\kappa}) + \tilde{O}(\sigma^{2}/(\phi\epsilon))$.

\paragraph{The Final Shortcut Quality.} From the analysis above, one inductive step gives multiplicative length slack $\lambda = \hat{\lambda} + 5\epsilon$ (after merging the additive slack $4\epsilon h$ into the multiplicative slack), congestion slack $\kappa = \tilde{O}(1)\cdot \hat{\kappa} + \tilde{O}(\sigma^{2}/(\phi\epsilon))$ and step $t = \tilde{O}(1/\epsilon)\cdot \hat{t}$. Suppose we are in the ideal case that all graphs in the whole induction tree are good LC-expanders. Note that the induction depth is $z = \log_{\sigma}n = O(\log^{0.25} n)$ because the number of vertices drops by a factor of $\sigma = 2^{\log^{0.75} n}$ in each inductive step. Therefore, the final shortcut has length slack $\lambda = 1+ O(\log^{0.25} n)\cdot \epsilon$, congestion slack $\kappa = \tilde{O}(1)^{O(\log^{0.25} n)}\cdot \tilde{O}(2^{3\log^{0.75} n}/\epsilon)$ and step $t = \tilde{O}(1/\epsilon)^{O(\log^{0.25} n)}$. This implies $\lambda = 1 + \epsilon$, $\kappa = \hat{O}(1)$ and $t = \hat{O}(1)$.

\subsubsection{$h$-LC-Flow Shortcuts for General Graphs}
\label{Overview:GeneralShortcuts}

We will leverage \emph{LC-expander decompositions} and \emph{LC-expander hierarchies} to extend the argument in \Cref{Overview:ExpanderShortcuts} to general graphs, and now we give their definitions. 

The first is LC-expander decompositions. An \emph{$h_{C}$-length moving cut} $C$ is a function that assigns each edge $e$ a cut value $C(e)\in\{0,1/h_{C},2/h_{C},...,1\}$. The \emph{graph $G$ with $C$ applied}, denoted by $G-C$, is the same graph but with different edge lengths $\ell_{G-C} = \ell_{G} + h_{C}\cdot C$. An \emph{$(h_{\ed},s)$-length $\phi$-expander decomposition} of a node-weighting $A$ in $G$ is an $(h_{\ed}\cdot s)$-length moving cut $C$ such that $A$ is $(h_{\ed},s)$-length $\phi$-expanding in $G-C$ (recall the term ``expanding'' in \Cref{Overview:ExpanderShortcuts}). From now we fix 
\[
s = \log n\text{ and }\phi = 1/2^{\log^{0.75} n}.
\]
It is known that (see e.g. \cite{HaeuplerRG22}), given any node-weighting $A$ and length parameter $h_{\ed}$, there exists an $(h_{\ed},s)$-length $\phi$-expander decomposition $C$ of $A$ in $G$ with $|C| \leq \tilde{O}(\phi)\cdot |A|$,
where $|A| = \sum_{v\in V}A(v)$ and $|C| = \sum_{e\in E}C(e)\cdot U(e)$ are sizes of $A$ and $C$ respectively. An interpretation of an LC-expander decomposition is that we can make $A$ expand in $G$ by increasing the edge lengths a bit.

Next we consider LC-expander hierarchies. For simplicity, we only formally define a two-level LC-expander hierarchy, and assume the graph $G$ admits such a two-level hierarchy. We fix
\[
h_{\diam} = \epsilon h/s,\ h_{\cov} = h_{\diam}/\log n
\]
to be consistent with step 1 in \Cref{Overview:ExpanderShortcuts}.%
\begin{itemize}
\item At the first level, we start with the degree node-weighting $A_{1}:=A_{\deg}$. Then we select $h_{\ed,1} = \epsilon(1+2\epsilon)h_{\cov}/s$, and take an $(h_{\ed,1},s)$-length $\phi$-expander decomposition $C_{1}$ of $A_{1}$. Note that $C_{1}$ is an $(h_{\ed,1}\cdot s)$-length moving cut. 
\item Turn to the second level. We define the \emph{cut node-weighting} of $C_{1}$, denoted by $A_{C_{1}}$, by setting $A_{C_{1}}(v) = \sum_{e\in E\text{ incident to }v}C_{1}(e)\cdot U(e)$ for each vertex $v$. Then the second-level node-weighting $A_{2}$ is exactly $A_{C_{1}}$, and we assume $A_{2}$ is $(h_{\ed,2},s)$-length $\phi$-expanding in $G$ for $h_{\ed,2} = h_{\diam}$ (so the hierarchy ends at the second level).
\end{itemize}

\begin{figure}[htbp]
    \centering
    \includegraphics[]{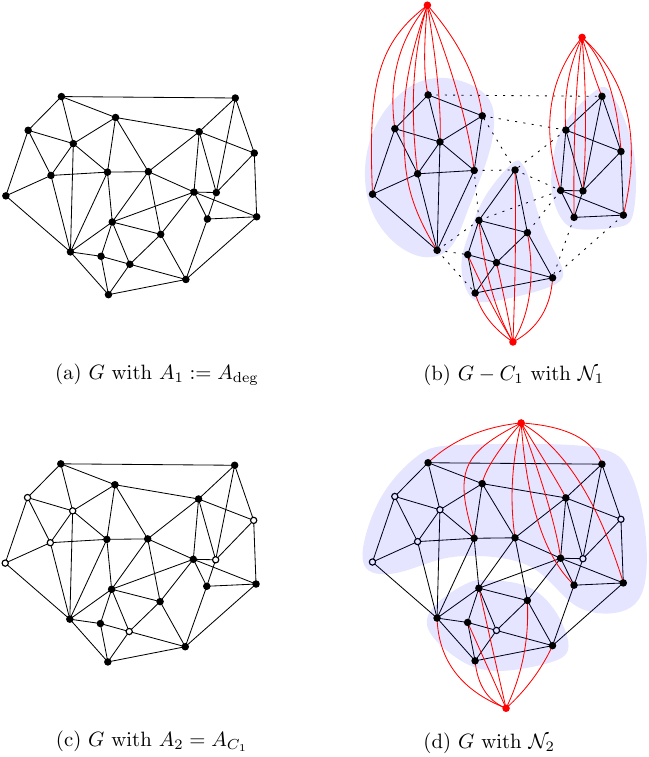}
    \caption{An example of a two-level expander decomposition hierarchy with shortcuts over large-clusters. For simplicity, we omit small-clusters and inter-cluster edges.
    (b) Dashed edges represent cut edges $e$ with $C_{1}(e) = 1$ (in this figure, we assume each $C_{1}(e)$ is either $0$ or $1$).
    (c) Vertices are not in the support set of $A_{2}$ is marked with hollow circles. }
    \label{fig:ED_hierarchy_example}
\end{figure}

\paragraph{Shortcut Construction.} The construction of the shortcut is almost identical to that in \Cref{Overview:ExpanderShortcuts}, except for the following. See \Cref{fig:ED_hierarchy_example} for an illustration.

\medskip

\noindent{\underline{Step 0.}} We add a Step 0 before Step 1. Assume \emph{inductively}, for $h_{1} = (1+2\epsilon) h_{\cov}$, there exists an $h_{1}$-LC-flow shortcut $H_{1}$ for the graph $G-C_{1}$ with length slack $\tilde{\lambda}$, congestion slack $\tilde{\kappa}$ and step $\tilde{t}$. Recall that $A_{1} = A_{\deg}$ is $(h_{\ed,1},s)$-length $\phi$-expanding in $G-C_{1}$ for $h_{\ed,1} = \epsilon h_{1}/s$, so this is exactly the same scenario with when building an $h_{1}$-LC-flow shortcut for $G-C_{1}$ in \Cref{Overview:ExpanderShortcuts}. 

\medskip

This completes the shortcut construction for level $1$, and we will move on to level 2, by following Steps 1 to 4 in \Cref{Overview:ExpanderShortcuts} (but replacing Step 3 with $3'$).

\medskip

\noindent{\underline{Step $3'$.}} For each large cluster $S$ in ${\cal N}_{2}$ (where ${\cal N}_{2}$ is the neighborhood cover from Step $1$), in the star graph $H_{S}$, the edge connecting the center $x_{S}$ and a vertex $v\in S$ now has capacity $A_{2}(v)$ instead of $A_{\deg}(v)$ (if some $v$ has $A_{2}(v) = 0$, this edge does not exist). This is a natural modification since we only have $A_{2}$ expanding in $G$ now.

\medskip

\noindent{\underline{Step $5'$.}} Further add the first-level shortcut $H_{1}$ from Step 0 into the final shortcut $H$.

\paragraph{The Forward Mapping.} There is no big change to the backward mapping, but the forward mapping becomes more involved. Consider a path $P$ from $F^{G}$, and we still start with partitioning $P$ into $\tilde{O}(1/\epsilon)$ subpaths $\hat{P}$ of length \emph{exactly}\footnote{We assume this to simplify the discussion, but generally it should be at most.} $h_{\cov}$. The small-cluster shortcuts from Step 2 can still take care of all the small-cluster subpaths. 

However, now the jumping subpath $P[w_{L},w_{R}]$ cannot be simply rerouted using the star graphs from Step $3'$ and the inter-cluster edges from Step $4$, because the star graphs is for the node-weighting $A_{2}$. In the extreme case that $A_{2}(w_{L}) = 0$, the star edge connecting $w_{L}$ and $x_{S_{1}}$ does not even exist. 

Fixing this needs to exploit the first-level shortcut $H_{1}$ from Step 0. We further classify large-cluster subpaths into \emph{light subpaths} and \emph{heavy subpaths}, where light subpaths are those $\hat{P}$ with $C_{1}(\hat{P}) = \sum_{e\in\hat{P}}C_{1}(e)\leq 1$ and heavy subpaths are the rest.

\medskip

\noindent{\underline{Light Large-Cluster Subpaths.}} Each light subpath $\hat{P}$ can be rerouted using the $h_{1}$-LC-flow shortcut $H_{1}$ of $G-C_{1}$ because in $G-C_{1}$, $\hat{P}$ has length $\ell_{G-C_{1}}(\hat{P}) = \ell_{G}(\hat{P}) + (h_{\ed,1}\cdot s)C_{1}(\hat{P}) \leq h_{\cov} + \epsilon h_{1}\leq h_{1}$ by $h_{1} = (1+2\epsilon) h_{\cov} $ and $h_{\ed,1} = \epsilon h_{1}/s$. After rerouting, the length and step become $\tilde{\lambda}h_{1}$ and $\tilde{t}$, so the multiplicative length slack w.r.t. $\ell_{G}(\hat{P})=h_{\cov}$ is $\tilde{\lambda}h_{1}/h_{\cov} \leq (1+2\epsilon)\tilde{\lambda}$. 

\medskip
\begin{figure}[htbp]
    \centering
    \includegraphics[]{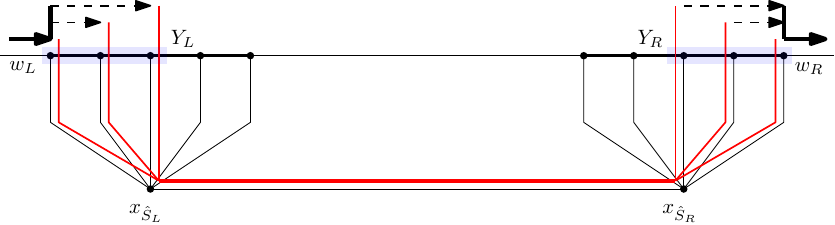}
    \caption{An example of the jumping subpath between heavy large-cluster subpaths. The blue region marks $Y_{L}$ and $Y_{R}$ in the leftmost and rightmost heavy large-cluster subpaths, respectively. The black dashed arrows represent the relocation using the first-level shortcut $H_{1}$, where we omit for simplicity.}
    \label{fig:fractional_jumping_subpath}
\end{figure}
\noindent{\underline{Heavy Large-Cluster Subpaths.}} The subpaths that have not been shortcut are heavy large-cluster subpaths. Let $\hat{P}_{L}$ and $\hat{P}_{R}$ be the leftmost and rightmost heavy large-cluster subpaths (say $\hat{S}_{L}$ and $\hat{S}_{R}$ are the clusters containing $\hat{P}_{L}$ and $\hat{P}_{R}$ respectively). The \emph{jumping subpath} $P[w_{L},w_{R}]$ is between the left endpoint $w_{L}$ of $\hat{P}_{L}$ and the right endpoint $w_{R}$ of $\hat{P}_{R}$. Since now $\hat{P}_{L}$ is further heavy (i.e. $A_{2}(\hat{P}_{L})\geq C_{1}(\hat{P}_{L})\geq 1$ is properly large), we can upload the flow units to $x_{S_{1}}$ from the $\hat{P}_{L}$-vertices $y_{1}$ (called \emph{portals}) with enough $A_{2}(y_{L})$, assuming that we can \emph{relocate} the flow units from $w_{L}$ to these $y_{L}$.

To choose these portals $y_{L}$, we first take the \emph{shortest} prefix $\hat{P}'_{L}$ of $\hat{P}_{L}$ such that $C_{1}(\hat{P}'_{L})\geq 1$. Then let $Y_{L} = \{\text{left endpoint $y_{L,e}$ of $e$}\mid e\in \hat{P}'_{L}, C_{1}(e)>0\}$ be the portals, and we restrict that each portal $y_{L,e}$ can receive up to $C_{1}(e)\cdot \vvalue(P)$ flow units relocated from $w_{1}$, which will not hurt since only $\vvalue(P)$ flow units need to be relocated. This restriction is to ensure that we will not violate the capacities of star edges at the very end (i.e. ensure the obtained $F^{H}$ has congestion $1$). For the right side, we choose $Y_{R}$ and apply the restriction similarly.

Rerouting the flow units from left portals $Y_{L}$ to right portals $Y_{R}$ can be done using star edges and inter-cluster edges similar to the argument in \Cref{Overview:ExpanderShortcuts}. It remains to relocate the units from $w_{L}$ to $Y_{L}$ (also from $Y_{R}$ to $w_{R}$). We can do this using the first-level shortcut $H_{1}$ again similar to the argument for light large-cluster subpaths, because for any $y_{L}\in Y_{L}$, the subpath $P[w_{L},y_{L}]$ has length in $G-C_{1}$ at most 
\[
\ell_{G}(P[w_{L},y_{L}]) + (h_{\ed,1}\cdot s)\cdot C_{1}(P[w_{L},y_{L}])\leq h_{\cov} + \epsilon(1+2\epsilon)h_{\cov}\leq 
h_{1}.
\]
Recall that $h_{1} = (1+2\epsilon) h_{\cov} $ and $h_{\ed,1} = \epsilon h_{1}/s$.
Note that $C_{1}(P[w_{L},y_{L}])\leq 1$ because of the ``shortest'' property of $\hat{P}'_{L}$.

\medskip

We directly state the length slack and the step bound of this inductive step without detailed explanations (which is quite straightforward following the argument in the expander-case in \Cref{Overview:ExpanderShortcuts}). The multiplicative length slack is $\lambda = \max\{\hat{\lambda}, (1+2\epsilon)\tilde{\lambda},1+\epsilon\} + O(\epsilon)$, where $\hat{\lambda}$ is from the small-cluster subpaths, $(1+2\epsilon)\tilde{\lambda}$ is from the light large-cluster subpaths, and $1+\epsilon$ and $O(\epsilon)$ are multiplicative and additive length slacks from rerouting (from left to right portals of) the jumping subpath. The step bound is $t = \tilde{O}(1/\epsilon)\cdot \max\{\hat{t},\tilde{t}\}$. For the backward mapping, the congestion slack is $\kappa = \tilde{O}(\hat{\kappa}) + \tilde{O}(\sigma^{2}/(\phi\epsilon)) + \tilde{\kappa}$ (just add one more term $\tilde{\kappa}$ from the first-level shortcut $H_{1}$).

\paragraph{The Final Shortcut Quality.} Although we focus on a two-level hierarchy above, generally the hierarchy has $r = O(\log^{0.25} n)$ levels because going up one level will reduce the size of node-weighting by a factor of roughly $2^{\log^{0.75} n}$ (recall that $|A_{2}| = 2|C_{1}|\leq \tilde{O}(\phi)|A_{1}|$ and we choose $\phi = 1/2^{\log^{0.75} n}$). Also we have shown that the first-dimensional induction has $z = \log_{\sigma} n = O(\log^{0.25} n)$ levels. Let $\lambda_{(\alpha,\beta)}$, $\kappa_{(\alpha,\beta)}$ and $t_{(\alpha,\beta)}$ be the shortcut quality when we are at level $\alpha\in[1,z]$ of the first-dimensional induction and level $\beta\in[1,r]$ of the hierarchy. The argument above gives the following recurrence relations:
\begin{align*}
&\lambda_{(\alpha,\beta)} = \max\{\lambda_{(\alpha-1,r)},(1+2\epsilon)\lambda_{(\alpha,\beta-1)},1+\epsilon\} + O(\epsilon),\\
&\kappa_{(\alpha,\beta)} = \tilde{O}(\kappa_{(\alpha-1,r)}) + \kappa_{(\alpha,\beta-1)} + \tilde{O}(\sigma^{2}/(\phi\epsilon)),\\
&t_{(\alpha,\beta)} = \tilde{O}(1/\epsilon)\cdot t_{(\alpha,\beta-1)},
\end{align*}
solving which gives $\lambda_{(z,r)}  = 1 + O(\log^{0.5} n)\cdot\epsilon$, $\kappa_{(z,r)} = \tilde{O}(1)^{O(\log^{0.5} n)}\cdot \tilde{O}(2^{3\log^{0.75} n}/\epsilon)$ and $t_{(z,r)} = \tilde{O}(1/\epsilon)^{O(\log^{0.5}n)}$. Substituting the above $\epsilon$ with an $\epsilon' = \epsilon/\log n$ gives a $h$-LC-flow shortcut with length slack $1 + \epsilon$, congestion slack $\hat{O}(1)$ and step $\hat{O}(1)$.

\subsubsection{Algorithmic Shortcuts for Vertex-Capacitated Graphs}
\label{Overview:VertexShortcut}

Lastly, we make some comments on making the above shortcut construction algorithmic and work for undirected graphs with vertex capacities and lengths. In fact, there is no significant barrier, except that we need \emph{length-constrained vertex expander decompositions} and they should be algorithmic. The existence of such decompositions was shown in a recent work \cite{HaeuplerLSW25}, but there is still no fast parallel algorithm for them to date. Therefore, we develop the first close-to-optimal parallel length-constrained vertex expander decomposition algorithm based on the LC-(edge)-expander decomposition algorithm in \cite{HaeuplerHT24}.

\begin{theorem}
\label{thm:OverviewVertexED}
Let $G$ be an undirected graph with vertex capacities and lengths. Given a node-weighting $A$ and parameters $h\geq 1, \epsilon\in(0,1), \phi\in(0,1)$, there is an algorithm that computes an $(h,s)$-length $\phi$-vertex-expander decomposition $C$ of $A$ in $G$ such that $s = \exp(\exp(O(\log \epsilon^{-1})\cdot \epsilon^{-1}))$ and $|C|\leq n^{O(\epsilon)}|A|/\phi$. The work is $\poly(h)\cdot m\cdot n^{\poly(\epsilon)}$ and the depth is $\poly(h)\cdot n^{\poly(\epsilon)}$.
\end{theorem}

There are two small issues when applying \Cref{thm:OverviewVertexED}. First, the tradeoff between the length slack $s = \exp(\exp(O(\log \epsilon^{-1})\cdot \epsilon^{-1}))$ and the coefficient $n^{O(\epsilon)}$ (called the \emph{cut slack}) in the bound of $|C|$ is much worse than the existential bound \cite{HaeuplerLSW25}. Hence, we need to tune the parameter in the argument accordingly. We point out that in \Cref{sect:Shortcut}, for convenience our full argument deviates a bit from the above overview: (1) we actually consider \emph{pure additive length slacks} when analysing the forward mapping, and (2) the parameters we select are very loose (but still get the same asymptotic bounds at the end).

Second, note that the work and depth of \Cref{thm:OverviewVertexED} depend on $\poly(h)$, which means we cannot compute $h$-LC-flow shortcuts for large $h$. Therefore, to get (general) LC-flow shortcuts, we use a standard technique called \emph{stacking/bootstrapping} (which has been used extensively, see e.g. \cite{DBLP:journals/jacm/Cohen00,DBLP:conf/focs/BernsteinGS21,DBLP:conf/stoc/HaeuplerH0RS24}). Concretely, we will bootstrap $h$-LC-flow shortcuts (for some subpolynomial $h$) by building them on top of each other.

\subsection{Low-Step Multi-Commodity Mincost Flows in Low Depth}
\label{Overview:LowStepFlow}

In this part, given a demand $D$, we want to compute a flow $F$ routing $D$ in $G$ with low congestion (say $\tilde{O}(1)$ or $\hat{O}(1)$), such that the total length (i.e. cost) of $F$ is within a $1+O(\epsilon)$ multiplicative factor from the min-total-length $t$-step feasible flow $F^{*}$ routing $D$. 

This problem can be reduced to the \emph{(approximate) length-constrained maximum flow} problem, which has been studied by \cite{HaeuplerHS23} but mainly in the single-commodity setting. Intuitively, the reduction is a simple greedy strategy. Since congestion $\tilde{O}(1)$ is allowed, we can consider length parameter $h = (1+\epsilon)^{i}$ in an increasing order of $i$. For each such $h$, we try to route as much demand as possible using a feasible $h$-length. Concretely, we partially route $D$ via a feasible $(1+\epsilon)h$-length flow $F_{h}$ with value (approximately) at least the maximum $t$-step $h$-length feasible flow $F^{*}_{h}$ partially routing $D$. Add this $F_{h}$ into $F$, subtract from $D$ the demand that has been routed, and then go to the next $h$. A small issue is that the LC-max flow subroutine will have work and depth depending on $\poly(h)$ (as we will see soon), so we cannot compute $F_{h}$ directly when $h$ is large. However, since we only need $\vvalue(F_{h})$ comparable to the value of the $t$-step $h$-length $F^{*}_{h}$, we actually allow $\epsilon h/t$ additive errors on the edge lengths. Hence we can just scale down the edge length by $\epsilon h/t$ (then round up to integers) and compute a $(t/\epsilon + t)$-length maxflow instead.

Now we discuss our result on approximate multi-commodity LC-maxflows, which is our second intermediate result which may be of independent interest.

\begin{theorem}
Let $G$ be a directed graph with edge lengths and capacities. Let $D$\footnote{The demand $D$ is in this form since we are considering \emph{non-concurrent} multi-commodity flows.} be a given demand such that at most $k$ pairs of vertices $(u,v)$ have $D(u,v) = \infty$ and the rest has $D(u,v) = 0$. Given parameters $h\geq 1$ and $0<\epsilon<1$, there is an algorithm that computes a $(1+\epsilon)$-approximate feasible $h$-length maxflow $F$ partially routing $D$. The work is $\tilde{O}(mk)\cdot \poly(h,\epsilon^{-1})$ and depth is $\tilde{O}(1)\cdot \poly(h,\epsilon^{-1})$.
\end{theorem}

We note that, assuming subpolynomial $h$ and $\epsilon^{-1}$, we are the first to show such an algorithm with work $\hat{O}(mk)$ and depth $\hat{O}(1)$ independent from $k$. Previously, \cite{HaeuplerHS23} mentioned a multi-commodity generalization of their single-commodity LC-maxflow result (in an almost black-box way), but their algorithm has depth $\tilde{O}(k)\cdot \poly(h,\epsilon^{-1})$. 

Our result is obtained by generalizing the single-commodity LC-maxflow algorithm \cite{HaeuplerHS23} in a white-box manner. Roughly speaking, \cite{HaeuplerHS23} showed that, in the single-commodity setting, computing $(1+\epsilon)$-approximate LC-maxflow can be reduced to computing blocking flows in an $h$-layer DAG, and \cite{DBLP:journals/siamcomp/Cohen95} can solve the latter using the idea of \emph{path count flows}. We show that this reduction also works in the multi-commodity setting (in low depth), and observe that Cohen's path count flows can naturally be computed for $k$ commodities in parallel.
\section{Preliminaries}
\label{sect:Preliminaries}

Throughout the paper, all graphs are undirected unless specified to be directed. 

We always use $n$ and $m$ to denote the number of vertices and edges of the original graph in the context (for example, the original graph of each main theorem is its input graph). For a real function $F:D\to \bbR$ and $S\subseteq D$, we write $F(S)=\sum_{s\in S}F(s)$. 

\paragraph{Polynomially-Bounded Integral Functions and Positive Lengths and Capacities.}
Original graphs might have length, capacities, demands, node-weightings, etc, and we assume that all these values (which could be real numbers), after scaling up to integers, are upper bounded by a parameter denoted as $N$. We assume $N=\poly(n)$. 

Especially, we assume that all length and capacities are \emph{positive} integers. This is almost without loss of generality. We can ignore vertices/edges with zero capacity. For vertices/edges with zero length, we can add a very tiny length, and then scale up the whole length function. But we note that this scaling does not work in the context sensitive to absolute lengths, for example, in some algorithms whose work and depth depend on $h$, where $h$ is some absolute length parameter. In this case, we will carefully enforce the input graphs have positive integral lengths.

\paragraph{Undirected graphs with vertex lengths and capacities.}
A undirected graph with vertex lengths and capacities, denoted by $G=(V,E)$ is an undirected graph where each vertex is associated with a length and a capacity from $\mathbb{N}^+$. We define the length function $\ell:V\to \mathbb{N}^{+}$ and the capacity function $U:V\to \mathbb{N}^{+}$ as mapping a vertex to its associated vertex length and capacity. We use $\ell_G,U_G$ to emphasize they are the length and capacity function for $G$. We use $V(G),E(G)$ to denote the vertex and edge set of $G$.

For convenience, we will not assign lengths and capacities to edges for such graphs. Namely, each edge has length $0$ and infinite capacity. This is without loss of generality because we can always transform edge lengths and capacities into vertex lengths and capacities by putting a vertex on each edge (and our work and depth bound are not sensitive to number of vertices).

\paragraph{Basic graph terminologies.} We define $\delta_G(u)=\{(u,v)\in E\}$ to be all adjacent edges of $u$ in $G$. We can extend this definition to $\delta_H(u)=\{(u,v)\in H\}$ where $H\subseteq E$ is an edge subset (or a subgraph). We use $N_G(u)$ to denote the vertex set of the neighborhood of $u$. We write $N_G(S)=\cup_{v\in S}N_G(v)-S$. We write $N_G[v]=N_G(v)\cup \{v\}$ and $N_G[S]=N_G(S)\cup S$.

\paragraph{Path and distances.} 
A simple path $P$ in $G$ is a sequence of distinct vertices $(v_1,...,v_k)$ such that $(v_i,v_{i+1})\in E$ for every $i\in[k-1]$.  We define $\step(P)=k$.%
We define $V(P)=\{v_1,...,v_k\}$, $E(P)=\{(v_i,v_{i+1})\in E\mid i\in[k-1]\}$, and $\Start(P)=v_1,\End(P)=v_k$. We define $\Succ_P(v_i)=v_{i+1}$ and $\Prec_P(v_i)=v_{i-1}$. 
For two vertices $v_i,v_j\in V(P)$ with $i\le j$, we define $P[v_i,v_j]$ as the subpath $\{v_i,v_{i+1},...,v_j\}$. Notice that $E(P)$ only contains undirected edges, if we want to emphasis the path direction, we use $\vec{E}(P)$ to denote the set of ordered pairs $(v_i,v_{i+1})$ for $i\in[k-1]$.

We define the length of the path $P$ as $\ell(P) = \ell(V(P))$(also denoted by $\leng(P)$), and its step $|P|$ as the number of edges in $E(P)$. %
When we want to emphasize that $\leng(P)$ is with respect to $G$, we will write $\leng(P)$ as $\leng(P,G)$.
We define $\dist_G(u,v)$ as the minimum length of simple paths from $u$ to $v$. We define $\dist_G(S,T)$ for two vertex sets $S,T$ as $\min_{s\in S,t\in T}\dist_G(s,t)$. We say $S,T$ are $h$-far if $\dist_G(S,T)\ge h$. We define the \emph{diameter} of $S$ as $\max_{u,v\in S}\dist_G(u,v)$. When $G$ is clear from the context, we omit $G$ from the subscript.

\paragraph{Graph Arboricity.} Give graph $G$, a \emph{forest cover} of $G$ consists of sub-graphs $F_1, F_2, \ldots, F_k$ of $G$ where each $F_i$ is an edge-disjoint forest and every edge of $G$ occurs in some $F_i$. $k$ is called the size of the forest cover and the \emph{arboricity} $\alpha$ of graph $G$ is the minimum size of a forest cover of $G$.

\subsection{Flows and Demands}
\label{sect:PrelimFlow}

\paragraph{(Multicommodity) Flows.} A (multicommodity) \emph{flow} $F$ in $G$ is a function that assigns to each (directed) simple path $P$ in $G$ a flow value $F(P)\in\bbR_{\ge 0}$. A path $P$ is a \emph{flow path} of $F$ if $F(P)>0$, and we use $\path(F)$ to denote the set of all flow paths of $F$. Without ambiguity, each flow path $P\in\path(F)$ is also a multicommodity flow (which has only one flow path $P$ with value $F(P)$) and all the notations below can apply on $P$. The value of $F$ is $\vvalue(F) = \sum_{P\in\path(F)} F(P)$. 
We have the following parameters for a flow $F$.
\begin{itemize}
    \item \emph{Length}: The length of $F$ is $\leng(F) = \max_{P\in\path(F)}\leng(P)$, and again we write $\leng(F)$ as $\leng(F,G)$ to emphasize that it is with respect to $G$. We say $F$ is an $h$-length flow if $\leng(F)\leq h$. 
    \item \emph{Step}: The step of $F$ is $\step(F) = \max_{P\in\path(F)}\step(P)$. 
    \item \emph{Congestion}: For each vertex $v \in V(G)$, the amount of \emph{flow units} allocated on $v$ by $F$ is $F(v) = \sum_{v\in V(P),P\in\path(F)} F(P)$.
The congestion of $F$ is $\cong(F) = \max_{v\in V(G)}F(v)/U(v)$. 
\end{itemize}

For a flow path $P\in \path(F)$, a \emph{flow subpath} $\hat{P}$ of $P$ is a subpath of $P$, and $\hat{P}$ also represents a multicommodity flow which has only one flow path $\hat{P}$ with value $\vvalue(P)$. A flow $F'$ is a \emph{subflow} of $F$ if $F'(P)\leq F(P)$ for each simple path $P$. Given two flows $F_{1}$ and $F_{2}$, $F_{1}+F_{2}$ denotes the flow where $(F_{1}+F_{2})(P) = F_{1}(P) + F_{2}(P)$ for each simple path $P$.

A \emph{path representation} of a flow represents the flow by storing all its flow paths and associated values. For \emph{edge representation}, we define $\vec{E}$ to be the directed counterpart of $E$, i.e. we put two opposite directed edges into $\vec{E}$ for each edge in $E$. The \emph{edge representation} of a flow $F$ is representing the flow by storing $F(e)$ for every edge $e\in \vec{E}$, where 
\[
F(e) = \sum_{P\in\path(F),\vec{E}(P)\ni e} F(P).
\]
Note that a flow path $P$ has a direction, and we use $\vec{E}(P)$ denote the directed edges on $P$ consistent with $P$'s direction.

\medskip\noindent{\textbf{(Multi-Commodity) Demands.}}
A \emph{demand} $D:V\times V\rightarrow\mathbb{R}_{\ge0}$ assigns
a non-negative value $D(v,w) \ge 0$ to each ordered pair of vertices in $V$ to specify the units of demand $D(v,w)$ vertex $v$ wants to send to vertex $w$. We note that it is always safe to assume $D(u,u) = 0$.
The size of a demand is written as $|D|$ and is defined as $\sum_{v,w} D(v,w)$. The \emph{support} of the demand $D$ is defined as $\supp(D)=\{u,v\in V\mid D(u,v)>0\}$. For each vertex $v$, we let $D(v,\cdot) = \sum_{w}D(v,w)$ and $D(\cdot,v) = \sum_{w}D(w,v)$ denote the total demand starting and ending at $v$ respectively. 
The vertex set of $D$, denoted by $V(D)$, is the set of all vertices $v$ with $D(v,\cdot)>0$ or $D(\cdot,v)>0$. 
We say a demand $D$ is
\begin{itemize}
    \item \emph{h-length constrained}: if for each $v,w \in V$ with $D(v,w) > 0$, we have $\dist_G(v,w) \le h$;
    \item \emph{sub-demand for $D'$}: if for any pairs of $v,w \in V$, $D(v,w) \le D'(v,w)$;
    \item \emph{$A$-respecting for some node-weighting $A$}: if for any $v\in V$, $\max\{D(v,\cdot),D(\cdot,v)\}\leq A(v)$;
\end{itemize}

For a flow $F$, the demand routed by/corresponding to $F$ is denoted by $\Dem(F)$ or $D_{F}$, where for each $u,v\in V(G)$, $D_{F}(u,v) = \sum_{P\in\path(F)~\text{s.t. $P$ is a $(u,v)$-path}} F(P)$.
We say a demand $D$ is routable with length $h$, congestion $\gamma$ and step $t$ if there is a flow $F$ routing $D$ with length $h$, congestion $\gamma$ and step $t$.

\paragraph{Single-Commodity Demands and Flows.} 

A \emph{single-commodity demand} $D:V\to\bbR$ assigns a (possibly negative) value $D(v)$ to each vertex in $V$, satisfying that $\sum_{v\in V} D(v) = 0$. We say a single-commodity demand $D_{\sg}$ is \emph{corresponding to} a multi-commodity demand $D_{\mt}$ if for each $v\in V$,
\[
D_{\sg}(v) = D_{\mt}(v,\cdot) - D_{\mt}(\cdot,v).
\]

For a single-commodity demand $D$, we say a flow $F$ is a \emph{single-commodity flow routing $D$} if $\Dem(F)$ is corresponding to $D$. %

\medskip\noindent{\textbf{$k$-Commodity Demands and Flows.}} 
For an integer $k$, a \emph{$k$-commodity demand} $D$ is a collection of $k$ single-commodity demand $D_{1},D_{2},...,D_{k}$. A \emph{$k$-commodity flow} $F$ is a collection of $k$ flows $\{F_{i}\mid i\in[k]\}$, and we call $F_{i}$ the \emph{$i$-th commodity flow of $F$}. We say a $k$-commodity flow $F$ routes a $k$-commodity demand $D$ if each $F_{i}$ routes $D_{i}$. The \emph{($k$-commodity) edge representation} of $F$ stores the edge representation of each $F_{i}$.

\subsection{Neighborhood covers.} 
We define neighborhood cover as follows.
\begin{definition}[Neighborhood covers]\label{def:neicov}
    Given a graph $G=(V,E)$ with vertex length, a \emph{neighborhood cover} $\cN$ is a collection of sets called \emph{clusterings} $\cS_1,...,\cS_\omega$ where each $\cS_i$ contains vertex disjoint subsets of $V$, we care about the following parameters of a neighborhood cover.
    \begin{enumerate}
        \item \textbf{Covering radius $h_{\cov}$:} Every subset $A\subseteq V$ with diameter at most $h_{cov}$ is a subset of some $S\in\cS_i$ for some $i$.
        \item \textbf{Diameter $h_{\diam}$:} For every $i$ and every $S\in\cS_i$, $S$ has diameter at most $h_{\diam}$.
        \item \textbf{Separation-factor $s$:} For every $i$, every two different sets in $\cS_i$ are $\left(s\cdot h_{\diam}\right)$-far.
        \item \textbf{Width:} the number of clustering, denoted by $w$.
    \end{enumerate}
\end{definition}

The following theorem is an algorithm for neighborhood cover for graphs with vertex length. It is implied by previous works about neighborhood cover for graphs with edge lengths.

\begin{theorem}\label{thm:neicov}
    There exists an algorithm that is given a graph $G$ with vertex lengths, parameters $h_{\cov},s>1$ and $\eps\in(0,1)$, output a neighborhood cover with covering radius $h_{\cov}$, separation-factor $s$, diameter and width
    \[h_{\diam}=\frac{1}{\eps}\cdot O(s)^{O(1/\eps)}\cdot h_{\cov}\qquad \qquad \omega=n^{O(\eps)}\cdot \log n\]
    in $h_{\diam}\cdot n^{O(\eps)}$ depth and $m\cdot h_{\diam}\cdot n^{O(\eps)}$ work with high probability. 
\end{theorem}

\begin{proof}
    The theorem follows from the algorithm by the Lemma 8.15 of \cite{HaeuplerRG22}, which is restated as Theorem 5.32 from \cite{HaeuplerHT24} with the same parameters under the setting of edge length. We show that we could transfer a vertex graph into a corresponding graph with edge length, and apply the algorithm from \cite{HaeuplerRG22} to get the same result.
    Namely, let $G_{v} = (V, E)$ be a graph with vertex length $\l_{G_{v}}$, we create a copy of it and denote it as $G_{e} = (V, E)$. For any edge $\{u, w\}$ in $E$, assign $(\l_{G_{v}}(u) + \l_{G_{v}}(w))/2$ to its length $\l_{G_{e}}(\{u, w\})$.
    We note that from the construction, we have that for any pair of vertices $u, w \in V$, $\dist_{G_{e}}(u,w) \le \dist_{G_{v}}(u,w) \le 2\dist_{G_{e}}(u,w)$.
    We then apply the algorithm from \cite{HaeuplerRG22} with the same covering radius $h_{\cov}$ and separation-factor $s$ over the graph $G_{e}$, the resulting neighborhood cover is also a valid one for the graph $G_{v}$ with the required parameters.
\end{proof}

When we do not care about the separation factor, there is a more efficient algorithm.

\begin{restatable}{lemma}{neicov}
Given an undirected graph $G$ with vertex lengths and parameters $h_{\cov}$ and $\beta$, there is an algorithm that computes a neighborhood cover ${\cal N}$ in $G$ with covering radius $h_{cov}$, diameter $h_{\diam} = \beta\cdot h_{\cov}$ and width $\omega = n^{O(1/\beta)}\log n$ in $h_{\diam}\cdot n^{O(1/\beta)}$ depth and $m\cdot h_{\diam}\cdot n^{O(1/\beta)}$ work.
\label{lemma:NeighborhoodCover}
\end{restatable}

\subsection{Length-Constrained Vertex Expander Decomposition}\label{subsec:preliminaryexpander}

The goal of this section is to introduce the concept of length-Constrained vertex expander decomposition. We start with defining concepts related to length-constrained vertex expanders.

\paragraph{Node-Weightings}
A \emph{node-weighting} $A:V\rightarrow\mathbb{R}_{\ge0}$ assigns a non-negative value $A(v)$ to a vertex $v$. The \emph{size} of $A$ is denoted by $|A|=\sum_{v}A(v)$ and let $\supp(A) := \{v : A(v) > 0\}$. For two node-weightings $A,A'$ we define $\min(A,A')$
and $A+A'$ as pointwise operations, and we write $A \preceq A'$ if $A$ is pointwise at most $A'$. %
We say $A$ restricting to a vertex set $S$ (denoted by $A_S$) as another node-weighting which assigns the same weight on every vertex in $S$ as $A$, and $0$ weight for every other vertices.

For a demand $D$ and a node-weighting $A$, we say $D$ is $A$-respecting if for each vertex $v$, $D(v,\cdot) + D(\cdot,v)\leq A(v)$.
For a flow $F$, We say $F$ routes $A$ to $B$ for some node-weightings $A,B$ with $|A|=|B|$ if $F$ routes a demand $D$ with $D(v,\cdot)=A(v)$ and $D(\cdot,v)=B(v)$.

\paragraph{Vertex Moving Cuts.} 

An $h$-length (moving) cut $C: V \rightarrow \{0, \frac{1}{h}, \frac{2}{h},\ldots\}$ on a graph $G$ assigns to each vertex $v$ a fractional cut value which is a multiple of $\frac{1}{h}$. The size of $C$ is defined as $|C| =\sum_{v}U_G(v)\cdot C(v)$. 
The length increase associated with the $h$-length moving cut $C$ is defined as $\l_{C,h}(v) = h \cdot C(v)$.

\begin{definition}
    For a graph $G$ with length function $\l$ we define $G - C$ (or $G+\l_{C,h}(v)$) as the graph with an updated length function $\l_{G-C}$ which assigns each vertex $v$ the length $\l_{G-C}(v) = \l(v) + \l_{C,h}(v)$.
\end{definition}
\begin{definition}[$h$-Length Separated Demand]
    For any $h$-length moving cut $C$ and demand $D$ on graph $G$, we define the amount of $h$-length separated demand as the sum of demands between vertices that are $h$-length separated by $C$. We denote this quantity with $\sep_h(C,D)$, i.e.,
    \begin{align*}
        \sep_{h}(C,D) = \sum_{u,v : \dist_{G-C}(u,v)>h} D(u,v).
    \end{align*}
\end{definition}

\begin{remark}
    In the notation of $\sep_h(C,D)$, we hide the inherent graph $G$ where the cut $C$ and demand $D$ is on. When using the definition, we assume we are working on the graph $G$ where $C$ and $D$ are defined. When we want to emphasize it is on a specific graph $G$, we will add it in the subscript. The same simplification applies to the rest of the definitions.
\end{remark}

\begin{definition}[$h$-Length Sparsity of a Cut $C$ for Demand $D$]\label{dfn:CDSparse}
For any demand $D$ and any $h$-length moving cut $C$ on graph $G$ with $\sep_{h}(C,D)>0$, the $h$-length sparsity of $C$ with respect to $D$ is the ratio of $C$'s size to how much demand it $h$-length separates i.e.,
\begin{align*}
    \spa_{h}(C,D) = \frac{|C|}{\sep_{h}(C,D)}.
\end{align*}
\end{definition}

\begin{definition}[$(h,s)$-Length Sparsity of a Cut w.r.t.\ a Node-Weighting]\label{def:sparsity}
The $(h,s)$-length sparsity of any $h\cdot s$-length moving cut $C$ with respect to a node-weighting $A$ is defined as:
\begin{align*}
    \spa_{(h,s)}(C,A) = \min_{A\text{-respecting h-length demand}\ D} \spa_{h \cdot s} (C,D).
\end{align*}

We say $C$ is an $(h,s)$-length $\phi$-sparse cut (with respect to $A$ on $G$) if $\spa_{(h,s)}(C,A)\le \phi$.
\end{definition}
Intuitively, $(h\cdot s)$-length sparsity of a cut measures how much it $h\cdot s$-length separates $h$-length demand w.r.t its own size.
Furthermore, for a given node-weighting, we associate the sparsest cut w.r.t the node-weighting with its conductance.

\begin{definition}[$(h,s)$-Length Conductance of a Node-Weighting]
    The $(h,s)$-length conductance of a node-weighting $A$ in a graph $G$ is defined as the $(h,s)$-length sparsity of the sparsest $h\cdot s$-length moving cut $C$ with respect to $A$, i.e.,
    \begin{align*}
        \cond_{(h,s)}(A) = \min_{h \cdot s\text{-length moving cut } C} \spa_{(h,s)}(C,A).
    \end{align*}
\end{definition}

\begin{definition}[$(h,s)$-Length $\phi$-Expanding Node-Weightings]\label{def:LC directed expansion}
    We say a node-weighting $A$ is $(h,s)$-length $\phi$-expanding if the $(h,s)$-length conductance of $A$ in $G$ is at least $\phi$.
\end{definition}

The above definition of length-constrained expanders characterizes them in terms of conductance. The below fact from \cite{HaeuplerLSW25} (see their Theorem 4.2) exactly characterizes length-constrained expanders as those graphs that admit both low congestion and low dilation routings under the setting of vertex capacity adapted from \cite{HaeuplerRG22}.

\begin{theorem}
[Routing Characterization of Length-Constrained Expanders, \cite{HaeuplerLSW25}]\label{thm:flow character} Given graph $G$ with vertex lengths and capacities and a node-weighting $A$, for any $h \geq 1$, $\phi < 1$ and $s \geq 1$ we have:
\begin{itemize}
\item \textbf{Length-Constrained Expanders Have Good Routings} If $A$ is $(h,s)$-length $\phi$-expanding in $G$, then every $h$-length $A$-respecting demand can be routed in $G$ with congestion at most $O(\frac{\log n}{\phi})$ and dilation at most $h\cdot s$.
\item \textbf{Not Length-Constrained Expanders Have a Bad Demand} If $A$ is not $(h,s)$-length $\phi$-expanding in $G$, then some $h$-length $A$-respecting demand cannot be routed in $G$ with congestion at most $\frac{1}{2\phi}$ and dilation at most $\frac{h\cdot s}{2}$.
\end{itemize}
\end{theorem}

The following definition is about length-constrained Vertex Expander Decomposition.
\begin{definition}[Length-Constrained Vertex Expander Decomposition]
    Given a graph $G = (V,E)$, an $(h,s)$-length $\phi$-expander decomposition for a node-weighting $A$ with length slack $s$ and cut slack $\kappa$ is an $h \cdot s$-length cut $C$ of size at most $\kappa \cdot \phi|A|$ such that $A$ is $(h,s)$-length $\phi$-expanding in $G - C$.
\end{definition}

\section{Length-Constrained Vertex Expander Decompositions}
\label{sect:LCVED}

We show the first efficient parallel length-constrained vertex expander decomposition algorithm, which is also associated with an expander routing algorithm. The proof is deferred to the appendix.

\begin{restatable}{theorem}{vertexLCED}
    \label{thm:vertexLC-ED}
    There is a randomized algorithm $\textsc{VertexLC-ED}(G,A,h,\phi,\eps)$ that takes inputs an undirected graph $G$ with vertex capacities and lengths, a node-weighting $A$, length bound $h\ge 1$, conductance $\phi\in(0,1)$ and a parameter $\eps$ satisfying $0< \eps<1$, outputs a $(h,s)$-length $\phi$-expander decomposition $C$ for $G,A$ with cut slackness $\kappa$ where
    \begin{align*}
    \kappa = n^{O(\eps)} \qquad \qquad s = \exp(\exp(\frac{O(\log(1/\eps))}{\eps}))
    \end{align*}
    in $\poly(h)\cdot m\cdot n^{\poly(\eps)}$ work and $\poly(h)\cdot n^{\poly(\eps)}$ depth. \footnote{We note that our length slackness parameter $s$ is doubly exponential in $1/\varepsilon$, whereas in \cite{HaeuplerHT24} a theorem on length-constrained expander decompositions for graphs with edge lengths and capacities claimed that $s$ is only single exponential in $1/\varepsilon$. Importantly, this discrepancy does not arise from our focus on vertex lengths and capacities. After discussions with the authors of \cite{HaeuplerHT24}, we have confirmed that their paper contains an error, and that a straightforward correction yields a length slackness that is doubly exponential in $1/\varepsilon$.}
    
    Moreover, given an $h$-length $A$-respecting demand $D$, there is an algorithm that outputs a path representation of a flow routing $D$ in $G-C$ with length $h\cdot s$ 
    and congesiton $n^{\poly(\eps)}/\phi$ 
    in work $(|\supp(D)|+m)\cdot \poly(h)\cdot n^{\poly(\eps)}$ and depth $\poly(h)\cdot n^{\poly(\eps)}$

\end{restatable}

\section{Length-Constrained Flow Shortcut with Length Slack $(1+\eps)$}
\label{sect:Shortcut}

In this section, we will give an algorithm for computing an LC-flow shortcut with length slack $(1+\eps)$ for arbitrarily small $\eps$ on undirected graphs with vertex lengths and congestion. The definition of LC-flow shortcuts is given in \Cref{def:LCFlowShortcut}. 

Before the definition, we need the following concept of \emph{path-mappings between flows}.  
Let $F^{G}$ and $F^{H}$ be two flows (which can be on different graphs $G$ and $H$) routing the same demand. A \emph{path-mapping} from $F^{G}$ to $F^{H}$, denoted by $\pi$, is a mapping from flow paths $P\in\path(F^{G})$ to subflows $F^{H}_{P}$ of $F^{H}$, satisfying that 
\begin{itemize}
\item $\Dem(P)=\Dem(F^H_P)$ (which implies $\vvalue(P) = \vvalue(F^{H}_{P})$) for each $P\in\path(F^{G})$, and
\item $\sum_{P\in\path(F^{G})} F^{H}_{P} = F^{H}$.
\end{itemize}
This path-mapping $\pi$ has length slack $(\lambda,\delta)$ if each $P\in\path(F^{G})$ has $\leng(F^{H}_{P})\leq \lambda\cdot\leng(P) + \delta$. Here we call $\lambda$ and $\delta$ the multiplicative and additive length slack respectively. When the additive length slack $\delta$ is zero, we simply say $\pi$ has length slack $\lambda$. When the multiplicative length slack $\lambda$ is one, we simply say $\pi$ has additive length slack $\delta$.

\begin{definition}
[Length-Constrained Flow Shortcut]
\label{def:LCFlowShortcut}
Given a graph $G=(V,E)$ with vertex length and capacities, we say an edge set $H$ (possibly with endpoints outside $V$) is an \emph{$t$-step LC-flow shortcut} of $G$ with \emph{length slack} $\lambda$ and \emph{congestion slack} $\kappa$ if the following holds.

\begin{itemize}

\item (Forward Mapping) For every feasible flow $F^G$ in $G$, there is a flow $F^H$ in $G\cup H$ routing $\Dem(F^G)$ with congestion 1 and step $t$. Furthermore, there is a path-mapping from $F^G$ to $F^H$ with length slack $\lambda$.

\item (Backward Mapping) For every feasible flow $F^H$ in $G\cup H$ such that $V(\Dem(F^H))\subseteq V(G)$, there is a flow $F^G$ in $G$ routing $\Dem(F^H)$ with congestion $\kappa$. Furthermore, there is a path-mapping from $F^{H}$ to $F^{G}$ with length slack $1$.

\end{itemize}

We define a \emph{backward mapping algorithm} associated with the LC-flow shortcut as follows. The input is an integer $k$ and the edge representation of a $k$-commodity flow $F^{H}$ in $G\cup H$ routing some $k$-commodity demand $D$ s.t. $V(D)\subseteq V$. The output is the edge representation of a $k$-commodity flow $F^{G}$ in $G$ routing $D$ with $\cong(F^{G})\leq \kappa\cdot \cong(F^{H})$ and $\totlen(F^{G}) \leq \totlen(F^{H})$.

\end{definition}

Before we state our result on LC-flow shortcuts, we make some comments on \Cref{def:LCFlowShortcut}. 
\begin{itemize}
\item We talk about path-mappings between flows in both forward and backward mappings because we want the length slacks to be \emph{path-competitive}. Existing a path-mapping with length slack $\lambda$ is a strong and flexible property in the sense that, taking the forward mapping as an example, it will imply claims such as $\leng(F^{H})\leq \lambda\leng(F^{G})$ and $\totlen(F^{H})\leq \lambda \totlen(F^{G})$. We note that our analysis for forward/backward mappings will actually map each flow path separately, so we get path-mappings for free.
\item In the backward mapping algorithm, we talk about $k$-commodity flows because we want better work and depth when $k$ is small. Also, when a flow is given in its edge representation, we lost the information of flow paths, so it no longer makes sense to talk about the length of the flow. However, we can still talk about total length (i.e. cost) because the total length is known given the edge representation. Moreover, we note that total length is exactly what we need when using LC-flow shortcuts in the future sections.
\end{itemize}

\Cref{thm:emulator} states our result on LC-flow shortcuts.

\begin{restatable}{theorem}{thmShortcut}\label{thm:emulator}
Let $G$ be an $n$-vertex $m$-edge undirected graph with vertex lengths and capacities. Let $1/\polylog(n)<\epsilon_{\ell}<1$ and $\Omega(1/\log\log\log n)< \epsilon_{\kappa},\epsilon_{h}<1$ be given parameters.
There is an algorithm that computes an LC-flow shortcut $H$ of $G$ of size $n^{1+\epsk}$ with length slack $1+\epsl$, congestion slack $n^{\epsk}$ and step 
\[
t = \left(1/(\epsl\epsh)\right)^{O\left(1/(\epsk\epsh)^2\right)}\cdot \exp(\exp(\log (1/(\epsk\epsh))\cdot O(1/(\epsk\epsh))))
\]
in work $n^{O(\epsh + \epsk)}m$ and depth $n^{O(\epsh)}$.
    
Moreover, the shortcut is associated with a backward mapping algorithm with work $n^{O(\epsh+\epsk)}(m+nk)$ and depth $n^{O(\epsh)}$.

\end{restatable}

We make some comments on the work bound of the backward mapping algorithm. Recall that it takes the edge representation of a $k$-commodity flow $F^{H}$ in $G\cup H$ as input, which means the input size is already $mk$. However, we can achieve $m+nk$ time because we assume the edge representation of $F^{H}$ is already stored in the memory, and our algorithm just needs to update part of it. Roughly speaking, we just perform backward mapping for the restriction of $F^{H}$ on $H$, and $H$ only has size roughly $n$. We point out that a roughly $mk$ work bound is already sufficient when we apply this theorem later.

Also, the restrictions $\epsl>1/\polylog(n)$ and $\epsk,\epsh>\Omega(1/\log\log\log n)$ are mainly added for convenience (so they are not tight). That is, this enable us to hide some terms purely depending on $\epsl,\epsk,\epsh$ into $n^{O(\epsk)}$ or $n^{O(\epsh)}$. For example, we can hide $\exp(\exp(1/\epsk))$ into $n^{\epsk}$.

By choosing parameters $\epsk = \epsh = \Theta(1/\log\log\log n)$, we immediate obtain the following cleaner version of \Cref{thm:emulator}.

\begin{corollary}
\label{coro:emulator}
Let $G$ be an $n$-vertex $m$-edge undirected graph with vertex lengths and capacities. Let $1/\polylog(n)<\epsilon<1$ be given parameters.
There is an algorithm that computes an LC-flow shortcut $H$ of $G$ of size $\hat{O}(n)$ with length slack $1+\eps$, congestion slack $\hat{O}(1)$ and step $\hat{O}(1)$
in work $\hat{O}(m)$ and depth $\hat{O}(1)$.
Moreover, the shortcut is associated with a backward mapping algorithm with work $\hat{O}(m+nk)$ and depth $\hat{O}(1)$.
\end{corollary}

\subsection{An Algorithm for Small Length Bound}\label{sec:smalllengthemulator}

In this section, we will provide an algorithm for efficiently constructing an LC-flow shortcut only for small enough (subpolynomial) length bound $h$. The exact trade-off between parameters is described in \Cref{lem:smallhemulator}.

For some technical reason, we need to strengthen the definition of shortcut \Cref{def:LCFlowShortcut}. The difference is that we require the forward mapping to preserve the congestion of the original vertices in $V$, which will be exploited when showing the congestion slack of the forward mapping. %

\begin{definition}
[Path-mapping $V$-Preserving LC-Flow Shortcut]
\label{def:hLCFlowShortcut}
Given a graph $G=(V,E)$ with vertex length and capacities, we say an edge set $H$ (possibly with endpoints outside $V$) is a \emph{path-mapping $V$-preserving $h$-length $t$-step LC-flow shortcut} of $G$ with \emph{additive length slack} $\delta$ and \emph{congestion slack} $\kappa$ if 

\begin{itemize}

\item (Forward Mapping) for every flow $F^G$ in $G$ with congestion $1$ and length $h$, there is a flow $F^H$ with congestion $1$ and step $t$ in $G\cup H$ with $\Dem(F^G)=\Dem(F^H)$ and for every $v\in V$ we have $F^H(v)\le F^G(v)$. Furthermore, there is a path-mapping from $F^G$ to $F^H$ with additive length slack $\delta$.

\item (Backward Mapping) for every flow $F^H$ in $G\cup H$ with congestion 1 such that $V(\Dem(F^H))\subseteq V(G)$, there is a flow $F^G$ in $G$ routing $\Dem(F^H)$ with congestion $\kappa$.
Furthermore, there is a path-mapping from $F^{H}$ to $F^{G}$ with length slack $1$.

\end{itemize}

\end{definition}

\begin{lemma}\label{lem:smallhemulator}
Let $G$ be an $n$-vertex $m$-edge undirected graph with vertex lengths and capacities. Let $1/\polylog(n)<\epsl<1$, $1/\Omega(\log\log\log n)<\epsk<1$, and $h\geq 1$ be given parameters. There is an algorithm $\Emulator(G,\epsl,\epsk,h)$ computes a path-mapping $V$-preserving $h$-length $t$-step LC-flow shortcut of $G$ of size $n^{1+O(\epsk)}$ with additive length slack $\epsl\cdot h$ and congestion slack $n^{\epsk}$ where
\[t=\left(1/\epsl\right)^{O\left(1/\epsk^2\right)}\cdot \exp(\exp(\log (1/\epsk)\cdot O(1/\epsk))).\]
in $\poly(h)\cdot n^{\poly(\epsk)}$ depth and $\poly(h)\cdot m\cdot n^{\poly(\epsk)}$ work. 

Moreover, the shortcut is associated with a backward mapping algorithm with $\poly(h)\cdot n^{\poly(\epsk)}$ depth and $\poly(h)\cdot (m+nk)\cdot n^{\poly(\epsk)}$ work. %
\end{lemma}

\paragraph{Building Blocks.} There are three important building blocks for our algorithm: length-constrained vertex expander decomposition \Cref{thm:vertexLC-ED}, neighborhood cover \Cref{lemma:NeighborhoodCover}, and approximate length-constrained single-commodity max flows. We state them as below for convenience.

\newtheorem*{restateED}{\textbf{Theorem \ref{thm:vertexLC-ED}}}
\begin{restateED}
    There is a randomized algorithm $\textsc{VertexLC-ED}(G,A,h,\phi,\eps)$ that takes inputs an undirected graph $G$ with vertex capacities and lengths, a node-weighting $A$, length bound $h\ge 1$, conductance $\phi\in(0,1)$ and a parameter $\eps$ satisfying $0< \eps<1$, outputs a $(h,s)$-length $\phi$-expander decomposition $C$ for $G,A$ with cut slackness $\kappa$ where
    \begin{align*}
    \kappa = n^{O(\eps)} \qquad \qquad s = \exp(\exp(\frac{O(\log(1/\eps))}{\eps}))
    \end{align*}
    in $\poly(h)\cdot m\cdot n^{\poly(\eps)}$ work and $\poly(h)\cdot n^{\poly(\eps)}$ depth.
    
    Moreover, given an $h$-length $A$-respecting demand $D$, there is an algorithm that outputs a path representation of a flow routing $D$ in $G-C$ with length $h\cdot s$ 
    and congesiton $n^{\poly(\eps)}/\phi$ 
    in work $(|\supp(D)|+m)\cdot \poly(h)\cdot n^{\poly(\eps)}$ and depth $\poly(h)\cdot n^{\poly(\eps)}$
\end{restateED}

Another important building block is neighborhood cover \Cref{lemma:NeighborhoodCover}. 

\neicov*

The last building block is a $0.5$-approximate length-constrained single-commodity max flow algorithm. We present here a simplified yet sufficient version in \Cref{lemma:LCstFlows}, and refer to \cite{HaeuplerHT24} or \Cref{thm:ApproxLCMCMF} for detailed and generalized versions. We note that in the context of \Cref{sect:PrelimFlow}, an $h$-length single-commodity $s$-$t$ flow $F_{s,t}$ is just an $h$-length flow \emph{partially} routing the single-commodity demand $D$ where $D(s) = \infty$, $D(t) = -\infty$, and $D(v) = 0$ for other vertices $v$. It is feasible if it has congestion $1$. It is $0.5$-approximately maximum if $\vvalue(F_{s,t})$ is at least half of the maximum value of all such flows.

\begin{lemma}[See \cite{HaeuplerHT24} or \Cref{thm:ApproxLCMCMF}]
\label{lemma:LCstFlows}
Let $G$ be an undirected graph with vertex lengths and capacities. Given a source vertex $s$, a sink vertex $t$, and a length bound $h$, there is an algorithm that computes a feasible $h$-length single-commodity $s$-$t$ flow $F_{s,t}$ which is $0.5$-approximately maximum. The flow $F_{s,t}$ is given in its path representation, which contains $\tilde{O}(m\cdot\poly(h))$ flow paths. The work and depth are $\tilde{O}(m\cdot \poly(h))$ and $\tilde{O}(\poly(h))$ respectively.
\end{lemma}

Now we are ready to describe the algorithm $\Emulator(G,\epsl,\epsk,h)$. The algorithm involves recursive calls to itself with a reduced number of vertices for the input graph $G$.

\newcommand{\tepsl}{\tilde{\epsl}}
\paragraph{Global variables.} %

We fix $n$ to be the number of vertices of the input graph of the root computation of the recursive algorithm, i.e., $n$ does not change during the recursive calls. We let $c$ to be a sufficiently large constant.

We will use the parameters $\kappa,s$ from \Cref{thm:vertexLC-ED} in our algorithm, which only depends on $n,\eps$. We fix $\eps$ to be $\epsk/c^2$ to $\Emulator(G,\epsl,\epsk,h)$, among all recursive calls. i.e., we always have
\[\eps=\epsk/c^2,\qquad\qquad\kappa\le n^{\epsk/c},\qquad\qquad s=\exp\exp(\frac{O(\log(1/\epsk))}{\epsk}).\]
When we use \Cref{thm:vertexLC-ED}, the input conductance $\phi$ is a fixed value depending only on $n,\eps$, which we set to be
\[\phi=\frac{1}{n^{\frac{\epsk}{9}}}.\]
It follows that, 
\[\phi\kappa=\frac{1}{n^{\frac{\epsk}{9}}}\cdot n^{\epsk/c}\le \frac{1}{n^{\epsk/10}}.\]

We will use the parameters $\beta,\omega$ from \Cref{thm:neicov}. Here $\omega$ depends on $\beta,n$ where we define $\beta$ to be
\[\beta=\frac{c}{\epsk^2},\]
so that
\[\omega=n^{O(1/\beta)}\le n^{\epsk^2/10}.\]

In the algorithm, we will use $\sigma$ to denote the recursion parameter, i.e., every recursion decreases the graph size by a factor of $\sigma$. We fix $\sigma$ to be
\[\sigma=n^{\epsk/6}.\]
The number of iterations in the algorithm is defined as 
\[r=c\cdot \log_{\frac{1}{\kappa\phi}}N=O(\frac{1}{\epsk}).\] Remember that $N=\poly(n)$ is an upper bound for the summation of all the capacities.

It will be clear from the analysis why we define these parameters in this way, but let us give a high-level explanation here. As discussed in \Cref{sect:OverviewPart1}, our argument can be viewed as a two-dimensional induction. To distinguish steps along these two dimensions, we refer to the first as \emph{recursions} and the second as \emph{iterations}. We hope to bound both the recursion depth and number of iterations by $O(1/\epsk)$, and as we will see, the former is controlled by $\sigma$ and the latter is controlled by $\kappa\phi$, so we set them in this way. Regarding $\omega$, the intuition is that the total graph size at each recursion level increases by a factor of $\omega$, so we want $\omega \leq n^{O(\epsk^{2})}$ to obtain a size overhead of $n^{O(\epsk)}$ after $O(1/\epsk)$ levels.

\paragraph{Base case.} The base case is when $G$ has fewer than $\sigma$ vertices. However, we do not need to handle it separately because the following recursive algorithm $\Emulator(G,\epsl,\epsk,h)$ automatically covers this case (in which no further recursion is performed). 

For better understanding, we describe a simpler and alternative treatment informally as follows. Because the number of vertices is at most $\sigma$, which is upper bounded by the desired congestion slack $n^{\epsk}$. We can simply create shortcuts between each pair of vertices. This introduces a factor of $\sigma^{2}$ to the congestion slack, which is acceptable. Formally, we implement this by performing Step 5 once with each vertex treated as a singleton large cluster.

\paragraph{Recursive algorithm $\Emulator(G,\epsl,\epsk,h)$.} Suppose the size of $G$ is larger than a sufficiently large constant. The algorithm consists of 
$r$
many iterations.
We set a node-weighting $A_1(v)=U_G(v)$. 

We define the relax slackness for $\epsl$ as
\[\tepsl=\frac{\epsl}{100\cdot 50^{r}}.\]

We define $h_r=h$, and we define recursively for each $i\in[r],$
\[h_{\diam,i}=\frac{\tepsl\cdot h_i}{s},\qquad\qquad h_{\cov,i}=\frac{h_{\diam,i}}{\beta},\qquad\qquad h_{i-1}=h_{\cov,i} = \frac{\tepsl\cdot h_{i}}{\beta s}.\]

Now we formally describe the $i$-th iteration as follows. The $i$-th interaction will construct the shortcut $H_i$. Initially let $H_i=\emptyset$.

\newcommand{\hed}{h_{ED}}
\paragraph{Step 1 (LC-expander decomposition).} We use \Cref{thm:vertexLC-ED} to call 
\[C_i\leftarrow \textsc{VertexLC-ED}(G,A_i,h_{\diam,i},\phi,\eps).\]
We set the node-weighting for the next iteration $A_{i+1}(v)=C_i(v)\cdot U_G(v)$. 

\paragraph{Step 2 (neighborhood cover).} Define $G_i=G-C_i$. We use \Cref{thm:neicov} with input graph $G_i$, and covering radius $h_{\cov,i}$. \Cref{lemma:NeighborhoodCover} returns a neighborhood cover $\cN_i$. We define $\cS_i=\cup_{\cS\in\cN}\cS$.

\paragraph{Step 3 (recursions on small clusters).} Define $\cS_{\ssmall,i}=\{S\in\cS_i\mid |S|\le n/\sigma\}$.
For every $S\in\cS_{\ssmall,i}$, we make the recursive call 
\[H_S\leftarrow \Emulator(G[S],\tepsl,\epsk,h_{\cov,i}),\]
and add the returned LC-flow shortcut $H_S$ to $H_i$. 

\paragraph{Step 4 (stars on large clusters).} Define $\cS_{\big,i}=\cS_i-\cS_{\ssmall,i}$. For each $S\in\cS_{\big,i}$, we construct a new artificial vertex $x_S$ with capacity $\sum_{v\in V}U_G(v)$ and length $1$. For each $v\in x_S$, we build a new artificial vertex $x_{S,v}$ with length $h_{\diam,i}\cdot s$ and capacity $A_i(v)$. We build edges $(x_S,x_{S,v})$ and $(x_{S,v},v)$ . We denote the set of these edges as $H_S$. We add $H_S$ to $H_i$.

\paragraph{Step 5 (edges connecting large clusters).} 
For each pair of large clusters $S_1,S_2\in\cS_{\big,i}$, construct 
\[\bar{k_i}=\log_{1+\tepsl}(3h_i)\]
nodes. For each $1\leq k\leq \bar{k}_{i}$, the $k$-th shortcut node denoted by $\shortcut(x_{S_{1}},x_{S_{2}},k)$, has 

\begin{itemize}[leftmargin=6em,labelsep=1.5em]
    \item[length] $(1+\tepsl)^{k} + 2h_{\diam,i}\cdot s$, and 
    \item[capacity] $2$ times the value of the $0.5$-approximate maximum $((1+\tepsl)^{k} + 2h_{\diam,i}\cdot s)$-length single-commodity flow (with congestion 1) from $x_{S_{1}}$ to $x_{S_{2}}$ in $G_i\cup H_{S_{1}}\cup H_{S_{2}}$.
\end{itemize}
\begin{remark}
    In the definitions, we require the length to be a positive integer. In the algorithm, we sometimes face a fractional length (for example $((1+\tepsl)^{k} + 2h_{\diam,i}\cdot s)$. In this case, we can simply round it up or down with a small loss. For simplicity, we do not reflect the change in the algorithm unless the length is small enough so that rounding results in a large loss.
\end{remark}
We can use \Cref{lemma:LCstFlows} to compute the capacity. Then we construct edges connecting $x_{S_1}$ to the $\bar{k_i}$ nodes, and edges connecting the $\bar{k_i}$ nodes to $x_{S_2}$. %

\paragraph{Outputs.} The returned LC-flow shortcut $H$ for $\Emulator(G,\epsl,\epsk,h)$ is the union of $H_i$ for all iterations $i\in[r]$. 

\paragraph{Backward mapping algorithm.} We will describe and analyze the algorithm after the backward mapping correctness section in \Cref{subsec:backwardmappingalgorithm} since it is almost identical to the backward mapping correctness proof.

\paragraph{Complexity of $\Emulator$.} Before proving the correctness, let us first calculate the complexity for the LC-flow shortcut part.

Step 1 takes $\poly(h_{\diam,i})\cdot n^{\poly(\epsk)}=\poly(h)\cdot n^{\poly(\epsk)}$ depth and $\poly(h)\cdot n^{\poly(\epsk)}\cdot m$ work according to \Cref{thm:vertexLC-ED}.

Step 2 takes $\poly(h_{\diam,i})\cdot \omega=\poly(h)\cdot n^{O(\epsk^2)}$ depth and $\poly(h_{\diam,i})\cdot m\cdot \omega=\poly(h)\cdot m\cdot n^{O(\epsk^2)} $ work according to \Cref{thm:neicov}.

Step 3 is a recursive procedure. We will analyze it in the end.

Step 4 can be done trivially.

Step 5 depends on the number of large clusters, which is calculated in the following lemma.

\begin{lemma}\label{lem:bigcluster}
    For every $i\in[r]$, $|\cS_{\big,i}|\le \omega\sigma \leq n^{O(\epsk)}$.
\end{lemma}
\begin{proof}
    According to the definition of neighborhood cover \Cref{def:neicov}, a neighborhood cover contains $\omega$ layers of clustering, where in each layer, different clusters are vertex disjoint. Thus, since each large cluster $S\in\cS_{\big,i}$ contains at least $n/\sigma$ vertices, the number of large clusters in each layer is $\sigma$. There are $\omega$ layers, so the lemma follows.
\end{proof}

Step 5 requires $(\omega\sigma)^2=n^{O(\epsk)}$ calling of approximate length-constrained single-commodity flow. Notice that the length bound is always at most $O(h_i)=O(h)$. Thus, according to \Cref{lemma:LCstFlows}, the depth for each flow computation is $\tilde{O}(\poly(h))$ depth and $\tilde{O}(\poly(h)\cdot m)$ work. In total, the work becomes $n^{O(\epsk)}\cdot \poly(h)\cdot m$.

Summing up all the steps, the depth is $\poly(h)\cdot n^{\poly(\epsk)}$ and work is $\poly(h)\cdot m\cdot n^{\poly(\epsk)}$ besides the recursive call.

Now, we analyze the recursive procedure. According to the definition of neighborhood cover, each vertex is included in at most $\omega$ many different clusters $S\in S_i$. There are $r$ iterations. Thus, the $z$-th recursion layer contains $(\omega r)^z n$ many vertices and $(\omega r)^zm$ many edges. There are $\log_\sigma n=O(1/\epsk)$ many layers, thus, the total number of vertices and edges are at most $(\omega r)^{\log_\sigma n}=n^{O(\epsk)}$ factor of more vertices and edges than the original graph. For the input parameters, we decreased the input $\epsl$ and $h$ in each layer of recursion. Decreasing the input for $h$ can only decrease the complexity, and $\epsl$ does not affect the complexity. Thus, the complexity stated in \Cref{lem:smallhemulator} follows.

\paragraph{Size of $H$.} Let us analyze the size of the output shortcut $H$. We only add edges to $H$ by Step 4 and Step 5 besides recursion.

In Step 4, we add $O(|S|)$ edges for each large cluster $S$. According to \Cref{lem:bigcluster}, there are at most $\omega\sigma r n=n^{1+O(\epsk)}$ many edges added in this step.

In Step 5, we build $\bar{k_i}=\tO{1/\tepsl} = \tilde{O}(1)$ (by our lower bounds on $\epsl$ and $\epsk$) edges between every two large clusters. This part is dominated by the number of edges in Step 4.

For recursion, as shown in the complexity analysis, the total number of vertices and edges among all recursions are $n^{O(\epsk)}\cdot n$ and $n^{O(\epsk)}\cdot m$. Thus, the size of the returned shortcut $H$ for the root recursion is $n^{1+O(\epsk)}$.

\medskip

In the next two sections, we will prove that the returned edge set is indeed a LC-flow shortcut satisfying the parameters specified by \Cref{lem:smallhemulator}. According to the definition of LC-flow shortcut \Cref{def:hLCFlowShortcut}, we need to verify forward mapping and backward mapping.

\paragraph{Induction hypothesis and base cases.} The proofs are based on induction on the size of the input graph to $\Emulator$ during the recursive calls, for which we denote by $G'$. 

\begin{lemma}[Induction hypothesis for $\Emulator(G',\epsl',\epsk,h)$]\label{lem:inductionforsmallh}
    Let $G'$ be a graph with $n'\leq n$ vertices. The algorithm $\Emulator(G',\epsl',\epsk,h)$ outputs a path-mapping $V$-preserving $h$-length LC-flow shortcut of $G'$ with additive length slack $\epsl'\cdot h$ and the following congestion slack and step.  %
    Let $d' = \lceil\log_{\sigma}n'\rceil$. Then
    \begin{align*}
    \text{step}~~~~~~~~~t(n',\epsl') &= \left(\frac{40\beta s}{\epsl'}\right)^{(r+1)\cdot d'}\cdot 50^{O(r^{2}\cdot (d')^{2})},~\text{and}\\
    \text{congestion}~~~~\kappa(n') &= (\omega r+1)^{d'}\cdot n^{\epsk}.
    \end{align*}
\end{lemma}

When $n'=n$ and $\epsl'=\epsl$, i.e., the root case, we plug in $\sigma = n^{\epsk/6}, \beta = O(1/\epsk^{2}), s = \exp\exp(\log(1/\epsk)\cdot O(1/\epsk)), r = O(1/\epsk)$ and $\omega \leq n^{\epsk^{2}/10}$, and then obtain
\begin{align*}
t(n,\epsl) = (1/\epsl)^{O(1/\epsk^{2})}\cdot \exp\exp(\log(1/\epsk)\cdot O(1/\epsk)),\qquad\text{and}\qquad \kappa(n) = n^{O(\epsk)}
\end{align*}
exactly as stated in \Cref{lem:smallhemulator} (after scaling down $\epsk$ by a large constant initially). 

\paragraph{Base case.} The base case for \Cref{lem:inductionforsmallh} is when $G$ has fewer than $\sigma$ vertices. Recall that we can choose not to handle it separately, and the following full argument shows that the returned LC-flow shortcut has desired guarantee. Or we can use the simpler alternative treatment. In this case, the output is clearly a $V$-preserving LC-flow shortcut with length slack $1+\epsl$, congestion slack $n^{\epsk}$

\medskip

In the next several sections, we will assume all the recursive calls are correct (since they reduced the graph size by at least $1$) and assume $G$ has $n'\le n$ vertices, the input length slackness is $\epsl$ to prove \Cref{lem:inductionforsmallh}.

\subsection{Proof of Correctness: Forward Mapping}

We will use induction to prove the following lemma holds for every $i$. Recall that we defined $G_i=G-C_i$ where $C_i$ is a $(h_{\diam,i},s)$-length moving cut returned in step $1$. 

\newcommand{\tepsli}[1]{\epsl^{(#1)}}
\begin{lemma}[Induction hypothesis for forward mapping]\label{lem:forwardinduction}
    For every $0\le i\le r$, for every $3 h_i$-length flow $F^G$ in $G_i$ with congestion $1$, there is a flow $F^H$ in $G\cup \left(\cup_{j\le i}{H_j}\right)$ routing $\Dem(F^G)$ with step $t^{(i)}$ and congestion $1$. Furthermore, there is a path-mapping from $F^G$ to $F^H$ with additive length slackness $\tepsli{i}\cdot h_i$
    such that
    \[\tepsli{i}=\frac{\epsl}{2\cdot 50^{r-i}}\qquad\qquad t^{(i)}=\left(\frac{40\beta s}{\tepsl}\right)^{i+1}\cdot t(n'/\sigma,\tepsl)\]
    where $t(n'/\sigma,\tepsl)$ denotes the function specified in \Cref{lem:inductionforsmallh}. We do not write $\epsk$ as a parameter to $t$ because $\epsk$ never changes during the recursive calls.

    Moreover, for every $v\in V$, we have $F^H(v)\le F^G(v)$. 
\end{lemma}

\paragraph{Base case.} For convenience, we define $G_0=G-C_0$ where $C_0$ is a $(h_{\diam,0},s)$-length moving cut such that $C_0(v)=1$ for every $v$, even though there is no `$0$-th iteration'. The reason for defining $G_0$ is that \Cref{lem:forwardinduction} is easy to verify for the base case when $i=0$. However, we point out that this is just for formality. One can also view $i=1$ as the base case and prove its guarantees directly using (a simplified version) of the general-case argument.

In $G-C_0$, every vertex has length at least $sh_{\diam,0}$ since $C_0(v)=1$ for every $v\in V$. Thus, every $3h_0$-length flow $F^G$ in $G-C_0=G_0$ has step at most $3/\tepsl$, and $F^{G}$ automatically has step at most $3/\tepsl$ in $G$, satisfying \Cref{lem:forwardinduction}. In the rest of the section, we will prove \Cref{lem:forwardinduction} for $i>0$ while assuming it is correct for every value less than $i$.

\paragraph{Ending case.} Before going into the induction step, let us first see \Cref{lem:forwardinduction} implies the required forward mapping when $i=r$. 

Remember that we set $A_{i+1}(v)=C_i(v)\cdot U(v)$. According to \Cref{thm:vertexLC-ED}, the size of $A_{i+1}$ decreases by a factor of $\kappa\phi$ in each iteration. Initially $A_{1}$ has size $N$, so after $r$ iterations, the size of $A_{r}$ decreases to at most $\frac{1}{N}$. This implies that the size of $C_{r}$ is at most $\frac{1}{N}$. Notice that in $G_{r}=G-C_{r}$, the total length increase on all vertices is at most $\sum_{v}C_{r}(v)\cdot h_{r}\le \frac{h_{r}}{N}$. Thus, we get that for every $h$-length flow $F^G$ in $G$, which is a $h_r+\frac{h_r}{N}$-length flow in $G_r$, with congestion $1$, there is a flow $F^H$ in $G\cup H$ routing $\Dem(F^G)$ with congestion $1$ and step $t^{(r)}$. Furthermore, there is a path-mapping from $F^G$ for $F^H$ with additive length slackness \[\tepsli{r}\cdot h_{r}+\frac{h_r}{N}\le \epsl\cdot h\] and step
\[t^{(r)}=\left(\frac{40\beta s}{\tepsl}\right)^{r+1}\cdot t(n'/\sigma,\tepsl)\le t(n',\epsl)\]

This is exactly the definition of forward mapping in \Cref{def:hLCFlowShortcut}.

\subsubsection{Construction of the Forward Mapping $F^H$}\label{subsubsec:constructionofforward}
Consider an arbitrary $3 h_{i}$-length flow $F^{G}$ in $G_i$. For each flow path $P\in\path(F^{G})$, we will construct a flow $F^{H}$ in $H$ routing $\Dem(P)$ as follows.

\paragraph{{Decompose the Flow Path $P$.}} We first decompose $P$ into flow subpaths $\hat{\cal P}$ using the following lemma.

\begin{lemma}\label{lem:decompose}
    A path $P$ can be decomposed into $z \leq 10\leng(P)/h_{\cov,i}$ subpaths $\hat{\cal P}=\{\hat{P}_1,...,\hat{P}_z\}$ such that
    \begin{enumerate}
        \item $\Start(\hat{P}_1)=\Start(P), \End(\hat{P}_z)=\End(P)$, for every $i\in[z-1]$, $\End(\hat{P}_i)=\Prec(\Start(\hat{P}_{i+1}))$,
        \item each $\hat{P}_i$ either
        \begin{itemize}
            \item contains exactly $1$ nodes, which we call a \emph{trivial path}, or
            \item has length at most $h_{\cov,i}$.
        \end{itemize}
    \end{enumerate}
\end{lemma}

\begin{proof}
    We will construct $\hat{P}_i$ iteratively. Suppose we have constructed $\hat{P}_i$ where $v=\End(\hat{P}_i)$, we will construct $\hat{P}_{i+1}$ based on $v$ in the following way (if $i=0$, we define $v=\Start(P)$): if $\ell(\Succ(v))>h_{\cov,i}/2$, then we let $\hat{P}_{i+1}$ be the only vertex $\Succ(v)$, which is a trivial path; otherwise, we let the subpath $\hat{P}_{i+1}$ to start with vertex $\Succ(v)$ and end at the last vertex on $P$ satisfying $\leng(\hat{P}_{i+1})\le h_{\cov,i}$. 

    Now we show that the number of paths is small. It is clear that paths are vertex-disjoint. The total length of all subpaths is at most $\leng(P)$. Thus, the number of trivial subpaths constructed in the first way is at most $2\leng(P)/h_{\cov,i}$ since each path has length at least $h_{\cov,i}/2$. Moreover, let $\hat{P}_i$ be a non-trivial path where $i<z$ and $\hat{P}_{i+1}$ is also a non-trivial path, then $\leng(\hat{P}_i)\ge h_{\cov,i}/2$ since otherwise we can augment $\hat{P}_i$ further. This means there are at most $2\leng(P)/h_{\cov,i}+1$ number of non-trivial paths even if we exclude all non-trivial paths before each trivial path. The total number of subpaths is at most $10\leng(P)/h_{\cov,i}$.
\end{proof}

Remember that $\cS_i$ contains all the clusters of a neighborhood cover with covering radius $h_{\cov,i}$ on the graph $G_i$, thus, each nontrivial subpath $\hat{P}\in\hat{\cal P}$ corresponds to a cluster $S_{\hat{P}}\in \cS_{i}$ such that $\hat{P}$ is totally inside $G[S_{\hat{P}}]$, and we call $\hat{P}$ a \emph{large-cluster subpath} (resp. \emph{small-cluster subpath}) if $S_{\hat{P}}$ is a large (resp. small) cluster.

We classify subpaths into \emph{light subpaths} and \emph{heavy subpaths}. For each subpath $\hat{P}\in \hat{\cal P}$, it is a light subpath if 
\[
C_{i-1}(\hat{P})\leq 1,
\]
otherwise $\hat{P}$ is a heavy subpath. %

Let $\hat{P}_{\HvLg,L}$ (resp. $\hat{P}_{\HvLg,R}$) be the leftmost (resp. rightmost) heavy and large-cluster subpath. Let $w_{L}$ be the left endpoint of $\hat{P}_{\HvLg,L}$ and let $w_{R}$ be the right endpoint of $\hat{P}_{\HvLg,R}$. 

We substitute the subpaths in $\hat{\cal P}$ between $\hat{P}_{\HvLg,L}$ and $\hat{P}_{\HvLg,R}$ (including them) with one flow subpath $\hat{P}_{\jump}$ (which is the subpath of $P$ starting at $w_L$ and ending at $w_R$), called the \emph{jumping subpath}. Note that after substituting, each subpath $\hat{P}\in\hat{\cal P}$ now belongs to one of the following types: trivial, small-cluster, light, and jumping. 

\medskip

\paragraph{Route the Flow Subpaths.} We will construct a routing $F^{H}_{P}$ of $\Dem(P)$ in $G\cup H$ by routing $\Dem(\hat{P})$ of each subpath $\hat{P}\in\hat{\cal P}$ separately and concatenating them at the end. 
Consider a subpath $\hat{P}\in\hat{\cal P}$.

\medskip

\noindent\underline{{Trivial Subpaths.}} If $\hat{P}$ is a trivial subpath, then we do not do anything: its routing $F^H_{\hat{P}}=\hat{P}$.

\noindent\underline{{Small-cluster Subpaths.}} Suppose $\hat{P}$ is a small-cluster subpath. Since $\hat{P}$ is not a trivial subpath, its length is at most $h_{\cov,i}$. Let $S=S_{\hat{P}}$ be the cluster containing $\hat{P}$. Remember that $H_S$ is returned in Step 3 of the algorithm by $\Emulator(G[S],\tepsl,\epsk,h_{\cov,i})$. We collect all the flow subpaths contained in $S$ and define

\[F_S=\sum_{P\in\path(F^G)}\sum_{\substack{\hat{P}\in \hat{\cP}\\S_{\hat{P}}=S}}\hat{P}\]

Notice that $F_S$ is a $h_{\cov,i}$-length flow. According to the definition of $H_S$, $\Dem(F_S)$ can be routed in $G\cup H_S$ by a flow $F^{H}_S$ with path mapping, congestion, and step guaranteed by \Cref{lem:smallhemulator}. We will analyze the routing quality later. We define the routing for $\hat{P}$ as the subflow in $F^H_S$ routing the demand $\Dem(\hat{P})$, denoted by $F^H_{\hP}$, notice that it satisfies the path mapping parameter stated in \Cref{lem:smallhemulator}.

\medskip
\noindent\underline{{Light Subpaths.}} Suppose $\hat{P}$ is a light subpath. Notice that $\hat{P}$ has length increase in $G_{i-1}=G-C_{i-1}$ at most $s\cdot h_{\diam,i-1}$ according to the definition of light subpaths and that $C_{i-1}$ is a $(h_{\diam,i-1},s)$-length moving cut. Thus, $\hat{P}$ in $G_{i-1}$ has length at most 
\[
h_{\cov,i}+s\cdot h_{\diam,i-1}\le 3 h_{i-1},
\]
because we define $h_{\cov,i} = h_{i-1}$ and $h_{\diam,i-1} = \tepsl\cdot h_{i}/s$. We collect all the light flow subpaths 

\[F_{\light}=\sum_{P\in\path(F^G)}\sum_{\substack{\hat{P}\in \hat{\cP} \\ \hat{P}\text{ is light}}}\hat{P}\]

Notice that $F_{\light}$ has length at most $3 h_{i-1}$. According to the induction hypothesis \Cref{lem:forwardinduction}, there is a flow $F^H_{\light}$ and a path mapping from $F_{\light}$ to $F^H_{\light}$ routing $\Dem(F_\light)$ in $G\cup H$ with the quality specified by \Cref{lem:forwardinduction}. We will analyze the quality later. We define $F^H_{\hat{P}}$ as the mapped flow path in $F^H_{\light}$ that routes the demand $\Dem(\hat{P})$.

\medskip

\noindent\underline{{Jumping Subpaths.}} Lastly, consider the case that $\hat{P}$ is the jumping subpath $\hat{P}_{\jump} = P[w_{L},w_{R}]$. This is the most complicated case and we need to first make some definitions.
\begin{description}
    \item[Definition of flows routing $w_L$ to vertices in $\hat{P}_{\HvLg,L}$.] Let $\hat{P}_{\jump,\pre}$ be the \emph{shortest} prefix of $\hat{P}_{\HvLg,L}$ (so it is also a prefix of $\hat{P}_{\jump}$) such that $C_{i-1}(\hat{P}_{\jump,\pre})\geq 1$. By the definition of heavy path, such a prefix always exists. Notice that $C_{i-1}(\hat{P}_{\jump,\pre})\le 2$ since $C_{i-1}$ assigns at most $1$ to each vertex (otherwise we decrease it to $1$ and it is still an length constraint expander decomposition). For every node $v\in V(\hat{P}_{\jump,\pre})$, we define the flow path $\hat{P}_{\pre,v}$ as the subpath $P[w_L,v]$ carrying flow value $\vvalue(\hat{P}_{\pre,v})=\vvalue(P)\cdot C_{i-1}(v)$. Notice that the length of $\hat{P}_{\pre,v}$ on the graph $G_{i-1}$ is at most 
\[h_{\cov,i}+2\cdot h_{\diam,i-1}s\le 3h_{i-1}\]
because $\hat{P}_{\pre,v}$ has length at most $h_{\cov,i}$ on $G$, and $C_{i-1}$ is a $(h_{\diam,i-1},s)$-length cut. 
We will also guarantee that the total flow value is exactly $\vvalue(P)$, i.e.,
\[\sum_{v\in\hat{P}_{\jump,\pre}}\vvalue(\hat{P}_{\pre,v})=\vvalue(P)\]
This can be easily achieved by reducing the flow value of $\hat{P}_{\pre,v}$ for some vertices $v$.

We collect all $\hat{P}_{\pre,v}$ among all possible $P,v$ and define
\[F_{\jump,\pre}=\sum_{P\in\path(F^G)}\sum_{v\in\hat{P}_{\jump,\pre}} \hat{P}_{\pre,v}\]

Notice that $F_{\jump,\pre}$ is $(3h_{i-1})$-length. According to the induction hypothesis \Cref{lem:forwardinduction}, there is a flow $F^H_{\jump,\pre}$ and a path mapping routing $\Dem(F_{\jump,\pre})$ in $G\cup H$ with the quality specified by \Cref{lem:forwardinduction}. We will analyze the quality later. We define $F^H_{\pre,v}$ as the mapped flow path in $F^H_{\jump,\pre}$ that routes the demand $\Dem(\hat{P}_{\pre,v})$.
    \item[Definition of flows routing vertices in $\hat{P}_{\HvLg,R}$ to $w_R$.] The definition is almost the same after changing $\pre$ to $\suf$. For completeness we provide the full definition again as follows.

Let $\hat{P}_{\jump,\suf}$ be the \emph{shortest} sufix of $\hat{P}_{\HvLg,R}$ (so it is also a sufix of $\hat{P}_{\jump}$) such that $C_{i-1}(\hat{P}_{\jump,\suf})\geq 1$. By the definition of heavy path, such a prefix always exists. Notice that $C_{i-1}(\hat{P}_{\jump,\suf})\le 2$ since $C_{i-1}$ assigns at most $1$ to each vertex. For every node $v\in V(\hat{P}_{\jump,\suf})$, we define the flow path $\hat{P}_{\suf,v}$ as the subpath $P[v,w_R]$ carrying flow value $\vvalue(\hat{P}_{\suf,v})=\vvalue(P)\cdot C_{i-1}(v)$. Notice that the length of $\hat{P}_{\suf,v}$ on the graph $G_{i-1}$ is at most 

\[\leng(\hat{P}_{\suf,v})\le h_{\cov,i}+2\cdot h_{\diam,i-1}s\le 3h_{i-1}\]

The first inequality is because $\hat{P}_{\suf,v}$ has length at most $h_{\cov,i}$ on $G$, and $C_{i-1}$ is a $(h_{\diam,i-1},s)$-length cut. 
We will also guarantee that the total flow value is exactly $\vvalue(P)$, i.e.,

\[\sum_{v\in\hat{P}_{\jump,\suf}}\vvalue(\hat{P}_{\suf,v})=\vvalue(P)\]
This can be easily achieved by reducing the flow value of $\hat{P}_{\suf,v}$ for some vertices $v$.
We collect all $\hat{P}_{\suf,v}$ among all possible $P,v$ and define

\[F_{\jump,\suf}=\sum_{P\in\path(F^G)}\sum_{v\in\hat{P}_{\jump,\suf}} \hat{P}_{\suf,v}\]

Notice that $F_{\jump,\suf}$ is $(3h_{i-1})$-length. According to the induction hypothesis \Cref{lem:forwardinduction}, there is a flow $F^H_{\jump,\suf}$ routing $\Dem(F_{\jump,\suf})$ in $G\cup H$ with the quality specified by \Cref{lem:forwardinduction}. We will analyze the quality later. We define $F^H_{\suf,v}$ as the mapped flow path in $F^H_{\jump,\suf}$ that routes the demand $\Dem(\hat{P}_{\suf,v})$.
\end{description}
\paragraph{} 

Now we are ready to define the flow $F^H_{\hat{P}_\jump}$, which is a routing for $\Dem(\hat{P}_{\jump})$ ($\vvalue(P)$ flows from $w_L$ to $w_R$). We will route $\Dem(\hat{P}_{\jump})$ by the following three steps.
\begin{description}
    \item[{Routing from $w_L$ to $\hat{P}_{\jump,\pre}$.}]
    For every $v\in \hat{P}_{\jump,\pre}$, we send flow through the flow path $F^H_{\pre,v}$. This guarantees that $\vvalue(P)$ units of flow in $w_L$ are distributed to vertices in $\hat{P}_{\jump,\pre}$ so that each vertex $v$ gets $\vvalue(F^H_{\pre,v})\le C_{i-1}(v)\cdot \vvalue(P)$ units of flow.
    
    \item[{Routing from $\hat{P}_{\jump,\pre}$ to $\hat{P}_{\jump,\suf}$.}] Denote the cluster in $\cS_i$ containing $\hat{P}_{\HvLg,L}$ by $S_L$, denote the cluster in $\cS_i$ containing $\hat{P}_{\HvLg,R}$ be $S_R$. 
    
    The routing first happens on the star shortcut $H_{S_L}$: for each vertex $v\in \hat{P}_{\jump,\pre}$, we route the $\vvalue(F^H_{\pre,v})$ units of flow from $v$ to $x_{S_L}$ through the vertex $x_{S_L,v}$.

    Then the routing happens through $\shortcut(x_{S_L},x_{S_R},k)$: Notice that now $x_{S_L}$ cumulates $\vvalue(P)$ unites of flow. We pick the smallest integer $k$ such that $(1+\tepsl)^k\ge \leng(\hat{P}_{\jump},G_i)$, and route the $\vvalue(P)$ units of flow from $x_{S_L}$ to $x_{S_R}$ through the vertex $\shortcut(x_{S_L},x_{S_R},k)$.

    At last the routing happens on the star shortcut $H_{S_R}$: for each vertex $v\in \hat{P}_{\jump,\suf}$, we route $\vvalue(F^H_{\suf,v})$ units of flow from $x_{S_R}$ to $v$ through the vertex $x_{S_R,v}$.

    \item[{Routing from $\hat{P}_{\jump,\suf}$ to $w_R$.}]
    For every $v\in \hat{P}_{\jump,\suf}$, we send flow through the flow path $F^H_{\suf,v}$. This guarantees that $w_R$ gets $\vvalue(P)$ units of flow and the routing finishes.

\end{description}

\paragraph{Concatenating the routing for flow subpaths.} We just defined $F^H_{\hat{P}}$ for each $\hat{P}\in\hat{\cP}$ as a routing of value $\vvalue(F^H_{\hat{P}})=\vvalue(P)$ from $\Start(\hat{P})$ to $\End(\hat{P})$. We construct the routing for $\Dem(P)$ by concatenating all $F^H_{\hat{P}}$: sending $\vvalue(P)$ of flow from $\End(\hat{P})$ to $\Start(\hat{P}')$ where $\hat{P}'$ is the next subpath following $\hat{P}$ in $\hat{\cP}$, using the original edge connecting $\End(\hat{P})$ to $\Start(\hat{P}')$. 

This finishes the description of routing the demand $\Dem(P)$, for which we denoted by $F^H_P$. The flow $F^H$ is the union of all the flows $F^H_P$ for every flow path $P\in F^G$ and the path mapping maps $P$ to $F^H_P$. 
\subsubsection{Quality of the Forward Mapping $F^H$}\label{subsubsec:qualityofforward}
It remains to analyze the quality of $F^H$.

\paragraph{Length of $F^H_P$.} We will prove that there is a path-mapping from $P$ to $F^H_P$ with additive length slackness $\tepsli{i}\cdot h_i$.

Consider a flow path $P\in\path(F^{G})$. We have that
\[\leng(F^H_{P})=\sum_{\hat{P}\in\hat{\cP}}\leng(F^H_{\hat{P}})
\]
by the concatenation definition of $F^H_P$. We will analyze the length $\leng(F^H_{\hat{P}},G\cup H)$ compared to $\leng(\hat{P},G_i)$ for different types of $\hat{P}$.

If $\hat{P}$ is a \underline{trivial subpath}, then $\leng(F^H_{\hat{P}},G\cup H)\le \leng(\hat{P},G_i)$ since $G_i$ only increases the length.

If $\hat{P}$ is a \underline{small-cluster subpath}, then $F^H_{\hat{P}}$ is a subflow of $F^H_{S_{\hat{P}}}$. According to \Cref{lem:inductionforsmallh}, we have that
\[
\leng(F^H_{\hat{P}},G\cup H)\le \leng(\hat{P},G_i)+\tepsl\cdot h_{\cov,i}\leq \tepsl \cdot \frac{\tepsl\cdot h_{i}}{s\beta}
\]

If $\hat{P}$ is a \underline{light subpath}, then on the graph $G_{i-1}$, $\hat{P}$ has length at most $\leng(\hat{P},G)+s\cdot h_{\diam,i-1}$ because $C_{i-1}$ is a $(h_{\diam,i-1},s)$-length moving cut which has value at most $1$ on $\hat{P}$. According to the induction hypothesis \Cref{lem:forwardinduction}, we have

\begin{align*}
    \leng(F^H_{\hat{P}},G\cup H)&\le \leng(\hat{P},G)+s\cdot h_{\diam,i-1}+\tepsli{i-1}\cdot h_{i-1}\\
    & \le \leng(\hat{P},G_i)+\left(\tepsl+\tepsli{i-1}\right)\cdot \frac{\tepsl\cdot h_i}{s\beta}
\end{align*}

If $\hat{P}$ is a \underline{jumping subpath}, then $F^H_{\hat{P}}$ is composed by three parts. For the first and third parts, i.e., from $w_L$ to $\hat{P_{\jump,\pre}}$ and from $\hat{P}_{\jump,\suf}$ to $w_R$, their lengths is by routing using \Cref{lem:forwardinduction}. On $G_{i-1}$, the length of the path $\hat{P}_{\jump,\pre}$ and $\hat{P}_{\jump,\suf}$ is at most $h_{\cov,i}+2sh_{\diam,i-1}$, because both $\hat{P}_{\jump,\pre}$ and $\hat{P}_{\jump,\suf}$ have length at most $h_{\cov,i}$ on $G$ and the value of the $(h_{\diam,i-1},s)$-length cut $C_{i-1}$ on $\hat{P}_{\jump,\pre}$ and $\hat{P}_{\jump,\suf}$ is at most $2$. According to \Cref{lem:forwardinduction}, the length of $F^H_{\hat{P}}$ of the first and the thrid path is (similar to the light subpath)

\[h_{\cov,i}+2sh_{\diam,i-1}+\tepsli{i-1}\cdot h_{i-1}\]

Now we analyze the second part of $F^H_{\hat{P}}$, i.e., from $\hat{P}_{\jump,\pre}$ to $\hat{P}_{\jump,\suf}$. This part consists of $5$ vertices. (i) The vertex on the star $x_{S,v}$ for some $v\in V(\hat{P}_{\jump,\pre})$ and $S\in\cS_i$ containing $\hat{P}_{\jump,\pre}$, this vertex has length $sh_{\diam,i}$. (ii) The vertex on the center of the star $x_S$ which has length $1$. (iii) The vertex $\shortcut(S_1,S_2,k)$ for some $S_1,S_2\in\cS_i$ containing $\hat{P}_{\jump,\pre},\hat{P}_{\jump,\suf}$, and $k$ is the smallest integer such that $(1+\tepsl)^k\ge \leng(\hat{P}_{\jump},G_i)$, which implies $(1+\tepsl)^k\le (1+\tepsl)\leng(\hat{P},G_i)$. Thus, the length of $\shortcut(S_1,S_2,k)$ is at most $(1+\tepsl)\leng(\hat{P},G_i)+2sh_{\diam,i}$. (iv) and (v), the last two vertices are symmetric to (i) and (ii). In total, this path has length 
\[2(sh_{\diam,i}+1)+(1+\tepsl)\leng(\hat{P})+2sh_{\diam,i}.\]

Thus, we have

\begin{align*}
    \leng(F^H_{\hat{P}})  &\le(h_{\cov,i}+2sh_{\diam,i-1}+\tepsli{i-1}\cdot h_{i-1}) +(2(sh_{\diam,i}+1)+(1+\tepsl)\leng(\hat{P},G_i)+2sh_{\diam,i})\\
    &\le \left(1+2\tepsl+\tepsli{i-1}\right)\cdot \frac{\tepsl\cdot h_i}{s\beta} + \left(4\tepsl h_{i}+2+\leng(\hat{P},G_i)+\tepsl\leng(\hat{P},G_i)\right)\\
    &\leq \leng(\hat{P},G_i)+\tepsl\leng(\hat{P},G_i) + \left(\frac{1+2\tepsl+\tepsli{i-1}}{s\beta} + 4\right)\cdot \tepsl \cdot h_{i} + 2.
\end{align*}

\underline{Summary.} Summing up all the values above, we get the length of $F^H_P$. Remember that $\leng(P)\le 3h_i$ and the the total number of subpaths is at most 
\[z\le 10\cdot \frac{3h_i}{h_{\cov,i}}=30\beta s/\tepsl\]
according to \Cref{lem:decompose}. There is exactly one jumping subpath and could be many small-cluster and light subpaths.

\begin{align*}
    \leng(F^H_P,G\cup H)-\leng(P,G_i)\le &\underbrace{\frac{30\beta s}{\tepsl}\cdot \tepsl \cdot \frac{\tepsl\cdot h_{i}}{s\beta}}_{\text{small-cluster subpath}}+\underbrace{\frac{30\beta s}{\tepsl}\left(\tepsl+\tepsli{i-1}\right)\cdot \frac{\tepsl\cdot h_i}{s\beta}}_{\text{light subpaths}}+\\
    &\underbrace{\tepsl\cdot \leng(P,G_i)+\left(\frac{1+2\tepsl+\tepsli{i-1}}{s\beta} + 4\right)\cdot \tepsl \cdot h_{i} + 2}_{\text{one jumping subpath}}\\
\end{align*}

By plugging in $\tepsl=\frac{\epsl}{100\cdot 50^{r}},\tepsli{i-1}=\frac{\epsl}{2\cdot 50^{r-i+1}}$ and $\leng(P,G_i)\le 3h_i$, we get 

\begin{align*}
\leng(F^H_P)-\leng(P)&\le \left(72\cdot\tepsl+31\cdot \tepsli{i-1}\right)\cdot h_i\\
&\le \frac{\epsl}{2\cdot 50^{r-i}}\cdot h_i\\
&\le \tepsli{i}\cdot h_i
\end{align*}

This finishes the length proof of \Cref{lem:forwardinduction}.

\paragraph{Congestion of $F^H$} Given that $F^{G}$ has congestion $1$, we will show that (1) $F^{H}$ has congestion $1$ in $G\cup H$, i.e., for every vertex $v\in G\cup H$, $F^H(v)\le U(v)$, and (2) the $V$-preserving-property, i.e.,  $F^H(v)\le F^G(v)$ if $v\in V$. To this end, we categorize the vertices in $G\cup H$ into different types.

\medskip

\noindent\underline{Original vertices in $G$.} For a vertex $v\in V$ where $V$ is the vertex set of $G$, we will show that $F^H(v)\le F^G(v)$. It suffices to show that for every $P\in\path(F^G)$, we have $F^H_P(v)\le F^G_P(v)$ since $F^H=\sum_{P\in\path(F^G)}F^H_P$ and $F^G=\sum_{P\in\path(F^G)}F^G_P$ (here we use $F^G_P$ to denote the flow path $P$ of $F^G$). Recall that we decompose $P$ into vertex disjoint subpaths $\hat{\cP}$ and $F^H_P$ is the concatenation of $F^H_{\hat{P}}$ for $\hat{P}\in\hat{\cP}$, we have $F^H_P(v)=\sum_{\hat{P}\in\hat{\cP}}F^H_{\hat{P}}(v)$. Thus, it suffices to show $F^H_{\hat{P}}(v)\le F^G_{\hat{P}}(v)$ where $F^G_{\hat{P}}$ is the flow subpath $\hat{P}$. 

If $\hat{P}$ is a trivial subpath, then $F^H_{\hat{P}}=\hat{P}$, so $F^H_{\hat{P}}(v)\le F^G_{\hat{P}}(v)$.

If $\hat{P}$ is a small-cluster subpath, then according to the induction hypothesis \Cref{lem:inductionforsmallh}, $F^H_{\hat{P}}(v)\le F^G_{\hat{P}}(v)$.

If $\hat{P}$ is a light subpath, then according to the induction hypothesis \Cref{lem:forwardinduction}, $F^H_{\hat{P}}(v)\le F^G_{\hat{P}}(v)$.

If $\hat{P}$ is a jumping subpath, according to the definition of jumping subpath, $F^H_{\hat{P}}(v)>0$ only if $v$ is on the subpath $\hat{P}$ and $F^H_{\hat{P}}(v)\le\vvalue(P)$. Thus, we have $F^H_{\hat{P}}(v)\le F^G_{\hat{P}}(v)$.

\medskip

\noindent\underline{Vertices (in $H_j$) added in the $j$-th iteration for $j<i$.} Let $v$ be a vertex not in $V$ but in $H_j$ for some $j<i$. 
One can verify that only routing for light subpaths can go through $v$. This is because the routings for small-cluster subpaths use shortcut $H_i$ (which is disjoint from $H_j$), and the routings for jumping subpaths use shortcuts constructed in Step 4 and 5, which is also in $H_i$. Thus, according to induction hypothesis \Cref{lem:forwardinduction}, we have $F^H_S(v)\le U(v)$. 

\medskip

\noindent\underline{Vertices (in $H_i$) added in Step 3 (recursions on small clusters).} Remember that in step 3, we add the LC-flow shortcut $H_S$ for some $S\in \cS_i$ to $H_i$ where $H_S$ might contain vertices outside $V$. Let $v$ be such a vertex. We need to show that $F^H(v)\le F^G(v)$. 

Recall that we define $F_S$ as the union of all small-cluster flow subpaths whose corresponding small-cluster is $S$. According to the definition of flow subpaths, we have $F_S(w)\le F^G(w)$ for every $w\in V$. Thus, $F_S$ has congestion $1$ on $G$. 
Remember that $F^H_S$ is defined as the routing for $\Dem(F_S)$ according to the induction hypothesis \Cref{lem:inductionforsmallh}. 
Moreover, one can verify that only flow in $F^H_S$ can go through $v\in H_S$: the routings for light subpaths are using induction hypothesis \Cref{lem:forwardinduction} which only uses the shortcut $H_j$ with $j<i$ but we have $H_S\subseteq H_i$; the routing for jumping subpath is only using the shortcut for large clusters.
Thus, we have $F^H(v)\le U(v)$ according to the induction hypothesis \Cref{lem:inductionforsmallh}. 

\medskip

\noindent\underline{Vertices (in $H_i$) added in Step 4 (stars on large clusters).} Let $v$ be a vertex $X_S$ for some $S\in\cS_i$, then $X_S$ has capacity $\sum_{v\in V}U_G(v)$, which is at least the value of $F^G$, so we must have $F^H(x_S)\le U(x_S)$. We only need to consider the vertex $x_{S,w}$ for some $w\in S$. 

One can verify that only the routing $F^H_{\hat{P}}$ for jumping subpaths $\hat{P}$ can go through $x_{S,w}$: the routing for small-cluster subpaths only use $H_S$ for small cluster $S$, and the routing for light subpaths only use $H_j$ for $j<i$. Let $\hat{P}$ be a jumping subpath, then $F^H_{\hat{P}}(x_{S,w})=\vvalue(P)\cdot C_{i-1}(v)$ according to the definition of $F^H_{\hat{P}}$. The total value of flow paths in $F^G$ going through $w$ is at most $U(w)$ since $F^G$ has congestion $1$. Thus, if we sum up all jumping subpaths, the value of flow going through $x_{S,w}$ is at most $U(w)\cdot C_{i-1}(v)$. Remember that we define $A_i(v)=C_{i-1}(v)\cdot U(v)$, we have that $F^H(x_{S,w})\le A_i(v)=U(x_{S,w})$, so the congestion on $x_{S,w}$ is at most $1$.

\medskip

\noindent\underline{Vertices (in $H_i$) added in Step 5 (edges connecting large clusters).} Let $v$ be a vertex $\shortcut(x_{S_1},x_{S_2},k)$ for some large clusters $S_1,S_2\in \cS_i$ and $k\in[\bar{k_i}]$. One can verify that only the routing $F^H_{\hat{P}}$ for jumping subpaths $\hat{P}$ can go through $\shortcut(x_{S_1},x_{S_2},k)$: the routing for small-cluster subpaths only use $H_S$ for small cluster $S$, and the routing for light subpaths only use $H_j$ for $j<i$. 

Suppose $F^H_{\hat{P}}$ for a jumping subpath $\hat{P}$ goes through $v=\shortcut(x_{S_1},x_{S_2},k)$, according to the definition of $F^H_{\hat{P}}$, we have $F^H_{\hat{P}}(v)=\vvalue(\hat{P})$ and $\leng(\hat{P})\le (1+\tepsl)^k$. We define a flow $F_{\hat{P}}(S_1,S_2)$ which routes $\vvalue(\hat{P})$ units of flow from $x_{S_1}$ to $x_{S_2}$ with length $\leng(\hat{P})$ in the following way
\begin{enumerate}
    \item for every $w\in P_{\jump,\pre}$, $x_{S_1}$ sends at most $C_{i-1}(w)\cdot \vvalue(\hat{P})$ units of flow to $w$,
    \item along with path $\hat{P}$, the flows on $P_{\jump,\pre}$ goes to $P_{\jump,\suf}$ such that every node $w\in P_{\jump,\suf}$ gets at most $C_{i-1}(w)\cdot \vvalue(\hat{P})$ units of flow,
    \item for every $w\in P_{\jump,\suf}$, $w$ sends at most $C_{i-1}(w)\cdot \vvalue(\hat{P})$ units of flow to $x_{S_2}$.
\end{enumerate}

Let $F(S_1,S_2)$ be the summation of flows $F_{\hat{P}}(S_1,S_2)$ for all jumping subpaths $\hat{P}$ with $\leng(\hat{P})\le (1+\tepsl)^k$ such that $F^H_{\hat{P}}$ goes through $\shortcut(x_{S_1},x_{S_2},k)$. It follows from the definition that $\leng(F(S_1,S_2),G_i)\le 2h_{\diam,i}+(1+\tepsl)^k$. Moreover, every flow $F_{\hat{P}}(S_1,S_2)$ corresponds to a flow path $P\in\path(F^G)$, so $F(S_1,S_2)$ cause congestion $1$ on $G$. For star vertice, $F(S_1,S_2)$ causes congestion $1$ on them by the same argument as our analysis for vertices added in step 4. Thus, $F(S_1,S_2)$ is a single-commodity flow with congestion $1$ from $x_{S_1}$ to $x_{S_2}$ in $G_i\cup H_{S_1}\cup H_{S_2}$. By the definition of $\shortcut(x_{S_1},x_{S_2},k)$, the congestion caused by $F^H$ on it is at most $1$.

\paragraph{Step of $F^H$} By the definition of path decomposition \Cref{lem:decompose}, for $P\in\path(F^G)$, $\step(F^H_P)$ is the summation of $\step(F^H_{\hat{P}})$ for all $\hat{P}\in\hat{\cP}$. We will analyze $\step(F^H_{\hat{P}})$ based on different types of $\hat{P}$.

If $\hat{P}$ is a \underline{trivial subpath}, then $\step(F^H_{\hat{P}})=1$ since a trivial subpath contains exactly $1$ vertex.

If $\hat{P}$ is a \underline{small-cluster subpath}, remember that $\step(F^H_{\hat{P}})$ is a subflow of $F_S$, where $F_S$ is from a recursive call to \Cref{lem:smallhemulator}. Thus, by the induction hypothesis \Cref{lem:inductionforsmallh}, we have
\[\step(F^H_{\hat{P}})\le t(n'/\sigma,\tepsl)\]

where we treat $t$ as a function of $n$ and $\epsl$ according to the parameter specified in \Cref{lem:smallhemulator}. We do not write $\epsk$ as a parameter to $t$ because $\epsk$ never changes during the recursive calls.

If $\hat{P}$ is a \underline{light subpath}, remember that $F^H_{\hP}$ is a subpath of $F_{\light}$, which according to the induction hypothesis \Cref{lem:forwardinduction} has step

\[\step(F^H_{\hP})\le t^{(i-1)}\]

If $\hat{P}$ is a \underline{Jumping subpath}, according to the construction, the step is two times the step $t^{(i-1)}$ plus a constant length path (with $5$ vertices), thus,

\[\step(F^H_{\hP})\le 5+2t^{(i-1)}\]

\underline{Summary.} According to \Cref{lem:decompose}, the total number of subpath of a flowpath in $\path(F^H)$ is $z\le \frac{30\beta s}{\tepsl}$, thus,

\begin{align*}
    \step(F^H)&\le z\cdot \left(1+t(n'/\sigma,\tepsl)+t^{(i-1)}\right)+5+2t^{(i-1)}\\
    &\le \frac{31\beta s}{\tepsl}\cdot \left(t(n'/\sigma,\tepsl)+t^{(i-1)}\right)\\
    &\le \frac{31\beta s}{\tepsl}\cdot t(n'/\sigma,\tepsl)\cdot \left(1+\left(\frac{40\beta s}{\tepsl}\right)^{i}\right)\\
    &\le t(n'/\sigma,\tepsl)\cdot \left(\frac{40\beta s}{\tepsl}\right)^{i+1}\\
    &=t^{(i)}
\end{align*}

\subsection{Proof of correctness: Backward Mapping}\label{subsec:backwardmappingcorrectness}

Let $F^H$ be an arbitrary flow $F^{H}$ in $G\cup H$ with congestion $1$. We will prove the existence of a flow $F^G$ in $G$ with congestion $\kappa$ and a path-mapping from $F^H$ to $F^G$ with length slack $1$. In the next section, we will show the correctness of the backward mapping algorithm, which is almost identical to the proof of backward mapping correctness.

\paragraph{Eliminating $\shortcut(x_{S_1},x_{S_2},k)$.} We first construct a flow $F^{H'}$ in $G\cup H$ routing the demand $\Dem(F^H)$ such that every flow path of $F^{H'}$ does not contain Steiner vertices constructed in Step 5, i.e., $\shortcut(x_{S_1},x_{S_2},k)$.

Let $P\in\path(F^H)$ and let $\hat{P}$ be a flow subpath of $P$ such that

\[\hat{P}=(x_{S_1},\shortcut(x_{S_1},x_{S_2},k),x_{S_2})\]

for some large clusters $S_1,S_2$ and some $k$. We will construct $F^{H'}_{\hP}$ routing the demand $\Dem(\hP)$, then replace $\hP$ by $F^{H'}_{\hP}$ for every such $\hP$. This defines the flow $F^{H'}$. Now we define the routing $F^{H'}_{\hP}$.

We define the largest $((1+\tepsl)^k+2h_{\diam,i}\cdot s)$-length single commodity flow from $x_{S_1}$ to $x_{S_2}$ in the graph $G\cup H_{S_1}\cup H_{S_2}$ as $F^{H'}_{S_1,S_2,k}$. We define $F^{H'}_{\hP}$ as scaling the flow $F^{H'}_{S_1,S_2,k}$ to have flow value exactly $\vvalue(\hP)$, i.e.,
\[F^{H'}_{\hP}=\frac{\vvalue(\hP)}{\vvalue(F^{H'}_{S_1,S_2,k})}\cdot F^{H'}_{S_1,S_2,k}\]

It is easy to see that $\vvalue(F^{H'}_{\hP})=\vvalue(\hP)$ and $\Dem(F^{H'}_{\hP})=\Dem(\hP)$. 

After defining $F^{H'}$, we now continue to define the flow $F^G$ that routes the demand $\Dem(F^{H'})$. For each flow path $P\in\path(F^{H'})$, we will define a flow $F^{G}_{P}$ routing the demand $\Dem(P)$. The final $F^G$ will be the union of $F^G_P$ for $P\in\path(F^{H'})$.

\paragraph{Decompose the Flow Path $P$.} We decompose $P$ into subpaths $\hat{\cal P}=\{\hat{P}_1,...,\hat{P}_z\}$ such that
\begin{itemize}
    \item $\Start(\hat{P}_1)=\Start(P),\End(\hat{P}_z)=\End(P)$ and for every $i\in[z-1]$, we have $\End(\hat{P}_i)=\Start(\hat{P}_{i+1})$,
    \item for every $\hat{P}\in \hat{\cP}$, both $\Start(\hat{P}),\End(\hat{P})$ are in the original vertex set $V(G)$, all inner vertices $V(\hat{P})-\{\Start(\hat{P}),\End(\hat{P})\}$ are not in the original vertex set $V(G)$, i.e., they only contain Steiner nodes $V(G\cup H)-V(G)$.
\end{itemize}

Such a decomposition trivially exists. %
Now we categorize subpaths based on their inner vertices.

If $\cP\in\hat{\cP}$ only contains original vertices $V(G)$, we call $\cP$ a \emph{trivial subpath}.

For a small-cluster $S\in\cS_{\ssmall,i}$, notice that the way we construct $H_S$ guarantees that Steiner nodes $V(H_S)-V(G)$ only have edges to $V(H_S\cup G)$. Thus, if $\hP\in\hat{\cP}$ has inner vertices in $V(H_S)-V(G)$, then all inner vertices of $\hP$ must be in $V(H_S)-V(G)$. We call such path \emph{small-cluster subpath}, and use $S_{\hP}$ to denote the corresponding small cluster.

If $\cP$ contains Steiner vertices as inner vertices, but they are not in $V(H_S)-V(G)$ for some small cluster $S\in\cS_{\ssmall,i}$. Remember the algorithm: step 5 constructs edges not connecting to the original vertices, step 3 constructs edges for small clusters, step 4 construct star shortcut connecting a large cluster. It must be the case that the second vertex of $\cP$ is $x_{S,v}$ for some large cluster $S\in\cS_{\big,i}$. Moreover, $\cP$ does not contain vertices constructed in Step 5. Thus, all inner vertices of $\cP$ are in $V(H_S)-V(G)$. According to step 4, in this case, there are only three inner vertices of $\cP$: $x_{S,v},x_S,x_{S,w}$ for some $v,w\in S$. We call $\cP$ a \emph{large-cluster subpath}, and use $S_{\hP}$ to denote the corresponding large cluster.

\paragraph{Route the Flow Subpaths.} To construct the flow $F^{G}_{P}$, we will first construct, for each subpath $\hat{P}\in\hat{\cal P}$, a flow $F^{G}_{\hat{P}}$ of $\Dem(\hat{P})$ in $G$, and finally let $F^{G}_{P}$ be the concatenation of all these flows $F^{G}_{\hat{P}}$. Consider a subpath $\hat{P}\in{\cal P}$.

\medskip

\noindent{\underline{Trivial subpaths.}} If $\cP$ is a trivial subpath, we set $F^G_{\hat{P}}$ to be just $\cP$.

\medskip

\noindent{\underline{Small-cluster subpaths.}} If $\cP$ is a small-cluster subpath with the corresponding small cluster $S$, we define the flow $F_S$ collecting all the small-cluster subpath as follows

\[F^{H'}_S=\sum_{P\in\path(F^{H'})}\sum_{\substack{\hP\in\hat{\cP}\\ S_{\hat{P}}=S}}\hat{P}\]

According to the induction hypothesis \Cref{lem:inductionforsmallh}, there is a flow $F^G_S$ on the original graph $G$ routing the demand $\Dem(F^{H'}_S)$ with the parameter specified in \Cref{lem:smallhemulator,def:LCFlowShortcut}. We define $F^G_{\hP}$ as the subflow in $F^G_S$ that routes the demand $\Dem(\hat{P})$. 

\medskip

\noindent{\underline{Large-cluster subpaths.}} If $\cP$ is a large-cluster suchpath with the corresponding large cluster $S\in\cS_i$ for some $i\in[r]$, we define the demand $D_S$ collecting all the demand $\Dem(\cP)$ for all large-cluster suchpath as follows. 

\[D_S=\sum_{P\in\path(F^{H'})}\sum_{\substack{\hP\in\hat{\cP}\\ S_{\hat{P}}=S}}\Dem(\hP)\]

The following lemma shows that $D_S$ is not too large compared to node-weighting $A_i$ so we can hope to use \Cref{thm:flow character}.
\begin{lemma}\label{lem:demandcongestion}
    Suppose $S\in\cS_i$ for some $i\in[r]$. Then $D_S$ is $A'_i$-respecting for
    \[A'_i=\cong(F^{H'})\cdot A_i\]
\end{lemma}
\begin{proof}
    We need to show that for every $v\in S$, we have $D_S(v,\cdot)\le \cong(F^{H'})\cdot A_i(v)$. The other direction $D_S(\cdot,v)\le \cdot A_i(v)$ is symmetric.

    Notice that for a large-cluster subpaths $\cP$ such that $S_{\cP}=S$, it can contribute $\vvalue(\cP)$ to $D_S(v,\cdot)$ only if $\cP=(v, x_{S,v},x_S,...)$. In other words, $\cP$ puts $\vvalue(\cP)$ amount of flow on the vertex $x_{S,v}$. Notice that the total number of flow on $x_{S,v}$ cannot exceed $U(x_{S,v})\cdot \cong(F^{H'})$, we have
    
    \[D_S(v,\cdot )\le U(x_{S,v})\cdot \cong(F^{H'})=A_i(v)\cdot \cong(F^{H'})\]

    This finishes the proof.
\end{proof}

Then we can use \Cref{thm:flow character} to prove the following lemma.

\begin{lemma}\label{lem:routabledemand}
    The demand $D_S$ can be routed by a flow $F^G_S$ of congestion $O(\cong(F^{H'})\cdot \frac{\log n}{\phi})$ and length $h_{\diam,i}\cdot s$.
\end{lemma}
\begin{proof}
    Given \Cref{lem:demandcongestion}, we have that $\frac{1}{\cong(F^{H'})}\cdot D_S$ is $A_i$-respecting. According to Step 1 (LC-expander decomposition), we have that $G_i$ is a $(h_{\diam,i},s)$-length $\phi$-expander. Notice that $D_S$ is a $h_{\diam,i}$-length demand since the diameter of $S\in\cS_i$ is guaranteed to be bounded by $h_{\diam,i}$. According to \Cref{thm:flow character}, $\frac{1}{\cong(F^{H'})}\cdot D_S$ can be routed in $G_i$ with congestion at most $O(\frac{\log n}{\phi})$ and length at most $h_{\diam,i}\cdot s$. Notice that vertex length in $G_i$ is at least the vertex length in $G$. The lemma follows by scaling the flow up by a factor of $\cong(F^{H'})$.
\end{proof}

We define the flow $F^G_{\hat{P}}$ as the subflow of $F^G_S$ routing the demand $\Dem(\hat{P})$.

\paragraph{Length of $F^G$.} We first analyze the length of $F^{H'}$. Notice that $F^{H'}_{S_1,S_2,k}$ has length at most $((1+\tepsl)^k+2h_{\diam,i}\cdot s)=\ell(\shortcut(x_{S_1},x_{S_2},k))$. Thus, by substituting the subpath $\hat{P}=(x_{S_1},\shortcut(x_{S_1},x_{S_2},k),x_{S_2})$ with $F^{H'}_{S_1,S_2,k}$, then length does not increase. We have $\leng(F^{H'})\le \leng(F^H)$. This length upper bound is path-wise, so the path mapping holds.

Compared to $F^{H'}$, the routing in $F^{G}$ for trivial subpaths does not increase the length, for small-cluster subpaths does not increase the length because of the induction hypothesis \Cref{lem:inductionforsmallh}. The only non-trivial part is for large-cluster subpaths.

According to \Cref{lem:routabledemand}, for $S\in\cS_i$, $\leng(F^G_S)\le h_{\diam,i}\cdot s$. Let $\hat{P}$ be a large-cluster subpath with $S_{\hP}=S$, according to the definition, we must have 
\[\leng(\hP)\ge h_{\diam,i}\cdot s\ge \leng(F^G_S)\]

The first inequality is because $\hP$ must contain a vertex $x_{S,v}$ for some $v\in S$ with length $h_{\diam,i}\cdot s$. Thus, substituting $\hP$ by a subflow of $F^G_S$ does not increase the length.

\paragraph{Congestion of $F^G$.} The rerouting for trivial subpaths does not change the congestion. 

\medskip

\noindent\underline{Small-cluster subpath.} We first need the following lemma certifying $F^{H'}$ does not increase the congestion of small-cluster subpath compared to $F^H$.
\begin{lemma}\label{lem:nonincreasesmallcluster}
    For every small cluster $S$, the collection of all small-cluster subpath $\hat{P}$ with $S_{\hP}=S$, denoted by $F^{H'}_S$, satisfies $\cong(F^{H'}_S)=1$.
\end{lemma}
\begin{proof}
    For every small-cluster subpath $\hat{P}$, it corresponds to a subpath of $F^H$ because $F^{H'}$ only creates the flow $F^{H'}_{S_1,S_2,k}$ which does not contain any Steiner vertices of $H_S$ for a small cluster $S$. Thus, according to $\cong(F^H)=1$, we have $\cong(F^{H'}_S)=1$.
\end{proof}
Let $\cN_i$ be the neighborhood cover computed in the $i$-th iteration and let $\cS\in\cN_i$ be a vertex disjoint vertex sets. Let $S\in\cS_i$ be a small cluster, according to the induction hypothesis \Cref{lem:inductionforsmallh}, since $F^{H'}_S$ has congestions $1$, the congestion of $F^G_S$ is
\[\cong(F^{G}_{S})\leq \kappa(n'/\sigma)\]

Notice that $F^G_S$ only uses vertices in $S$. Vertex sets in $\cS_i$ are disjoint. Thus, we can collect all small clusters and get
\[\cong\left(F^G_{\ssmall}:=\sum_{i\in[r]}\sum_{\cS\in\cN_i}\sum_{\substack{S\in\cS\\S\text{ is a small cluster}}}F^G_S\right)\le r\omega\cdot \kappa(n'/\sigma)\]
\medskip

\medskip

\noindent\underline{Large-cluster subpath.} We first give the congestion bound of $F^{H'}$.
\begin{lemma}\label{lem:Hprimecongestion}
    $\cong(F^{H'})\le 1+r\cdot \left(\omega\sigma\right)^2\cdot \bar{k}_r$. 
\end{lemma}
\begin{proof}

Let $\shortcut(x_{S_1},x_{S_2},k)$ be a Steiner vertex constructed in Step 5, let $\cP_{S_1,S_2,k}$ contain all subpath $\hat{P}$ of some path $P\in\path(F^H)$ such that $\hP=(x_{S_1},\shortcut(x_{S_1},x_{S_2},k),x_{S_2})$, then we have
    \begin{align*}
        \sum_{\hP\in\cP_{S_1,S_2,k}}F^{H'}_{\hP}&=\sum_{\hP\in\cP_{S_1,S_2,k}}\frac{\vvalue(\hP)}{\vvalue(F^{H'}_{S_1,S_2,k})}\cdot F^{H'}_{S_1,S_2,k}\\
        &\le \frac{U(\shortcut(x_{S_1},x_{S_2},k))}{\vvalue(F^{H'}_{S_1,S_2,k})}\cdot F^{H'}_{S_1,S_2,k}\\
        &\le 2\cdot F^{H'}_{S_1,S_2,k}
    \end{align*}
The first inequality is because $F^H$ has congestion $1$. The second inequality is because of the definition of $\shortcut(x_{S_1},x_{S_2},k)$ and $F^{H'}_{S_1,S_2,k}$. 

Notice that $F^{H'}_{S_1,S_2,k}$ has congestion $1$, so $\sum_{\hP\in\cP_{S_1,S_2,k}}F^{H'}_{\hP}$ has congestion $2$. We have

\[\cong(F^{H'})\le 1+\sum_{i\in[r]}\sum_{S_1,S_2\in\cS_{\big,i}}\sum_{k\in[\bar{k}_i]}\cong(\sum_{\hP\in\cP_{S_1,S_2,k}}F^{H'}_{\hP})\]

The first term $1$ comes from the original congestion from $F^H$. The second term is summing over all possible $S_1,S_2,k$. The size of $S_{\big,i}$ is provided in \Cref{lem:bigcluster}.

Finally, we get  (remember that $\bar{k_i}$ is increasing on $i$)

\[\cong(F^{H'})\le 1+r\cdot \left(\omega \sigma\right)^2\cdot \bar{k}_r\]

Notice that the flow $F^{H'}_{S_1,S_2,k}$ does not contain any Steiner vertices of $H_S$ for a small cluster $S$. Thus, the lemma follows.
\end{proof}

By combining \Cref{lem:bigcluster,lem:routabledemand,lem:Hprimecongestion}, we can collect all the routings of large clusters and get

\[\cong\left(F^G_{\big}:=\sum_{i\in[r]}\sum_{S\in\cS_{\big,i}}F^G_S\right)\le r\cdot \omega\sigma\cdot \left(1+r\cdot (\omega\sigma)^2\cdot\bar{k_r}\right)\cdot \frac{\log n}{\phi}\]

\medskip

\noindent\underline{Summary.} By summing up the congestion of $F^G_{\ssmall}$ and $F^G_{\big}$ and plugging in the corresponding parameters, we get
\begin{align*}
    \cong(F^G)&\le F^G_{\ssmall}+F^G_{\big}\\
    &\le r\omega\cdot\kappa(n'/\sigma)+r\cdot \omega\sigma\cdot \left(1+r\cdot (\omega\sigma)^2\cdot\bar{k_r}\right)\cdot \frac{\log n}{\phi}
\end{align*}

Remember that $\omega\leq n^{\epsk^{2}/10}$, $\sigma = n^{\epsk/6}$, $\phi = 1/n^{\epsk/9}$, $\Omega(\log\log\log n)\leq \epsk\leq 1$ and
\[
\bar{k_r}=\log_{1+\tepsl}(3h)=\tO{1/\tepsl} = \tilde{O}(1)
\]
because $1/\tepsl = \tilde{O}(1)$ for \emph{all} recursive calls in the recursion tree. Concretely, it is not hard to see $1/\tepsl \leq 50^{O(r\cdot \lceil \log_{\sigma}n\rceil)}/\epsl= 50^{O(\epsk^{2})}/\epsl = \tilde{O}(1)$, where we slightly abuse the notation by letting $\epsl$ denote the $\epsl$ of the \emph{root} instance, which satisfies $1/\epsl = \poly\log n$, the condition required in \Cref{lem:smallhemulator}.

Therefore, we have
\[\cong(F^{G})\leq (\omega r)\cdot \kappa(n'/\sigma) + n^{\epsk}\leq \kappa(n')\]
as required by \Cref{lem:inductionforsmallh} (recall that $\kappa(n') = (\omega r + 1)^{\lceil\log_{\sigma}n'\rceil}\cdot n^{\epsk}$).

\subsection{Backward Mapping Algorithm}\label{subsec:backwardmappingalgorithm}

We now describe a backward mapping algorithm for \Cref{lem:smallhemulator}
\[F^G\leftarrow \textsc{BackwardMapping}(G,H,F^H).\]

This algorithm is given $G$, the shortcut $H$ constructed by the procedure $\Emulator(G,\epsl,\epsk,h)$, and the edge representation of a $k$-commodity flow $F^{H}$ on $G\cup H$. From the edge representation of $F^{H}$, we can obtain the $k$-commodity demand $D$ routed by $F^{H}$, and the edge representation of each $q$-th commodity flow $F^{H}_{q}$ of $F^{H}$. We intend to output a $k$-commodity flow $F^{G}$ in $G$ routing $D$, with $\cong(F^{G})\leq n^{\epsk} \cdot \cong(F^{H})$ and $\totlen(F^{G})\leq \totlen(F^{H})$.

    \paragraph{Eliminating $\shortcut(x_{S_1},x_{S_2},j)$.} We will construct an edge representation of a $k$-commodity flow $F^{H'}$ in $G\cup H$ (in fact, $F^{H'}$ is obtained by \emph{updating} $F^{H}$), routing the same $k$-commodity demands as $F^H$ without any flow passing through the shortcut vertices constructed in Step 5. 
    
    For every vertex $\shortcut(x_{S_1},x_{S_2},j)$, suppose there are $\gamma$ units of $k$-commodity flows passing through this vertex in the direction of $x_{S_1}$ to $x_{S_2}$ (we do the same for the reverse direction $x_{S_{2}}$ to $x_{S_{1}}$). We first use \Cref{lemma:LCstFlows} to compute a $0.5$-approximate $((1+\tepsl)^j+2h_{\diam,i}\cdot s)$-length single-commodity max flow $F_{\sg}$ from $x_{S_1}$ to $x_{S_2}$ in the graph $G\cup H_{S_1}\cup H_{S_2}$. Note that $F_{\sg}$ is given in its path representation with $\tilde{O}(m\cdot\poly(h))$ flow paths. Then we scale $F_{\sg}$ to make sure that $\vvalue(F_{\sg}) = \gamma$, i.e., the $k$-commodity flow value passing through $\shortcut(x_{S_1},x_{S_2},j)$.
    
    Finally, we use $F_{\sg}$ to allocate the flow values of every commodity from $x_{S_1}$ to $x_{S_2}$ passing through $\shortcut(x_{S_1},x_{S_2},j)$, as follows. For each commodity $1\leq q\leq k$, suppose $\gamma_{q}$ units of flow are passing through $\shortcut(x_{S_1},x_{S_2},j)$ (note that $\sum_{1\leq q\leq k}\gamma_{q} = \gamma$). The allocation can be done by a straightforward \emph{sequential} algorithm, which just scans flow paths of $F_{\sg}$ one by one to allocate the $\gamma_{q}$ units of the $q$-th commodity flow in increasing order of $q$ (namely, once $\gamma_{q}$ units of the $q$-th commodity are fully allocated, we proceed to $\gamma_{q+1}$). Furthermore, this sequential algorithm can be easily simulated by a parallel algorithm, which takes $\tilde{O}((|\path(F_{\sg})|+k)h)$ work and $\tilde{O}(1)$ depth. The extra $h$ factor in work because we go through each $O(h)$-length flow path in $\path(F_{\sg})$ to update $F^{H}$.

    \paragraph{Recursive routing.} For every small cluster $S$, we run

    \[F^G_S\leftarrow \textsc{BackwardMapping}(G[S],H_S,F^H[H_S])\]

    where $F^H[H_S]$ is a $k$-commodity flow only preserving the flow of $F^{H}$ on edges in $H_S$. We replace the flow $F^H[H_S]$ by $F^G_S$ for every small cluster $S$ to get $F^{H}_{\big}$. Notice that $F^H_{\big}$ routes the same demand as $F^H$.

    \paragraph{Expander routing.} It remains to map the $F^{H}_{\big}$-flows on star edges back to $G$, and we will exploit expander routing. For every large cluster $S$ and every commodity of $F^H_{\big}$ defined as $F^H_{\big,q}$ (for the $q$-th commodity), let $F^{S}_{\big,q}$ be the restriction of $F^{H}_{\big,q}$ on the star graph $H_{S}$. Let $D_{\sg,S,q}$ be a single-commodity demand on vertices $S$, where for each $v\in S$, $D_{\sg,S,q}(v)$ is the net $F^{S}_{\big,q}$-flow value going out $v$. We take a multi-commodity demand $D_{S,q}$ corresponding to $D_{\sg,S,q}$ (recall the definition of \emph{corresponding} in \Cref{sect:PrelimFlow}), such that $\supp(D_{S,q}) = O(|S|)$ (note that such $D_{S,q}$ always exists). Note that, when we get a (multi-commodity) flow routing $D_{S,q}$ later, it is a single-commodity flow routing $D_{\sg,S,q}$.

    Let $D_S$ be the union of all $D_{S,q}$ for $q\in[k]$. We have $\supp(D_S)=O(|S|\cdot k)$.    
    For $i\in [r]$ and $\cS\in\cN_i$, define 
    \[D_{i,\cS,\big}=\sum_{S\in\cS\cap\cS_{\big,i}} D_S\]
    
    Notice that $D_{i,\cS,\big}$ is $h_{\diam,i}$-length because each large cluster has a diameter $h_{\diam,i}$ according to the definition of neighborhood cover. Moreover, $D_{i,\cS,\big}$ is $A_i$ respecting because clusters in $\cS$ are vertex disjoint, so every vertex $v$ has at most $1$ adjacent edges in $\cup_{S\in\cS}H_S$, which is $(v,x_{S,v})$.
    $F^H$ has congestion $1$, so that $F^H_{\big}$ put at most $A_i(u)$ flow on the vertex $x_{S,u}$, this means $D_{i,\cS,\big}$ is $A_i$-respecting.
    
    According to \Cref{thm:vertexLC-ED}, there is an algorithm that outputs a path representation (which can be easily transformed into an edge representation) of a flow routing $D_{i,\cS,\big}$ with the length and congestion specified in \Cref{thm:vertexLC-ED}, let the flow be $F_{i,\cS,\big}$. We replace flow $F^H_{\big}$ on $\cup_{S\in\cS\cap\cS_{i,\big}}H_S$ (i.e. those $F^{S}_{\big}$) by $F_{i,\cS,\big}$ for every $\cS\in\cN_{i}$ and $i\in[r]$ (taking the different commodities into consideration). The resulting flow $F_G$ is on the graph $G$ and will be returned.

\paragraph{Correctness.} Now, we proceed to the correctness. The correctness is based on induction on the size of $G$. The base case is trivial since we do not need to reroute. Now suppose the algorithm is correct for every graph with size less than $G$.

\paragraph{Correctness: congestion.} 
Without loss of generality, assume $\cong(F^{H}) = 1$.

Notice that the way we eliminate $\shortcut(x_{S_1},x_{S_2},k)$ is the same as in \Cref{subsec:backwardmappingcorrectness} with at most a constant factor of loss in congestion. Thus, according to the same proof as in \Cref{lem:Hprimecongestion}, we have 
\[\cong(F^{H'})\le O(r\cdot (\omega\sigma)^2\cdot \bar{k}_r)\]

For the expander routing, according to \Cref{thm:vertexLC-ED}, the congestion applied to $G$ in this step is  at most 
\[r\cdot \omega \cdot n^{\poly(\epsk)}\cdot n^{O(1/\epsk)}=O(n^{\poly(\epsk)})\]

Finally, analogous to the proof of \Cref{lem:inductionforsmallh}, the whole recursion introduces an overhead of $(\omega r+1)^{\lceil \log_{\sigma}n\rceil} = n^{O(\epsk)}$ to the congestion. Therefore, the congestion of $F^{G}$ in $G$ is
\[
\cong(F^{G}) = n^{O(\epsk)} \cdot O(r\cdot (\omega\sigma)^2\cdot \bar{k}_r)\cdot O(n^{\poly(\epsk)})\le n^{\poly(\epsk)}\leq n^{O(\epsk)}.
\]

\paragraph{Correctness: total length.}

Firstly, we show that $F^{H'}$ does not increase the total length of $F^H$. Notice that the way we eliminate $\shortcut(x_{S_1},x_{S_2},k)$ is by replacing a flow containing three vertices $(x_{S_1},\shortcut(x_{S_1},x_{S_2},k),x_{S_2}),x_{S_2})$ by a flow path with length at most $\ell(\shortcut(x_{S_1},x_{S_2},k))$ and the same value $q$. The total length is at most $q\cdot \ell(\shortcut(x_{S_1},x_{S_2},k))$, which is at most the total length of the original flow.

Then, consider the recursive routing step, according to the induction hypothesis, the total length does not increase by rerouting the flow in $H_S$ for small clusters in $H$ to $G$.

Lastly, we show that the expander routing does not increase the total length. Consider an expander routing to the demand $D_{i,\cS,\big}$ which corresponds to an original flow of value $q$ such that each unit of flow must go through $x_{S,v}\in V(H_S)$ for some large cluster $S\in\cS_{i,\big}$. Remember that the length of $x_{S,v}$ is $h_{\diam,i}\cdot s$, thus, the total length is at least $q\cdot h_{\diam,i}\cdot s$.
According to \Cref{thm:vertexLC-ED}, the expander routing results in a flow with length bound $h_{\diam,i}\cdot s$ with the same flow value. Thus, the total length of $F^G$ is $q\cdot h_{\diam,i},\cdot s$. This proves that the total length does not get increased by the expander routing.

\paragraph{Complexity.}
Consider the step for eliminating $\shortcut(x_{S_1},x_{S_2},k)$. For every pair of large clusters in $\cS_{i,\big}$ and $k\in [\bar{k}_i]$ for $i\in [r]$, we have one call to \Cref{lemma:LCstFlows} and a parallel subroutine to allocate the $k$-commodity flow values to the single-commodity flow given by \Cref{lemma:LCstFlows}, both of which take $\tilde{O}(\poly(h)\cdot m)$ work and $\tilde{O}(\poly(h))$ depth.  \Cref{lem:bigcluster} shows that the number of large clusters is $n^{O(\epsk)}$, which dominates $\bar{k}_{i}$ and $r$. Hence, the number of calls to \Cref{lemma:LCstFlows} is $n^{O(\epsk)}$, which takes totally $n^{O(\epsk)}\cdot \poly(h)\cdot m$ work and $\tilde{O}(\poly(h))$ depth. 

The step for recursive routing involves recursive calls, which we will analyze later.

The step for expander routing involves one call to \Cref{thm:vertexLC-ED} expander routing algorithm per $\cS\in\cN_i$ and $i\in [r]$. Each of time involves routing a demand with support size $O(nk)$, so the total work is $(kn+m)\cdot \poly(h)\cdot n^{\poly(\epsk)}\cdot \omega r=(kn+m)\cdot \poly(h)\cdot n^{\poly(\epsk)}$ and the total depth is $\poly(h)\cdot n^{\poly(\epsk)}$. 

The recursion, as analyzed in the complexity of $\Emulator$, in total apply a factor of $n^{O(\epsk)}$ to the total number of edges and vertices. The parameter $k$ does not change in every recursion since the number of flow commodities never increases. Thus, the total work is $(kn+m)\cdot \poly(h)\cdot n^{\poly(\epsk)}$ and total depth is $\poly(h)\cdot n^{\poly(\epsk)}$.

\subsection{Bootstrapping: Proof of \Cref{thm:emulator}}\label{subsec:bootstrapping}

In \Cref{lem:LCompSmallhEmulator}, we first turn the additive length slack in \Cref{lem:smallhemulator} into a multiplicative length slack.

\begin{lemma}\label{lem:LCompSmallhEmulator}
Let $G$ be an undirected graph with vertex lengths and capacities. Let $1/\polylog(n)<\epsl<1$, $1/\Omega(\log\log\log n)<\epsk<1$, and $h\geq 1$ be given parameters. There is an algorithm $\Emulator(G,\epsl,\epsk,h)$ computes an $h$-length $t$-step LC-flow shortcut of $G$ of size $n^{1+O(\epsk)}$ with length slack $(1+2\epsl)$ and congestion slack $n^{\epsk}\cdot O(\log n)$ where
\[t=\left(1/\epsl\right)^{O\left(1/\epsk^2\right)}\cdot \exp(\exp(\log (1/\epsk)\cdot O(1/\epsk)))\]
in $\poly(h)\cdot n^{\poly(\epsk)}$ depth and $\poly(h)\cdot m\cdot n^{\poly(\epsk)}$ work. 

Moreover, the shortcut is associated with a backward mapping algorithm with $\poly(h)\cdot n^{\poly(\epsk)}$ depth and $\poly(h)\cdot (m+k)\cdot n^{\poly(\epsk)}$ work.
\end{lemma}
\begin{proof}
    The general idea is that we compute $h'$-length flow shortcut with addtive length slack for all $h' \le h$.
    Namely, let $h_{j} = 2^{j}$ for $0 \le j \le \lceil\log_{2} h \rceil$, and we apply \Cref{lem:smallhemulator} to compute $H_{j} = \Emulator(G,\epsl,\epsk,h_{j})$ for each $j$.
    We denote the resulting flow shortcut as $H = \bigcup_{j}H_j$.
    Since there are at most $O(\log n)$ many shortcuts $H_{j}$, the overall edge size is still $n^{1+O(\epsk)}$.

\medskip
    
    \noindent\textbf{Forward Mapping.} For any flow $F^{G}$ in $G$ with congestion $1$, we decompose it into $F_j^{G}$ such that $h_{j-1} \le \leng(F_j^{G}, G) \le h_{j}$ for each $j$,
    and thus $F^{G} = \sum_{j}F_j^{G}$.
    Since $H_{j}$ is an $h_{j}$-length flow shortcut by \Cref{lem:smallhemulator}, then we have a flow $F^{H_{j}}$ in $G \cup H$ with $\Dem(F_j^{G}) = \Dem(F^{H_{j}})$ and $F^{H_{j}}$ has congestion $1$, 
    and step $t$.
    Note that the congestion from $F^{H_{j}}$ remains the same for parts in graph $G$, and congestion over each $H_{j}$ are independent from each other.
    Furthermore, there is a path-mapping from $F_{j}^{G}$ to $F^{H_{j}}$ with additive length slack $\delta = \epsl\cdot h_{j}$.
    Namely, for each path $P \in \path(F_{j}^{G})$, the path-mapping maps it to some subflows $F_{P}^{H_{j}}$ of $F^{H_{j}}$ with the same flow value of $P$ and \[\leng(F_{P}^{H_{j}},G\cup H_{j}) \le \leng(P, G) + \epsl\cdot h_{j} \le \leng(P, G) + 2\epsl\cdot \leng(P, G) \le (1+2\epsl)\leng(P, G),\]
    which certifies the path-mapping to admit a multiplicative length slack of $(1+2\epsl)$.
    By the union of $F^{H_{j}}$ for all $j$, we have a flow $F^{H}$ routing $\Dem(F^{G})$ with the same congestion $1$ and step $t$.
    What is more, similarly by the union of path-mapping from each $H_{j}$, there is a path-mapping from $F^{G}$ to $F^{H}$ with length slack $(1+2\epsl)$.

\medskip

    \noindent\textbf{Backward Mapping.}  
    Given the flow $F^{H}$ in $G \cup H$ with congestion $1$,
    for any path $P \in \path(F^{H})$, let $P^{G}$ and $P^{H_{j}}$ denote the path restrict in $G$ and $H_{j}$ respectively.
    We further let $F^{H_{j}}$ be the flow collecting the flow subpaths restricted in $H_{j}$ from each $P$,
    and $F_{0}^{G}$ be the one collecting parts from $P$ in $G$.
    Then by the shortcut property of $H_{j}$, we have a flow $F_{j}^{G}$ in $G$ routing the same demand with congestion $n^{\epsk}$.
    Again by concatenating flows $F_{j}^{G}$ for all $j$, we have a flow $F^{G}$ routing the demand $\Dem(F^{H})$. 
    We note that $\lceil\log_{2} h \rceil = O(\log n)$, and since we just take the union of $O(\log n)$ many flow shortcuts with the same congestion slack $n^{\epsk}$, the overall congestion of $F^{G}$ is at most $n^{\epsk}\cdot O(\log n)$.
    Since each $P^{H_{j}}$ is mapped to some subflows in $G$ with at most the same length by path-mapping from each $H_{j}$, we can concatenate those subflows together and the overall length is at most $\leng(P, G\cup H)$.

    Moreover, the backward mapping algorithm works similarily as the analysis above, which concludes the backward mapping property for $H$.

\medskip

    \noindent\textbf{Running Time.} 
    Each shortcut $H_{j}$ can be computed independently, thus there is only an extra $O(\log n)$ factor in the total work compared to \Cref{lem:smallhemulator}, which can be incorporated in the polynomial factor over $\epsk$.
    The argument also applies to the running time of the backward mapping algorithm.
\end{proof}

Now we start the bootstrapping. At the beginning, we receive parameters $1/\poly\log(n)<\epsl<1$ and $1/\Omega(\log\log\log n)<\epsk, \epsh<1$ from the input of \Cref{thm:emulator}. Let $c$ be some sufficiently large constant. Let
\[
\tepsl = \epsilon_{h}\cdot \epsl/c\qquad\text{and}\qquad\tepsk = \epsilon_{h}\cdot \epsk/c.
\]
Then we refer to
\begin{align*}
\lambda_{0} &= 1 + 2\tepsl,\\
\kappa_{0} &= n^{\tepsk}\cdot O(\log n) = n^{O(\tepsk)},\\
t_{0} &= \left(1/\tepsl\right)^{O\left(1/\tepsk^2\right)}\cdot \exp(\exp(\log (1/\tepsk)\cdot O(1/\tepsk)))
\end{align*}
as the length slack, congestion slack and step bound from \Cref{lem:LCompSmallhEmulator} with $\tepsl$ and $\tepsk$ as input. 

Next, we let $h = n^{\epsh}$, $d = \lceil \log_{h}N\rceil = O(1/\epsh)$ and 
\[
h_{0} = \frac{20(\lambda_{0}+\tepsl)^{d}\cdot t_{0}\cdot h}{\tepsl} = n^{O(\epsh)}
\]

\paragraph{Construct Shortcut $H$.} We let ${\cal A}_{0}$ be the algorithm $\Emulator(G_{0},\tepsl,\tepsk,h_{0})$ we will bootstrap. Note that given a graph $G_{0}$, ${\cal A}_{0}$ will compute an $h_{0}$-length flow shortcut $H_{0}$ of $G_{0}$ with length slack $\lambda_{0}$, congestion slack $\kappa_{0}$ and step $t_{0}$. 

Recall that $G$ denotes the input graph of \Cref{thm:emulator}. We let $V(G^{a}) = \{v\in V(G)\mid \l_{G}(v) \le a\}$ and denote by $G^{a}$ the graph induced by $V(G^{a})$.
Consider the following bootstrapping procedure. Initially, let $H_{1}$ be an $h_{0}$-LC flow shortcut of $G_{1}:=G^{h}$ by calling ${\cal A}_{0}$ on $G_{1}$, and let $\bar{H}_{1} = G_{1}\cup H_{1}$. For each $2\leq i\leq d$, 
\begin{itemize}
\item Let $G_{i} = \bar{H}_{i-1}\cup (G^{h^{i}} \setminus G^{h^{i-1}})$.%
\item Let $G'_{i}$ be $G_{i}$ but with different vertex lengths, and $\l_{i} = \tepsl\cdot h^{i-1}/(10t_{0})$ denote the length regularization factor for $i$, for each vertex $v\in V(G'_{i})$,
\[
\ell_{G'_{i}}(v)= \lceil \ell_{G_{i}}(v)/\l_{i} \rceil.
\]
\item Let $H'_{i}$ be an $h_{0}$-length flow shortcut of $G'_{i}$ by calling ${\cal A}_{0}$ on $G'_{i}$.
\item Let $H_{i}$ be $H'_{i}$ but with different vertex lengths:  for each vertex $v\in V(H_{i})$,
\[
\ell_{{H}_{i}}(v)= \ell_{ H'_{i}}(v)\cdot \l_{i}.
\]
\item Let $\bar{H}_{i} = G_{i} \cup H_{i}$.
\end{itemize}

We let the final shortcut $H$ be
\[ H = \bigcup_{j\leq d} H_{j}.\] A key observation is that $G\cup H$ is exactly the same as $G_{d}\cup H_{d}$, because we choose $d$ such that $h^{d}$ is at least the maximum edge length.

\paragraph{Shortcut Size.} When $i=1$ we have $|V(G_{1})|\leq n$. For each $i\geq 2$, we have
\[
|V(G_{i})| \leq |V(\bar{H}_{i-1})| + n\leq |V(G_{i-1})| + |V(H_{i-1})| + n\leq |V(G_{i-1})| + (|V(G_{i-1})|)^{1+O(\tepsk)} + n,
\]
where the last inequality is by \Cref{lem:LCompSmallhEmulator}. Therefore, we have
\[
|V(G_{i})|\leq n^{1+O(i\cdot\tepsk)}\leq n^{1+O(\tepsk/\epsh)}
\]
since $d = O(1/\epsh)$. Using \Cref{lem:LCompSmallhEmulator} again, we have
\[
|H| = \sum_{i}|H_{i}|\leq \sum_{i}|V(G_{i})|^{1+O(\tepsk)}\leq n^{1+O(\tepsk/\epsh)}/\epsh\leq n^{1+\epsk}
\]
by our choice of $\tepsk$.

\paragraph{Shortcut Construction Work and Depth.} We call the backward mapping algorithm of ${\cal A}_{0}$ on the graphs $G'_{i}$ sequentially. Each $G'_{i}$ has at most $n^{1+\epsk}$ vertices and $m + n^{1+\epsk}$ edges from our argument on shortcut size above. By \Cref{lem:LCompSmallhEmulator}, the work is 
\[
\poly(h_{0})\cdot (m+n^{1+\epsk})\cdot (n^{1+\epsk})^{\poly(\tepsk)}\cdot O(1/\epsh) = n^{O(\epsh+\epsk)}\cdot m,
\]
and the depth is
\[
\poly(h_{0})\cdot (n^{1+\epsk})^{\poly(\tepsk)}\cdot O(1/\epsh) = n^{O(\epsh)}.
\]

\paragraph{Shortcut Quality.} The quality of $H$ is given by the lemma below, and we defer the proof to \Cref{sect:ProofBootstrapQuality}.
\begin{lemma}
\label{lemma:BootstrapQuality}
$H$ is an LC-flow shortcut of $G$ with length slack $\lambda = (\lambda_{0} + \tepsl)^{d} + \tepsl/5$, congestion slack $\kappa = (\kappa_{0}+1)^{d}$ and step $t = t_{0}$.
\end{lemma}

Recall that $d = O(1/\epsh)$, $\lambda_{0} = 1 + 2\tepsl$, $\tepsl = \epsh\cdot\epsl/c$ and $\tepsk = \epsh\cdot \epsk/c$ for some sufficiently large $c$. For $H$, the length slack is $1+\epsl$, the congestion slack is $n^{\epsk}$, and the step is
\[
t = \left(1/(\epsl\epsh)\right)^{O\left(1/(\epsk\epsh)^2\right)}\cdot \exp(\exp(\log (1/(\epsk\epsh))\cdot O(1/(\epsk\epsh))))
\]

\paragraph{Algorithmic Backward Mapping.}

Recall that the input of the backward mapping algorithm is the edge representation of a $k$-commodity flow $F^{H}$ in $G\cup H$. We want to output the edge representation of a $k$-commodity flow $F^{G}$ in $G$ routing the same $k$-commodity demand, satisfying that $\cong(F^{G})\leq \kappa\cdot\cong(F^{H})$ and $\totlen(F^{G})\leq \totlen(F^{H})$. We note that all the operations between flows (such as addition and subtraction) are under their edge representations (the operations are natural so we will not define them formally).

The algorithm is as follows. Initially $F^{H}$ is a flow in $G_{d}\cup H_{d}$. Start from the last $i=d$ and let $F^{H_{d}} = F^{H}$.
We iterate $i$ from $d$ to $1$.
\begin{enumerate}
\item Let $F^{H_{i}}_{\true}$ be the restriction of $F^{H_{i}}$ on $H_{i}$. Let $F^{G_{i}}_{\old} = F^{H_{i}} - F^{H_{i}}_{\true}$ be the restriction of $F^{H_{i}}$ on $G_{i}$.
\item Apply the associated backward mapping algorithm of the shortcut $H_{i}$ (precisely, it should be $H'_{i}$) with $F^{H_{i}}_{\true}$ as the input. The output is the edge representation of a $k$-commodity flow $F^{G_{i}}_{\new}$ in $G'_{i}$ (so also $G_{i}$) with $\totlen(F^{G_{i}}_{\new},G'_{i})\leq \totlen(F^{H_{i}}_{\true}, H'_{i})$ and $\cong(F^{G_{i}}_{\new},G'_{i})\leq \kappa_{0}\cdot \cong(F^{H_{i}}_{\true},H'_{i})$. Because $\totlen(F^{H_{i}}_{\true}, H_{i}) = \ell_{i}\cdot  \totlen(F^{H_{i}}_{\true}, H'_{i})$ and $\totlen(F^{G_{i}}_{\new}, G_{i})\leq \ell_{i}\cdot \totlen(F^{G_{i}}_{\new}, G'_{i})$, we have \[
\totlen(F^{G_{i}}_{\new},G_{i})\leq \totlen(F^{H_{i}}_{\true},H_{i}).\]
Also, we have\[
\cong(F^{G_{i}}_{\new},G_{i})\leq \kappa_{0}\cdot \cong(F^{H_{i}}_{\true},H_{i})
\]
because the capacity of $G_{i}$ and $G'_{i}$ (also $H_{i}$ and $H'_{i}$) is the same.

\item Let $F^{G_{i}} = F^{G_{i}}_{\old} + F^{G_{i}}_{\new}$. Intuitively, we concatenate two flows $F^{G_{i}}_{\old}$ and $F^{G_{i}}_{\new}$, and observe that the concatenation indeed forms a flow. Naturally, we have
\[
\totlen(F^{G_{i}},G_{i})\leq \totlen(F^{G_{i}}_{\old},G_{i}) + \totlen(F^{H_{i}}_{\true},H_{i}) = \totlen(F^{H_{i}},G_{i}\cup H_{i}),\text{ and}
\]
\[
\cong(F^{G_{i}},G_{i})\leq \cong(F^{G_{i}}_{\old},G_{i}) + \kappa_{0}\cdot\cong(F^{H_{i}}_{\true},H_{i}) \leq (\kappa_{0}+1)\cdot \cong(F^{H_{i}},G_{i}\cup H_{i}),
\]

\item Let $F^{H_{i-1}}$ be the restriction of $F^{G_{i}}$ on $\bar{H}_{i-1} = G_{i-1}\cup H_{i-1}$. Let $F^{G,i}_{\output}$ be the restriction of $F^{G_{i}}$ on $(G^{h^{i}}\setminus G^{h^{i-1}})$. This means
\[
\totlen(F^{H_{i-1}},G_{i-1}\cup H_{i-1}) = \totlen(F^{G_{i}},G_{i}) - \totlen(F^{G,i}_{\output},G^{h^{i}}\setminus G^{h^{i-1}}),
\]
\[
\cong(F^{H_{i-1}},G_{i-1}\cup H_{i-1}) \leq \cong(F^{G_{i}},G_{i})\text{ and }\cong(F^{G,i}_{\output},G^{h^{i}}\setminus G^{h^{i-1}}) \leq \cong(F^{G_{i}},G_{i})
\]
\end{enumerate}
Let the final output be $F^{G} = \sum_{1\leq i\leq d} F^{G,i}_{\output}$. From the analysis above, we know
$\totlen(F^{G},G)\leq \totlen(F^{H},G\cup H)$ and $\cong(F^{G},G)\leq (\kappa_{0}+1)^{d}\cdot \cong(F^{H},G\cup H)$. So our backward mapping algorithm has the desired quality.

Regarding the work and depth, note that we call the backward mapping algorithm of ${\cal A}_{0}$ on the graphs $G'_{i}$ sequentially. Each $G'_{i}$ has at most $n^{1+\epsk}$ vertices and $m + n^{1+\epsk}$ edges from our argument on shortcut size above. By \Cref{lem:LCompSmallhEmulator}, the work is 
\[
\poly(h_{0})\cdot (m+n^{1+\epsk} + (n^{1+\epsk})k)\cdot (n^{1+\epsk})^{\poly(\tepsk)}\cdot O(1/\epsh) = n^{O(\epsh+\epsk)}\cdot (m + nk),
\]
and the depth is
\[
\poly(h_{0})\cdot (n^{1+\epsk})^{\poly(\tepsk)}\cdot O(1/\epsh) = n^{O(\epsh)}.
\]

\subsubsection{Proof of \Cref{lemma:BootstrapQuality}}
\label{sect:ProofBootstrapQuality}
\begin{lemma}[Forward Mapping]
The following holds for each $1\leq i\leq d$. Let $F^{G}$ be an $h^{i}$-length flow. Then there is a flow $F^{H_{i}}$ in $\bar{H}_{i}$ routing $\Dem(F^{G})$ with $\cong(F^{{H}_{i}})\leq \cong(F^{G})$, 
and $\step(F^{{H}_{i}})\leq t_{0}$. 
Moreover, there is a path-mapping from $F^{G}$ to $F^{H_{i}}$ with length slack $(\lambda_{i},\delta_{i})$
where $\lambda_{i} = (\lambda_{0} + \tepsl)^{i}$ and $\delta_{i} = \tepsl h^{i-1}/5$.
\label{lemma:BootstrappingForward}
\end{lemma}
As mentioned in the construction, we take $H = H_{\le d}$ as the resulting flow shortcut, and for any flow $F^{G}$,
we decompose it as $F^{G} = \bigcup_{i} F^{G_{i}}$ such that each path $P \in \path(F^{G_{i}})$ has length $h^{i-1} \le \leng(P, G) \le h^{i}$.
From \cref{lemma:BootstrappingForward}, for each $F^{G_{i}}$, there is a flow $F^{H_{i}}$ routing the same demand with $\cong(F^{H_{i}}) \le \cong(F^{G_{i}})\le \cong(F^{G})$ and $\step(F^{H_{i}}) \le t_{0}$.
And there exists a path-mapping from $F^{G_{i}}$ to $F^{H_{i}}$ s.t. for each path $P \in \path(F^{G_{i}})$ and the corresponding subflows $F_{P}^{H_{i}}$ in $G \cup H$, we have
\[\leng(F_{P}^{H_{i}}, G \cup H) \le (\lambda_{0} + \tepsl)^{i}\cdot\leng(P,G) + \tepsl h^{i-1}/5.\]
Note that $\leng(P, G) \ge h^{i-1}$, we further have
\[\leng(F_{P}^{H_{i}}, G \cup H) \le (\lambda_{0} + \tepsl)^{d}\cdot\leng(P,G) + (\tepsl/5)\cdot\leng(P,G),\]
which shows that the path-mapping from each $H_{i}$ has a length slack $(\lambda_{0}+\tepsl)^{d} + \tepsl/5$. 
Let $F^{H} = \bigcup_{i}F^{H_{i}}$, since each shortcut $H_{i}$ is disjoint, and the congestion over vertices in $G$ is unchanged, we conclude $\cong(F^{H}) \le \cong(F^{G})$, $\step(F^{H}) \le t_{0}$ and the overall path-mapping has length slack $(\lambda_{0}+\tepsl)^{d} + \tepsl/5$.

\paragraph{Forward Mapping: Proof of \Cref{lemma:BootstrappingForward}}

For $i=1$, we note that $h_{0} > h$, and thus the statement holds by the way we construct $H_{1}$ and \Cref{lem:LCompSmallhEmulator}. To prove the statement for a general $i\geq 2$, we will use induction and assume the statement holds for $i-1$.

\paragraph{Decompose Flow Paths of $F^{G}$.} For each flow path $P\in\path(F^{G})$, similarly as \cref{lem:decompose}, we decompose it into flow subpaths $\hat{\cal P}$ such that each subpath $\hat{P}\in\hat{\cal P}$ either has $\leng(\hat{P}, G)\leq h^{i-1}$ (called a \emph{non-trivial subpath}) or contains exactly one vertex $v$ with vertex length larger than $h^{i-1}$ (called a \emph{trivial subpath}). 
We note that a simple greedy argument shows there exists a decomposition $\hat{\cal P}$ of size 
\[
|\hat{\cal P}|\leq 2\leng(P)/h^{i-1} \leq 2h+1.
\]
Let $F^{G}_{\sshort}$ collect all non-trivial flow subpaths.

\paragraph{Route $\Dem(F^{G}_{\sshort})$ in $\bar{H}_{i-1}$.} By the construction, $F^{G}_{\sshort}$ have length at most $h^{i-1}$ and congestion at most $\cong(F^{G})$. By the induction hypothesis, there is a flow $F^{H_{i-1}}_{\sshort}$ in $\bar{H}_{i-1}$ routing $\Dem(F^{G}_{\sshort})$ on $\bar{H}_{i-1}$ with $\cong(F^{H_{i-1}}_{\sshort})\leq \cong(F^{G})$ and $\step(F^{H_{i-1}}_{\sshort})\leq t_{0}$.
And there is a path-mapping from $F^{G}_{\sshort}$ to $F^{H_{i-1}}_{\sshort}$ with length slack $(\lambda_{i-1},\delta_{i-1})$.

\paragraph{Route $\Dem(F^{G})$ in $G_{i}$ and $G'_{i}$.} We construct a routing $F^{G_{i}}$ of $\Dem(F^{G})$ in $G_{i}$ by, for each flow path $P\in\path(F^{G})$, routing $\Dem(P)$ in $G_{i}$ via $F^{G_{i}}_{P}$ defined as follows (we use $F^{G'_{i}}$ denote the same routing on $G'_{i}$).  
\begin{itemize}
\item For each non-trivial subpath $\hat{P}\in{\cal \hat{P}}$, the corresponding flow $F^{H_{i-1}}_{\hat{P}}$ in $\bar{H}_{i-1}\subseteq G_{i}$ routes $\Dem(\hat{P})$. Note that $\step(F^{H_{i-1}}_{\hat{P}})\leq \step(F^{H_{i-1}}_{\sshort})\leq t_{0}$, and the length slack of the path-mapping guarantees
\[
\leng(F^{H_{i-1}}_{\hat{P}}, H_{i-1})\leq \lambda_{i-1}\cdot \leng(\hat{P}, G) + \delta_{i-1}
\]
\item For each trivial subpath $\hat{P}\in\hat{\cal P}$ containing vertex $v$, clearly $h^{i-1} < \l_{G}(v) \le h^{i}$, and we route $\Dem(\hat{P})$ in $G_{i}$ using the same vertex $v$.
\end{itemize}
Lastly, we get $F^{G_{i}}_{P}$ by replacing each non-trivial subpath with $F^{H_{i-1}}_{\hat{P}}$ and keeping each trivial subpath unchanged. 
Clearly we can bound the the congestion of $F^{G_{i}}$ by $\cong(F^{G_{i}})\leq \max\{\cong(F^{H_{i-1}}_{\sshort}),1\} = \cong(F^{H_{i-1}}_{\sshort})\leq \cong(F^{G})$. 
Note that we can naturally define a path-mapping $\pi_{F^{G}\to F^{G'_{i}}}$ by mapping each $P\in \path(F^{G})$ to $F^{G_{i}}_{P}$.
We let $F^{G'_{i}}_{P}$ be the same routing as $F^{G_{i}}_{P}$ but with the different length function.
By the way we define the vertex lengths $\ell_{G'_{i}}(\cdot)$ of $G'_{i}$, for each flow path $P\in\path(F^{G})$, we have
\begin{align*}
\leng(F^{G'_{i}}_{P},G'_{i})&\leq \sum_{\substack{\text{non-trivial subpath}\\\hat{P}\in \hat{\cal P}}}\left( \frac{\leng(F^{H_{i-1}}_{\hat{P}},\bar{H}_{i-1})}{\l_{i}} + \step(F^{H_{i-1}}_{\hat{P}})\right) + \sum_{\substack{\text{trivial subpath}\\\hat{P}=\{v\}\in\cal{\hat{P}}}}\left(\frac{\ell_{G}(v)}{\l_{i}}+1\right)\\
&\leq \sum_{\substack{\text{non-trivial subpath}\\\hat{P}\in \hat{\cal P}}}\left( \frac{\lambda_{i-1}\cdot\leng(\hat{P},G) +\delta_{i-1}}{\l_{i}} + t_{0}\right) + \sum_{\substack{\text{trivial subpath}\\\hat{P}=\{v\}\in\cal{\hat{P}}}}\left(\frac{\ell_{G}(v)}{\l_{i}}+1\right)\\
&\leq \frac{\lambda_{i-1}\cdot \leng(P,G)}{\l_{i}} + \frac{|\hat{\cal P}|\cdot \delta_{i-1}}{\l_{i}} + |\hat{\cal P}|\cdot t_{0}\\
&\leq  \frac{\lambda_{i-1}\cdot \leng(P,G)}{l_{i}}+ |\hat{\cal P}|\cdot (2/h+1)t_{0}
\end{align*}
where the last inequality we use $\delta_{i-1} = \tepsl h^{i-2}/5$ and $\l_{i} = \tepsl h^{i-1}/(10t_{0})$.
Plugging in $\leng(P,G)\leq h^{i}$, and $|\hat{\cal P}|\leq 2h$, we further have
\[
\leng(F^{G_{i}}_{P},G'_{i})\leq \frac{10\lambda_{i-1}t_{0}}{\tepsl}\cdot h + (2h+1)\cdot (2/h+1)t_{0}\leq \frac{20(\lambda_{0} + \tepsl)^{d}t_{0}}{\tepsl}\cdot h = h_{0}.
\]
That is, $F^{G'_{i}}$ is an $h_{0}$-length flow on $G'_{i}$.

\paragraph{Route $\Dem(F^{G})$ in $H'_{i}$ and $H_{i}$.} By the induction hypothesis and because $F^{G'_{i}}$ is an $h_{0}$-length flow on $G'_{i}$, there is a flow $F^{H'_{i}}$ in $G'_{i} \cup H'_{i}$ routing $\Dem(F^{G})$ (we use $F^{H_{i}}$ denote the same routing in $G_{i} \cup H_{i}$) with $\cong(F^{H'_{i}})\leq \cong(F^{G'_{i}})\leq \cong(F^{G})$ and $\step(F^{H'_{i}})\leq t_{0}$ as desired. 
Furthermore, there is a path-mapping $\pi_{F^{G'_{i}}\to F^{H'_{i}}}$ from $F^{G'_{i}}$ to $F^{H'_{i}}$ with length slack $\lambda_{0}$. 

Namely, for each $P\in\path(F^{G})$, from above we can map it to $F_{P}^{G'_{i}}$ with path-mapping $\pi_{F^{G}\to F^{G'_{i}}}$.
Further by composing it with $\pi_{F^{G'_{i}}\to F^{H'_{i}}}$, we obtain a path-mapping $\pi_{F^{G}\to F^{H'_{i}}}$ that maps each $P\in\path(F^{G})$ to a flow $F^{H'_{i}}_{P}$ in $G'_{i} \cup H'_{i}$
such that
\[
\leng(F^{H'_{i}}_{P},G'_{i} \cup H'_{i})\leq \lambda_{0}\cdot \leng(F^{G'_{i}}_{P},G'_{i}).
\]
Since $F^{H_{i}}$ and $F^{H'_{i}}$ are the same routing again, we use $\pi_{F^{G}\to F^{H_{i}}} = \pi_{F^{G}\to F^{H'_{i}}}$ denote the same path-mapping from $F^{G}$ to $F^{H_{i}}$.
By the way we define the vertex length $\ell_{H_{i}}(\cdot)$ of $H_{i}$, for each path $P \in \path(F^{G})$, 
we have
\begin{align*}
\leng(F^{H_{i}}_{P},G_{i} \cup H_{i}) &\le \leng(F^{H'_{i}}_{P},G'_{i}\cup H'_{i})\cdot \l_{i}\\
&\leq \lambda_{0}\cdot\leng(F^{G_{i}}_{P},G'_{i})\cdot \l_{i}\\
&\leq \lambda_{0}\cdot \left(\frac{\lambda_{i-1}\cdot \leng(P,G)}{\l_{i}} + |\hat{\cal P}|\cdot(2/h + 1) t_{0} \right)\cdot \l_{i}\\
& \leq \lambda_{0} \cdot \lambda_{i-1} \leng(P,G) + \lambda_{0}\cdot\left(\frac{2\leng(P,G)}{h^{i-1}} + 1\right)\cdot(2/h + 1)\cdot t_{0}\l_{i}\\
& \le (\lambda_{0}\lambda_{i-1} + \tepsl)\leng(P,G) + \lambda_{0}\cdot(2/h + 1) \cdot \tepsl h^{i-1}/10\\
& \leq (\lambda_{0} + \tepsl)\cdot\lambda_{i-1}\leng(P,G) + \tepsl h^{i-1}/5.
\end{align*}
For the fourth inequality, we use the fact that $|\hat{\cal P}| \le 2\leng(P,G)/h^{i-1}+1$, and we apply $\l_{i} = \tepsl h^{i-1}/(10t_{0})$ in the fifth inequality. Note that we assume $\lambda_{0}$ is close to $1$ and $h$ is large enough such that $\lambda_{0}\cdot(2/h + 1) \le 2$. 
This concludes the \cref{lemma:BootstrappingForward}.

\begin{lemma}[Backward Mapping]\label{lemma:BootstrappingBackward}
The following holds for each $1\leq i\leq d$. Let $F^{H_{i}}$ be a flow with $V(\Dem(F^{H_{i}})) \subseteq V(G)$ and congestion $1$ in $G_{i}\cup H_{i}$. Then there is a flow $F^{G}$ in $G$ routing $\Dem(F^{H_{i}})$ with $\cong(F^{G})\leq (\kappa_{0}+1)^{i}$. 
And here is a path-mapping from $F^{H_{i}}$ to $F^{G}$ with length slack $1$.
\end{lemma}

\paragraph{Backward Mapping: Proof of \Cref{lemma:BootstrappingBackward}}

For $i = 1$, let $G_{1} = G^{h}$, and we note that $H_{1}$ is the $h_{0}$-length flow shortcut over $G_{1}$,
thus the statement similarly holds from the definition of LC-flow shortcut. For $i \ge 2$, we then assume the statement holds for $i-1$ and use induction to prove the lemma.

Given the flow $F^{H_{i}}$ with $V(\Dem(F^{H_{i}})) \subseteq V(G)$ and congestion $1$ in $G_{i} \cup H_{i}$, 
we would like to show that the demand $\Dem(F^{H_{i}})$ can be routed in $G$ with congestion $\kappa_{0}^{i}$ and at most the same length.
    
\paragraph{Route $\Dem(F^{H_{i}})$ in $G'_{i} \cup H_{i}'$.} By the way we define the vertex length between $G_{i}, H_{i}$ and $G'_{i}, H'_{i}$ respectively,
this naturally gives the same routing $F^{H'_{i}}$ in $G'_{i} \cup H'_{i}$ with the same congestion $1$ but different length. 
Further we let the trivial path-mapping between them as $\pi_{F^{H_{i}}\to F^{H'_{i}}}$.
    
\paragraph{Route $\Dem(F^{H_{i}})$ in $G_{i}'$.}
For each path $P \in \path(F^{H'_{i}})$, we decompose it into two collections of subpaths $\hat{\cP}_{H'_{i}}$ and $\hat{\cP}_{G'_{i}}$ s.t. subpaths from them are restricted on $H'_{i}$ and $G'_{i}$ respectively. 
Let ${F}_{\mathrm{sub}}^{H'_{i}}$ and $F_{\mathrm{sub}}^{G'_{i}}$ collect all flow paths from $\hat{\cP}_{H'_{i}}$ and  $\hat{\cP}_{G'_{i}}$ respectively for each $P$. 

Note that the decomposition of each path $P$ takes place at the boundary between ${H'}_{i}$ and $G'_{i}$, we can have $V(\Dem({F}_{\mathrm{sub}}^{H'_{i}})) \subset V(G'_{i})$.
Then since $H_{i}'$ is an LC-flow shortcut over $G'_{i}$, given ${F}_{\mathrm{sub}}^{H'_{i}}$,
there exists a flow $\hat{F}^{G'_{i}}$ routing $\Dem(\Dem({F}_{\mathrm{sub}}^{H'_{i}}))$
with congestion $\kappa_{0}$.

We construct $F^{G'_{i}}$ by concatenating $\hat{F}^{G'_{i}}$ and $F_{\mathrm{sub}}^{G'_{i}}$,
and we have $\cong(F^{G'_{i}}) \le \cong(\hat{F}^{G'_{i}}) + \cong(F_{\mathrm{sub}}^{G'_{i}})\le \kappa_{0} + 1$.
Further, from the backward path-mapping from $H'_{i}$,
for each path $P \in F^{H'_{i}}$, each subpath $\hat{P}$ restricted on $H'_{i}$ is mapped to subflows $F_{\hat{P}}^{G'_{i}}$ with the length at most $\leng(\hat{P},G'_{i}\cup H'_{i})$. 
In general, by replacing each subpath $\hat{P}$ with $F_{\hat{P}}^{G'_{i}}$, we can map each path $P$ to some subflows $F_{P}^{G'_{i}}$, which gives a path-mapping $\pi_{F^{H'_{i}}\to F^{G'_{i}}}$ from $F^{H'_{i}}$ to $F^{G'_{i}}$.

\paragraph{Route $\Dem(F^{H_{i}})$ in $G_{i}$.} Considering $G_{i}$ is the same graph as $G'_{i}$ 
where the length of each vertex is increased by a multiplicative factor of at most $\l_{i}$.
then we also have a flow $F^{G_{i}}$ directly from $F^{G'_{i}}$ with the same congestion. 
Let the trivial path-mapping between $G'_{i}$ and $G_{i}$ as $\pi_{F^{G'_{i}}\to F^{G_{i}}}$, by composing it with all path-mappings above, we have a path-mapping $\pi_{F^{H_{i}}\to F^{G_{i}}}$ from $F^{H_{i}}$ to $F^{G_{i}}$.

For each path $P \in F^{H_{i}}$(then also in $F^{H'_{i}}$), we can map it to subflows $F_{P}^{G_{i}}$. 
Let $\hat{\cP}_{H_{i}}$ and $\hat{\cP}_{G_{i}}$ defined similarly as above, 
we note that 
\[\sum_{\hat{P} \in \hat{\cP}_{{H}_{i}} \sqcup \hat{\cP}_{G_{i}}} \leng(\hat{P}, G_{i}\cup H_{i}) = \leng(P, G_{i}\cup H_{i}).\]
Given any subpath $\hat{P} \in \hat{\cP}_{{H}_{i}}$, it is also in the $\hat{\cP}_{{H}'_{i}}$.
As discussed in the previous step, it can be mapped to some subflows $F_{\hat{P}}^{G'_{i}}$(then also in $F^{G_{i}}$) with length at most $\leng(\hat{P}, G'_{i}\cup H'_{i})$.
We have
\[\leng(F_{\hat{P}}^{G_{i}},G_{i}) \le \leng(F_{\hat{P}}^{G'_{i}},G'_{i})\cdot \l_{i} \le \leng(\hat{P}, G'_{i}\cup H'_{i}) \cdot \l_{i} \le \leng(\hat{P}, G_{i}\cup H_{i})\]
where the last inequality use the fact that $\hat{P}$ is restricted in $H'_{i}$.
Further, any subpath $\hat{P}$ restricted in $G_{i}$ remains the same under the path-mapping, and thus contribute the same length in $\leng(F_{P}^{G_{i}})$.
We conclude that $\leng(F_{P}^{G_{i}}) \le \leng(P,G_{i}\cup H_{i})$, which shows the path-mapping $\pi_{F^{H_{i}}\to F^{G_{i}}}$ has length slack $1$.

\paragraph{Decompose Flow Paths of $F^{G_{i}}$.} Note that $G_{i} = \bar{H}_{i-1}\cup (G^{h^{i}} \setminus G^{h^{i-1}})$ where $\bar{H}_{i-1} = G_{i-1}\cup H_{i-1}$.
We apply the same technique to each flow path $P \in \path(F^{G_{i}})$.
Namely, we decompose  it into two collections of subpaths $\hat{\cP}_{\bar{H}_{i-1}}$ and $\hat{\cP}_{G'}$ s.t. subpaths from them are restricted on the $\bar{H}_{i-1}$ and $(G^{h^{i}} \setminus G^{h^{i-1}})$ respectively. 

Let $F^{H_{i-1}}$ collect all flow paths from $\hat{\cP}_{\bar{H}_{i-1}}$ for each $P \in \path(F^{G_{i}})$, 
we can similarly have $V(\Dem(F^{{H}_{i-1}})) \subset V(G)$, and then by induction hypothesis there is a flow $F^{G''}$ in $G$ routing the same demand with $\cong(F^{G''}) \le (\kappa_{0}+1)^{i-1}\cdot \cong(F^{\bar{H}_{i-1}})\le (\kappa_{0}+1)^{i-1}\cdot \cong(F^{G_{i}}) \le (\kappa_{0}+1)^{i}$.
We further note that the congestion from it only involves vertex with length samller than $h^{i-1}$.
    
We then similarily define $F^{G'}$ to collect all flow paths from $\hat{\cP}_{G'}$ for each $P$, and by concatenating it with $F^{G''}$, we have a flow $F^{G}$ routing the demand $\Dem(F^{G_{i}}) = \Dem(F^{H_{i}})$ with congestion $\max\{\cong(F^{G''}),\cong(F^{G'})\} \le (\kappa_{0}+1)^{i}$.
    
Considering the path-mapping $\pi_{F^{H_{i-1}}\to F^{G''}}$ from induction hypothesis, for each path $P \in \path(F^{G_{i}})$,
each subpath $\hat{P} \in \hat{\cP}_{\bar{H}_{i-1}}$ corresponds to subflows $F_{\hat{P}}^{G''}$ with $\leng(F_{\hat{P}}^{G''}, G) \le \leng(\hat{P},G_{i})$.
By concatenating all subpaths $\hat{P} \in \hat{\cP}_{G'}$, each path $P$ corresponds to subflows $F_{P}^{G}$ in $G$ with length $\leng(F_{P}^{G}) \le \leng(P,G_{i})$.
This gives a path-mapping $\pi_{F^{G_{i}}\to F^{G}}$ from $F^{G_{i}}$ to $F^{G}$ with length slack $1$.

\paragraph{Route $\Dem(F^{H_{i}})$ in $G$.} We conclude by combining previous steps. 
Given a feasible flow $F^{H_{i}}$ in $G_{i}\cup H_{i}$, there exists flow $F^{G}$ in $G$ routing the same demand with congestion $(\kappa_{0}+1)^i$.
Further by composing $\pi_{F^{H_{i}}\to F^{G_{i}}}$ and $\pi_{F^{G_{i}}\to F^{G}}$, we have a path-mapping from $F^{H_{i}}$ to $F^{G}$ with length slack~$1$.
This concludes the \cref{lemma:BootstrappingBackward}.

\section{Length-Constrained Multi-Commodity Maxflows in Low Depths}
\label{Sect:ApproxLCMCFlow}

In this section we focus on directed graphs with edge lengths and capacities. Note that this setting is stronger than undirected graphs with vertex lengths and capacities because of the standard vertex-splitting reduction.

In this section, the notations of flows are slightly different from \Cref{sect:PrelimFlow}, because we want the notations to be consistent with a previous work \cite{HaeuplerHS23}. Precisely, we will talk about multi-commodity flows in a directed graph with edge lengths $\ell$, capacities $U$, and source-sink pairs $\{(S_{i},T_{i})\mid 1\leq i\leq k,\  S_{i},T_{i}\subseteq V(G)\text{ are disjoint}\}$. Consider a length parameter $h\geq 1$. Let ${\cal P}_{h}(S_{i},T_{i})$ collect all simple $S_{i}$-$T_{i}$ paths with length at most $h$. Then an \emph{$h$-length multi-commodity flow} $F = \{F_{i}\mid 1\leq i\leq k\}$ and \emph{moving cuts} $w$ are formalized by the following LP and its dual.
\begin{align*}
~~~~~~~~\max\ &\vvalue(F) =\sum_{1\leq i\leq k}\sum_{P\in {\cal P}_{h}(S_{i},T_{i})} F_{i}(P) ~~~~~~~~~~~~~~~~~~~~~~~~~~~~\min\ |w|=\sum_{e\in E(G)} U(e)w(e)\\
\text{s.t. }&\sum_{i}\sum_{P\ni e} F_{i}(P)\leq U(e), \forall e\in E(G)~~~~~~~~~~~~~~~~~~\text{s.t. }\sum_{e\in P}w(e)\geq 1,\forall P\in\bigcup_{i}{\cal P}_{i}(S_{i},T_{i})\\
&F_{i}(P)\geq 0, \forall i\in[k],P\in{\cal P}_{h}(S_{i},T_{i})~~~~~~~~~~~~~~~~~~~~~~~~~~~~w(e)\geq 0,\forall e\in E(G)
\end{align*}
An $(1+\epsilon)$-approximate $h$-length multi-commodity flow and moving cut pair $(F,w)$ is such that $|w|\leq (1+\epsilon)\vvalue(F)$.

In fact, in the context of \Cref{sect:PrelimFlow}, an ($h$-length) multi-commodity flow $F$ defined as above is equivalent to an ($h$-length) $k$-commodity flow $F$ \emph{partially} routing the $k$-commodity demand $D$ defined as follows. For each commodity $1\leq i\leq k$, the $i$-th commodity demand $D_{i}$ of $D$ is, for each $v\in V(G)$, if $v\in S_{i}$, $D_{i}(v) = \infty$; if $v\in T_{i}$, $D_{i}(v) = -\infty$; otherwise, $D_{i} = 0$. Intuitively, the multi-commodity flows in this section are non-concurrent.

The following is the main result of this section.

\begin{theorem}
\label{thm:ApproxLCMCMF}
Let $G = (V(G),E(G))$ be a directed graph with positive integral edge lengths $\ell$, integral capacities $U$ and source-sink pairs $\{(S_{i},T_{i})\subseteq V(G)\mid 1\leq i \leq k\}$. Given length parameter $h\geq 1$ and approximation parameter $0<\epsilon < 1$, there is a parallel algorithm that computes a feasible multi-commodity $h$-length flow, moving cut pair $(F,w)$ that is $(1+\epsilon)$-approximate. Furthermore, $F$ is $1/O(kh^{2}\log n/\epsilon^{3})$-fractional and it is outputted in its path representation, where the number of flow paths for each commodity is $O(mh^{5}\log n/\epsilon^{4})$.
The work and depth are $\tilde{O}(mkh^{8}/\epsilon^{4})$ and $\tilde{O}(h^{5}/\epsilon^{3})$ respectively.
\end{theorem}

We point out that, \cite{HaeuplerHS23} showed a single-commodity version of \Cref{thm:ApproxLCMCMF}. They also generalized their single-commodity result to the multi-commodity setting, but with work $\tilde{O}(mk)\cdot \poly(h,\epsilon^{-1})$ and depth $\tilde{O}(k)\cdot \poly(h,\epsilon^{-1})$. Our multi-commodity result is stronger in the sense that we get depth $\tilde{O}(\poly(h,\epsilon^{-1}))$ independent on $k$. To the best of our knowledge, we are the first to achieve this result and it may be of independent interest.

However, the flows they output have better integrality, i.e. their solution is $1/\tilde{O}(\epsilon^{-2})$-fractional, while ours is $1/\tilde{O}(k\cdot \poly(h,\epsilon^{-1}))$-fractional. But we note that we will not exploit the fractionality guarantee in \Cref{thm:ApproxLCMCMF} in the later sections.

Our approach is a while-box generalization of \cite{HaeuplerHS23}, so the techniques have a large overlap with theirs.

\subsection{Section Preliminaries}

\begin{lemma}[Sparse Decomposition of Acyclic Flows, Theorem 10.1 in \cite{HaeuplerHS23}]
\label{lemma:FlowDecomposition}
Let $G$ be an $h$-layer DAG with source and sink vertices $S,T\subseteq V(G)$. Given a $(1/\mu)$-fractional edge flow function $f$, there is an algorithm that computes the path representation of a $(1/\mu)$-fractional single-commodity flow $F$ with $|\path(F)|\leq m$ whose edge representation is $f$. The work is $\tilde{O}(mh)$ and depth is $\tilde{O}(h)$.
\end{lemma}

\begin{lemma}[Parallel Flow Rounding \cite{DBLP:journals/siamcomp/Cohen95}]
\label{lemma:FlowRounding}
Let $G$ be a directed graph with source and sink vertices $S,T\subseteq V(G)$. Given the edge representation of a fractional single-commodity flow $F$ in $G$, there is an algorithm that computes the edge representation of an integral single-commodity flow $\hat{F}$ in $G$ satisfying that
\begin{itemize}
\item $\lceil\vvalue(F) - 1/\poly(n)\rceil\leq \vvalue(\hat{F})\leq \lceil \vvalue(F)\rceil$,
\item for each edge $e\in E(G)$, $\lfloor F(e)\rfloor\leq \hat{F}(e)\leq \lceil F(e)\rceil$.
\end{itemize}
The work is $\tilde{O}(m)$ and the depth is $\tilde{O}(1)$.
\end{lemma}

\begin{corollary}
\label{coro:FlowRounding}
Let $G$ be a directed graph with source and sink vertices $S,T\subseteq V(G)$. Given a rounding parameter $\mu\leq \poly(n)$ which is a power of $2$, and the edge representation of a fractional single-commodity flow $F$ in $G$, there is an algorithm that computes the edge representation of a $(1/\mu)$-fractional single-commodity flow $\hat{F}$ in $G$ satisfying that
\begin{itemize}
\item $\lceil\mu\cdot\vvalue(F) - 1/\poly(n)\rceil\leq \mu\cdot\vvalue(\hat{F})\leq \lceil \mu\cdot\vvalue(F)\rceil$,
\item for each edge $e\in E(G)$, $\lfloor \mu\cdot F(e)\rfloor\leq \mu\cdot \hat{F}(e)\leq \lceil \mu\cdot F(e)\rceil$.
\end{itemize}
The work is $\tilde{O}(m)$ and the depth is $\tilde{O}(1)$.
\end{corollary}
\begin{proof}
Round $\mu\cdot F$ to an integral $\hat{F}'$ using \Cref{lemma:FlowRounding}. Return $\hat{F} = \hat{F}'/\mu$.
\end{proof}

\subsection{Multi-Commodity $h$-Length $(1+\epsilon)$-Lightest Path Blockers}
\begin{definition}
Let $G = (V(G),E(G))$ be a graph with lengths $\ell$, capacities $U$, weights $w$ and source-sink pairs $\{(S_{i},T_{i})\subseteq V(G)\mid 1\leq i\leq k\}$. Fix parameters $\epsilon > 0$, $h\geq 1$, $\lambda\leq \min_{1\leq i\leq k}d^{(h)}_{w}(S_{i},T_{i})$ where $d^{(h)}_{w}(S_{i},T_{i}):= \min_{P \in \cP_{h}(S_{i},T_{i})}w(P)$, $\mu\geq 1$ and $0\leq \alpha\leq 1$, an $h$-length $(1/\mu)$-fractional multi-commodity flow $F$ is an $(1/\mu)$-fractional $h$-length $(1+\epsilon)$-lightest path $\alpha$-blocker if
\begin{itemize}
\item Each flow path $P\in\path(F)$ has $w(P)\leq (1+2\epsilon)\lambda$, and
\item For each $P\in \bigcup_{1\leq i\leq k}{\cal P}_{h}(S_{i},T_{i})$ such that $w(P)\leq (1+\epsilon)\lambda$, there is some $e\in E(P)$ such that $F(e) \geq \alpha\cdot U(e)$.
\end{itemize}
\end{definition}

\begin{theorem}
\label{thm:PathBlockers}
Let $G = (V(G),E(G))$ be a directed graph with positive integral lengths $\ell$, integral capacities $U$, (possibly fractional) weights $w$ and source-sink pairs $\{(S_{i},T_{i})\subseteq V(G)\mid 1\leq i\leq k\}$. Given parameters $h\geq 1$, $0<\epsilon < 1$, and $\lambda \leq \min_{1\leq i\leq k}d^{(h)}_{w}(S_{i},T_{i})$, there is an algorithm that computes the path representation of a $(1/\mu)$-fractional $h$-length $(1+\epsilon)$-lightest path $\alpha$-blocker $F$, where $\alpha = \Omega(\epsilon/h^{2})$ and $\mu = O(kh^{2}/\epsilon)$. Furthermore, the number of flow paths in the path representation is $|\path(F)|\leq O(mh^{2}/\epsilon)$. The work and depth are $\tilde{O}(mkh^{5}/\epsilon)$ and $\tilde{O}(h^{2})$.
\end{theorem}

\subsubsection{Length-Weight Expanded DAG}
\label{sect:LWExpandDAG}

Let $G$ be a directed graph with positive integral edge lengths $\ell$, capacities $U$, weights $w$ and source-sink pairs $\{(S_{i},T_{i})\subseteq V(G)\mid 1\leq i\leq k\}$, where edge lengths are positive integers. Given parameters $h\geq 1$, $0<\epsilon < 1$ and $\lambda \leq \min_{1\leq i\leq k}d^{(h)}_{w}(S_{i},T_{i})$, its \emph{length-weight expanded DAG} $D^{(h,\lambda)}$ is an $(h+1)$-layer directed graph with edge capacities $U'$ and source-sink pairs $\{(S'_{i},T'_{i})\mid 1\leq i\leq k\}$.

Before defining $D^{(h,\lambda)}$, we first define another edge weight $\tilde{w}$ which rounds up $w$ to the nearest multiple of $(\epsilon/h)\cdot\lambda$, i.e. for each $e\in E(G)$, 
\[
\tilde{w}(e) = (\epsilon/h)\cdot\lambda\cdot\left\lceil \frac{w(e)}{(\epsilon/h)\cdot\lambda}\right\rceil.
\]
Intuitively, $\tilde{w}$ discretizes $w$.
Now we define $D^{(h,\lambda)}$ as follows.
\begin{itemize}
\item The vertex set $V(D^{(h,\lambda)})$ is constructed by, for each $v\in V(G)$, making copies $v(x,h')$ where $0\leq h'\leq h$ and $x$ ranges over multiples of $\frac{\epsilon}{h}\cdot\lambda$ between $0$ and $(1+2\epsilon)\cdot \lambda$. Therefore, each original vertex $v\in V(G)$ has $\kappa = (h+1)\cdot (\lfloor\frac{(1+2\epsilon)\cdot\lambda}{(\epsilon/h)\cdot\lambda}\rfloor+1) = O(h^{2}/\epsilon)$ copies.

\item For each source-sink pair $(S_{i},T_{i})$ of $G$. The corresponding source-sink pair $(S'_{i},T'_{i})$ of $D^{(h,\lambda)}$ is defined by $S'_{i} = \{v(0,0)\mid v\in S_{i}\}$ and $T'_{i} = \{\text{all copies of $v$ in $V(D^{(h,\lambda)})$}\mid v\in T_{i}\}$.

\item Consider the edge set $E(D^{(h,\lambda)})$. For each original edge $e=(u,v)\in E(G)$, we add an edge from $u(x,h')$ to $v(x+\tilde{w}(e), h'+\ell(e))$ (called a copy of $e$) for each $x,h'$ such that $x+\tilde{w}(e)\leq (1+2\epsilon)\cdot\lambda$ and $h'+\ell(e)\leq h$. We let $E'(e)$ denote the set of all copies of $e$ and note that $|E'(e)|\leq \kappa$.

\item For each edge $e'\in E(D^{(h,\lambda)})$ which is a copy of an original edge $e\in E(G)$, the capacity of $e'$ is $U'(e') = U(e)$.

\end{itemize}

Next, we discuss how to project paths between $G$ and $D^{(h,\lambda)}$. We first consider the forward projection. 

\medskip

\noindent{\textbf{Forward Projection.}} For each $1\leq i\leq k$ and original path $P\in {\cal P}_{h}(S_{i},T_{i})$ with $w(P)\leq \lambda\cdot(1+\epsilon)$, we define its projection $P'$ on $D^{(h,\lambda)}$ as follows. Say $P$ visits original vertices $v_{1},v_{2},...,v_{k}$ where $v_{1}\in S_{i}$ and $v_{k}\in T_{i}$. Then $P'$ starts at the copy $v_{1}(0,0)$ of $v_{1}$, walks along copies of edges in $P$, and ends at some copy of $v_{k}$. Note that $P'$ must exist because (1) $x$ can be (almost) as large as $(1+2\epsilon)\lambda$ when creating vertex copies, and (2) $\tilde{w}(e)$ overestimates $w(e)$ by at most $(\epsilon/h)\cdot \lambda$ when creating edge copies (it means $\tilde{w}(P)\leq w(P) + \epsilon\lambda \leq (1+2\epsilon)\lambda$ because $P$ is $h$-length).

\medskip

\noindent{\textbf{Backward Projection.}} For each $1\leq i\leq k$ and an $S'_{i}$-$T'_{i}$ path $P'$ in $D^{(h,\lambda)}$, we can project it back to an original path $P$ in $G$ by simply taking the original edge of each edge in $P'$. Observe that $w(P)\leq (1+2\epsilon)\lambda$ because $x$ is at most $(1+2\epsilon)\lambda$ when creating vertex copies.

\subsubsection{Multi-Commodity Blocking Flows in $h$-Layer DAGs}

\begin{lemma}
\label{lemma:ApproxBlockingFlow}
Let $G$ be an $h$-layer DAG with integral capacities $U$ and source-sink pairs $\{(S_{i},T_{i})\mid 1\leq i\leq k\}$.
Given a constant $0<\alpha<1$, there is an algorithm that computes the edge representation of a multi-commodity $(1/O(k))$-fractional $\alpha$-blocking flow $F$. The work is $\tilde{O}(mkh^{3})$ and the depth is $\tilde{O}(h^{2})$.
\end{lemma}
\begin{proof}

\noindent{\textbf{Path Counts and the Path Count Flow.}}
We first introduce the notion of \emph{path counts} and the \emph{path count flow} with respect to an arbitrary $(1/\poly(n))$-fractional capacity function $U'$ of $G$. We point out that it is possible that $U'(e)=0$ for some edges $e\in E(G)$, and we will ignore such edges in the following definition.

For any simple path $P$ in $G$, we let $c(P) = \prod_{e\in P} U'(e)$.
For each $1\leq i\leq k$ and edge $e\in E(G)$, we define its \emph{commodity-$i$ path count} to be 
\[
c(e,i) = \sum_{P\in{\cal P}(S_{i},T_{i})\text{ s.t. }e\in P} c(P).
\]
Equivalently, we can interpret $c(e,i)$ as follows. For each vertex $v\in V(G)$, let
$c^{-}(v,i) = \sum_{P\in{\cal P}(S_{i},v)}c(P)$ and $C^{+}(v,i) = \sum_{P\in{\cal P}(v,T_{i})}c(P)$. Then for each edge $e=(u,v)\in E(G)$, we have
\[
c(e,i) = c^{-}(u,i)\cdot U'(e)\cdot c^{+}(v,i).
\]
For each $e\in E(G)$, we let $c(e) = \sum_{1\leq i\leq k}c(e,i)$ be its total path counts of all commodities. Let $c_{\max} = \max_{e\in E(G)}c(e)$ be the maximum path counts.

The \emph{commodity-$i$ path count flow} w.r.t. $U'$, denoted by $\wtilde{F}_{i}$ is defined by the following edge representation: for each $e\in E(G)$, 
\[
\wtilde{F}_{i}(e) = U'(e)\cdot c(e,i)/c_{\max}.
\]
The \emph{multi-commodity path count flow} is then $\wtilde{F} = \{\wtilde{F}_{i}\mid 1\leq i\leq k\}$. Note that $\wtilde{F}$ is a feasible multi-commodity flow w.r.t. $U'$ in $G$ since the the flow conservation is guaranteed, and it can be computed efficiently as shown in \Cref{lemma:PathCountFlows}.

\begin{lemma}
\label{lemma:PathCountFlows}
Given a $(1/\poly(n))$-fractional capacity function $U'$ of $G$, there is an algorithm that computes the edge representation of $\wtilde{F}_{i}$ for all $1\leq i\leq k$ in $\tilde{O}(mkh^{2})$ work and $\tilde{O}(h)$ depth.
\end{lemma}
\begin{proof}
As we discussed above, computing the edge representation of $\wtilde{F}_{i}$ can be reduced to computing $c^{+}(v,i)$ and $c^{-}(v,i)$ for all $1\leq i\leq k$ and vertices $v\in V(G)$. It is easy to see that, for a fixed $1\leq i\leq k$, a dynamic-programming-type algorithm can compute $c^{-}(v,i)$ and $c^{+}(v,i)$ for all vertices $v$ in $\tilde{O}(mh^{2})$ work and $\tilde{O}(h)$ depth. Run the dynamic programming for each $1\leq i\leq k$ in parallel gives the lemma.

In fact, the dynamic programming takes $O(m)$ work and $O(h)$ depth assuming multiplying two large numbers only takes constant work. Without this assumption, note that each of these $c^{-}(v,i),c^{+}(v,i)$ and $c(e,i)$ can be represented precisely using $\tilde{O}(h)$ bits, because the number of simple paths is $n^{O(h)}$ and each $c(P)$ is $1/n^{O(h)}$-fractional and at most $n^{O(h)}$. Multiplying two $\tilde{O}(h)$-bit numbers can be done trivially in $\tilde{O}(h^{2})$ work and $\tilde{O}(1)$ depth, which becomes the overheads to the final work and depth.
\end{proof}

\medskip

\noindent{\textbf{The Algorithm.}} Now we are ready to computes the multi-commodity $(1/O(k))$-fractional $\alpha$-blocking flow $F$ with respect to the graph $G$ and capacities $U$. On a high level, the algorithm iteratively computes a path count flow, rounds it to be $(1/O(k))$-fractional, and subtracts it from the capacity. 

Formally speaking, we start with the initial capacities $U^{(0)} = U$. At each iteration $j\geq 0$, the algorithm do the following.
\begin{enumerate}
\item If the maximum path counts $c^{(j)}_{\max}$ w.r.t. $U^{(j)}$ is $0$, terminate the algorithm. Otherwise, compute the multi-commodity path count flow $\wtilde{F}^{(j)} = \{\wtilde{F}^{(j)}_{i}\mid 1\leq i\leq k\}$ w.r.t. $U^{(j)}$ using \Cref{lemma:PathCountFlows}. 
\item For all $1\leq i\leq k$ in parallel, Round $\wtilde{F}^{(j)}_{i}$ to a $(1/\mu)$-fractional flow $\hat{F}^{(j)}_{i}$ using \Cref{coro:FlowRounding}, where $\mu$ is the smallest power of $2$ such that $\mu\geq 4k/(1-\alpha)$. It guarantees that for each $e\in E(G)$, $\lfloor \mu\cdot\wtilde{F}^{(j)}_{i}(e)\rfloor\leq \mu\cdot \hat{F}^{(j)}_{i}(e)\leq \lceil \mu\cdot\wtilde{F}^{j}_{i}(e)\rceil$, which implies $\wtilde{F}^{(j)}_{i}(e) - 1/\mu\leq \hat{F}^{(j)}_{i}(e) \leq \wtilde{F}^{j}_{i}(e) + 1/\mu$.
\item Let $F^{(j)} = \{\hat{F}^{(j)}_{i}/2\mid 1\leq i\leq k\}$. For each edge $e\in E(G)$, update its capacity to
\begin{equation*}
U^{(j+1)}(e) = \left\{
\begin{aligned}
&0, ~~~~~~~~~~~~~~~~~~~~~\text{ if }U^{(j)}(e) - F^{(j)}(e)\leq 1-\alpha\\
&U^{(j)}(e) - F^{(j)}(e),~\text{otherwise}\\
\end{aligned}
\right.
\end{equation*}
\end{enumerate}
At the end, let $F = \sum_{j\geq 0} F^{(j)}$ be the output. 

\medskip

\noindent{\textbf{Correctness.}} We now argue that $F$ is indeed a multi-commodity $(1/O(k))$-fractional $\alpha$-blocking flow. 

First, we show that $F$ is a feasible flow in $G$ w.r.t. $U$. It suffices to show that for each $j$, $F^{(j)}$ is feasible w.r.t. $U^{(j)}$. To see this, for each edge $e\in E(G)$, we have $\wtilde{F}^{(j)}(e)\leq U^{(j)}(e)$ since we have shown that $\wtilde{F}^{(j)}$ is feasible w.r.t. $U^{(j)}(e)$. Then, we have 
\[
F^{(j)}(e) = \sum_{1\leq i\leq k}\hat{F}^{(j)}_{i}(e)/2 \leq  (\wtilde{F}^{(j)}(e) + k/\mu)/2 \leq (U^{(j)}(e) + 1-\alpha)/2\leq U^{(j)}(e),
\]
where the last inequality is because each non-zero $U^{(j)}(e)$ is at least $1-\alpha$.

Second, $F$ is $\alpha$-blocking by the following reason. Let $\bar{j}$ be the last iteration such that $c^{(\bar{j})}_{\max} = 0$. Consider an arbitrary path $P\in\bigcup_{1\leq i\leq k}{\cal P}(S_{i},T_{i})$. Note that $\prod_{e\in P}U^{(j)}(e) = 0$ because $c^{(\bar{j})}_{\max} = 0$. It means there is an edge $e\in P$ such that $U^{(\bar{j})}(e) = 0$. Then, by the update rule of capacities, we have $U(e) - \sum_{j<\bar{j}}F^{(j)}(e)\leq 1-\alpha$, which implies $F(e)\geq U(e)-(1-\alpha)\geq \alpha\cdot U(e)$ (since $U(e)\geq 1$). Therefore, $P$ is $\alpha$-blocked by $F$. 

Third, $F$ is $1/(O(k))$-fractional because each $\hat{F}^{j}_{i}$ is $(1/\mu)$-fractional, each $F^{(j)}$ is $(1/(2\mu))$-fractional, and $\mu = 4k/(1-\alpha) = O(k)$.

\medskip

\noindent{\textbf{Work and Depth.}} It remains to argue that the algorithm above computes $F$ in required work and depth. It suffices to show that the number of iterations is at most $\tilde{O}(h)$. 

The key observation is that, in each iteration $j$, we have $F^{(j)}(e)\geq U^{(j)}(e)/8$ for each edge $e\in E(G)$ s.t. $c^{(j)}(e)\geq c^{(j)}_{\max}/2$. To see this, for each such edge $e$, we have $\wtilde{F}^{(j)}(e)\geq U^{(j)}(e)/2$ by the definition. Thus 
\[
\hat{F}^{(j)}(e) = \sum_{1\leq i\leq k}\hat{F}^{j}_{i}(e)\geq \wtilde{F}^{(j)}(e) - k/\mu\geq \wtilde{F}^{(j)}(e) - (1-\alpha)/4\geq U^{(j)}(e)/4,
\]
where the last inequality is because $U^{(j)}(e)\geq 1-\alpha$. Finally, $F^{(j)}(e) = \hat{F}^{(j)}(e)/2 \geq U^{(j)}(e)/8$.

Providing this observation, we show that $c_{\max}$ will be halved every $O(\log U_{\max})$ iterations. Formally speaking, for each $j_{1}$ such that $c_{\max}^{(j_{1})}>0$, we have $c^{(j_{2})}_{\max}\leq c^{(j_{1})}_{\max}/2$ where $j_{2} = j_{1} + \Theta(\log U_{\max})$. Assume the opposite for contradiction. Then some edge $e$ has $c^{(j)}(e)\geq c_{\max}^{(j_1)}/2$ during for each $j_{1}\leq j\leq j_{2}$. By the observation, every $j_{1}\leq j\leq j_{2}-1$ has $U^{(j+1)}(e)\leq 7 U^{(j)}(e)/8$, meaning that $U^{(j_{2})}(e) < 1-\alpha$. Hence $U^{(j_{2})}(e) = 0$ and $c^{(j_2)}(e) = 0$, contradicting that $c^{(j_2)}(e) \geq c^{(j_{1})}_{\max}/2$.

Therefore, after $O(h\log(n\cdot U_{\max}))\cdot O(\log U_{\max}) = \tilde{O}(h)$ iterations, $c_{\max}$ will drop below $(1-\alpha)^{h}$ (note that initially $c^{(0)}_{\max}\leq n^{O(h)}\cdot (U_{\max})^{h}$). When $c_{\max}<(1-\alpha)^{h}$, we have $c_{\max} = 0$ because the edge capacity is either at least $1-\alpha$ or $0$.

\end{proof}

\subsubsection{Computing Path Blockers: Proof of \Cref{thm:PathBlockers}}

\noindent{\textbf{The Algorithm.}} The algorithm is as follows.
\begin{enumerate}
\item Construct the length-weight expanded DAG $D^{(h,\lambda)}$.
\item Compute the path representation of a multi-commodity $(1/O(k))$-fractional $0.5$-blocking flow $F'$ in $D^{(h,\lambda)}$ by applying \Cref{lemma:ApproxBlockingFlow} and \Cref{lemma:FlowDecomposition}.
\item Initialize $F$ to be a zero flow. For each flow path $P'\in\path(F')$, project it to an original path $P$ in $G$ (according to the backward projection in \Cref{sect:LWExpandDAG}), and add $P$ to the flow $F$ (i.e. update $F(P)\gets F(P) + F'(P')/\kappa$), where $\kappa = O(h^{2}/\epsilon)$ is from the definition of $D^{(h,\lambda)}$. 
\item Return $F$. 

\end{enumerate}

\noindent{\textbf{Correctness.}} Now we show that $F$ is a $(1/\mu)$-fractional $h$-length $(1+\epsilon)$-lightest path $\alpha$-blocker for $\alpha = \Omega(\epsilon/h^{2})$ and $\mu = O(kh^{2}/\epsilon)$. First $F$ is $(1/\mu)$-fractional because $F'$ is $1/O(k)$-fractional and $\kappa = O(h^{2}/\epsilon)$.
Next, each path $P\in\path(F)$ has $w(P)\leq (1+2\epsilon)\lambda$ as we discussed when introducing the backward path projection. It is also clear that the number of flow paths of each commodity in the path representation of $F$ is $O(mh^{2}/\epsilon)$, because of \Cref{lemma:FlowDecomposition} and because $D^{(h,\lambda)}$ has $O(mh^{2}/\epsilon)$ edges.

It remains to show that, for each $P\in \bigcup_{1\leq i\leq k}{\cal P}_{h}(S_{i},T_{i})$ such that $w(P)\leq (1+\epsilon)\lambda$, there is an edge $e\in P$ such that $F(e)\geq \alpha\cdot U(e)$. Recall that the forward projection in \Cref{sect:LWExpandDAG} maps $P$ to an $S'_{i}$-$T'_{i}$ path $P'$ in $D^{(h,\lambda)}$. Because $F'$ is a $0.5$-blocking flow in $D^{(h,\lambda)}$, there exists an edge $e'\in P'$ such that $F'(e)\geq U'(e)/2 = U(e)/2$. By the definition of $F$, the original edge $e\in P$ corresponding to $e'$ has $F(e)\geq F'(e)/\kappa\geq U(e)/(2\kappa) = U(e)\cdot \Omega(\epsilon/h^{2})$ as desired.

\medskip

\noindent{\textbf{The Work and Depth.}} The work and depth are dominated by applying \Cref{lemma:ApproxBlockingFlow}. Because $D^{(h,\lambda)}$ has size $O(mh^{2}/\epsilon)$, the final work and depth are $\tilde{O}(mkh^{5}/\epsilon)$ and $\tilde{O}(h^{2})$.

\subsection{$(1+\epsilon)$-Approx $h$-Length Multi-Commodity Maxflows: Proof of \Cref{thm:ApproxLCMCMF}}

The algorithm is described in \Cref{algo:LCMCMF}. To simplify the analysis, \Cref{algo:LCMCMF} only outputs a $(1+100\epsilon)$-approximate solution. To obtain a $(1+\epsilon)$-approximate solution required by \Cref{thm:ApproxLCMCMF}, we can update $\epsilon \gets \epsilon/100$ at the very beginning.

\begin{algorithm}[H]
\caption{Multi-Commodity Length-Constrained Flows and Moving Cuts}
\label{algo:LCMCMF}
\begin{algorithmic}
\Require A directed graph $G$ with edge lengths $\ell$, capacities $U$, and source-sink pairs $\{(S_{i},T_{i})\mid 1\leq i\leq k\}$; length parameter $h$ and constant $\epsilon \in (0,1/100)$.
\Ensure $(1\pm 100\epsilon)$-approximate $h$-length multi-commodity flow $F$ and moving cuts $w$.
\Statex Let $\zeta = 1/\epsilon$ and $\eta = \epsilon^{2}/((1+10\epsilon)\ln m)$.
\Statex Initialize $w(e) = 1/(m^{\zeta}\cdot U(e))$ for all edges $e\in E(G)$.
\Statex Initialize $\lambda = 1/m^{\zeta}$.
\Statex Initialize $F$ to be a zero multi-commodity flow.
\Statex Initialize $w_{\best} \gets w/\lambda$.
\While {$\lambda < 1$}
\For{$\Theta(h^{3}\log n/\epsilon^{3})$ iterations}
\Statex \hspace{3em}Compute a  $h$-length $(1+\epsilon)$-lightest path $\Omega(\epsilon/h^{2})$-blocker $\hat{F}$ by applying \Cref{thm:PathBlockers} \\\hspace{6em}on $G$ with weights $w$ and parameter $\lambda$.
\Statex \hspace{3em}$F\gets F + \eta \cdot \hat{F}$.
\Statex \hspace{3em}For each $e\in E(G)$, $w(e)\gets w(e)\cdot (1+\epsilon)^{\hat{F}(e)/U(e)}$.
\Statex \hspace{3em}If $\sum_{e}w(e)U(e)/\lambda\leq \sum_{e}w_{\best}(e)U(e)$, update $w_{\best}\gets w/\lambda$.
\EndFor
\Statex $\lambda \gets \lambda \cdot (1+\epsilon)$.
\EndWhile
\Statex Return $(F,w_{\best})$.
\end{algorithmic}
\end{algorithm}

\noindent{\textbf{Correctness.}} The feasibility and approximation ratio are given by \Cref{lemma:FeasibleLCMCMF} and \Cref{lemma:OptLCMCMF} respectively. Furthermore, $F$ is $1/O(kh^{2}\log n/\epsilon^{3})$-fractional because each $\hat{F}$ is $1/O(kh^{2}/\epsilon)$-fractional guaranteed by \Cref{thm:PathBlockers} and $\eta = \Theta(\epsilon^{2}/\log n)$. For each commodity, the number of flow paths in $\path(F)$ is at most $O(mh^{5}\log n/\epsilon^{4})$, because each $\hat{F}$ has $O(mh^{2}\epsilon)$ flow paths for this commodity and $F$ is the summation of $\Theta(h^{3}\log n/\epsilon^{3})$ such $\hat{F}$ (after scaling).

\begin{lemma}
\label{lemma:Lambda}
It always hold that $\lambda \leq \min_{1\leq i\leq k} d^{(h)}_{w}(S_{i},T_{i})$.
\end{lemma}
\begin{proof}
We will prove this by induction. Initially, the statement trivially holds. Now consider a phase of the outer while loop. Let $\lambda_{0}$ be the $\lambda$ at the beginning of this phase. Assume $\lambda_{0}\leq \min_{1\leq i\leq k} d^{(h)}_{w}(S_{i},T_{i})$ at the beginning of this phase, we will show that $(1+\epsilon)\lambda_{0}\leq \min_{1\leq i\leq k}d^{(h)}_{w}(S_{i},T_{i})$ holds at the end of this phase.

Assume for contradiction that some path $P\in\bigcup_{i}{\cal P}_{h}(S_{i},T_{i})$ has $w(P)\leq (1+\epsilon)\lambda_{0}$ during the whole phase. Then, in each of the $\Theta(h^{3}\log n/\epsilon^{3})$ iterations of this phase, there is some edge $e\in P$ such that $\hat{F}(e)\geq \Omega(\epsilon/h^{2})\cdot U(e)$. By averaging over at most $h$ edges in $P$, there is a fixed edge $e\in P$ such that $\hat{F}(e)\geq \Omega(\epsilon/h^{2})\cdot U(e)$ in at least $\Theta(h^{2}\log n/\epsilon^{3})$ iterations. 
Therefore, at the end of this phase, 
\begin{align*}
    w(e) & \ge \frac{1}{m^{\zeta}\cdot U(e)}\cdot(1+\epsilon)^{\Omega(\epsilon/h^{2})\cdot \Theta(h^{2}\log n/\epsilon^{3})}\\
    & \ge \frac{1}{m^{\zeta}\cdot \poly(n)}\cdot  (1+\epsilon)^{\Theta(\log n/\epsilon^{2})}\\
    & \ge \frac{n^{\Theta(1/\epsilon)}}{m^{1/\epsilon}\cdot\poly(n)}\\
    & \ge 1+ \epsilon
\end{align*}
where we use $U(e) \le U_{\max} \le \poly(n)$ in the second inequality and
the last inequality is because we set the hidden constant behind $\Theta$ 
to be sufficiently large. This contradicts $w(P)\leq (1+\epsilon)\lambda_{0}$ since $\lambda_{0}<1$.
\end{proof}

\begin{lemma}
\label{lemma:FeasibleLCMCMF}
The output $(F,w_{\best})$ is feasible.
\end{lemma}
\begin{proof}
First, $w_{\best}$ is feasible because $w_{\best}$ is $w/\lambda$ at some moment, and $\min_{1\leq i\leq k}d^{(h)}_{w}(S_{i},T_{i})\geq \lambda$ always holds by \Cref{lemma:Lambda}.

To see that $F$ is feasible, for each edge $e\in E(G)$ we prove $F(e)\leq U(e)$. By the update rule of $w$, we have $w(e) \ge 1/(m^{\zeta}\cdot U_{\max})\cdot (1+\epsilon)^{F(e)/(\eta\cdot U(e))}$. Also, observe that $w(e)\leq (1+2\epsilon)(1+\epsilon)$ because right before the last time $w(e)$ is updated, $w(e)\leq w(P)\leq (1+2\epsilon)\lambda\leq 1+2\epsilon$. Therefore, we have
\begin{align*}
    \frac{1}{m^{\zeta}\cdot \poly(n)}\cdot(1+\epsilon)^{F(e)/(\eta\cdot U(e))}&\leq (1+2\epsilon)(1+\epsilon)\\
    \frac{1}{m^{\zeta}\cdot \poly(n)}\cdot \exp(\frac{F(e)}{U(e)}\cdot\frac{(1+10\epsilon)\ln m}{\epsilon^2} \cdot\ln(1+\epsilon)) & \le O(1)\\
    m^{((F(e)/U(e))(1+10\epsilon)-1)/\epsilon} & \le \poly(n)\\
\end{align*}
Suppose $F(e)/U(e) > 1$, then we can find small enough $\epsilon$ to make the above inequality invalid.
Thus we have $F(e) \le U(e)$.
\end{proof}

\begin{lemma}
\label{lemma:OptLCMCMF}
The output $(F,w_{\best})$ satisfies $(1-100 \epsilon)\sum_{e\in E(G)}w_{\best}(e)\cdot U(e)\leq \vvalue(F)$.
\end{lemma}
\begin{proof}
Let $\beta = \sum_{e}w_{\best}(e)U(e)$ be the value of the best dual solution $w_{\best}$. Consider the $i$-th iteration of the inner loop (over the course of the whole algorithm), let $\lambda^{(i)}$ and $w^{(i)}$ denote the variable $\lambda$ and function $w$ at the beginning of this iteration. Let $D^{(i)} = \sum_{e\in E(G)}w^{(i)}(e)\cdot U(e)$. Note that $D^{(i)}/\lambda^{(i)} \geq \beta$ by the definition of $w_{\best}$. Let $\hat{F}_{i}$ denote the path blocker $\hat{F}$ computed in this iteration. We have
\begin{align*}
D^{(i+1)} &= \sum_{e\in E(G)} w^{(i)}(e)\cdot (1+\epsilon)^{\hat{F}_{i}(e)/U(e)}\cdot U(e)\\
&\leq \sum_{e\in E(G)} w^{(i)}(e)\cdot U(e)\cdot (1+\epsilon\cdot \hat{F}_{i}(e)/U(e))\\
&= D^{(i)} + \epsilon \cdot \sum_{e\in E(G)}w^{(i)}(e)\cdot \hat{F}_{i}(e)\\
&\leq D^{(i)} + \epsilon\cdot (1+2\epsilon)\lambda^{(i)}\cdot \vvalue(\hat{F}_{i})\\
&\leq D^{(i)} + \epsilon\cdot (1+2\epsilon)\cdot \vvalue(\hat{F}_{i})\cdot\frac{D^{(i)}}{\beta}\\
&\leq D^{(i)}\cdot \exp\left(\frac{\epsilon\cdot (1+2\epsilon)\cdot \vvalue(\hat{F}_{i})}{\beta}\right)
\end{align*}

Let $T-1$ be the index of the last iteration. Because $F = \eta\sum_{i}\hat{F}_{i}$, we have
\[
D^{(T)}\leq D^{(1)}\cdot \exp\left(\frac{\epsilon\cdot (1+2\epsilon)\cdot \vvalue(F)}{\eta\cdot \beta}\right).
\]
Note that $D^{(1)} = \sum_{e}w^{(1)}(e)U(e)= m^{1-\zeta}$, and $D^{(T)}\geq 1$. 
Therefore,
\begin{align*}
m^{1-\zeta}&\leq\exp\left(\frac{\epsilon\cdot (1+2\epsilon)\cdot \vvalue(F)}{\eta\cdot \beta}\right)\\
(\zeta - 1)\ln m &\leq \frac{\epsilon\cdot (1+2\epsilon)\cdot \vvalue(F)}{\eta\cdot \beta}\\
\frac{1-\epsilon}{(1+2\epsilon)(1+10\epsilon)}\cdot \beta&\leq \vvalue(F)\\
\frac{\beta}{1+100\epsilon}&\leq \vvalue(F),
\end{align*}
as desired.
\end{proof}

\noindent{\textbf{The Work and Depth.}} Basically the algorithm just applies \Cref{thm:PathBlockers} $\Theta(h^{3}\log n/\epsilon^{3})$ times, so the total work and depth are $\tilde{O}(mkh^{8}/\epsilon^{4})$ and $\tilde{O}(h^{5}/\epsilon^{3})$.
\section{Approximate Mincost Flows via Flow Boosting}
\label{sect:MincostFlowByBoosting}

In this section, we finally get parallel $(1+\epsilon)$-approximate mincost flows, for both the concurrent and non-concurrent versions. We first formally define the problems in \Cref{def:concurrent-flow} and \Cref{def:nonconcurrent-flow}, which are from \cite{DBLP:conf/stoc/HaeuplerH0RS24}. The main result is given by \Cref{thm:MainResult}.

\begin{definition}[Concurrent Flow Problem]\label{def:concurrent-flow}
    Let $G$ be a connected graph with edge capacities $u \geq 1$ and $D : V \times V \mapsto \mathbb{N}$ an integral demand. Let $k = |\supp(D)|$. The \textit{$k$-commodity concurrent flow problem} asks to find a capacity-respecting flow $F$ routing $\lambda D$ for maximum $\lambda$. An algorithm is a $C$-approximation for the concurrent flow problem if it always produces a capacity-respecting flow $F$ routing $\lambda D$, where $\lambda \geq  \lambda^{*}/C$ and $\lambda^{*}$ is the maximum value for which there exists a capacity-respecting $F^{*}$ routing $\lambda^{*} D$.

    In the concurrent flow problem with costs, each edge has a cost $b \geq 0$ and there is a total cost budget $B$: the produced flow $F$ must additionally have total cost $\sum_{P} F(P) \sum_{e \in P} b(e) \leq B$. An algorithm is a $C$-approximation for the concurrent flow problem with costs if it always produces a capacity-respecting flow $F$ of total cost at most $B$ routing $\lambda D$, where $\lambda \geq \lambda^{*}/C$ and $\lambda^{*}$ is the maximum value for which there exists a capacity-respecting $F^{*}$ of total cost at most $B$ routing $\lambda^{*} D$.
\end{definition}

\begin{definition}[Non-Concurrent Flow Problem]\label{def:nonconcurrent-flow}
    Let $G$ be a connected graph with edge capacities $u \geq 1$ and $S$ a set of vertex pairs. Let $k = |S|$. The \textit{$k$-commodity non-concurrent flow problem} asks to find a capacity-respecting flow $F$ routing flow between vertex pairs in $S$, i.e. $\supp(D_{F}) \subseteq S$ of maximum value. An algorithm is a $C$-approximation for the non-concurrent flow problem if it always produces a capacity-respecting flow $F$ routing flow between vertex pairs in $S$ of value $\val(F) \geq \val(F^{*})/C$, where $F^{*}$ is the maximum-value capacity-respecting flow routing flow between vertex pairs in $S$.

    In the non-concurrent flow problem with costs, each edge has a cost $b \geq 0$ and there is a total cost budget $B$: the produced flow $F$ must additionally have total cost $\sum_{P} F(P) \sum_{e \in P} b(e) \leq B$. An algorithm is a $C$-approximation for the concurrent flow problem with costs if it always produces a capacity-respecting flow $F$ of total cost at most $B$ routing flow between vertex pairs in $S$ of value $\val(F) \geq \val(F^{*})/C$, where $F^{*}$ is the maximum-value capacity-respecting flow of total cost at most $B$ routing flow between vertex pairs in $S$.
\end{definition}

\begin{restatable}{theorem}{MainResult}
\label{thm:MainResult}
Let $G$ be an undirected graph with vertex costs and capacities. Given a precision parameter $1/\poly\log n\leq \epsilon\leq 1$, there are $(1+\epsilon)$-approximate algorithms for both the $k$-commodity concurrent and non-concurrent mincost flow problems with work $\hat{O}(mk)$ and depth $\hat{O}(1)$. The flow solution is outputted explicitly in its edge representation.
\end{restatable}

\begin{remark}
\label{remark:deterministic}
We note that our algorithms are randomized and correct with high probability, but they have great potential to be deterministic. The only randomized components are the parallel neighborhood cover algorithm \Cref{lemma:NeighborhoodCover} and the parallel length-constrained vertex expander decomposition algorithm \Cref{thm:vertexLC-ED}. The latter is randomized only because the parallel neighborhood cover with separation \Cref{thm:neicov}. It seems that there have been deterministic results on parallel neighborhood covers \cite{DBLP:conf/focs/RozhonEGH22,DBLP:conf/stoc/RozhonGHZL22}, but we have not looked into them carefully.
\end{remark}

The high-level approach is to use the flow boosting technique \cite{DBLP:journals/siamcomp/GargK07,DBLP:conf/stoc/HaeuplerH0RS24}. In \Cref{sect:FlowBoosting}, we state the flow boosting template. In \Cref{sect:LowStepFlow}, we show how to compute bi-criteria approximate low-step flows with $(1+\eps)$ total-length slack and moderate congestion, using the approximate length-constrained maxflow algorithm \Cref{thm:ApproxLCMCMF}. In \Cref{sect:ApproxMTLFlow}, using the LC-flow shortcut \Cref{thm:emulator}, we can compute the bi-criteria approximate flow on the shortcut graph and map it back the original graph. Lastly, in \Cref{sect:ProofBootstrapQuality} we complete the proof of \Cref{thm:MainResult} by plugging the subroutine of \Cref{sect:ApproxMTLFlow} into the flow boosting template.

\subsection{Section Preliminaries}

\begin{lemma}[Parallel Mincost Flow Rounding \cite{DBLP:journals/siamcomp/Cohen95}]
\label{lemma:MincostFlowRounding}
Let $G$ be a directed graph with edge lengths and capacities. Let $S,T\subseteq V(G)$ be source and sink vertices. Given the edge representation of a fractional single-commodity flow $\wtilde{F}$ in $G$, there is an algorithm that computes the edge representation of an integral single-commodity flow $\hat{F}$ in $G$ satisfying that
\begin{enumerate}
\item\label{item:MincostFlowRounding1} $\totlen(\hat{F})\leq \totlen(\wtilde{F}) + 1/\poly(n)$,
\item\label{item:MincostFlowRounding2} for each edge $e\in E(G)$, $\lfloor F(e)\rfloor\leq \hat{F}(e)\leq \lceil F(e)\rceil$.
\end{enumerate}
The work is $\tilde{O}(m)$ and the depth is $\tilde{O}(1)$.
\end{lemma}

\begin{corollary}
\label{coro:MCMCFlowRounding}
Let $G$ be a directed graph with edge lengths and capacities, and let $\mu\geq 1$ be a given round parameter which is some power of $2$. Given a (multi-commodity) flow $\wtilde{F}$ in $G$, there is an algorithm that computes the edge representation of an $(1/\mu)$-fractional flow $\hat{F}$ satisfying the following.
\begin{enumerate}
\item\label{item:MCMCFlowRounding1} Let $\wtilde{D}$ and $\hat{D}$ denote $\Dem(\wtilde{F})$ and $\Dem(\hat{F})$ respectively. For each $(u,v)\in\wtilde{D}$, $\lfloor \wtilde{D}(u,v)\cdot \mu\rfloor/\mu\leq \hat{D}(u,v)\leq \lceil 
\wtilde{D}(u,v)\cdot \mu \rceil/\mu$.
\item\label{item:MCMCFlowRounding2} $\totlen(\hat{F})\leq \totlen(\wtilde{F}) + 1/\poly(n)$.
\item\label{item:MCMCFlowRounding3} For each edge $e\in E(G)$, $\wtilde{F}(e) - |\supp(D)|/\mu\leq \hat{F}(e)\leq \wtilde{F}(e) + |\supp(D)|/\mu$.
\end{enumerate}
The work is $\tilde{O}(m\cdot|\supp(D)|)$ and the depth is $\tilde{O}(1)$.
\end{corollary}
\begin{proof}
Roughly speaking, we just apply \Cref{lemma:MincostFlowRounding} on different commodities in parallel. Concretely, let $\wtilde{F}^{(i)}$ be the commodity-$i$ flow of $\wtilde{F}^{(i)}$ corresponding to the source-sink pair $(u,v)\in \supp(\wtilde{D})$. First, we create a graph $G^{(i)}$ by adding artificial source $u'$ and sink $v'$, and adding two edges $(u',u),(v,v')$ with length $0$. Second, augmenting $\wtilde{F}^{(i)}$ to be a $u'$-to-$v'$ single-commodity flow on $G^{(i)}$ by putting $\vvalue(\wtilde{F}^{(i)}) = \wtilde{D}(u,v)$ flow units on edges $(u',u)$ and $(v,v')$. Third, round $\wtilde{F}^{(i)}$ to a $(1/\mu)$-fractional flow $\hat{F}^{(i)}$ by applying \Cref{lemma:MincostFlowRounding} on $\mu\cdot \wtilde{F}^{(i)}$ and then scaling down the output by $\mu$.

\Cref{item:MCMCFlowRounding1,item:MCMCFlowRounding3} in \Cref{coro:MCMCFlowRounding} come from \Cref{item:MincostFlowRounding1} in \Cref{lemma:MincostFlowRounding}, and \Cref{item:MCMCFlowRounding2} in \Cref{coro:MCMCFlowRounding} comes from \Cref{item:MincostFlowRounding2} in \Cref{lemma:MincostFlowRounding}. The work and depth are obvious.
\end{proof}

\subsection{The Flow Boosting Template}
\label{sect:FlowBoosting}

The statement and the proof of the following flow boosting template are almost identical to those in \cite{DBLP:conf/stoc/HaeuplerH0RS24}, except that we transform them into undirected graphs with vertex lengths, costs and capacities.

\begin{lemma}
\label{lemma:FlowBoosting}
Let $G$ be an undirected graph with vertex capacity function $U$ and vertex cost function $b$, and let $B\geq 0$ be the cost budget. Let ${\cal F}$ be any convex set of flows in $G$ containing at least one capacity-respecting flow, and let $\epsilon,s,\kappa\geq 0$ be parameters. Suppose there is an algorithm ${\cal A}$ that, given any integral vertex length function $\ell:V\to\{1,2,...,O(m^{1/\epsilon}N/\epsilon)\}$, computes the edge representation of a (non necessarily capacity-respecting) flow $F\in{\cal F}$ such that
\begin{itemize}
\item Length slack $s$: $\sum_{v\in V(G)} F(v)\cdot \ell(v)\leq s\cdot\sum_{v\in V(G)}F^{*}(v)\cdot \ell(v)$ for any flow $F^{*}\in{\cal F}$ that is also capacity-respecting, and
\item Congestion slack $\kappa$: $F(v)\leq \kappa\cdot U(v)$ for all $v\in V(G)$, i.e., $F(v)/\kappa$ is capacity-respecting.
\end{itemize}
Then, there is an algorithm that makes $O(\kappa \epsilon^{-2}\log^{2}n)$ calls to the algorithm ${\cal A}$ and outputs the edge representation of a flow $\bar{F}\in{\cal F}$ and a scalar $\lambda\geq 0$ such that
\begin{enumerate}
\item Feasibility: The flow $\lambda\bar{F}$ is capacity-respecting with cost at most $b$, and \label{item:flowboostingFeasibility}
\item Approximation factor $\approx s$: Let $\lambda^{*}$ be the maximum value such that there exists flow $F^{*}\in {\cal F}$ where $\lambda^{*}F^{*}$ is capacity-respecting with cost at most $B$. Then $\lambda\geq \frac{1-O(\epsilon)}{s}\lambda^{*}$.\label{item:flowboostingApx}
\end{enumerate}
Moreover, the flow $\bar{F}$ is a convex combination of the flows returned by the algorithm ${\cal A}$, and this convex combination can be outputted as well. The algorithm takes $\tilde{O}(\kappa\epsilon^{-2}m)$ work and $\tilde{O}(\kappa\epsilon^{-2})$ depth besides calling ${\cal A}$.
\end{lemma}

The proof is from \cite{DBLP:conf/stoc/HaeuplerH0RS24}, and we append it here for completeness.
Let $\cK$ be the set of capacity-respecting flows in $G$, and consider the following flow LP of the graph $G$. We have a variable $x(F)\ge0$ for each $F\in\mathcal F\cap\mathcal K$ indicating that we send flow $F$ scaled by $x(F)$. To avoid clutter, we also define $b(F)=\sum_{v\in V}F(v)\cdot b(v)$ as the cost of the flow $F$.
\begin{align*}
\max \qquad & \sum_{F\in\mathcal F\cap\mathcal K} x(F)
\\\text{s.t.} \qquad & \sum_{F\in\mathcal F\cap\mathcal K} F(v) \cdot x(F) \le U(v) & \forall v\in V
\\ & \sum_{F\in\mathcal F\cap\mathcal K}b(F)\cdot x(F)\le B
\\ & x\ge0
\end{align*}
Let $\beta$ be the optimal value of this LP. Note that since $\mathcal F\cap\mathcal K$ is convex, there is an optimal solution with $x(F^*)=\beta$ for some $F^*\in\mathcal F$ and $x(F)=0$ elsewhere. It follows that $\beta=\lambda^*$.

The dual LP has a variable $y(v)\ge0$ for each vertex $v\in V$ as well as a variable $\phi\ge0$ of the cost constraint.
\begin{align*}
\min \qquad & \sum_{v\in V} U(v)\cdot y(v) + B\cdot\phi & =:D(y,\phi)
\\\text{s.t.} \qquad & \sum_{v\in V} F(v) \cdot(y(v)+b(v)\phi)\ge 1 & \forall F\in\mathcal F\cap\mathcal K
\\ & y\ge0,\,\phi\ge0
\end{align*}
Let $D(y,\phi)=\sum_{v\in V} U(v)\ell(v)+B\cdot\phi$ be the objective value of the dual LP. 
We note that for any vertex $v \in V$, we interpret the value $y(v) + b(v)\phi$ as a length over $v$,
and we further define $\alpha(y,\phi)$ as the minimum length taken over any flow $F\in\mathcal F\cap\mathcal K$ under length function $y+\phi b$:
\[ \alpha(\ell,\phi)=\min_{F\in\mathcal F\cap\mathcal K}\sum_{v\in V}F(v)\cdot(\ell(v)+b(v)\phi) .\]
We note that for any feasible solution $\{y_{0},\phi_{0}\}$, we have $\{\frac{y_{0}}{\alpha(y_{0}, \phi_{0})}, \frac{\phi_{0}}{\alpha(y_{0}, \phi_{0})}\}$ is also a feasible solution.
Thus by scaling, we can restate the dual LP as finding a length function $y+\phi b$ minimizing $D(y,\phi)/\alpha(y,\phi)$. 
By LP duality, the minimum is at least $\beta$, the optimal value of the primal LP, which we recall also equals $\lambda^*\ge1$.

\paragraph{Algorithm.} The algorithm initializes length functions $y^{(0)}(v)=\delta/U(v)$ and $\phi^{(0)}=\delta/B$ for parameter $\delta=n^{-1/\epsilon}$, and proceeds for a number of iterations. For each iteration $i$, the algorithm wishes to call the oracle $\mathcal O$ on length function $y^{(i-1)}+\phi^{(i-1)}b$, but the length function $y^{(i-1)}+\phi^{(i-1)}b$ is not integral. However, we will ensure that they are always in the range $[\delta/N,O(1)]$. So the algorithm first multiplies each length by $O(N/(\delta\epsilon))=O(n^{1/\epsilon}N/\epsilon)$ and then rounds the weights to integers so that each length is increased by at most roughly the same factor up to $(1+\epsilon)$ after scaled back.
The algorithm calls the oracle on these scaled, integral weights to obtain a flow $F^{(i)}$. On the original, unscaled graph, the flow satisfies the following two properties guaranteed by the oracle $\cO$:
 \begin{enumerate}
 \item \textup{Length slack $(1+\epsilon)s$:} $\sum_{v\in V}F^{(i)}(v)\cdot(y^{(i-1)}(v)+b(v)\phi^{(i-1)})\le(1+\epsilon)s\cdot\alpha(y^{(i-1)},\phi^{(i-1)})$, and
 \item \textup{Congestion slack $\kappa$:} $F(v)\le\kappa U(v)$ for all $v\in V$, i.e., $F/\kappa\in\mathcal K$.
 \end{enumerate}

Define $z^{(i)}=\min\{1,\,B/b(F^{(i)})\}$ so that $b(z^{(i)}F^{(i)})\le B$, i.e., the cost of the scaled flow $z^{(i)}F^{(i)}$ is within the budget $B$.
The lengths are then updated as
\[ y^{(i)}(v)=y^{(i-1)}(v)\bigg(1+\frac\epsilon\kappa\cdot\frac{z^{(i)}F^{(i)}(v)}{U(v)}\bigg) \quad\text{and}\quad \phi^{(i)}=\phi^{(i-1)}\bigg(1+\frac\epsilon\kappa\cdot\frac{b(z^{(i)}F^{(i)})}B\bigg) . \]
This concludes the description of a single iteration. The algorithm terminates upon reaching the first iteration $t$ for which $D(t)\ge1$ and outputs
\[ \bar F=\frac{z^{(1)}F^{(1)}+z^{(2)}F^{(2)}+\cdots+z^{(t-1)}F^{(t-1)}}{z^{(1)}+z^{(2)}+\cdots+z^{(t-1)}}\quad\text{and}\quad\lambda=\frac{z^{(1)}+z^{(2)}+\cdots+z^{(t-1)}}{\kappa\log_{1+\epsilon}1/\delta} .\]

\paragraph{Correctness Analysis.} We first show the feasibility of the flow $\bar{F}$.
\begin{lemma}
    The scaled down flow $\lambda\cdot \bar{F} = \frac1{\kappa\log_{1+\epsilon}1/\delta}(z^{(1)}F^{(1)}+z^{(2)}F^{(2)}+\cdots+z^{(t-1)}F^{(t-1)})$ is capacity-respecting with cost at most $B$.
\end{lemma}
\begin{proof}
    Consider any vertex $v$. On each iteration, we route $z^{(i)}F^{(i)}(v)\le\kappa U(v)$ units of flow through $v$ and increase its corresponding variable $y(v)$ by a factor $(1+\frac\epsilon\kappa\cdot\frac{z^{(i)}F^{(i)}(v)}{U(v)})\le1+\epsilon$. So for every $\kappa U(v)$ units of flow routed through $v$ over the iterations, we increase its variable by at least a factor $1+\epsilon$. 
    Initially, the value of its variable is $\delta/U(v)$, and after $t-1$ iterations, it satisfies $y^{(t-1)}(v)\le 1/U(v)$ because otherwise $D(y^{(t-1)},\phi^{(t-1)}) > 1$ which leads to an early termination. Therefore the total amount of flow through $v$ in the first $t-1$ phases is strictly less than $\kappa \log_{1+\epsilon}\frac{1/U(v)}{\delta/U(v)}\cdot U(v)=\kappa\log_{1+\epsilon}1/\delta \cdot U(v)$. Scaling the flow down by $\kappa\log_{1+\epsilon}1/\delta$, we obtain a capacity-respecting flow.

    We can similarly show that the scaled down flow has cost at most $B$. 
    On each iteration, we route a flow of cost $b(z^{(i)}F^{(i)})\le B$ and increase the variable $\phi^{(i)}$ by a factor $(1+\frac\epsilon\kappa\cdot\frac{b(z^{(i)}F^{(i)})}B)\le1+\epsilon/\kappa$ over the previous variable $\phi^{(i-1)}$. So for every $B$ cost of flow routed, we increase the variable $\phi^{(i)}$ by at least a factor $1+\epsilon/\kappa$. And for every $\kappa B$ cost of flow routed, the length increases by at least a factor $(1+\epsilon/\kappa)^\kappa\ge1+\epsilon$. Initially, $\phi^{(0)}=\delta/B$, and after $t-1$ iterations, we similarly have $\phi^{(t-1)}\le 1/B$. Therefore the total cost of flow in the first $t-1$ phases is strictly less than $\kappa B\log_{1+\epsilon}\frac{1/B}{\delta/B}=\kappa\log_{1+\epsilon}(1/\delta) \cdot B$. Scaling the flow down by $\kappa\log_{1+\epsilon}1/\delta$, we obtain a flow with cost at most $B$.
\end{proof}
Then we argue about the approximation factor for $\bar{F}$.
We will analyze the values of $D(y^{(i)},\phi^{(i)})$ and $\alpha(y^{(i)},\phi^{(i)})$ for the variables $y^{(i)},\phi^{(i)}$. To avoid clutter, we denote $D(i)=D(y^{(i)},\phi^{(i)})$ and $\alpha(i)=\alpha(y^{(i)},\phi^{(i)})$.  For each iteration $i$, we have
\begin{align*}
D(i)&=\sum_{v\in V} U(v)\cdot y^{(i)}(v)+B\cdot\phi^{(i)}
\\&=\sum_{v\in V}U(v)\cdot y^{(i-1)}(v)\bigg(1+\frac\epsilon\kappa\cdot\frac{z^{(i)}F^{(i)}(v)}{U(v)}\bigg)+B\cdot\phi^{(i-1)}\bigg(1+\frac\epsilon\kappa\cdot\frac{b(z^{(i)}F^{(i)})}B\bigg)
\\&=D(i-1)+\frac\epsilon\kappa\cdot\sum_{v\in V}z^{(i)}F^{(i)}(v)\cdot y^{(i-1)}(v)+\frac\epsilon\kappa\cdot z^{(i)}\cdot b(F^{(i)})\cdot\phi^{(i-1)}
\\&=D(i-1)+\frac\epsilon\kappa\cdot z^{(i)}\bigg(\sum_{v\in V}F^{(i)}(v)\cdot(y^{(i-1)}(v)+b(v)\phi^{(i-1)})\bigg)
\\&\le D(i-1)+\frac\epsilon\kappa\cdot z^{(i)}\cdot(1+\epsilon)s\cdot \alpha(i-1),
\end{align*}
where in the last inequality we use the length slack property of the flow from the oracle $\cO$.
Recall that $\beta$ is the optimal value of the primal LP problem, then we have $D(i-1)/\alpha(i-1)\ge\beta$. 
Further
\[ D(i) \le D(i-1)+\frac\epsilon\kappa\cdot z^{(i)}\cdot(1+\epsilon)s\cdot \frac{D(i-1)}\beta .\]
Define $\epsilon'=\epsilon(1+\epsilon)s/\kappa$ so that
\[ D(i)\le\bigg(1+\frac{\epsilon(1+\epsilon)sz^{(i)}}{\kappa\beta}\bigg)D(i-1) = \bigg(1+\frac{\epsilon'z^{(i)}}\beta\bigg)D(i-1) .\]
Since $D(0)=(n+1)\delta\le m\delta$, we have for $i\ge1$
\begin{align*}
D(i)\le\bigg(\prod_{j\le i}(1+\epsilon'z^{(j)}/\beta)\bigg)m\delta&=\bigg(1+\frac{\epsilon'z^{(i)}}\beta\bigg)m\delta\prod_{j\le i-1}\bigg(1+\frac{\epsilon'z^{(j)}}\beta\bigg)
\\&\le(1+\epsilon')m\delta\exp\bigg(\frac{\epsilon'\sum_{j\le i-1}z^{(j)}}\beta\bigg) ,
\end{align*}
where the last inequality uses our assumption that $\beta\ge1$ and the fact that $z^{(j)}\le1$ by definition. To avoid clutter, define $z^{(\le i)}=\sum_{j\le i}z^{(j)}$ for all $i$.

The procedure stops at the first iteration $t$ for which $D(t)\ge1$. Therefore,
\[ 1\le D(t)\le(1+\epsilon')m\delta\exp\bigg(\frac{\epsilon'z^{(\le t-1)}}\beta\bigg) ,\]
which implies
\begin{gather}
\frac\beta{z^{(\le t-1)}}\le\frac{\epsilon'}{\ln\frac1{(1+\epsilon')m\delta}} .\label{eq:to-establish-condition-2}
\end{gather}
We further note the fact that $\delta=n^{-1/\epsilon}$ to obtain
\[ \frac\lambda{\lambda^*}=\frac{z^{(\le t-1)}}{\kappa\log_{1+\epsilon}1/\delta}\cdot\frac1\beta\ge \frac{\ln\frac1{(1+\epsilon')m\delta}}{\epsilon'\cdot\kappa\log_{1+\epsilon}1/\delta} = \frac{\ln\frac1{(1+\epsilon')m\delta}}{\epsilon(1+\epsilon) s\log_{1+\epsilon}1/\delta} \ge \frac{1-O(\epsilon)}s .\]

\paragraph{Running time.}
Recall from above that $\lambda\bar F$ is capacity-respecting with cost at most $B$, so $\lambda\le\beta$. Since $\lambda=\frac{z^{(\le t-1)}}{\kappa\log_{1+\epsilon}1/\delta}$, we obtain $z^{(\le t-1)}\le\beta\kappa\log_{1+\epsilon}1/\delta$. On each iteration $i\le t-1$, either $z^{(i)}=1$ or $z^{(i)}<1$, and the latter case implies that $b(z^{(i)}F^{(i)})=B$, which means $\phi^{(i)}=\phi^{(i-1)}(1+\epsilon/\kappa)$. Initially, $\phi^{(0)}=\delta/B$, and after $t-1$ iterations, since $D(t-1)<1$, we have $\phi^{(t-1)}\le D(t-1)/B<1/B$. So the event $\phi^{(i)}=\phi^{(i-1)}(1+\epsilon/\kappa)$ can happen at most $\log_{1+\epsilon/\kappa}1/\delta$ times. It follows that $z^{(i)}<1$ for at most $\log_{1+\epsilon/\kappa}1/\delta$ values of $i\le t-1$. Since $z^{(\le t-1)}\le\beta\kappa\log_{1+\epsilon}1/\delta$, we have $z^{(i)}=1$ for at most $\beta\kappa\log_{1+\epsilon}1/\delta$ many values of $i\le t-1$. Therefore, the number of iterations $t$ is at most
\[ \log_{1+\epsilon/\kappa}1/\delta+\beta\kappa\log_{1+\epsilon}1/\delta+1=O(\beta\kappa\epsilon^{-2}\log m) .\]
By scaling all vertex capacities and costs by various powers of two, we can ensure that $\beta\in[1,2]$ on at least one guess, so the number of iterations is $O(\kappa\log_{1+\epsilon}1/\delta)=O(\kappa\epsilon^{-2}\log m)$. Doing so also ensures that $\lambda^*=\beta\ge1$ as we had previously assumed. For incorrect guesses, we terminate the algorithm above after $O(\kappa\epsilon^{-2}\log m)$ iterations to not waste further computation. Among all guesses, we take the one with maximum $\lambda$ that satisfies feasibility. Since there are $O(\log n)$ relevant powers of two, the running time picks up an overhead of $O(\log n)$.

\subsection{Approximate Low-Step Min-Total-Length Multi-Commodity Flows}
\label{sect:LowStepFlow}

In this subsection, we want to show \Cref{lemma:LowStepFlow}, but in fact we will prove a slightly more general version \Cref{lemma:DirectedLowStepFlow}. The proof of \Cref{lemma:LowStepFlow} is deferred to the end of this subsection after finishing proving \Cref{lemma:DirectedLowStepFlow}. For better understanding, we note that in \Cref{lemma:LowStepFlow}, choosing $\tau=1$ means we are in the non-concurrent setting, while choosing $\tau = |D|$ means the concurrent setting.

\begin{lemma}
\label{lemma:LowStepFlow}

Let $G$ be an undirected graph with positive integral vertex lengths $\ell$ and integral vertex capacities $U$. Given an integral demand $D$, step bound $t\geq 1$, a flow value threshold $\tau$ equal to either $1$ or $|D|$, suppose there exists a feasible flow $F^{*}$ partially routing $D$ with $\vvalue(F^{*}) = \tau$ and step $t$. Then further given a precision parameter $10/n\leq \epsilon<1$, there is an algorithm that computes the path representation of a flow $F$ satisfying the following.
\begin{enumerate}
\item $\Dem(F)\leq D$ and $\vvalue(F) = \tau$.
\item $F$ has congestion $O(\log^{2} n/\epsilon)$ in $G$.
\item $\totlen(F)\leq (1+\epsilon)^{5}\cdot\totlen(F^{*})$ for any feasible flow $F^{*}$ partially routing $D$ with $\vvalue(F^{*}) = \tau$ and step $t$.
\end{enumerate}
The work and depth are $\tilde{O}(mk\cdot t^{8}/\epsilon^{9})$ and $\tilde{O}(t^{5}/\epsilon^{6})$ respectively, where $k = |\supp(D)|$ is the number of commodities.
\end{lemma}

The intuition behind the proof of \Cref{lemma:DirectedLowStepFlow} is a simple greedy strategy, as mentioned in \Cref{Overview:LowStepFlow}. But there is a small technical issue. When we pick $\tau = |D|$, we are actually in the concurrent setting, while \Cref{thm:ApproxLCMCMF} we will use is an approximate algorithm in the non-concurrent setting. Roughly speaking, when we miss an $\epsilon$-fractional of total demands in the non-concurrent setting, we may miss some commodity entirely, which is not allowed in the concurrent setting. A simple fix is to iterate the non-concurrent approximate algorithm $\tilde{O}(\log n)$ times to reduce the missing fraction to $1/n$. This is what the second For loop in \Cref{algo:LowStepFlows} is doing.

\begin{lemma}
\label{lemma:DirectedLowStepFlow}
Let $G$ be a directed graph with positive integral edge lengths $\ell$ and  integral capacities $U$. Given an integral demand $D$, step bound $t\geq 1$, and a flow value threshold $1\leq \tau\leq |D|$, suppose there exists a feasible flow $F^{*}$ partially routing $D$ with $\vvalue(F^{*}) = \tau$ and step $t$. Then further given a precision parameter $10/n\leq \epsilon<1$, there is an algorithm that computes the path representation of a flow $F$ satisfying the following.
\begin{enumerate}
\item $\Dem(F)\leq D$ and $\vvalue(F) = \tau - 1/n$,
\item $F$ has congestion $O(\log^{2} n/\epsilon)$ in $G$,
\item $\totlen(F)\leq (1+\epsilon)^{4}\cdot\totlen(F^{*})$ for any feasible flow $F^{*}$ partially routing $D$ with $\vvalue(F^{*}) = \tau$ and step $t$.
\end{enumerate}
The work and depth are $\tilde{O}(mk\cdot t^{8}/\epsilon^{9})$ and $\tilde{O}(t^{5}/\epsilon^{6})$ respectively, where $k = |\supp(D)|$ is the number of commodities.
\end{lemma}
\begin{proof}
The algorithm is described in \Cref{algo:LowStepFlows}. Here we discuss in details the main steps on \Cref{line:ApproxLCflow} and \Cref{line:FlowRounding}. Let $\mu$ be the smallest power of $2$ at least $n^{4}$.

\begin{algorithm}[H]
\caption{Low-Step Multi-Commodity Flows: \textsc{LowStepFlows}$(G,D,t,\epsilon)$}
\label{algo:LowStepFlows}
\begin{algorithmic}[1]
\Require Graph $G$ with integral edge lengths $\ell$ and integral capacities $u$, integral demand $D$, step bound $t\geq 1$, and a tradeoff parameter $\epsilon$.
\Ensure A flow $F$ partially routing $D$ with $|D| - \vvalue(F)\leq 1/n$ and congestion $O(\log^{2} n/\epsilon)$, such that $\totlen(F)\leq (1+\epsilon)^{4}\cdot\totlen(F^{*})$ for any feasible flow $F^{*}$ routing $D$ with step $t$.
\State Let $\mu$ be the smallest power of $2$ such that $\mu\geq n^{4}$.
\State Initialize $F\gets 0$ and $D'\gets D$.
\For{$p$ from $0$ to $\bar{p} = \lceil\log_{1+\epsilon}(nN)\rceil$}
\State $h_{p} \gets (1+\epsilon)^{p}$.
\State Let $G'$ be $G$ with edge lengths $\ell'$ such that $\ell'(e) = \lceil \ell(e)/(\epsilon h_{p}/t)\rceil$
\For{$j$ from $1$ to $\bar{j}=\Theta(\log n)$}
\State\label{line:ApproxLCflow}Let $\wtilde{F}_{p,j}$ be a $1.5$-approximate $(t/\epsilon + t)$-length maxflow with $\Dem(\wtilde{F}_{p,j})\leq D'$ in $G'$.
\Statex\Comment{Apply \Cref{thm:ApproxLCMCMF}}
\State\label{line:TerminateJ}If $\vvalue(\wtilde{F}_{p,j})\leq 1/(2n)$, terminate the inner loop.
\State\label{line:FlowRounding}Round $\wtilde{F}_{p,j}$ to a $(1/\mu)$-fractional flow $\hat{F}_{p,j}$. \Comment{Apply \Cref{coro:MCMCFlowRounding}}
\State\label{line:lambda}Let $\lambda$ be the maximum value in $[0,1]$ such that $\vvalue(F + \lambda \hat{F}_{p,j})\leq \tau - 1/n$
\State\label{line:ScaleF}$F_{p,j}\gets \lambda \hat{F}_{p,j}$.
\State\label{line:UpdateF}$F\gets F + F_{p,j}$.
\State\label{line:TerminateAlgo}If $\lambda<1$, terminate the algorithm and return $F$.
\State\label{line:NewDPrime}$D'\gets D'-\Dem(\hat{F}_{p,j})$.
\EndFor
\EndFor
\State\label{line:Return}Return $F$.
\end{algorithmic}
\end{algorithm}

\noindent{\underline{\Cref{line:ApproxLCflow}.}} We will compute a feasible (fractional) flow $\wtilde{F}_{p,j}$ in $G'$ with $\Dem(\wtilde{F}_{p,j})\leq D'$ and length $t/\epsilon + t$ satisfying that $\vvalue(\wtilde{F}_{p,j})\geq \vvalue(\wtilde{F}^{*}_{p,j})/(1+1/2)$ for any feasible flow $\wtilde{F}^{*}_{p,j}$ in $G'$ with $\Dem(\wtilde{F}^{*}_{p,j})\leq D'$ and length $t/\epsilon + t$. To do this, 
\begin{itemize}
\item We first construct a graph $G''$ by starting from $G'$ and adding, for each demand pair $(u_{i},v_{i})\in \supp(D')$, two artificial vertices $s_{i},t_{i}$ and two artificial directed edges $(s_{i},u_{i}),(v_{i},t_{i})$ with capacity $D'(u_{i},v_{i})$ and length $1$.
\item We will see in a moment that \Cref{line:FlowRounding} guarantees that $D'$ is always $(1/\mu)$-fractional.
Hence, we apply \Cref{thm:ApproxLCMCMF} on $G''$ with capacities scaled up by $\mu$ (to ensure integral capacities), source-sink pairs $\{(s_{i},t_{i})\mid i\}$, length parameter $t/\epsilon + t + 2$ and approximation parameter $0.5$. Let $\wtilde{F}''_{p,j}$ be the output.
\item $\wtilde{F}_{p,j}$ is obtained by scaling down $\wtilde{F}''_{p,j}$ by $\mu$ and removing the flow on the artificial edges.
\end{itemize}

\noindent\underline{\Cref{line:FlowRounding}.} We simply apply \Cref{coro:MCMCFlowRounding} on $G'$ and $\wtilde{F}_{p,j}$ with rounding parameter $\mu$. The output $\hat{F}_{p,j}$ is a $(1/\mu)$-fractional flow on $G'$ satisfying that
\begin{enumerate}
\item $\lfloor\Dem(\wtilde{F}_{p,j})\cdot \mu\rfloor /\mu\leq \Dem(\hat{F}_{p,j})\leq \lceil\Dem(\wtilde{F}_{p,j})\mu\rceil/\mu\leq D'$ (since $D'$ is $(1/\mu)$-fractional), which implies $\vvalue(\hat{F}_{p,j})\geq \vvalue(\wtilde{F}_{p,j})-n^{2}/\mu\geq \vvalue(\wtilde{F}_{p,j})-1/n^{2}$
\item $\totlen(\hat{F}_{p,j})\leq \totlen(\wtilde{F}_{p,j}) + 1/\poly(n)$. 
\item $\hat{F}_{p,j}(e)\leq U(e) + n^{2}/\mu$ for each $e\in E(G')$, which implies that $\hat{F}_{p,j}$ has congestion $1+1/n^{2}$.
\end{enumerate}
Finally, because $\hat{F}_{p,j}$ is $(1/\mu)$-fractional, the new $D'$ on \Cref{line:NewDPrime} is also $(1/\mu)$-fractional inductively.

\medskip

\noindent{\textbf{Correctness.}} The congestion of $F$ is obviously $O(\log^{2}n/\epsilon)$ because there are $O(\log^{2}n/\epsilon)$ different pairs of $p,j$, and each $\hat{F}_{p,j}$ has congestion $O(1)$. Also, the algorithm directly gives $\Dem(F)\leq D$. It remains to show $\vvalue(F)= \tau - 1/n$, and $\totlen(F)\leq (1+\epsilon)\totlen(F^{*})$ for any feasible flow $F^{*}$ partially routing $D$ in $G$ with $\vvalue(F^{*}) = \tau$ and step $t$. 

We begin by setting up some notations. First, we emphasize that for each $0\leq p\leq \bar{p}$ and $1\leq j\leq \bar{j}$, we will regard $\wtilde{F}_{p,j}$ as a zero flow if it is not computed by the algorithm (same for $\hat{F}_{p,j}$ and $F_{p,j}$).
Let $F_{p} = \sum_{p'\leq p,j} F_{p',j}$ and similarly, let $F^{*}_{p}$ be the maximum $h_{p}$-length subflow of $F^{*}$. For convenience, we refer to $F_{-1}$ and $F^{*}_{-1}$ as zero flows. We first prove \Cref{fact:LowStepLengthRounding} and \Cref{lemma:Value}, which is useful for our argument later. 
\begin{fact}
\label{fact:LowStepLengthRounding}
For any $0\leq p\leq \bar{p}$, we have
\begin{enumerate}
\item\label{item:LowStep1} An $h_{p}$-length $t$-step flow $F$ in $G$ has length $t/\epsilon + t$ in $G'$,
\item\label{item:LowStep2} An $(t/\epsilon + t)$-length flow $F'$ in $G'$ has length $(1+\epsilon)h_{p}$ in $G$.
\end{enumerate}
\end{fact}
\begin{proof}
For \Cref{item:LowStep1}, consider any flow path $P\in\path(F)$. We know $P$ has step $t$ and length $\sum_{e\in P} \ell(e)\leq h_{p}$. Therefore, the length of $P$ in $G'$ is
\[
\sum_{e\in P}\ell'(e) = \sum_{e\in P}\lceil \ell(e)/(\epsilon h_{p}/t)\rceil\leq \sum_{e\in P}(\ell(e)/(\epsilon h_{p}/t)+1)\leq t/\epsilon + t.
\]

For \Cref{item:LowStep2}, consider any flow path $P'\in\path(F')$. The length of $P'$ in $G$ is $\sum_{e\in P'}\ell(e)\leq \sum_{e\in P'}(\epsilon h_{p}/t)\cdot \ell'(e)\leq (1+\epsilon)h_{p}$.
\end{proof}

\begin{lemma}
\label{lemma:Value}
For each $0\leq p\leq \bar{p}$, we have $\vvalue(F_{p})\geq \vvalue(F^{*}_{p})-1/n$.
\end{lemma}
\begin{proof}
Consider a fixed $p$. We set up some notations: for each $j\geq 0$, let $F_{p,\leq j} = F_{p-1} + \sum_{j'\leq j}F_{p,j'}$, and let $D'_{p,j}$ be the $D'$ at the beginning of iteration $j$.

\medskip

\noindent\underline{Case 1: terminate on \Cref{line:TerminateAlgo}.} Consider the case that the inner loop for some $p'\leq p$ terminates on \Cref{line:TerminateAlgo}. Then we directly have $\vvalue(F_{p}) = \vvalue(F_{p'}) = \tau - 1/n\geq \vvalue(F^{*}) - 1/n\geq \vvalue(F^{*}_{p}) - 1/n$ as desired.

\medskip

\noindent\underline{Case 2: terminate on \Cref{line:TerminateJ}.} Next, we consider the case that the inner loop terminates on \Cref{line:TerminateJ} in some iteration $j$. By the definition of $F^{*}_{p}$, there is an $h'_{p}$-length $t$-step feasible flow in $G$ routing $\min\{\Dem(F^{*}_{p}),D'_{p,j}\}$. By \Cref{fact:LowStepLengthRounding}, this flow is an $(t/\epsilon + t)$-length feasible flow in $G'$ routing $\min\{\Dem(F^{*}_{p}),D'_{p,j}\}$. Therefore, $\wtilde{F}_{p,j}$ has 
\[
\vvalue(\wtilde{F}_{p,j})\geq |\min\{\Dem(F^{*}_{p}),D'_{p,j}\}|/1.5\geq (\vvalue(F^{*}_{P}) - \vvalue(F_{p,\leq j-1}))/1.5.
\]
where the second inequality is because 
\begin{align*}
\min\{\Dem(F^{*}_{p}),D'_{p,j}\} &= \Dem(F^{*}_{p}) + D'_{p,j} - \max\{\Dem(F^{*}_{p}),D'_{p,j}\}\\
&\geq \Dem(F^{*}_{p}) + D'_{p,j} - D\\
&= \Dem(F^{*}_{p}) - \Dem(F_{p,\leq j-1}).
\end{align*}
Since $\vvalue(\wtilde{F}_{p,j})\leq 1/(2n)$ from \Cref{line:TerminateJ}, we have 
\[
\vvalue(F^{*}_{P}) - \vvalue(F_{p,\leq j-1})\leq 1/n
\]
as desired.

\medskip

\noindent\underline{Case 3: terminate after $\Theta(\log n)$ iterations.} Finally, we consider the case that the inner loop has $\Theta(\log n)$ iterations. In this case, we will show the following statement instead: for each $j\geq 0$, 
\[
\vvalue(F^{*}_{p}) - \vvalue(F_{p,\leq j})\leq \max\{\vvalue(F^{*}_{p})/2^{j},1/n\}.
\]
Providing this statement, when $j = \Theta(\log n) \geq \log_{2}(\vvalue(F^{*}_{p})\cdot n)$ (we choose a sufficiently large constant behind the $\Theta$ on \Cref{line:ApproxLCflow}), it gives $\vvalue(F^{*}_{p}) - \vvalue(F_{p})\leq 1/n$ as desired.

Now we prove this statement by induction. In the base case $j=0$, the statement holds trivially. Consider an inductive step $j\geq 1$, and assume the statement holds for $j-1$. We also assume $\vvalue(F^{*}_{p}) - \vvalue(F_{p,\leq j-1})>1/n$, otherwise this inductive step is trivial. 

On \Cref{line:ApproxLCflow}, an argument similar to Case 2 gives 
\[
\vvalue(\wtilde{F}_{p,j})\geq (\vvalue(F^{*}_{P}) - \vvalue(F_{p,\leq j-1}))/1.5
\]
On \Cref{line:FlowRounding}, the flow $\hat{F}_{p,j}$ has
\[
\vvalue(\hat{F}_{p,j})\geq \vvalue(\wtilde{F}_{p,j}) - 1/n^{2}\geq (\vvalue(F^{*}_{P}) - \vvalue(F_{p,\leq j-1}))/2
\]
where we use $\vvalue(F^{*}_{P}) - \vvalue(F_{p,\leq j-1})\geq 1/n$. Finally,
\begin{align*}
\vvalue(F^{*}_{p}) - \vvalue(F_{p,\leq j}) &= \vvalue(F^{*}_{p}) - \vvalue(F_{p,\leq j-1}) - \vvalue(F_{p,j})\\
&= \vvalue(F^{*}_{p}) - \vvalue(F_{p,\leq j-1}) - \vvalue(\hat{F}_{p,j})\\
&\leq 2\cdot (\vvalue(F^{*}_{p}) - \vvalue(F_{p,\leq j-1}))\\
&\leq \vvalue(F^{*}_{p})/2^{j},
\end{align*}
where the second equality is because $F_{p,j} = \hat{F}_{p,j}$.
\end{proof}

To see $\vvalue(F) = \tau-1/n$. Note that when $p = \lceil \log_{1+\epsilon}(nN)\rceil$, $h_{p} \geq nN$ is sufficiently large and we have $F = F_{p}$ and $F^{*} = F^{*}_{p}$. By \Cref{lemma:Value}, $\vvalue(F)\geq \vvalue(F^{*}) - 1/n = \tau - 1/n$. On the other hand, \Cref{line:lambda,line:ScaleF,line:UpdateF} give $\vvalue(F)\leq \tau - 1/n$. Therefore, we can conclude $\vvalue(F) = \tau - 1/n$.

Lastly, we prove $\totlen(F)\leq (1+\epsilon)^{4}\cdot \totlen(F^{*})$. We note the notation $\totlen(\cdot)$ is with respect to the length function $\ell$. We first show \Cref{claim:TotlenFhat} below.

\begin{claim}
\label{claim:TotlenFhat}
For each $0\leq p\leq \bar{p}$ and $1\leq j\leq \bar{j}$, $\totlen(\hat{F}_{p,j})\leq \vvalue(\hat{F}_{p,j})\cdot (1+\epsilon)^{3}\cdot h_{p}$.
\end{claim}
\begin{proof}
We assume that $\hat{F}_{p,j}$ is computed by the algorithm, otherwise we view $\hat{F}_{p,j}$ as a zero flow and the claim is trivial.

By \Cref{fact:LowStepLengthRounding}, we have 
\[
\totlen(\wtilde{F}_{p,j})\leq \vvalue(\wtilde{F}_{p,j})\cdot (1+\epsilon)h_{p}.
\]
Therefore, 
\begin{align*}
\totlen(\hat{F}_{p,j})&\leq \totlen(\wtilde{F}_{p,j}) + 1/\poly(n)\\
&\leq \vvalue(\wtilde{F}_{p,j})\cdot (1+\epsilon)h_{p} + 1/\poly(n)\\
&\leq (\vvalue(\hat{F}_{p,j}) + 1/n^{2})\cdot(1+\epsilon)h_{p} + 1/\poly(n)\\
&\leq \vvalue(\hat{F}_{p,j})\cdot (1+\epsilon)^{3} h_{p},
\end{align*}
where the last inequality is because $\epsilon \geq 10/n$ and $\vvalue(\hat{F}_{p,j})\geq \vvalue(\wtilde{F}_{p,j})-1/n^{2}\geq 1/(2n)-1/n^{2}$ (\Cref{line:TerminateJ} guarantees that $\vvalue(\wtilde{F}_{p,j})\geq 1/(2n)$).
\end{proof}

Providing \Cref{claim:TotlenFhat}, for each $0\leq p\leq \bar{p}$, we have
\[
\totlen(F_{p} - F_{p-1})\leq \vvalue(F_{p} - F_{p-1})\cdot (1+\epsilon)^{3}h_{p},
\]
because $F_{p} - F_{p-1} = \sum_{0\leq j\leq \bar{j}}F_{p,j}$ is a linear combination of $\{\hat{F}_{p,j}\mid 1\leq j\leq \bar{j}\}$. On the other hand, we have
\[
\totlen(F^{*}_{p} - F^{*}_{p-1})\geq \vvalue(F^{*}_{p} - F^{*}_{p-1})\cdot h_{p}/(1+\epsilon),
\]
because when $p\geq 1$, each flow path in $F^{*}_{p} - F^{*}_{p-1}$ has length at least $h_{p-1} = h_{p}/(1+\epsilon)$ by definition. Moreover, when $p=0$, each flow path has length at least $1=h_{0}$ because the edge length $\ell$ are positive integers.

Next we rewrite $\totlen(F)$ and $\totlen(F^{*})$ as follows. For convenience, we let $h_{-1} = 0$. 
\begin{align*}
\totlen(F) &= \sum_{0\leq p\leq \bar{p}}\totlen(F_{p} - F_{p-1})\\
&\leq \sum_{0\leq p\leq \bar{p}}\vvalue(F_{p} - F_{p-1})\cdot (1+\epsilon)^{3}\cdot h_{p}\\
&=\sum_{0\leq p\leq \bar{p}}(\vvalue(F - F_{p-1}) - \vvalue(F - F_{p}))\cdot (1+\epsilon)^{3}\cdot h_{p}\\
&=\sum_{0\leq p\leq \bar{p}}\vvalue(F-F_{p-1})\cdot (1+\epsilon)^{3}\cdot(h_{p} - h_{p-1}).
\end{align*}
Similarly,
\[
\totlen(F^{*})\geq \sum_{0\leq p\leq \bar{p}}\vvalue(F^{*} - F^{*}_{p-1})\cdot (h_{p} - h_{p-1})/(1+\epsilon).
\]
Furthermore, for each $0\leq p\leq \bar{p}$, we have
$\vvalue(F-F_{p-1})\leq \vvalue(F^{*}-F^{*}_{p-1})$
because
\begin{align*}
\vvalue(F-F_{p-1})- \vvalue(F^{*}-F^{*}_{p-1})&= \vvalue(F) - \vvalue(F^{*}) + (\vvalue(F^{*}_{p-1}) - \vvalue(F_{p-1}))\\
&\leq (\tau-1/n) - \tau + 1/n\\
&\leq 0,
\end{align*}
where the second inequality is by $\vvalue(F) = \tau - 1/n$, $\vvalue(F^{*}) = \tau$ and \Cref{lemma:Value}. Combining all these above, we can conclude $\totlen(F)\leq (1+\epsilon)^{4}\cdot \totlen(F^{*})$

\medskip

\noindent{\textbf{Work and Depth.}} Clearly the bottleneck is \Cref{line:ApproxLCflow}. A single execution takes $\tilde{O}(m|\supp(D)|\cdot(t/\epsilon)^{8})$ work and $\tilde{O}((t/\epsilon)^{5})$ depth. We need to execute it $\bar{p}\cdot \bar{j} = \tilde{O}(1/\epsilon)$ times sequentially. Therefore, the final work and depth are $\tilde{O}(m|\supp(D)|\cdot t^{8}/\epsilon^{9})$ and $\tilde{O}(t^{5}/\epsilon^{6})$.

\end{proof}

\begin{proof}[Proof of \Cref{lemma:LowStepFlow}]
We first construct a directed graph $G_{\dir}$ (with positive integral edge lengths and capacities) as a counterpart of $G$ by standard vertex splitting. Initially $G_{\dir}$ is empty.
\begin{itemize}
\item For each vertex $v\in V(G)$, we add two vertices $v_{\iin}$ and $v_{\out}$ to $G_{\dir}$ with a directed edge $e_{v}$ from $v_{\iin}$ to $v_{\out}$. The edge $e_{v}$ has length $\ell(e_{v}) = \ell(v)$ and capacity $U(e_{v}) = U(v)$.
\item For each edge $e=\{u,v\}\in V(G)$, we add two edges, one from $u_{\out}$ to $v_{\iin}$ and the other one from $v_{\out}$ to $v_{\iin}$. Both edges have length $0$ and capacity infinity (to be precise, set the capacity to be $10mN$ which is sufficiently large). %
\end{itemize}
The following \Cref{ob:UndirToDir} basically says we can transfer flows between $G$ and $G_{\dir}$.
\begin{observation}
\label{ob:UndirToDir}
We have the following:
\begin{itemize}
\item (Forward mapping) Any flow $F$ in $G$ can be mapped to a flow $F_{\dir}$ routing the same demand in $G_{\dir}$ with $\totlen(F_{\dir})=\totlen(F)$, $\conge(F_{\dir})=\cong(F)$ and $\step(F_{\dir})\leq 2\step(F)$. 
\item (Backward mapping) Any feasible flow $F_{\dir}$ in $G$ can be mapped to a feasible flow $F$ routing the same demand in $G$ with $\totlen(F)=\totlen(F_{\dir})$ and $\cong(F)=\cong(F_{\dir})$.
\end{itemize}
Moreover, the backward mapping can be done in $O(m)$ work and $O(1)$ depth.
\end{observation}

We then compute the desired flow $F$ as follows.
\begin{enumerate}
\item Apply \Cref{lemma:DirectedLowStepFlow} on $G_{\dir}$ with the same $D,\tau,\epsilon$ but a different step bound $2t$. Let $\wtilde{F}_{\dir}$ be the output with guarantees stated in \Cref{lemma:LowStepFlow}.
\item For each demand pair $(u_{i},v_{i})\in \supp(D)$, let $F^{(i)}_{\dir} = \lambda_{i}\cdot \wtilde{F}^{(i)}_{\dir}$, where 
\begin{equation*}
\lambda_{i} = \left\{
\begin{aligned}
&\tau/(\tau-1/n),\ \text{if}\ \tau=1\\
&D(u_{i},v_{i})/\vvalue(\wtilde{F}^{(i)}_{\dir}),\ \text{if}\ \tau=|D|
\end{aligned}
\right.
\end{equation*}
Let $F_{\dir}$ be the multi-commodity flow containing all $F_{\dir,i}$.
\item Compute $F$ by mapping $F_{\dir}$ back to $G$ using \Cref{ob:UndirToDir}.
\end{enumerate}

To see the correctness, first we clearly have $\Dem(F) \leq D$ and $\Dem(F) = \tau$ because of Step 2. Regarding the congestion and total length, note that the scaling factor in Step 2 is at most $1+\epsilon$ no matter $\tau$ equals to $1$ or $|D|$. In particular, when $\tau = |D|$, we have $D(u_{i},v_{i})/(D(u_{i},v_{i})-1/n)\leq 1/(1-1/n)\leq 1+\epsilon$ because $\Dem(\wtilde{F}_{\dir})\leq D$, $\Dem(\wtilde{F}_{\dir}) = \tau - 1/n$ and $D$ is an integral demand. Therefore, compared to \Cref{lemma:DirectedLowStepFlow}, the congestion and total length slack in \Cref{lemma:LowStepFlow} blow up by a factor of $(1+\epsilon)$.

The work and depth are the same as \Cref{lemma:DirectedLowStepFlow}

\end{proof}

\subsection{Approximate Min-Total-Length Multi-Commodity Flows}
\label{sect:ApproxMTLFlow}

\begin{lemma}
\label{lemma:ApproxMTLFlow}
Let $G$ be an undirected graph with vertex lengths $\ell$ and capacities $U$. Given an integral demand $D$, a flow value threshold $\tau$ equal to either $1$ or $|D|$, suppose there exists a feasible flow $F^{*}$ partially routing $D$ with $\vvalue(F^{*}) = \tau$. Given a precision parameter $1/\poly\log(n)\leq \epsilon\leq 1$, there is an algorithm that computes the edge representation of a flow $F$ such that
\begin{itemize}
\item $\Dem(F) \leq D$ and $\vvalue(F) = \tau$,
\item $F$ has congestion $\hat{O}(1)$, and
\item $\totlen(F)\leq (1+\epsilon)\totlen(F^{*})$ for any feasible flow $F^{*}$ partially routing $D$ with $\vvalue(F^{*}) = \tau$.
\end{itemize}
The work and depth are $\hat{O}(mk)$ and $\hat{O}(1)$ respectively, where $k = |\supp(D)|$ is the number of commodities.
\end{lemma}
\begin{proof}
Let $\epsilon' = \epsilon/100$. The algorithm is as follows.
\begin{enumerate}
\item We apply \Cref{coro:emulator} on $G$ to compute an LC-flow shortcut $H$ of $G$ with length slack $\lambda = 1+\epsilon'$, congestion slack $\kappa = \hat{O}(1)$ and step $t=\hat{O}(1)$. Let $G_{\sc} = G\cup H$ be the corresponding LC-flow shortcut graph.
\item Next, apply \Cref{lemma:LowStepFlow} on $G_{\sc}$ with demand $D$, step bound $t$, flow value threshold $\tau$ and precision parameter $\epsilon'$. Let $F_{\sc}$ be the output flow in $G_{\sc}$.
\item Finally, use the algorithmic backward mapping from \Cref{coro:emulator} to map $F_{\sc}$ back to a flow $F$ in $G$, and $F$ is the flow we desire.
\end{enumerate}

To see the correctness, let $F^{*}$ be an arbitrary feasible flow in $G$ partially routing $D$ with $\vvalue(F^{*}) = \tau$. In the first step, the forward mapping of \Cref{coro:emulator} shows that there exists a flow $F^{*}_{\sc}$ in $G_{\sc}$ routing $\Dem(F^{*})$ with 
\[
\cong(F^{*}_{\sc}) \leq 1,\ 
\totlen(F^{*}_{\sc})\leq (1+\epsilon')\totlen(F^{*})\text{ and step }t = \hat{O}(1).
\]
In the second step, \Cref{lemma:LowStepFlow} guarantees that $\Dem(F_{\sc})\leq D$, $\vvalue(F_{\sc}) = \tau$, the flow $F_{\sc}$ in $G_{\sc}$ has
\[
\cong(F_{\sc}) \leq \tilde{O}(1/\epsilon) = \tilde{O}(1)\text{ and }\totlen(F_{\sc})\leq (1+\epsilon')^{5}\cdot \totlen(F_{\sc}^{*}).
\]
Finally, in the third step, the backward mapping algorithm from \Cref{coro:emulator} guarantees that $F$ routes $\Dem(F_{\sc})$ and has
\[
\cong(F)\leq \hat{O}(1)\cdot \cong(F_{\sc})\text{ and }\totlen(F)\leq \totlen(F_{\sc}).
\]
Thus, we can conclude that $F$ partially routing $D$, $\vvalue(F) = \tau$ and
\[
\cong(F)\leq \hat{O}(1)\text{ and }\totlen(F)\leq (1+\epsilon')^{6}\totlen(F^{*})\leq (1+\epsilon)\totlen(F^{*})
\]
as desired.

Consider the work and depth. By \Cref{coro:emulator}, step 1 has work $\hat{O}(m)$ and depth $\hat{O}(1)$. By \Cref{lemma:LowStepFlow} and $t = \hat{O}(1)$, step 2 has work $\hat{O}(mk)$ and depth $\hat{O}(1)$. By \Cref{coro:emulator} again, step 3 has work $\hat{O}(m+nk)$ and depth $\hat{O}(1)$.

\end{proof}

\subsection{Approximate Multi-Commodity Mincost Flows: Proof of \Cref{thm:MainResult}}
\label{sect:MainResultProof}
\MainResult*

\begin{proof}
We first define the demand $D$ and the flow value threshold $\tau$. 
\begin{itemize}
\item For the concurrent version, the demand $D$ is given (with $|\supp(D)| = k$), and we set $\tau = |D|$. 
\item For the non-concurrent version, we are given a set $S$ of vertex pairs (with $|S| = k$), we define $D$ by setting, for each $(u,v)$, $D(u,v) = 1$ if $(u,v)\in S$, otherwise $D(u,v) = 0$. Set $\tau = 1$.
\end{itemize}

We now use the flow boosting template \Cref{lemma:FlowBoosting} with precision parameter
\[
\epsilon' = \epsilon/c
\]
for some sufficiently large constant $c$. We let the convex set ${\cal F}$ of flows collect all flows $F$ in $G$ routing $D$ with $\vvalue(F) \geq \tau$. To make sure that ${\cal F}$ contains at least one feasible flow (so that it satisfies the condition of \Cref{lemma:FlowBoosting}), we can simply work on the graph $G$ with vertex capacities scaled up by $\tau$, and then scale down the output flow by $\tau$ at the end.

Observe that we can use \Cref{lemma:ApproxMTLFlow} as the algorithm ${\cal A}$ in the flow boosting template \Cref{lemma:FlowBoosting}. Indeed, in graph $G$ with the vertex length function $\ell$ given by \Cref{lemma:FlowBoosting}, \Cref{lemma:ApproxMTLFlow} will compute the edge representation of a flow $F\in{\cal F}$ such that $\totlen(F)\leq (1+\epsilon')\totlen(F^{*})$ for any feasible flow $F^{*}$ partially routing $D$ with $\vvalue(F^{*}) = \tau$ (so also for all such $F^{*}$ with $\vvalue(F^{*}) \geq \tau$), i.e. for any capacity-respecting $F^{*}\in{\cal F}$. Furthermore, we have $\cong(F)\leq \hat{O}(1)$. Therefore, we set the length slack $s$ and congestion slack $\kappa$ of \Cref{lemma:FlowBoosting} to be
\[
s = 1+\epsilon'\text{ and }\kappa = \hat{O}(1)
\]

The output of the flow boosting template \Cref{lemma:FlowBoosting} is the edge representation of a flow $\bar{F}\in{\cal F}$ (so $\vvalue(F)\geq \tau$) and a scalar $\lambda\geq 0$ satisfying the following.
\begin{enumerate}
\item The flow $\lambda \bar{F}$ is capacity-respecting and has cost at most $B$.
\item Let $\lambda^{*}$ be the maximum value such that there exists flow $F^{*}\in{\cal F}$ where $\lambda^{*}F^{*}$ is capacity-respecting with cost at most $B$ (greedily we can assume $\vvalue(F^{*}) = \tau$). Then $\lambda \geq \frac{1-O(\epsilon')}{s}\lambda^{*} = \frac{1-O(\epsilon')}{1+\epsilon'}\lambda^{*}$.
\end{enumerate}
Let $F = \lambda \bar{F}$ be our final output.

Observe that in the concurrent version, the $\lambda$ and $\lambda^{*}$ above are exactly those in \Cref{def:concurrent-flow}, so the approximation factor is $\frac{1-O(\epsilon')}{1+\epsilon'}\leq 1+\epsilon$. In the non-concurrent version, the $\lambda$ and $\lambda^{*}$ above are the $\vvalue(F)$ and $\vvalue(F^{*})$ in \Cref{def:concurrent-flow}, so the approximation factor is also $1+\epsilon$.

The work and depth are $\hat{O}(mk)$ and $\hat{O}(1)$ by \Cref{lemma:FlowBoosting}, $\kappa = \hat{O}(1)$ and \Cref{lemma:ApproxMTLFlow}.
\end{proof}

\section{Conclusions and Open Problems}

In this work, we present a parallel $(1+\epsilon)$-approximate multi-commodity mincost flow algorithm with almost optimal work and depth, which applies to undirected graphs with capacities and costs on both edges and vertices. Our key technical contribution is devising an LC-flow shortcut with $1+\epsilon$ length slack, $\hat{O}(1)$ congestion slack, and $\hat{O}(1)$ step, which has size $\hat{O}(n)$ and can be computed in $\hat{O}(m)$ work and $\hat{O}(1)$ depth. Our work also leaves several interesting open problems.

\paragraph{Polylogarithmic Depth.} A natural and important open problem is to further reduce the depth from subpolynomial to polylogarithmic while retaining almost-optimal work or improving toward nearly-optimal work, even in the single-commodity setting. As we discussed in \Cref{sect:Introduction}, $(1+\epsilon)$-approximate parallel algorithms with polylogarithmic depth and nearly-linear work are known for \emph{shortest paths} \cite{DBLP:conf/stoc/Li20,DBLP:conf/stoc/AndoniSZ20,DBLP:conf/stoc/RozhonGHZL22} and \emph{(edge-capacitated) maxflows} \cite{DBLP:conf/soda/AgarwalKLPWWZ24}. However, it remains unclear whether similar results (particularly polylogarithmic depth) can be achieved for the (edge-capacitated) mincost flow problem or vertex-capacitated maxflow problem.

\begin{question}
Are there parallel $(1+\epsilon)$-approximate algorithms for either (edge-capacitated) mincost flows or vertex-capacitated maximum flows, with nearly-linear work and polylogarithmic depth?
\end{question}

\paragraph{Better LC-flow Shortcuts.} The second open problem is about the quality of the LC-flow shortcuts. Recall that LC-flow shortcuts have the following quality measures: length slack, congestion slack, step, and size. In our approach to obtaining approximate maximum/mincost flows using LC-flow shortcuts, the approximation depends on the length slack, the depth overhead depends on congestion slack and step, and the work overhead depends on congestion slack, step and size. 

Therefore, it is interesting to optimize the trade-offs between these measures, particularly between congestion slack and step, even for existential constructions without efficient algorithms. Note that our construction shows the existence\footnote{Use existential length-constrained expander decompositions with optimal tradeoffs.} of LC-flow shortcuts with $1+\epsilon$ length slack, $n^{O(1/k)}$ congestion slack,  $(1/\epsilon)^{\poly(k)}$ step and $n^{1+O(1/k)}$ size, for $\epsilon \leq 1$ and $k\geq 1$. It would be interesting to reduce the dependence on $k$ in the step bound to subexponential or even polynomial. We point out that such better step bounds are known for some weaker notions of shortcuts:
\begin{itemize}
\item Shortcuts with only length control \cite{DBLP:journals/jacm/Cohen00}: $1+\epsilon$ length slack, $(\frac{\log n}{\epsilon})^{O(\log k)}$ step and $n^{1+O(1/k)}$ size.
\item Flow shortcuts with only (edge-)congestion control (e.g., tree flow sparsifiers \cite{harrelson2003polynomial,racke2014computing}): $\tilde{O}(1)$ congestion slack, $\tilde{O}(1)$ step and $\tilde{O}(n)$ size.
\item LC-flow shortcuts with edge-congestion control and large length slack \cite{DBLP:conf/stoc/HaeuplerH0RS24}: $\poly(k)$ length slack, $n^{O(1/k)}$ congestion slack, $\poly(k)$ step and $n^{1+O(1/k)}$ size.
\end{itemize}
As stepping stones, the following simpler notions of flow shortcuts are also open to improved trade-offs. The first is LC-flow shortcuts with only \emph{edge}-congestion control and $1+\epsilon$ length slack, for which the best known trade-off is still achieved by our work. The second is flow shortcuts with vertex-congestion control (without length control), for which the best known tradeoff is $n^{O(1/k)}$ congestion slack, $2^{O(k)}$ step, and $\tilde{O}(n)$ size, as shown implicitly\footnote{In fact, \cite{HaeuplerLSW25} presents LC-flow shortcuts with vertex-congestion control and large length slack.} in \cite{HaeuplerLSW25}.

\paragraph{Distributed Mincost Flows.} Parallel and distributed algorithms are closely related, and there have been many technical transfers between the two settings. For example, \cite{DBLP:conf/stoc/RozhonGHZL22} studies the $(1+\epsilon)$-approximate shortest path problem in parallel and distributed settings simultaneously. 

As we give the almost-optimal parallel algorithm for $(1+\epsilon)$-approximate mincost flow problem, a natural next step is to solve it almost-optimally in the distributed setting. We note that there has been distributed $(1+\epsilon)$-approximate (single-commodity) maxflow algorithms with almost-optimal round complexity in the \textsf{CONGEST} model \cite{ghaffari2015near}. However, analogous results for the min-cost flow problem remain unclear, though there have been some efforts on distributed exact min-cost flow in directed graphs \cite{forster2022minor,de2023minimum}.

\begin{question}
Is there distributed $(1+\epsilon)$-approximate mincost flow algorithm with almost-optimal round complexity in the \textsf{CONGEST} model?
\end{question}

We note that our approach may potentially transfer to the distributed setting, since the main building blocks, including length-constrained expander decompositions, sparse neighborhood cover, and approximate length-constrained maxflows, all admit efficient distributed algorithms.
%\bibliographystyle{alpha}
%\bibliography{ref}
\newcommand{\etalchar}[1]{$^{#1}$}

\appendix

\section{Notations for Expander Decomposition Algorithm}\label{sec:notationsexpander}

We will prove \Cref{thm:vertexLC-ED} in the appendix. \Cref{thm:vertexLC-ED} involves an expander routing algorithm, which will be achieved by an expander witness. We define it as follows. 

\paragraph{Routers and expander witness.}
We will need the definition of routers as follows. Notice that the congestion is defined on the edges.
\begin{definition}[Routers]\label{def:routers}
Given an edge-capacitated graph $G=(V,E)$ and a node-weighting $A$ on $V$, we say that $G$ is a $t$-step $\kappa$-router for $A$ if every $A$-respecting demand can be routed in $G$ by a $t$-step flow $f$ such that for every edge $e\in E$ we have $\sum_{e\in p}f(p)\le \kappa\cdot U(e)$ flow paths in $f$. In this case, we say $f$ has \emph{edge-congestion} $\kappa$ on $G$. 

\end{definition}
 The following definition shows how to embed a router into a graph.

 \begin{definition}[Graph Embedding]\label{def:graphEmbedding}
     An \emph{$h$-length edge-to-vertex embedding with congestion $\kappa$} from an edge-capacitated graph $H$ to $G$ with vertex lengths and capacities where $V(H)=V(G)$, is an assignment of each edge $(u,v)\in E(H)$ a flow $F_{u,v}$ on $G$ routing the demand $D(u,v)=U(u,v)$, such that $\sum_{e\in E(H)}F_e$ has length at most $h$ and congestion at most $\kappa$.
 \end{definition}
\paragraph{Length-Constrained Expansion Witnesses}

Below, for node-weighting $A$ and nodes $S$, we let $A_S$ be $A$ restricted to nodes in $S$.

\begin{definition}[Length-Constrained Expansion Witness]\label{def:LCExpWitness}
    A $(h,s)$-length $\phi$-expansion witness for graph $G$ with vertex capacities and lengths and node-weighting $A$ consists of four numbers $s_0,s_1,\kappa_0,\kappa_1$ such that $s_0 \cdot s_1 \leq s,\kappa_0 \cdot \kappa_1 \leq 1/\phi$ and the following:
    \begin{itemize}
        \item \textbf{Neighborhood Cover:} a neighborhood cover $\mcN$ of $G$ with covering radius $h$;
        \item \textbf{Routers:} an $s_0$-step $\kappa_0$-router $R_S$ of $A_S$ for each $S \in \cS\in \mcN$;
        \item \textbf{Embedding of Routers:} an $(h \cdot s_1)$-length edge-to-vertex embedding of $\cup_{S\in\cS\in\cN}R_S$ into $G$ with congestion $\kappa_1$.
    \end{itemize}
\end{definition} 

We can define a witnessed expander decomposition as follows.
\begin{definition}[Witnessed Length-Constrained (Vertex) Expander Decomposition]
    Given a graph $G = (V,E)$, a \emph{witnessed $(h,s)$-length $\phi$-expander decomposition} for a node-weighting $A$ with length slack $s$ and cut slack $\kappa$ is an $h \cdot s$-length cut $C$ of size at most $\kappa \cdot \phi|A|$ such that $A$ is $(h,s)$-length $\phi$-expanding in $G -C$ along with a $(h,s)$-length $\phi$-expansion witness for $G-C$.
\end{definition}

The following theorem is the main focus of this appendix. It shows an algorithm computing a length-constrained expander decomposition with a witness after applying the decomposition. 

\begin{theorem}[Witnessed Length-Constrained Expander Decompositions]
    \label{thm:vertexLC-EDwitness}
    There is a randomized algorithm $\textsc{VertexLC-ED}(G,A,h,\phi,\eps)$ that takes inputs an undirected graph $G$ with vertex capacities and lengths, a node-weighting $A$, length bound $h\ge 1$, conductance $\phi\in(0,1)$ and a parameter $\eps$ satisfying $0< \eps<1$, outputs a witnessed $(h,s)$-length $\phi$-expander decomposition $C$ for $G,A$ with cut slackness $\kappa$ where
    \begin{align*}
    \kappa = n^{O(\eps)} \qquad \qquad s = \exp(\exp(\frac{O(\log(1/\eps))}{\eps}))
    \end{align*}

    The algorithm takes $\poly(h)\cdot m\cdot n^{\poly(\eps)}$ work and $\poly(h)\cdot n^{\poly(\eps)}$ depth. 
\end{theorem}

We will show that \Cref{thm:vertexLC-EDwitness} implies \Cref{thm:vertexLC-ED}, so it suffices to prove \Cref{thm:vertexLC-EDwitness} in the appendix.

\begin{proof}[Proof of \Cref{thm:vertexLC-ED}]
    The first part (output an expander decomposition) directly follows from \Cref{thm:vertexLC-EDwitness}. We only need to prove an expander routing algorithm.
    
    Suppose we are given a $h$-length $A$-respecting demand $D$ on $G-C$, we will show how to route $D$ with length $O(\exp(\poly(1/\eps))\cdot hs)$ and congestion $n^{\poly(\eps)}/\phi$. We assume $\eps=\omega(1/\log \log n)$ as otherwise $s$ becomes $\poly(n)$ and there is no length bound, which makes this problem converge to normal expander routing. 

    Notice that from \Cref{thm:vertexLC-EDwitness}, we also get a witness defined in \Cref{def:LCExpWitness}. For each pair $(u,v)$ such that $D(u,v)>0$, since $D$ is $h$-length, the neighborhood cover $\cN$ contains a cluster $S\in\cS\in\cN$ such that $u,v\in S$, and we build the demand on $R_S$ by putting the same demand $D(u,v)$ on $R_S$. We do this for every pair $(u,v)$, and get a demand $D_S$ on $R_S$ for every $S\in\cS\in\cN$ so that $\sum_{S\in\cS\in\cN}D_S=D$, and $D_S$ is trivially $A_S$-respecting on $S$ because $D$ is $A$-respecting.

    Now we can apply Theorem 7.1 \cite{DBLP:conf/stoc/HaeuplerH0RS24} for each $R_S$ and $D_S$. For clarity, we restate Theorem 7.1 of \cite{DBLP:conf/stoc/HaeuplerH0RS24} as follows.

    \begin{theorem}[Theorem 7.1 of \cite{DBLP:conf/stoc/HaeuplerH0RS24}]\label{thm:routerrouting}
        Given a $s_0$-step $\kappa_0$-router $R_S$ for a node-weithting $A$, an $A$-respecting demand $D_S$ and $\eps\in(\log^{-c}N,1)$ for some sufficiently small constant $c$, one can compute a path-representation of flow $F$ routing $D$ on $R_S$ with 
        \begin{itemize}
            \item edge congestion $n^{\poly(\eps)}\cdot \kappa_0$ and step $\exp(\poly(1/\eps))\cdot s_0$, and
            \item $|\path(F)|\le (|E(R)|+|\supp(D)|)\cdot n^{\poly(\eps)}$
        \end{itemize} 
        in $(|E(R)|+|\supp(D)|)\cdot \poly(s_0)\cdot n^{\poly(\eps)}$ work and $\poly(s_0)\cdot n^{\poly(\eps)}$ depth.
    \end{theorem}

    We apply \Cref{thm:routerrouting} on $R_S$ for every $S\in\cS\in\cN$ to get a flow $F^R_S$ in $R$ routing $D_S$. To get a flow $F_S$ in $G-C$ routing $D_S$, we use the \emph{embedding of routers}: for every flow path $P\in F^R_S$ and edge $e\in E(P) $, we embed this edge into $G-C$ with flow $F_e$ of the same value as $P$, and concatenate all the flows $F_e$ to get the flow $F_P$, and take the union of all $F_P$ to get the flow $F_S$.

    \paragraph{Quality of the returned flow.} We first analyze the length. According to \Cref{thm:routerrouting}, the flow $F^R_S$ for every $S$ has step $\exp(\poly(1/\eps))\cdot s_0$. Thus, each flow path in $F_S$ is concatenated by $\exp(\poly(1/\eps))\cdot s_0$ many flows of length $h\cdot s_1$. Thus, $F_S$ has length at most $O(\exp(\poly(1/\eps))\cdot s_0\cdot h\cdot s_1)=O(\exp(\poly(1/\eps))\cdot hs)$. To get the length $hs$ stated in \Cref{thm:vertexLC-ED}, it suffices to relax $s$ and $\phi$ a bit to $s=\exp\exp(\frac{O(\log (1/\eps))}{\eps})\cdot O(\exp(\poly(1/\eps))=\exp\exp(\frac{O(\log (1/\eps))}{\eps})$. Notice that $C$ is still a valid $(h,s)$-length $\phi$-expander decomposition after relaxing (increasing) $s$ and decreasing $\phi$ by a constant according to the monotonicity \Cref{lem:monotonicity}.

    Then, we analyze the congestion. Notice that 
    \[\cong(\sum_{S\in\cS\in\cN}F_S)\le 2\cong(\sum_{S\in\cS\in\cN}\sum_{P\in\path F^R_S}\sum_{e\in E(P)}F_e)\]
    notice that $F_e$ is routing $U(e)$ many flow from one end point of $e$ to another. Moreover, according to \Cref{thm:routerrouting}, the value of the total amount of flow from one end point of $e$ to another is at most $n^{\poly(\eps)}\cdot \kappa_0\cdot U(e)$. Thus, after scaling the flow $\sum_{S\in\cS\in\cN}F_S$ down by $n^{\poly(\eps)}\cdot \kappa_0$, it becomes a subflow of $\sum_{S\in\cS\in\cN}\sum_{e\in E(R_S)}F_e$, which according to the definition of witness, has congestion at most $\kappa_1$. The congestion of the original flow $\sum_{S\in\cS\in\cN}F_S$ is thus at most $n^{\poly(\eps)}\cdot \kappa_0\cdot \kappa_1\le n^{\poly(\eps)}/\phi$.

    \paragraph{Complexity.} \Cref{thm:vertexLC-EDwitness} takes $\poly(h)\cdot m\cdot n^{\poly(\eps)}$ work and $\poly(h)\cdot n^{\poly(\eps)}$ depth. 
    
    After getting the witness, we apply \Cref{thm:routerrouting} on each router $R_S$. Notice that $R_S$ for all $S\in\cS\in\cN$ are explicit outputs of algorithm \Cref{thm:vertexLC-EDwitness}, so the size of $\cup_{S\in\cS\in\cN}R_S$ is at most the work, i.e., $\poly(h)\cdot m\cdot n^{\poly(\eps)}$. Each pair of demand in $D$ contributes to at most one demand pair in $D_S$, which implies $\sum_{S\in\cS\in\cN}\supp(D_S)\le \supp(D)$. Thus, According to \Cref{thm:routerrouting}, the routing takes in total 
    \[(\poly(h)\cdot m\cdot n^{\poly(\eps)}+|\supp(D)|)\cdot \poly(s_0)\cdot n^{\poly(\eps)}\]

    work and $\poly(s_0)\cdot n^{\poly(\eps)}$ depth.

    For each flow path of $F^R_S$, we map the (at most $\exp(\poly(1/\eps))\cdot s_0$ many) edges of the path to $G-C$ of (at most $hs$ many edges). According to the support flow size specified in \Cref{thm:routerrouting}, the total work of this step is

    \[\exp(\poly(1/\eps))\cdot s_0\cdot hs\cdot (\poly(h)\cdot m\cdot n^{\poly(\eps)}+|\supp(D)|)\cdot n^{\poly(\eps)}\]

    and depth $\tO{1}$. By accumulate all the work above, we get the total work (remember that $\eps=\omega(1/\log\log n)$ so that $n^{\poly(\eps)}\ge \poly(s)$

    \[(m+|\supp(D)|)\cdot \poly(h)\cdot n^{\poly(\eps)}\]

    and depth 
    \[\poly(h)\cdot n^{\poly(\eps)}\]
\end{proof}

\paragraph{Organization.} The algorithm is similar to classical expander decomposition and is based on repeatedly cutting the most balanced sparse cut. We will define this concept in the length constraint setting in this section. In \Cref{sec:unionofspasecut}, we will show some useful properties of sparse cut that will be used in the analysis of the algorithm. In \Cref{sec:apxDLSC}, we will give an algorithm for finding the (approximately) `most balanced' sparse cut. In \Cref{sec:expanderfromDLSC}, we will apply the algorithm in \Cref{sec:apxDLSC} repeatedly to finally prove \Cref{thm:vertexLC-ED}.

We first state the basic fact about the monotonicity of length constraint cuts when $h$ increases, $s$ decreases, and $\phi$ increases.

\begin{lemma}\label{lem:monotonicity}
    Fix a node-weighting $A$ and a (moving cut) $C$ on the vertex set $V$, for every two graphs $G,G'$ where length function of $G'$ is at most length function of $G$, and every $h'=ch,2\le s'\le s,\phi'\ge 2\phi$ where $c\in\bbN^+$, we have the following monotonicity.
    
    \begin{enumerate}
        \item If $C$ is an $(h,s)$-length $\phi$-sparse cut in $G$, then $2C$ is an $(h',s')$-length $\phi'$-sparse cut in $G'$.\label{item:mono_cut_item1}
        \item $DLSC_{G',A}(\phi',h',s')\ge DLSC_{G,A}(\phi,h,s)$.\label{item:mono_cut_item2}
        \item $G'$ is $(h',s')$-length $\phi'$-expanding implies $G$ is $(h,s)$-length $\phi$-expanding.\label{item:mono_cut_item3}
    \end{enumerate}

\end{lemma}
\begin{proof}
    Given $C$ is an $(h,s)$-length $\phi$-sparse cut, denote by $D$ its witnessing demand that has the demand-size of $C$, we have $A_{(h,s)}(C) = |D| \ge |C|/\phi$ and
    $\leng(P_{u,v}, G+\l_{C,hs}) > hs$
    for any demand pairs $D(u,v) > 0$.
    We show in the following that $2C$ also $h's'$-separates the $h'$-length demand $D$ since $h' \ge h$, which yields a sparsity upper bound of at most $2\phi$.
    Namely, from $\leng(P_{u,v}, G+\l_{C,hs}) > hs$, we know that the length increase along path $P_{u,v}$ is at least $hs-h$. 
    With the assumption of $s\ge 2$, the length increase from $2C$ is larger than $hs$.
    We can further assume that $h's' = c'hs$ where $c' \in \bbN^+$, and view $C$ as an $(h', s')$-length moving cut with the same cut value but larger length increase on the graph.
    Then for any graph $G'$ is a copy of $G$ but with length function at most the one of $G$, we ahve
    $\leng(P_{u,v}, G'+\l_{2C,h's'}) > h's'$.
    Thus the sparisty of the $(h',s')$-length moving cut $2C$ in $G'$ is upper bounded by $2|C|/|D| \le 2\phi$.
    This immediately concludes the \textbf{\cref{item:mono_cut_item1}}, and implies the following ones.

    \noindent\textbf{Proof of \cref{item:mono_cut_item2}}: For the second point, since $|D| = DLSC_{G,A}(\phi,h,s)$, and we have an $(h',s')$-length $\phi'$-sparse cut $2C$ also fully separates $D$, it shows that $DLSC_{G',A}(\phi',h',s') \ge |D|$.

    \noindent\textbf{Proof of \cref{item:mono_cut_item3}}: For the last one, we show it by contraposition, suppose $G$ is not $(h,s)$-length $\phi$-expanding, then there exists an $(h,s)$-length $\phi$-sparse cut $C$ in $G$.
    As shown above, we can construct an $(h',s')$-length $\phi'$-sparse cut $2C$ in $G'$, which implies that $G'$ is not $(h',s')$-length $\phi'$-expanding.
\end{proof}

\paragraph{Concepts of demand sizes.} We will define demand-size-largest sparse cut, which corresponds to the concept of `most balanced sparse cut' in the classical setting, i.e., demand-size corresponds to the balance, as defined below.

\begin{definition}[$(h,s)$-Separated Demand-Size]\label{def:demandSize}
    For a graph $G$, an $h$-length cut $C$, and a node-weighting $A$, we define the $(h,s)$-length demand-size of $C$ (with respect to $A$ on $G$) as the size of the largest $A$-respecting $h$-length demand which is $(hs)$-separated by $C$. We denote this ``demand-size'' by $A_G(h,s,C)$. When $G$ is clear from the context, we simply write $A(h,s,C)$. 
\end{definition}

\begin{definition}[Demand-Size-Largest Sparse Cut, \qLDSC]\label{def:LDSC}
    For a graph $G$, a node weighting $A$, and an $(h,s)$-length $\phi$-sparse cut $C$, we call $C$ the demand-size-largest $(h,s)$-length $\phi$-sparse cut if its demand-size is maximum among all $(h,s)$-length $\phi$-sparse cuts. We notate the size of this cut as \begin{align*}
        \qLDSC_{G,A}(\phi, h, s) := A_{G}(h,s,C)
    \end{align*}

    When $G,A$ are clear from the context, we omit them and simply write $\qLDSC(\phi,h,s)$. 
\end{definition}

\begin{restatable}[Approximately Demand-Size-Largest Sparse Cut]{definition}{apxsparsecut}\label{dfn:apxLargeCut}
    $C$ is an $\alpha$-approximate demand-size-largest $(h,s)$-length $\phi$-sparse cut for a node-weighting $A$ with length approximation $\alpha_h$, length slackness approximation $\alpha_s$ and sparsity approximation $\alpha_{\phi}$ if it is an $(\alpha_{h}\cdot h, s)$-length $\phi$-sparse cut and we have
    \begin{align*}
        A_{(h,s)}(C) \geq \frac{1}{\alpha} \cdot \qLDSC(\phi/\alpha_\phi, h/\alpha_s, s \cdot \alpha_s).
    \end{align*}
\end{restatable}

\section{Union of Sparse Vertex Moving Cut Sequence is Sparse}\label{sec:unionofspasecut}

Formally, we consider a sequence of moving cuts, defined as follows.
\begin{definition}[Sequence of Moving Cuts]\label{dfn:movingCutSequence}
    Given graph $G = (V,E)$ and node-weighting $A$, a sequence of moving cuts is
    $(C_1, C_2, \ldots)$. 
    Let $G - \sum_{j < i} C_j$ denote the graph that is applied with cuts from $C_1$ to $C_{i-1}$ sequentially.    
    We refer to the $(h,s)$-length sparsity of $C_i$ with respect to $A$ in $G - \sum_{j < i} C_j$ as its $(h,s)$-length sparsity in the sequence. 
    We say the sequence $(C_1, C_2, \ldots)$ is $(h,s)$-length $\phi$-sparse if each $C_i$ is $(h,s)$-length $\phi$-sparse in the sequence.
\end{definition}

\begin{definition}[Sequence of Witnessing Demands]\label{dfn:witnessDemandSequence}
    Given graph $G = (V,E)$, a node-weighting $A$, and an $(h,s)$-length $\phi$-sparse sequence of moving cuts $(C_1, C_2, \ldots)$, we define the sequence of witnessing demands as $\mcD = (D_{1}, D_{2}, \ldots)$ such that $D_{i}$ is an $h$-length $A$-respecting demand in $G - \sum_{j < i} C_j$ and $\spa_{hs}(C_{i}, D_{i}) \le \phi$. 
    Furthermore, we assume $D_{i}$ is fully separated, i.e., $\sep_{hs}(C_{i}, D_{i}) = |D_{i}|$.
\end{definition}
Intuitively, witnessing demand $D_{i}$ indicates close vertex pairs $hs$-separated by $C_i$. 
In particular, if $D_{i}(u,v) >0$, then we have $\dist_{G - \sum_{j < i} C_j}(u, v) \le h$ but $\dist_{G - \sum_{j \le i} C_j}(u, v) \ge hs$.

The main result we would like to show in this section is that the union of sparse length-constrained vertex moving cuts is still sparse:
\begin{restatable}[Union of Sparse Vertex Moving Cuts is a Sparse Vertex Moving Cut]{theorem}{unionOfCuts}
	\label{thm:unionOfMovingCuts} Let $(C_1, \ldots, C_k)$ be a sequence of moving cuts where $C_i$ is an $(h,s)$-length $\phi_i$-sparse cut in $G - \sum_{j < i} C_j$ w.r.t.\ node-weighting $A$. Then the moving cut $\sum_i C_i$ is an $(h',s')$-length $\phi'$-sparse cut w.r.t.\ $A$ where $h' = 2h$, $s' = \frac{(s-2)}{2}$ and $\phi' = s^3 \cdot \log^3 n \cdot n^{O(1/s)} \cdot \frac{\sum_i |C_i|}{\sum_i |C_i|/\phi_i}$.
\end{restatable}

\subsection{Low Arboricity Demand Matching Graph via Parallel Greedy Spanners}
We follow the idea by \cite{HaeuplerHT24} to construct a graph called demand matching graph.
It helps to construct an overall demand that witnesses the sparsity of the union of moving cuts in the sequence of moving cuts.
Given node-weighting $A$ and demand $D$, we can scale them (arbitrarily close to being) integral.
We can thus assume, w.l.o.g., that $A$ and $D$ are integral.
Informally, the graph creates $A(v)$ copies for each vertex $v$. 

\begin{definition}[Demand Matching Graph\cite{HaeuplerHT24}]\label{def:demandMatching}
Given a graph $G = (V,E)$ with vertex capacity and length, a node-weighting $A$ and $A$-respecting demands $\mcD = (D_1, D_2, \ldots )$, we denote $2A(v)$ unique ``copies'' of $v$ by $\copies_{A}(v) := \{v_{i} : i \in \{1,2,\ldots,2A(v)\}\}$.
We further define the demand matching (multi)-graph $G(\mcD) = (V', E')$ with unit edge length as follows:
\begin{itemize}%
    \item  \textbf{Vertices:} $G(\mcD)$ has vertices $V' = \bigsqcup_v \copies_{A}(v)$.
    \item \textbf{Edges:} For each demand $D_i$, let $M_i$ be any matching where the number of edges between $\copies_{A}(u)$ and $\copies_{A}(v)$ for each $u,v \in V$ is $D_i(u,v)$. Then $E' = \bigcup_i M_i$.
\end{itemize}
\end{definition}
We note there always exists a matching $M_{i}$ defined as above for each demand $D_{i}$ since $D_{i}$ is $A$-respecting and we create $2A(v)$ unique ``copies'' of each vertex $v$. 
Given a sequence of witnessing demands, 
we can construct a corresponding demand matching graph as defined above.
The key property is that it has a low arboricity.

\begin{lemma}[Bounded Arboricity of Demands of Sequence of Cuts\cite{HaeuplerHT24}]\label{lem:arbBound} Let $(C_1, C_2, \ldots)$ be a sequence of $(h, s)$-length cuts with witnessing demands $\mcD = (D_1, D_2, \ldots)$. Then $G(\mcD)$ has arboricity at most $s^3 \cdot \log^3 N \cdot N^{O(1/s)}$.
\end{lemma}
\begin{proof}
    We say that a sequence $(E_1, E_2, \ldots,E_k)$ of edge sets on $V$ is $s$-$\pg$ (abbreviation for $s$-parallel greedy) for some integer $s\ge 2$, iff for each $1\le i\le k$,
    \begin{itemize}
        \item set $E_i$ is a matching on $V$; and
        \item if we denote by $G_{i-1}$ the graph on $V$ induced by edges in $\bigcup_{1\le j\le i-1}E_j$, then for every edge $(u,v)\in E_i$, $d_{G_{i-1}}(u,v)> s$.
    \end{itemize}
    Equivalently, a sequence $(E_1,\ldots,E_k)$ is $s$-$\pg$ iff for each $1\le i\le k$, every cycle in $G_i$ of length at most $s+1$ contains at least two edges in $E_i$.
    We say that a graph $G$ is $s$-$\pg$ iff its edge set $E(G)$ is the union of some $s$-$\pg$ sequence on $V(G)$. \cite{HaeuplerHT23} proves that every $s$-\pg graph on $n$ vertices has arboricity $s^3 \cdot \log^3 n \cdot n^{O(1/s)}$.
    
    We now show that this fact implies \Cref{lem:arbBound} by showing that $G(\mcD)$ is a $s$-$\pg$ graph. 
    Specifically, given witnessing demands $\mcD = (D_{1}, D_{2},\ldots, D_{k})$, we would like to show that the reversed sequence $(M_{k}, M_{k-1},\ldots, M_{1})$ is $s$-$\pg$.
    
    For each $i$, let $G_{i}(\mcD)$ be the graph induced by edges in $\bigcup_{k-i+1\le j\le k}M_{k}$.
    Consider an edge in $M_{k-i}$, the matching corresponding to demand $D_{k-i}$. 
    By definition, it suffices to show that for each $(u,v)\in M_{k-i}$, there is no length-at-most-$s$ path in $G_{i}(\mcD)$ containing $u,v$. 
    Assume for contradiction that there exists a path $P=(u,x_1,\ldots,x_{s-1},v)$ in $G_{i}(\mcD)$. 
    This means that in graph $G-\sum_{1\le j\le k-i}C_{j}$, every pair in $(u,x_1),(x_1,x_2),\ldots,(x_{s-1},v)$ is at distance at most $h$. %
    By triangle inequality, this implies that the distance between $u,v$ in $G-\sum_{1\le j\le k-i}C_{j}$ is less than $hs$. However, as the moving cut $C_{k-i}$ separates all pairs in $D_{k-i}$ to distance more than $hs$, as an edge in $M_{k-i}$, $u$ and $v$ should be at distance more than $hs$ in $G-\sum_{1\le j\le k-i}C_{j}$, a contradiction.
\end{proof}

\subsection{Matching-Dispersed Demand}
We describe how to construct the witness demand for the union of moving cuts from the demand matching graph constructed previously.
\begin{definition}[Tree Matching Demand\cite{HaeuplerHT24}]\label{def:treeMatchDemand}
Given tree $T = (V,E)$ we define the tree-matching demand on $T$ as follows. Root $T$ arbitrarily. For each vertex $v$ with children $C_v$ do the following. If $|C_v|$ is odd let $U_v = C_v \cup \{v\}$, otherwise let $U_v = C_v$. Let $M_v$ be an arbitrary perfect matching on $U_v$ and define the demand associated with $v$ as 
\begin{align*}
    D_{v}(u_1, u_2) = D_{v}(u_2, u_1):= 
    \begin{cases}
        1 & \text{if $\{u_1,u_2\} \in M_v$}\\
        0 & \text{otherwise}.
    \end{cases}
\end{align*}
where each edge in $M_v$ has an arbitrary canonical $u_1$ and $u_2$. Then, the tree matching demand for $T$ is defined as
\begin{align*}
    D_{T} := \sum_{v \text{ internal in }T} D_v
\end{align*}
\end{definition}

We observe that a tree matching demand has size equal to the input size (up to constants).
\begin{lemma} \label{lem:treeMatchSize}
Let $T$ be a tree with $n$ edges. Then $|D_T| \geq n$.
\end{lemma}

We can construct a matching-dispersed demand from a forest cover by applying the tree matching demand to each tree and scaling down by the arboricity.
\begin{definition}[Matching-Dispersed Demand\cite{HaeuplerHT24}]\label{dfn:matchingDemand}
Given graph $G$, node-weighting $A$ and demands $\mcD = (D_1, D_2, \ldots)$, let $G(\mcD)$ be the demand matching graph (\Cref{def:demandMatching}), let $T_1, T_2, \ldots$ be the trees of a minimum size forest cover with $\alpha$ forests of $G(\mcD)$ (\Cref{def:demandMatching}) and let $D_{T_1}, D_{T_2}, \ldots$ be the corresponding tree matching demands (\Cref{def:treeMatchDemand}). Then, the matching-dispersed demand on nodes $u,v \in V$ is
\begin{align*}
    \mathrm{MD}_{\mcD, A}(u,v) := \frac{1}{4\alpha} \cdot \sum_i \sum_{u' \in \copies_{A}(u)}  \sum_{v' \in \copies_{A}(v)}D_{T_i}(u', v')
\end{align*}
\end{definition}

We begin with a simple helper lemma that observes that the matching-dispersed demand has size essentially equal to the input demands (up to the arboricity).
\begin{lemma}[\cite{HaeuplerHT24}]\label{lem:matchingDemandSize}
    Given graph $G$, node-weighting $A$ and and demands $\mcD = (D_1, D_2, \ldots)$ where $G(\mcD)$ has arboricity $\alpha$, we have that the matching-dispersed demand $\mathrm{MD}_{\mcD, A}$ satisfies $|\mathrm{MD}_{\mcD, A}| \geq \frac{1}{4 \alpha} \sum_i |D_i|$.
\end{lemma}
\begin{proof}
    Observe that the number of edges in $G(\mcD)$ is exactly $\sum_i |D_i|$ and so summing over each tree $T_j$ in our forest cover and applying \Cref{lem:treeMatchSize} gives
    \begin{align*}
        \sum_j |D_{T_j}| \geq  \sum_i |D_i|
    \end{align*}
    Combining this with the definition of $\mathrm{MD}_{\mcD, A}$ (\Cref{dfn:matchingDemand}) gives the claim.
\end{proof}

\subsection{Sparsity Proof for Union of Sparse Vertex Moving Cut}

We now argue the key properties of the matching-dispersed demand which will allow us to argue that it can be used as a witnessing demand for $\sum_i D_i$.
\begin{lemma}[Properties of Matching-Dispersed Demand\cite{HaeuplerHT24}]\label{lem:sparseOfMatching}
    Given graph $G = (V, E)$ and node-weighting $A$, let $C_1, C_2, \ldots$ be a sequence of moving cuts where $C_i$ is $(h, s)$-length $\phi_i$-sparse in $G - \sum_{j < i} C_i$ w.r.t.\ $A$ with witnessing demands $\mcD = (D_1, D_2, \ldots)$. Then the matching dispersed demand $\mathrm{MD}_{\mcD, A}$ is:
    \begin{enumerate}
        \item\label{item:property_md1} a $2h$-length $A$-respecting demand;
        \item\label{item:property_md2} $h \cdot (s-2)$-separated by $\sum_i C_i$ and;
        \item\label{item:property_md3} of size $|\mathrm{MD}_{\mcD,A}| \geq \frac{1}{s^3 \cdot \log^3 N \cdot N^{O(1/s)}} \sum_i \frac{|C_i|}{\phi_i} $.
    \end{enumerate}
\end{lemma}
\begin{proof}
    We show the each property separatedly.

    \noindent{\textbf{Proof of \Cref{item:property_md1}.}} For a pair of vertices $u$ and $v$, we note $\mathrm{MD}_{\mcD, A}(u, v) > 0$ only if there is a path consisting of at most two edges between a node in $\copies_{A}(u)$ and a node in $\copies_{A}(v)$ in the demand matching graph $G(\mcD)$. 
    Furthermore, $u' \in \copies_{A}(u)$ and $v' \in \copies_{A}(v)$ is connected by an edge in $G(\mcD)$ only if there is some witness demand $D_{i}$ such that $D_{i}(u, v) > 0$.
    Since each $D_{i}$ is $h$-length, we have $\dist_{G}(u, v) \le h$, and thus by triangle inequality $\mathrm{MD}_{\mcD, A}(u, v)$ is $2h$-length in $G$.

    To see that $\mathrm{MD}_{\mcD, A}(u, v)$ is $A$-respecting, we observe that each vertex in $G(\mcD)$ is incident to at most $2\alpha$ matchings across all of the tree matching demands.
    Thus for any $u \in V$ since $|\copies_{A}(u)| = 2A(u)$ we have
    \[\sum_{u' \in \copies_{A}(u)}\sum_{j}\sum_{v}\sum_{v' \in \copies_{A}(v)} D_{T_j}(u', v') \le \sum_{u' \in \copies_{A}(u)}2\alpha \le 4\alpha\cdot A(u).\]
    Then for any $u \in V$ in matching-dispersed demand we have,
    \[\sum_{v}\mathrm{MD}_{\mcD, A}(u, v) = \sum_{v} \frac{1}{4\alpha} \sum_{j}\sum_{u' \in \copies_{A}(u)}\sum_{v' \in \copies_{A}(v)} D_{T_j}(u', v') \le A(u)\]
    A symmetric argument shows that $\sum_{v}\mathrm{MD}_{\mcD, A}(v, u) \le A(u)$, and so we have that $\mathrm{MD}_{\mcD, A}$ is $A$-respecting.
    
    \noindent{\textbf{Proof of \Cref{item:property_md2}.}} Consider an arbitrary pair of vertices $u$ and $v$ such that $\mathrm{MD}_{\mcD, A}(u,v) > 0$; it suffices to argue that $\sum_i C_i$ indeed $h(s-2)$-separates $u$ and $v$.
    Given $\mathrm{MD}_{\mcD, A}(u, v) > 0$,
    if $u' \in \copies_{A}(u)$, $v' \in \copies_{A}(v)$ and $(u', v')$ is an edge in $G(\mcD)$, then we know some $D_i$ has value between $u$ and $v$ and $C_{i}$ $hs$-separates them.
    
    Otherwise, there is a path $(u', w', v')$ in $G(\mcD)$ where $u' \in \copies_{A}(u)$, $v' \in \copies_{A}(v)$ and for some $w \in V$ we have $w' \in \copies_{A}(w)$.
    It shows that there exists some demand $D_{i}$ and $D_{j}$ that has value between $\{u,w\}$ and $\{w,v\}$ respectively.
    By definition of $G(\mcD)$ (\Cref{def:demandMatching}), each $D_i$ corresponds to a different matching in $G(\mcD)$ and so since $\{u', w'\}$ and $\{w', v'\}$ share the vertex $w'$, we may assume $i \neq j$ and without loss of generality that $i < j$. 
    Let $G_{\leq i}$ be $G$ with $\sum_{l \leq i} C_l$ applied.

    Since $D_i$ is $hs$-separated by $C_{\leq i}$, we know that
    \begin{align}
        \dist_{G_{\leq i}}(u, w) \geq hs.\label{eq:a}
    \end{align}

    On the other hand, since $D_j$ is an $h$-length demand, and $j > i$, we know that the distance between $w$ and $v$ in $G_{\leq i}$ is
    \begin{align}
        d_{G_{\leq i}}(w, v) \leq h. \label{eq:b}
    \end{align}

    Thus, it follows that $C_{\leq i}$ must $h(s-2)$ separate $u$ and $v$ since otherwise we would know that $d_{G_{\leq i}}(u,w) \leq h(s-2)$ and so combining this with \Cref{eq:b} and the triangle inequality we get $d_{G_{\leq i}}(u, w) \leq hs - h$, contradicting \Cref{eq:a}. Thus, $\sum_{i} C_i$ must $h(s-2)$ vertices $u$ and $v$.
    
    \noindent{\textbf{Proof of \Cref{item:property_md3}.}} By \Cref{lem:matchingDemandSize}, we know that $|\mathrm{MD}_{\mcD, A}| \ge \frac{1}{4\alpha} \sum_{i}|D_{i}|$ where $\alpha$ is the arboricity of $G(\mcD)$. 
    From \Cref{lem:arbBound}, we have a bound of $s^3 \cdot \log^3 N \cdot N^{O(1/s)}$ on the arboricity of $G(\mcD)$.
    Further, each $D_{i}$ is a witnessing demand for a $\phi_{i}$-sparse cut $C_{i}$ such that $|D_{i}| \ge \frac{|C_{i}|}{\phi_{i}}$,
    we conclude
    \[|\mathrm{MD}_{\mcD, A}| \ge \frac{1}{4\cdot s^3 \cdot \log^3 N \cdot N^{O(1/s)}} \sum_{i} \frac{|C_i|}{\phi_{i}}\]
\end{proof}

We conclude with the proof of \Cref{thm:unionOfMovingCuts}:
\unionOfCuts*
\begin{proof}
Let $\widehat{C}$ be the union of sparse cuts and defined as $\widehat{C} = \sum_{i}C_{i}$.
Recall that to demonstrate that $\widehat{C}$ is a $\phi'$-sparse $(h', s')$-length sparse cut, it suffices to argue that there exists an $h'$-length $A$-respecting demand $D$ that is $h's'$-separated by $\widehat{C}$ where $|D| \geq  \frac{|\widehat{C}|}{\phi'}$.

\Cref{lem:sparseOfMatching} demonstrates the existence of exactly such a demand---namely the matching dispersed demand denoted by $D$ as defined in \Cref{dfn:matchingDemand}---for $h' = 2h$, $s' = \frac{(s-2)}{2}$ and $\phi' = s^3 \cdot \log^3 N \cdot N^{O(1/s)} \cdot \frac{\sum_i |C_i|}{\sum_i |C_i| / \phi_i}$.
\end{proof}
\begin{remark}
    We note that when taking the union of sparse cuts, we scale the cut value of $\widehat{C}$ to be the multiple of $1/h's'$ such that the length increase remains the same as $\sum_{i}C_{i}$.
    This gives $|\widehat{C}| = \frac{s}{s-2}\cdot \sum_{i} |C_{i}|$.
    For simplicity of presentation, we take $|\widehat{C}| = \sum_{i}|C_{i}|$ whenever the $\frac{s}{s-2}$ factor can be hidden in the $O(\cdot)$ notation. 
\end{remark}

\subsection{Application: Demand-Size-Largest Sequence of Sparse Cuts}
\begin{definition}[\qLDSCS] Fix a graph $G$, node-weighting $A$ and parameters $h$, $s$ and $\phi$. $\LDSCS$ is the demand-size of the demand-size-largest sequence of $(h,s)$-length $\phi$-sparse moving cuts. Specifically,
    \begin{align*}
        \LDSCS(\phi, h, s)  := \sum_i A_{(h,s)}(C_i)
    \end{align*}
    where above $A_{(h,s)}(C_i)$ is computed after applying all $C_j$ for $j < i$ and $(C_1, C_2, \ldots)$ is the $(h,s)$-length $\phi$-expanding moving cut sequence maximizing $\sum_i A_{(h,s)}(C_i)$.
\end{definition}

\begin{restatable}{lemma}{LDSCSAtMostLDSC}\label{lem:LDSCSAtMostLDSC}
    Given graph $G$ and node-weighting $A$, we have that 
    \begin{align*}
    \LDSCS(\phi, h, s) \leq s^3 \cdot \log^3 n \cdot n^{O(1/s)} \cdot \LDSC(\phi', h', s')
    \end{align*}
    where $h' = 2h$, $s' = \frac{(s-2)}{2}$ and $\phi' = s^3 \cdot \log^3 n \cdot n^{O(1/s)} \cdot \phi$.
\end{restatable}
\begin{proof}
    Let $(C_1, C_2, \ldots)$ be a sequence of $(h,s)$-length $\phi$-sparse cuts such that $\sum_i A_{(h,s)}(C_i) = \qLDSCS(\phi, h, s)$ where $\phi_i \leq \phi$ is the minimum value for which $C_i$ is $\phi_i$-sparse. Let $D_1, D_2, \ldots$ be the demands witnessing these cuts so that for all $i$ we have
    \begin{align*}
        |C_i|/\phi_{i} = |D_i|.
    \end{align*}
    Let $\widehat{C} = \sum_i C_i$ where $|\widehat{C}| = \sum_{i} |C_{i}|$. 
    By \Cref{thm:unionOfMovingCuts} we know that $\widehat{C}$ is an $(h',s')$-length $\phi'$-sparse cut for $A$.
    What is more, we can further tighten the sparsity of $\widehat{C}$ by each $\phi_{i}$, and we have
    \[\spa_{(h',s')}(\widehat{C}, A) = s^3 \cdot \log^3 n \cdot n^{O(1/s)} \cdot \frac{\sum_{i} |C_{i}|}{\sum_{i}|C_{i}|/\phi_{i}},\]
    and so there must be some $h'$-length demand $D$ which is $h's'$-separated by $\widehat{C}$, and we further have
    \begin{align*}
        s^3 \cdot \log^3 n \cdot n^{O(1/s)} \cdot \frac{|\widehat{C}|}{\spa_{(h',s')}(\widehat{C}, A)} = |\widehat{C}|\cdot\frac{\sum_i |C_i|/\phi_i}{\sum_i |C_i|} = \sum_i |D_i|.
    \end{align*}
    In other words, the demand-size of $\widehat{C}$ multiplied by a sparsity slack is at least $\sum_i |D_i|$, as required.
\end{proof}
\subsection{Application: Demand-Size Largest Sparse Cut at Most Largest Expanders' Complement}
\begin{definition}[Largest Expander's Complement Size]\label{def:LEC}
    $\qLEC(\phi, h, s)$ is $\phi$ times the size of the complement of the largest $(h,s)$-length $\phi$-expanding subset of $A$. That is, let $\hat{A}$ be the $(h,s)$-length $\phi$-expanding node-weighting on $G$ satisfying $\hat{A} \preceq A$ with largest size and let $\bar{A} = A - \hat{A}$ be its complement. Then 
    \begin{align*}
        \qLEC(\phi, h, s) := \phi \cdot |\bar{A}|.
    \end{align*}\label{eqv: node weighting}
\end{definition}
Observe that applying our previous relations we can get a simple lower bound on the demand-size of the largest-demand-size length-constrained sparse cut.
\begin{restatable}{lemma}{lemDLSCAtMostLEC}\label{lem:DLSCAtMostLEC}
    Given graph $G$ and node-weighting $A$ and parameters $h,s,\phi$, we have that 
    \begin{align*}
        \qLDSC(\phi, h, s) \leq (1/\phi)\cdot \qLEC(\phi', h', s)
    \end{align*}
    where $\phi' = \tilde{O}(\phi \cdot s^3 \cdot n^{O(1/s)})$ and $h' = 2h$.
\end{restatable}

\subsubsection{Largest Weighted Sparse Cut Sequence At Most Largest Cut}
\begin{definition}[Largest Weighted Sparse Cut Sequence Size]\label{def:LWSC}
    $\qLWSC(\phi, h, s)$ is the largest weighted size of a sparse cut sequence, namely
    \begin{align*}
        \qLWSC(\phi, h, s) := \sum_i \frac{\phi}{\spa_{(h,s)}(C_i,A)} \cdot |C_i|
    \end{align*}
    where $(C_1, C_2, \ldots)$ is the $(h,s)$-length $\phi$-expanding moving cut sequence maximizing $\sum_i \frac{\phi}{\spa_{(h,s)}(C_i,A)} \cdot |C_i|$ and each $\spa_{(h,s)}(C_i, A)$ is computed after applying $C_j$ for $j <i$.
\end{definition}

\begin{definition}[Largest Sparse Cut Size]\label{def:LC}
    $\qLC(\phi, h, s)$ is the size of the largest $(h,s)$-length $\phi$-sparse cut in $G$ w.r.t.\ $A$. That is 
    \begin{align*}
        \qLC(\phi, h, s) := |C_0|
    \end{align*}
    where $C_0$ is the moving cut of largest size in the set $\{C : \spa_{(h,s)}(C,A) \leq \phi \}$.
\end{definition}

\begin{restatable}{lemma}{lemLWSCAtMostLC}\label{lem:LWSCAtMostLC}
Given graph $G$ and node-weighting $A$, we have that 
\begin{align*}
\qLWSC(\phi, h, s) \leq \qLC(\phi', h', s')
\end{align*}
where $\phi' = \phi \cdot s^3 \cdot \log^3 n \cdot n^{O(1/s)}$, $h' = 2h$ and $s' = \frac{(s-2)}{2}$.
\end{restatable}
\begin{proof}
    Let $(C_{1}, C_{2},\ldots)$ be the $(h,s)$-length $\phi$-sparse moving cut sequence w.r.t $A$ in $G$ maximizing the weighted cut sequence size as defined in \Cref{def:LWSC}. 
    Similarily take $\widehat{C} = \sum_{i} C_{i}$ where $|\widehat{C}| = \sum_{i} |C_{i}|$.
    By \Cref{thm:unionOfMovingCuts} we know that $\widehat{C}$ is an $(h',s')$-length $\phi'$-sparse cut for $A$.
    Naturally, the weighted cut sequence size is related to the size of the union of cuts.
    We note that,
    \begin{align*}
        \sum_i \frac{\phi}{\spa_{(h,s)}(C_i,A)} \cdot |C_i| & = \phi \cdot \sum_{i} |C_{i}|/\spa_{(h,s)}(C_{i}, A)\\
        & = \phi \cdot \frac{\sum_{i} |C_{i}|/\spa_{(h,s)}(C_{i}, A)}{\sum_{i} |C_{i}|} \cdot |\widehat{C}|\\
        & = \frac{\phi'}{\spa_{(h', s')}(\widehat{C}, A)}\cdot |\widehat{C}| .
    \end{align*}
    Then it suffices to argue that there exists some $(h', s')$-length $\phi'$-sparse moving cut $C'$ can achieve the above size.
    The basic idea to get $C'$ is to simply add arbitrary length increases to $\widehat{C}$ such that the size increases while $C'$ remains to be $\phi'$-sparse.
    Namely, based on $\widehat{C}$, we arbitrarily increase the cut value to get an $(h',s')$-length moving cut $C'$ where $|C'| = \frac{\phi'}{\spa_{(h', s')}(\widehat{C}, A)}\cdot |\widehat{C}|$. 
    We note that $C'$ also $h's'$-separates the witness demand $D$ of cut $\widehat{C}$, thus
    \[\spa_{(h', s')}(C', A) \le \frac{|C'|}{|D|} = \frac{\phi' \cdot |\widehat{C}|}{|D| \cdot \spa_{(h', s')}(\widehat{C}, A)} = \phi'\].
    This concludes the lemma.
\end{proof}

\subsubsection{Largest Cut At Most Largest Expander's Complement}
\begin{restatable}{lemma}{lemLCAtMostLEC}\label{lem:LCAtMostLEC}
Given graph $G$ and node-weighting $A$ and parameters $h,s,\phi$, we have that 
\begin{align*}
\qLC(\phi, h, s) \leq \qLEC(\phi', h, s)  
\end{align*}
where $\phi' = 3 \phi$.
\end{restatable}
Intuitively, if there exists a large expanding sub-node-weighting, then any sparse cut would have relatively small size, otherwise it contradicts the fact that no sparse cut can cut too much into an expanding sub-node-weighting.
To formally show the proof of above lemma, we first formalize the idea with projected down demand\cite{HaeuplerHT24}.
\paragraph{Projected Down Demands.}
\begin{definition}[Projected Down Demand\cite{HaeuplerHT24}]\label{def:projDemand}
Suppose we are given graph $G$, node-weighting $A$, $A$-respecting demand $D$ and $\hat{A} \preceq A$ where $\bar{A} := A - \hat{A}$ is the complement of $\hat{A}$. Then, let $D^+$ be any demand such that $\sum_v D^+(u, v) = \min(\bar{A}(u), \sum_v D(u,v))$ for every $u$ and  $D^+ \preceq D$. Symmetrically, let $D^-$ be any demand such that $\sum_v D(v,u) = \min(\bar{A}(u), \sum_v D(v,u))$ and $D^- \preceq D$. Then we define the demand $D$ projected down to $\hat{A}$ on $(u,v)$ as
\begin{align*}
    D^{\downharpoonright \hat{A}}(u,v) := \max(0, D(u,v) - D^+(u,v) - D^-(u,v)).
\end{align*}
\end{definition}
The following establishes the basic properties of $D^{\downharpoonright \hat{A}}$.
\begin{lemma}\label{lem:projDemProps} Given graph $G$, node-weighting $A$, $\hat{A} \preceq A$ where $\bar{A} := A - \hat{A}$, we have that $D^{\downharpoonright \hat{A}}$ is $\hat{A}$-respecting, $|D^{\downharpoonright \hat{A}}| \geq |D| - 2 |\bar{A}|$ and $D^{\downharpoonright \hat{A}} \preceq D$.
\end{lemma}
\begin{proof}
We do a case analysis by the 
minimizer of $\min(\bar{A}(u), \sum_v D(u,v))$.
\begin{itemize}
    \item If $\sum_v D^+(u,v) = \bar{A}(u)$ (where $D^+$ is defined in \Cref{def:projDemand}) then by the fact that $D$ is $A$-respecting we have
    \begin{align*}
        \sum_v D^{\downharpoonright \hat{A}}(u,v) &\leq \sum_v D(u,v) - \sum_{v} D^+(u,v) \\
        &\leq A(v) - \bar{A}(u) \\
        & = \hat{A}(u).
    \end{align*}
    \item On the other hand, if $\sum_v D^+(u,v) = \sum_v D(u,v)$ then by the non-negativity of node-weightings we have \begin{align*}
        \sum_v D^{\downharpoonright \hat{A}}(u,v) &\leq \sum_v D(u,v) - \sum_{v} D^+(u,v) \\
        &= 0 \\
        & = \hat{A}(u).
    \end{align*}
\end{itemize}

In either case we have $\sum_v D^{\downharpoonright \hat{A}}(u,v) \leq \hat{A}(u)$. A symmetric argument using  $D^-$ (where $D^-$ is defined in \Cref{def:projDemand}) shows that $\sum_v D^{\downharpoonright \hat{A}}(v,u) \leq \hat{A}(u)$ and so $D^{\downharpoonright \hat{A}}(u,v)$ is $\hat{A}$-respecting.

To see that $|D^{\downharpoonright \hat{A}}| \geq |D| - 2 |\bar{A}|$, observe that by definition, $|D^{\downharpoonright \hat{A}}| \geq |D|-|D^+| - |D^-|$. But, also by definition, $|D^+|, |D^-| \leq |\bar{A}|$, giving the claim. Lastly, observe that $D^{\downharpoonright \hat{A}} \preceq D$ trivially by construction.
\end{proof}
We can now formally prove the main result of this section.
\lemLCAtMostLEC*
\begin{proof}
    Let $C_0$ be the $(h,s)$-length $\phi$-sparse cut of largest size w.r.t.\ $A$ in $G$ and let $\bar{A}$ be the complement of the largest $(h,s)$-length $\phi'$-expanding subset $\hat{A} \preceq A$ as in \Cref{def:LEC}.

    Let $D$ be the demand that witnesses the $(h,s)$-length $\phi$-sparsity of $C_0$; that is, it is the minimizing $A$-respecting demand in \Cref{def:sparsity}. We may assume, without loss of generality, that $C_0$ $hs$-separates all of $D$; that is, $\sep_{hs}(C_0, D) = |D|$. Let $D^{\downharpoonright \hat{A}}$ be the projected down demand (as in \Cref{def:projDemand}). Recall that by \Cref{lem:projDemProps} we know that $D^{\downharpoonright \hat{A}}$ is $\hat{A}$-respecting, $|D^{\downharpoonright \hat{A}}| \geq |D| - 2 |\bar{A}|$ and $D^{\downharpoonright \hat{A}} \preceq D$.
    
    However, since $C_0$ $hs$-separates all of $D$ and $D^{\downharpoonright \hat{A}} \preceq D$ we know that $C_0$ must $hs$-separate all of $D^{\downharpoonright \hat{A}}$ and so applying this and $|D^{\downharpoonright \hat{A}}| \geq |D| - 2 |\bar{A}|$ we have
    \begin{align}\label{eq:ay}
        \spa_{s \cdot h}(C_0, D^{\downharpoonright \hat{A}}) = \frac{|C_0|}{\sep(C_0, D^{\downharpoonright \hat{A}})} = \frac{|C_0|}{|D^{\downharpoonright \hat{A}}|}\leq  \frac{|C_0|}{|D| - 2 |\bar{A}|}
    \end{align}
    where, as a reminder, $\spa$ is defined in \Cref{dfn:CDSparse} and \Cref{def:sparsity}.
    
    On the other hand, since $D^{\downharpoonright \hat{A}}$ is $\hat{A}$-respecting and $\hat{A}$ is $(h,s)$-length $\phi'$-expanding by definition, we know that no cut can be too sparse w.r.t.\ $\hat{A}$ and, in particular, we know that
    \begin{align}\label{eq:by}
        3\phi = \phi' \leq \spa_{(h,s)}(C_0, \hat{A}) \leq \spa_{s
        \cdot h}(C_0, D^{\downharpoonright \hat{A}})
    \end{align}
    
    Combining \Cref{eq:ay} and \Cref{eq:by} and solving for $\phi' \cdot |\bar{A}|$ we have
    \begin{align*}
        \frac{3\phi \cdot |D| - |C_0|}{2} \leq \phi' \cdot |\bar{A}|.
    \end{align*}
    
    However, recall that $C_0$ is an $(h,s)$-length $\phi$-sparse cut witnessed by $D$ and, in particular, this means that $\frac{1}{\phi} |C_0|\geq |D|$. Applying this we conclude that 
    \begin{align*}
        |C_0| \leq \phi' \cdot |\bar{A}|.
    \end{align*}
    as required.
\end{proof}

\subsubsection{Proof of \Cref{lem:DLSCAtMostLEC}.}
With \Cref{lem:LWSCAtMostLC,lem:LCAtMostLEC} as shown above, it is rather straight forward to prove \Cref{lem:DLSCAtMostLEC}.
\lemDLSCAtMostLEC*
\begin{proof}
    Let $C$ be the $(h,s)$-length $\phi$-sparse cut of largest demand-size. 
    Observe that $C$ is also a candidate for the largest weighted sparse cut sequence, and thus
    \[\frac{\phi}{\spa_{(h,s)}(C,A)}\cdot |C| = \phi \cdot \qLDSC(\phi, h, s) \leq \qLWSC(\phi, h, s)\]
    Further, by \Cref{lem:LWSCAtMostLC} we have 
    \begin{align*}
        \qLWSC(\phi, h, s) \leq \qLC(\phi', h', s')
    \end{align*}
    where $\phi' = \tilde{O}(\phi \cdot s^3 \cdot n^{O(1/s)})$, $h' = 2h$ and $s' = (s-2)/2$. Lastly, by \Cref{lem:LCAtMostLEC} we have 
    \begin{align*}
        \qLC(\phi', h', s') \leq \qLEC(3\phi', h', s').
    \end{align*}
    Combining the above and observing that $\qLEC(3\phi', h', s') \leq \qLEC(3\phi', h', s)$ gives the lemma.
\end{proof}

\section{Approximating Demand-Size-Largest Sparse Cuts (DLSC)}\label{sec:apxDLSC}

In this section, we will prove the following lemma. 
\begin{restatable}{lemma}{ApxSparsecut}\label{lem:approxsparsecut}
        There exists an algorithm $\ADLSC{}(G,A,h,s,\phi,\eps)$ that is given a \emph{vertex-capacitated graph} $G=(V,E,\l,u)$, a length slackness $s>2$, a length bound $h>s^{c/\eps}$ for some sufficiently large constant $c$, a conductance parameter $\phi\in(0,1)$, a parameter $\eps\in (1/\log^{0.1}n,1)$, either output an empty cut, or computes an $\alpha$-approximate demand-size-largest $(h,s)$-length $\phi$-sparse cut with sparsity approximation $\alpha_\phi$, length slackness approximation $\alpha_s$ where
        \[\alpha = s^3\cdot n^{O(\eps+\frac{1}{s})} \qquad \qquad \alpha_\phi = s^3 \cdot n^{O(\eps+\frac{1}{s})} \qquad \qquad \alpha_s = s^{O(1/\eps)}\]
        
        In the case when the algorithm outputs an empty cut, it also outputs a $(h/s^{O(1/\eps)},s/\eps)$-length $\phi/n^{O(\eps)}$-expansion witness for $G$.
        
        The algorithm uses $\poly(h)\cdot m\cdot n^{f(\eps)}$ work and $\poly(h)\cdot n^{f(\eps)}$ depth, where $\lim_{\eps\to 0}f(\eps)\to 0$.
\end{restatable}

To recall the definition of approximately demand-size-largest sparse cut, we restate \Cref{dfn:apxLargeCut} as follows.

\apxsparsecut*

The idea to prove \cref{lem:approxsparsecut} is to use the constant-hop cut-matching game introduced by \cite{HaeuplerH025}, which we will explain in the next section.

\subsection{Constant Steps Cut Matching Games}

The cut-matching game in the literature was used to find sparse cuts in the classical setting where the cut strategy simply returns a bi-partition of the vertex set. This classical cut-matching game will generate an expander with polylogarithmic hops, which is too much for our case where the hop corresponds to the length slackness, and we want the length slackness to be a constant. Thus, \cite{HaeuplerH025} proposed a generalized cut-matching game where instead of simply bi-partitioning the vertex set, they do multiple partitions (in our case partition the node-weighting). This generalized version allows them to achieve constant hops, We define them formally as follows.

\paragraph{Definition.} In this section, only for the definition of a cut-matching game, we will forget about the vertex capacity and length and only consider edge-capacitated graphs $G=(V,E)$ where each edge $e\in E$ is represented as $(\text{id}_e,\text{capacity}_e)$, i.e., associated with a capacity. We define $U:E\to\bbN^+$ as the edge capacity function. This is because the cut-matching game was only defined with congestion on edges, which also suffices for our case (this should be clear in the next section). 

\paragraph{Cut strategies.} A cut strategy is an algorithm that is given an undirected graph $G=(V,E)$ and a node-weighting $A$, outputs a set of node-weightings pairs $\{(A^{(j)},B^{(j)})\}_j$ satisfying
\begin{enumerate}
    \item for every $j$ we have $|A^{(j)}|=|B^{(j)}|$, and $A^{(j)}+B^{(j)}\preceq A$,
    \item $A\preceq \sum_{j}(A^{(j)}+B^{(j)})$.
\end{enumerate}

\paragraph{Matching strategies.} A matching strategy is an algorithm that is given a graph $G$ and node-weighting pairs $\{(A^{(j)},B^{(j)})\}_j$, outputs for each $j$ a set of capacitated edges $M^{(j)}\subseteq \supp(A^{(j)})\times\supp(B^{(j)})$ such that for every vertex $u$ we have $U(\delta_{M^{(j)}}(u))\le A^{(j)}(u)$ and $U(\delta_{M^{(j)}}(u))\le B^{(j)}(u)$.

\paragraph{Cut-matching games.} A cut-matching game is an algorithm that is given a cut strategy and matching strategy, starts from an empty edge-capacitated graph $G_0=(V,E_0=\emptyset)$ and runs in a certain number of rounds (denoted by $r$) to update $G_0$. In the $i$-the round (starting from $i=1$), the algorithm first applies the cut strategy on the graph $G_{i-1}$ and gets a set of node-weighting pairs $\{(A^{(j)}_i,B^{(j)}_i)\}_j$. Then the algorithm applies the matching strategy on $\{(A^{(j)}_i,B^{(j)}_i)\}_j$ to get $\{M^{(j)}_i\}_j$ and get the graph $G_i=(V,E_i)$ where $E_i=E_{i-1}\bigcup\left(\cup_jM^{(j)}_i\right)$. The resulting graph of a cut-matching game is $G_r$. We care about the following parameters of a cut-matching game.

\begin{itemize}
    \item \textbf{Rounds of Interaction:} The number of rounds $r$.
    \item \textbf{Cut Batch Size:} The maximum number of pairs the cut strategy plays in each round of interaction, i.e., $\max_i |\{(A_{i}^{(j)}, B_i^{(j)})\}_j|$. The classical cut matching games can be viewed as the special case when the cut batch size is $1$. We need larger cut batch size in order to achieve constant hop.
    \item \textbf{Matching Perfectness:} If each set of edges the matching player plays for a batch always has a total capacity of at least a $1-\alpha$ fraction of the total node-weighting then we say that the cut matching game is $(1-\alpha)$-perfect. That is, a cut matching game is $(1-\alpha)$-perfect if for every $i$ we have
        \[\sum_j U\left(M_i^{(j)}\right) \geq (1-\alpha) \cdot \sum_j|A_i^{(j)}| = (1-\alpha) \cdot \sum_j|B_i^{(j)}|\]
\end{itemize}

In \cite{HaeuplerHT24}, they proved a high-quality cut strategy. By combining Theorem 1.1 and Theorem 13.6 of \cite{HaeuplerHT24}, we get the following lemma.

\begin{lemma}[Cut strategy, \cite{HaeuplerHT24}]\label{lem:cut-strategy}
    For every $\eps\in\left(0,1\right)$ (possibly a function of $n$), there is a cut strategy with cut batch size $n^{O(\eps)}$ which when used in a cut matching game with $1/\eps$ rounds of interaction and an arbitrary $(1-\alpha)$-perfect matching strategy results in a $G_r$ that is a $1/\eps$-step and $n^{O(\eps)}$-router (see \Cref{def:routers}) for some $A' \preceq A$ of size $|A'| \geq (1-O(\frac{\alpha}{\eps})) \cdot |A|$. 
    This cut-strategy costs work $m\cdot n^{f(\eps)}\cdot \poly(h)$ and depth $n^{f(\eps)}\cdot\poly(h)$ where $f(\eps)\to 0$ as $\eps\to 0$.
\end{lemma}

In order to prove \Cref{lem:approxsparsecut}, i.e., find an approximately demand-size-largest sparse cut of $G$, we will use cut-matching games to construct a graph $H$ on the same vertex set but with a different edge set. Specifically, we will use the cut strategy provided by \Cref{lem:cut-strategy} to construct $H$. The high-level idea is to find a $(1-\alpha)$-perfect matching after each round of cut strategy and add them to $H$ so that we can use \Cref{lem:cut-strategy} to argue that $H$ becomes a good router in the end; moreover, we will make sure that each edge in $H$ (i.e., in each matching) corresponds to a path in the original graph $G$, such that these path are vertex length bounded with low vertex congestion. For that purpose, the matching player can be viewed as routing as many as possible demands in the original graph with length bound and low congestion, and each successfully routed demand corresponds to an edge in $H$. The following lemma implied by \cite{HaeuplerHS23} is for that purpose. 

\begin{lemma}[Application of Theorem 16.1 \cite{HaeuplerHS23}]\label{lem:matching-strategy}
    There is an algorithm that is given a graph $G=(V,E)$ with vertex capacities and length, two node-weightings $A,B\subseteq V$, length bound $h\ge 1$ and sparsity $0<\phi<1$, outputs a pair of flow and moving cut $(F,C)$ such that
    \begin{enumerate}
        \item $F$ has vertex congestion $\tO{1/\phi}$, vertex length $h$ and routes some $\hA\preceq A$ to $\tB\preceq B$.
        \item $C$ is a $h$-length vertex moving cut, $\{v\mid \hA(v)<A(v)\}$ and $\{v\mid \tB(v)<B(v)\}$ are $h$-far on $G-C$, and $C$ has size at most
        \[|C|\le \phi\cdot\left(|A|-\val(F)\right)\]
    \end{enumerate}
    The algorithm runs in $\tO{m\cdot\poly(h)}$ work and $\tO{\poly(h)}$ depth.
\end{lemma}
\begin{proof}
    We first apply standard transformation from vertex capacity to directed edge capacity in order to apply Theorem 16.1 of \cite{HaeuplerHS23}. We construct a \emph{directed graph} with edge capacities and length $G'=(V_{in}\cup V_{out}\cup\{s,t\},E')$ defined as follows.
    \begin{itemize}
        \item $V_{in},V_{out}$ are both copies of $V$, we write $v^{in},v^{out}$ for the corresponding copies for $v\in V$.
        \item $E'=\{(v^{in},v^{out})\mid v\in V\}\cup\{(u^{out},v^{in}\mid (u,v)\in E\}\cup\{(s,u^{in}),(v^{out},t)\mid u\in\supp(A),v\in\supp(B)\}$.
        \item Each edge $(v^{in},v^{out})$ get capacity $U(v)$ and length $\l(v)$; each edge $(u^{out},v^{in})$ get sufficiently large capacity, for example, $10\max_{v\in V}U(v)$, and length $1$\footnote{Actually we want the length to be $0$ here, but in \cite{HaeuplerHS23}, they have the assumption that lengths are at least 1}; each edge $(s,u^{in})$ gets capacity $A(u)$ and length $1$, each edge $(v^{out},t)$ gets capacity $B(v)$ and length $1$.
    \end{itemize} 
    We now apply Theorem 16.1 of \cite{HaeuplerHS23} with parameters $S=\{s\},T=\{t\}$ and $h,\phi$ to be $3\cdot\text{(our length bound h)},(\text{our sparsity }\phi)/3$ in \Cref{lem:matching-strategy}. We will get $(\hat{f},\hat{w})$ in $\tO{m\cdot\poly(h)}$ work and $\tO{\poly(h)}$ depth. 
    
    \paragraph{Getting $F$.} We transfer $\hat{f}$ into a flow in $F$ in $G$ in the following way. Every flow path $p\in\hat{f}$ must be like $(s,v_1^{in},v_1^{out},v_2^{in},v_2^{out},...,v^{in}_k,v_k^{out},t)$ for some sequence $(v_1,...,v_k)$ due to the definition of $G'$. We let $(v_1,...,v_k)$ be the corresponding flow in $F$ if $(s,v_1^{in})$ is saturated by $\hat{f}$. In this way, as long as $\hat{f}$ has edge-congestion $\tO{1/\phi}$ (as implied by Theorem 16.1 of \cite{HaeuplerHS23}), each vertex $v$ cannot have vertex congestion more than $\tO{1/\phi}\cdot U(v)$ in $F$ since each of them corresponds to an edge $v^{in},v^{out}$ used in $\hat{f}$. Thus, $F$ has vertex-congestion at most $\tO{1/\phi}$. Moreover, $F$ routes some $\hA\preceq A$ to $\tB\preceq B$. So $F$ is a valid output.

    \paragraph{Getting $C$.} For each $v\in V$, we assign the moving cut $C(v)$ as 
    \[C(v)/3=\hat{w}(v^{in},v^{out})+\left(\sum_{(v,u)\in E}\hat{w}(v^{out},u^{in})\right)+\hat{w}(s,v^{in})\] 
    The last term only exists when $v\in \supp(A)$. Clearly, we have 
    \[|C|\le 3\cdot \sum_{v\in V}C(v)\cdot U(v)\le 3\cdot \sum_{e\in E'}\hat{w}(e)\cdot U(e)\le \phi\cdot\left(\left(\sum_{(s,v^{in})\in E'}U(s,v^{in})\right)-\val(\hat{f})\right)\]
    Notice that the last inequality is from Theorem 16.1 of \cite{HaeuplerHS23}. According to the definition of $G'$, we have $\sum_{(s,v^{in})\in E'}U(s,v^{in})=|A|$, and $\val(\hat{f})=\val(F)$, so $|C|\le \phi\cdot\left(|A|-\val(F)\right)$. 
    
    It remains to verify $\{v\mid \hA(v)<A(v)\}$ and $\{v\mid \tB(v)<B(v)\}$ are $h$-far on $G-C$. According to Theorem 16.1 of \cite{HaeuplerHS23}, the distance of $s,t$ in $G'$ is at least $3h$ under the distance function where (i) all edges $(s,v^{in})$ or $(t,v^{out})$ that are saturated by $\hat{f}$ has distance $3h+1$, (ii) all other edges $e$ has distance in $G'$ increased by $\hat{w}(e)\cdot 3h$. Suppose there is a path $p_G$ in $G$ from $a\in\{v\mid \hA(v)<A(v)\}$ to $b\in\{v\mid \tB(v)<B(v)\}$ with distance less than $h$ in $G-C$, this corresponds to a path in $G'$ denoted as $p'_G=(s,a^{in},a^{out},...,b^{in},b^{out},t)$ where $(s,a^{in}),(b^{out},t)$ are both not saturated by $\hat{f}$. The edge length of $p$, according to the definition of $C$, is at most
    \[\l(p'_G)\le (|P|+1)+\l_{C,h}(P_G)<3h\]

    The first term $(|P|+1)$ is from the fact that we set the initial length of every edge to be at least $1$. The second term is because every length increases on edges by $\hat{w}$ is absorbed by the same length increase on its adjacent vertex by $C$. This is a contradiction that $s,t$ in $G'$ is at least $3h$-far.  
\end{proof}
\subsection{Approximating DLSC from Cut-Matching Games (Proof of \Cref{lem:approxsparsecut})}\label{subsec:apxDLSCcutstrategy}

This section is adapted from Section 13.2 of \cite{HaeuplerHT24}, with the simplification that we do not need to guarantee the cut strategy to work on a smaller node-weighting (for their paper they need this property to do recursion while gradually reducing the node-weighting size), and we do not do $(\le h,s)$-length decomposition as they do. 

In this section, we show how to combine the cut and matching strategy in \Cref{lem:cut-strategy,lem:matching-strategy} to get a cut-matching game that can prove \Cref{lem:approxsparsecut}. Intuitively, since a length-constraint expander is a low-step expander under every low-diameter cluster, we are going to apply our cut-matching game on every low-diameter cluster. Thus, \Cref{thm:neicov} will be used to get a collection of low-diameter clusters to apply cut-matching games. We describe our algorithm formally as follows.

\paragraph{Step 1: Create Clusters for Cut Matching Games.}\label{step:1} Apply \Cref{thm:neicov} to $G$ to compute a neighborhood cover $\cN$ with separation factor $2s$, covering radius $h_{\cov}=h/\left(\frac{1}{\eps} \cdot (2s)^{O(1/\eps)}\right)=h/s^{O(1/\eps)}$ (because $s>2$), cluster diameter $h_{\diam} = h$ and width $\omega = n^{O(\eps)}\cdot\log n$.

\paragraph{Step 2: Run Cut Matching Games.}\label{step:2} 
First, let 
\begin{align}\label{eq:phiPrime}
    \phi' := \phi / n^{O(\eps)}
\end{align}
be the (relaxed) sparsity with respect to which we will run our cut-matching game. 

Next, we do the following for each $\cS\in\cN$. Intuitively, we wish to run a cut-matching game on every cluster in $\cS$ which is a low-diameter cluster. We do so by running the cut player for each of them simultaneously and running the matching player by combining the source and sink node-weighting from all of them: difference clusters will not interfere with each other because they are far away from each other. Formally, we initialize empty graph $H_S=(S,\emptyset)$ for every $S\in\cS$ and we repeat the following from $i=1$ to $1/\eps$.

\begin{enumerate}
    \item \textbf{Run cut strategies:} For each $S \in \cS$, apply the the cut strategy (from \Cref{lem:cut-strategy}) to $H_S$. Let $\{A^{(j)}_{i,S},B^{(j)}_{i,S}\}_{j\in[b]}$ be the output pairs of node-weightings from the cut strategy for cluster $S$. Here we fix $b$ to be an upper bound for the cut batch size for all cut strategies over $S\in\cS$, which satisfies $b=n^{O(\eps)}$ according to \Cref{lem:cut-strategy}. For the cut strategy that returns less than $b$ pairs of node weightings, we simply augment it to $b$ pairs by appending empty node weightings to them.
    \item \textbf{Compute a cutmatch:} For each $j\in[b]$, we apply \Cref{lem:matching-strategy} on the graph $G$, two node-weightings $A^{(j)}_{i,\cS}:=\sum_{S\in\cS}A^{(j)}_{i,S},B^{(j)}_{i,\cS}:=\sum_{S\in\cS}B^{(j)}_{i,S}$, length bound $s\cdot h_{\diam}$ and sparsity $\phi'$, to get a flow and moving cut pair $(F^{(j)}_{i,\cS},C^{(j)}_{i,\cS})$. Define $F_{i,\cS}=\cup_jF^{(j)}_{i,\cS}$
    
    \item \textbf{Update graphs:} For each $S\in\cS$ and each $j\in[b]$, add a matching $M^{(j)}_{i,S}\subseteq \supp(A_{i,S}^{(j)})\times \supp(B_{i,S}^{(j)})$ to $H_S$ defined as follows: $(a,b)\in M^{(j)}_{i,S}$ if there is a flow path in $F_{i,\cS}$ from $a$ to $b$, and the capacity of this edge $(a,b)$ is the summation of $F_{i,\cS}(p)$ among all flow path $p$ from $a$ to $b$.
    
\end{enumerate}

\paragraph{Algorithm Output.}
We return as our cut $C^*$ the union of all cuts among all $C^{(j)}_{i,\cS}$ for $j\in[b],i\in[1/\eps]$ and $\cS\in\cN$. If $C^*$ is empty, then we return a witness:
\begin{itemize}
    \item \textbf{Neighborhood Cover:} we return the neighborhood cover $\cN$ constructed in Step 1, it has covering radius $h/s^{O(1/\eps)}$;
    \item \textbf{Routers:} an $s_0$-step $\kappa_0$-router of $A_S$ for each $S \in \cS\in \mcN$, for which we return the final graph $H_S$ in the cut-matching game. Notice that $s_0=1/\eps$ and $\kappa_0=n^{O(\eps)}$ according to \Cref{lem:routers}.
    
    \item \textbf{Embedding of Routers:} an $(h_{\cov} \cdot s_1)$-length edge-to-vertex embedding of $\cup_{S\in\cS\in\cN}R_S$ into $G$ with congestion $\kappa_1$. We take the embedding as in Step 2 (3): every edge in $H_S=R_S$ corresponds to a flow path in $F_{i,\cS}$ for some $i,\cS$. According to \Cref{lem:matching-strategy}, these flow path has length at most $h_{\diam}\cdot s$. For the congestion, notice that for each $j\in[b]$ and $i\in[1/\eps]$, the flow $F^{(j)}_{i,\cS}$ has congestion $\tO{1/\phi'}$ according to \Cref{lem:matching-strategy}. Thus, the total congestion is at most $n^{O(\eps)}/\phi$.
    Notice that $s_1=s$ and $\kappa_1=n^{O(\eps)}/\phi$ follows from the argument above.
\end{itemize}

\begin{lemma}\label{lem:routers}
    If $C^*$ is an empty cut, then $H_S$ is a $1/\eps$-step $n^{O(\eps)}$-router.
\end{lemma}
\begin{proof}
    As $C^*$ is an empty cut, every cut $C^{(j)}_{i,\cS}$ is an empty cu. According to \Cref{lem:cut-strategy}, that means $\hA^{(j)}_{i,\cS},\hB^{(j)}_{i,\cS}$ must be equal to $A^{(j)}_{i,\cS},B^{(j)}_{i,\cS}$ (here $\hA,\hB$ are the vertices saturated by the flow), otherwise there exists $a\in \supp(\hA^{(j)}_{i,\cS}),b\in \supp(\hB^{(j)}_{i,\cS})$ such that $a,b$ are $h_{\diam}\cdot s$-far. In this case, we must have $a,b$ in different clusters $S\in\cS$. In other words, if we look at the cluster $S$ that contains $a$, we have that all vertices in $B^{(j)}_{i,S}$ are saturated. Notice that $|A^{(j)}_{i,S}|=|B^{(j)}_{i,S}|$, so there must be a flow from $A^{(j)}_{i,S}$ to $B^{(j)}_{i,S'}$ for some $S'\not=S$. This is a contradiction as $S,S'$ are at least $h_{\diam}\cdot s$ far, but the flow has length at most $h_{\diam}\cdot s$. 
    
    Thus, we must have $\hA^{(j)}_{i,\cS},\hB^{(j)}_{i,\cS}$ equal to $A^{(j)}_{i,\cS},B^{(j)}_{i,\cS}$ for every $i,j,\cS$. This means the matching perfectness is $1$. According to \Cref{lem:cut-strategy}, $H_S$ finally is a router for $A_S$ with $1/\eps$ step and $n^{O(\eps)}$ congestion.

\end{proof}

\paragraph{Witness.} Suppose $C^*$ is empty, then the algorithm does not output $C^*$, but an expander witness. The correctness follows from \Cref{lem:routers} and the above arguments. 

\begin{proof}[Proof of \Cref{lem:approxsparsecut}]

Now we assume $C^*$ is not empty and we will prove that the algorithm returns an approximately demand-size-largest sparse cut. 

We first show the complexity argument.

\paragraph{Complexity.} Step 1 uses \Cref{thm:neicov}, which costs $h\cdot n^{O(\eps)}$ depth and $h\cdot n^{O(\eps)}\cdot m$ work. 

Step 2 consists of $1/\eps$ iterations. For each iteration, there are three steps.

Firstly, it simultaneously runs cut strategies for all $S\in\cS\in\cN$. Notice that according to \Cref{thm:neicov}the total number of vertices and edges for all $S\in\cS\in\cN$ are at most $n^{O(\eps)}\cdot n$ and $n^{O(\eps)}\cdot m$ because each vertex in $V$ is included in at most $n^{O(\eps)}$ many different $S$. Thus, according to \Cref{lem:cut-strategy}, all the cut-strategies in total uses depth $n^{f(\eps)}\cdot\poly(h)$ and work $m\cdot n^{f(\eps)}\cdot\poly(h)$ where $f$ is a function that goes to $0$ as $\eps$ goes to $0$.

Secondly, it simultaneously compute a cutmatch for each $j\in[b]$ and $\cS\in\cN$ by applying \Cref{lem:matching-strategy}, in total costs $\tO{m\cdot\poly(h)\cdot n^{O(\eps)}}$ work and $\tO{\poly(h)}$ depth. 

Thirdly, it updates the graphs, which is straightforward by using the computed flow, so the work and depth are subsumed by other operations. The output is by union of all the computed cut, so the cost is subsumed by other costs.

\paragraph{Correctness.} We now proceed to prove the correctness. 

The intuition for the correctness are as follows. According to \Cref{lem:LDSCSAtMostLDSC}, i.e., demand-size largest sparse cut (DLSC) is roughly at most the largest expander's complement size (LEC). Recall from \Cref{def:LEC} that largest expander's complement size (LED) is roughly less than $A-\hat{A}$ for every expanding node-weighting $\hat{A}\preceq A$. Thus, our plan is to construct a large expanding node-weighting to serve as an upper bound for DLSC, and then the demand-size of the returned cut $C^*$ will not be too smaller than that value because of the guarantee from \Cref{lem:matching-strategy}. 

For every $S\in\cS\in\cN$, let the matching perfectness of the cut-matching game on $H_S$ be $(1-\alpha_S)$. In other words, for every iteration $i$, we have $\sum_{j}U\left(M^{(j)}_{i,S}\right)\ge (1-\alpha_S)\cdot\sum_j|A^{(j)}_{i,S}|$. Moreover, there exists an iteration $i^*_S$ such that $\sum_{j}U\left(M^{(j)}_{i^*_S,S}\right)= (1-\alpha_S)\cdot\sum_j|A^{(j)}_{i^*_S,S}|$.

\paragraph{Correctness: Constructing an expanding node-weighting.} According to \Cref{lem:cut-strategy}, for every $S\in\cS\in\cN$, there exists a node-weighting $A'_S\preceq A_S$ (here $A_S$ is the node-weighting $A$ restricting to $S$) with $|A'_S|\ge (1-O(\frac{\alpha_S}{\eps}))\cdot|A|$ such that $H_S$ in the end becomes a $1/\eps$-step and $n^{O(\eps)}$-router for $A'_S$. We define $\hat{A}$ in the following way. For every node $v\in V$, if there exists $S\in\cS$ such that $A'_S(v)\le A_S(v)/2$, let $\hat{A}(v)=0$, otherwise let $\hat{A}(v)=A_S(v)$. Clearly, we have $\hA\preceq A$. We will show that $\hA$ is expanding by using \Cref{thm:flow character}.

\begin{lemma}\label{lem:expandingoftA}
    $\hA$ is $(h_{\cov},s^{O(1/\eps)})$-length $\phi/n^{O(\eps)}$-expanding on $G$.
\end{lemma}
\begin{proof}
    According to the second point of \Cref{thm:flow character}, we only need to show that, for every $h_{cov}$-length demand $D$ that is $\hA$ respecting, we can route $D$ on $G$ with congestion $n^{O(\eps)}/\phi$ and length (dilation) at most $h_{\cov}\cdot s^{O(1/\eps)}$.

    For every positive demand $D(u,v)>0$, since $\dist_G(u,v)\le h_{\cov}$, according to the definition of neighborhood cover, there exists $S\in\cS\in\cN$ such that both $u,v$ are in $S$. Moreover, since $D$ is $\hA$ respecting, we have $A'_S(v)\ge A_S(v)/2$. We route the demand $D(u,v)$ by using the $1/\eps$-step $n^{O(\eps)}$-router $H_S$ with an additional congestion factor of at most $2$ (because $D(u,v)$ can be as large as $A_S(v)$ but in the router we can only have $A'_S(v)$). We will get a routing path on $H_S$ from $u$ to $v$ and we need to map it back to the original graph $G$, this can be done by mapping each edge in $H_S$ to a path in $G$: remember that each edge in $H_S$ corresponds to a flow in step (update graphs), which according to the step (compute a cutmatch), has length bound $s\cdot h_{\diam}=s^{O(1/\eps)}\cdot h$. This gives the length bound.
    
    To see the total congestion, notice that every edge in $H_S$ is used at most $n^{O(\eps)}$ times, where edges in $H_S$ can be decomposed into at most $(1/\eps)\cdot b=(1/\eps)\cdot n^{O(\eps)}=n^{O(\eps)}$ (remember that $\eps>\log^{0.1}n$) matchings, each matching can be mapped to flow paths of the original graph $G$ with vertex congestion $\tO{1/\phi'}=n^{O(\eps)}/\phi$ according to the usage of \Cref{lem:matching-strategy}. Thus, routing in a single $S$ cause a total vertex congestion of $n^{O(\eps)}/\phi$ in $G$. For a single $\cS\in\cN$, it is easy to see that different routings in difference $S\in\cS$ do not share vertices with each other: every two clusters in $\cS$ are $2s\cdot h_{\diam}$-far, and remember that every flow path for the routing is $s\cdot h_{diam}$-length bounded. So we can route all the demands contributed by $\cS\in\cN$ in vertex congestion of $n^{O(\eps)}/\phi$. Notice that $|\cN|$ is the width of the neighborhood cover, bounded by $n^{O(\eps)}$, so the total vertex congestion is $n^{O(\eps)}/\phi$.
\end{proof}

According to the definition of LEC \Cref{def:LEC}, we get the following inequality from \Cref{lem:expandingoftA}.

\[LEC(\phi/n^{O(\eps)},h_{\cov},s^{O(1/\eps)})\le \left(\phi/n^{O(\eps)}\right)\cdot(|A|-|\hA|)\]

By using \Cref{lem:DLSCAtMostLEC}, we can get an upper bound for DLSC.

\begin{equation}\label{eq:DLSC}
    DLSC\left(\phi/\left(s^3\cdot n^{O(\eps+1/s)}\right),h_{\cov}/2,s^{O(1/\eps)}\right)\le \left(s^3\cdot n^{O(\eps+1/s)}\right)\cdot (|A|-|\hA|)
\end{equation}

Next, we will show the cut returned has demand-size close to $(|A|-|\hA|)$. 

\paragraph{Correctness: Returned cut is approximately demand-size-largest sparse.} We first show that the returned cut $C^*$ is sparse. In the proof we will define a demand $D$ to show the sparsity of $C^*$. $D$ will also be used later to bound the demand size of $C^*$. 

\begin{lemma}\label{lem:C*issparse}
    If $C^*$ is not empty, $C^*$ is a $(h,s)$-length $\phi$-sparse cut with respect to the node-weighting $A$.
\end{lemma}
\begin{proof}
    Recall that $C^*$ is the union of cuts $C^{(j)}_{i,\cS}$ for $j\in[b],i\in[1/\eps]$ and $\cS\in\cN$, and $C^{(j)}_{i,\cS}$ is computed from \Cref{lem:matching-strategy} which is guaranteed to satisfy

    \[|C^{(j)}_{i,\cS}|\le \phi'\cdot \left(|A^{(j)}_{i,\cS}|-\val(F^{(j)}_{i,\cS})\right)\]

    Moreover, suppose $F^{(j)}_{i,\cS}$ routes the demand $\tA^{(j)}_{i,\cS}\preceq A^{(j)}_{i,\cS}$ to $\tB^{(j)}_{i,\cS}\preceq  B^{(j)}_{i,\cS}$. Then, the two vertex sets which are not saturated by $F^{(j)}_{i,\cS}$, i.e., $V_1=\{v\mid \tA^{(j)}_{i,\cS}(v)<A^{(j)}_{i,\cS}(v)\}$ and $V_2=\{v\mid \tB^{(j)}_{i,\cS}(v)<B^{(j)}_{i,\cS}(v)\}$, they are $s\cdot h_{\diam}$-far. Thus, we can define a demand $D^{(j)}_{i,\cS}$ from $V_1$ to $V_2$ that when combined with the node-weighting saturated by $F^{(j)}_{i,\cS}$, becomes the full node-weighting $A^{(j)}_{i,\cS}$. $D^{(j)}_{i,\cS}$ is an $h_{\diam}$-length demand separated by $C^{(j)}_{i,\cS}$ and thus $C^*$, with separation length $s\cdot h_{\diam}=s \cdot h$ and demand size $\left(|A^{(j)}_{i,\cS}|-\val(F^{(j)}_{i,\cS})\right)$. We define the demand $D$ as the union of all demands $D^{(j)}_{i,\cS}$ scale down by $n^{O(\eps)}\cdot (1/\eps)=n^{O(\eps)}$ (to make sure it is $A$ respecting, because $i$ can iterate over $1/\eps$ choices and $j,\cS$ can iterate over $n^{O(\eps)}$ choices). According to the definition of sparse cut \Cref{def:sparsity}, we can give the sparsity bound of $C^*$ as (remember $h_{\diam}=h$)

    \begin{align*}
        \spa_{(h,s)}(C^*,A)&\le \frac{|C^*|}{|D|}\\
        &\le \frac{\sum_{i,j,\cS}|C^{(j)}_{i,\cS}|}{\sum_{i,j,\cS}|D^{(j)}_{i,\cS}|}\cdot n^{O(\eps)}\\
        &\le \phi'\cdot n^{O(\eps)}\\
        &\le \phi
    \end{align*}

    The third inequality comes from the fact that $|C^{(j)}_{i,\cS}|/|D^{(j)}_{i,\cS}|=|C^{(j)}_{i,\cS}|/\left(|A^{(j)}_{i,\cS}|-\val(F^{(j)}_{i,\cS})\right)\le\phi'$ for all $i,j,\cS$ from \Cref{lem:matching-strategy}.
\end{proof}

Then we bound the demand-size of $C^*$. From the proof of \Cref{lem:C*issparse}, we have already shown that $D$ is separated by $C^*$ by length $h\cdot s$. $D$ is $A$-respecting $h$-length, and has size 

\begin{equation}\label{eq:sizeofD}
    |D|=\frac{1}{n^{O(\eps)}}\cdot \sum_{i,j,\cS}|D^{(j)}_{i,\cS}|= \frac{1}{n^{O(\eps)}}\cdot \sum_{i,j,\cS}\left(\left(\sum_{S\in\cS}|A^{(j)}_{i,S}|\right)-\val(F^{(j)}_{i,\cS})\right)
\end{equation}

Recall the definition of $\hA$: a vertex $v$ has $\hA(v)<A(v)$ only happens when $\hA(v)=0$ and there exists some $S\in\cS\in\cN$ such that $A'_S(v)\le A_S(v)/2$. The definition of $A'_S$ is from \Cref{lem:cut-strategy}, i.e., an expanding sub-node-weighting of $A_S$ which has the size guarantee to be at least $(1-\frac{O(\alpha_S)}{\eps})|A_S|$ where $\alpha_S$ is the matching perfectness on the cut-matching game of $S$. Thus, every $v$ with $\hA(v)<A(v)$ contributes to the difference $\frac{O(\alpha_S)}{\eps}|A_S|$ between $|A_S|$ and $|A'_S|$ for some $S$, so we have
\[|A|-|\hA|\le \sum_{S\in\cS\in\cN}\frac{O(\alpha_S)}{\eps}|A_S|\]

According to the definition of matching perfectness, there exists an iteration $i^*_S$ for every $S$ such that $\sum_{j}U\left(M^{(j)}_{i^*_S,S}\right)= (1-\alpha_S)\cdot\sum_j|A^{(j)}_{i^*_S,S}|$, i.e., in that iteration we get

\begin{equation}\label{eq:alphas}
    \alpha_S\cdot\sum_j|A^{(j)}_{i^*_S,S}|=\sum_j|A^{(j)}_{i^*_S,S}|-\sum_{j}U\left(M^{(j)}_{i^*_S,S}\right)
\end{equation}

Recall from the definition of cut strategy that $\sum_j|A^{(j)}_{i^*_S,S}|\ge |A_S|/2$. We can take the sum over all $S$ and all iterations $i$, which gives us
\begin{align*}
    |A|-|\hA|&\le \sum_{S\in\cS\in\cN}\frac{O(\alpha_S)}{\eps}|A_S|\\
    &\le O(\frac{1}{\eps})\cdot \sum_{S\in\cS\in\cN}\alpha_S\cdot \sum_j|A^{(j)}_{i^*_S,S}|\\
    &\le^{\ref{eq:alphas}} O(\frac{1}{\eps})\cdot \sum_{S\in\cS\in\cN}\left(\sum_j|A^{(j)}_{i^*_S,S}|-\sum_{j}U\left(M^{(j)}_{i^*_S,S}\right)\right)\\
    &\le O(\frac{1}{\eps})\cdot \sum_{S\in\cS\in\cN}\sum_{i}\left(\sum_j|A^{(j)}_{i,S}|-\sum_{j}U\left(M^{(j)}_{i,S}\right)\right)\\
    &\le O(\frac{1}{\eps})\cdot \left(\sum_{i,j,S}|A^{(j)}_{i,S}|-\sum_{i,j,S}U\left(M^{(j)}_{i,S}\right)\right)\\
    &=O(\frac{1}{\eps})\cdot \left(\sum_{i,j,S}|A^{(j)}_{i,S}|-\sum_{i,j,\cS}\val\left(F^{(j)}_{i,\cS}\right)\right)
\end{align*}
\end{proof}

We can compare the last line with the size of $D$ as shown in \Cref{eq:sizeofD}, and get

\[|A|-|\hA|\le \frac{n^{O(\eps)}}{\eps}|D|\le n^{O(\eps)}\cdot |D| \]

Combine with \Cref{eq:DLSC}, we get

\[DLSC\left(\phi/\left(s^3\cdot n^{O(\eps+1/s)}\right),h_{\cov}/2,s^{O(1/\eps)}\right)\le \left(s^3\cdot n^{O(\eps+1/s)}\right)\cdot |D|\]

Recall that $D$ is $h$-length which is $hs$-far in $G-C^*$. Thus, we have that 

\[A_{(h,s)}(C^*)\ge |D|\ge \frac{1}{s^3\cdot n^{O(\eps+1/s)}}\cdot DLSC\left(\phi/\left(s^3\cdot n^{O(\eps+1/s)}\right),h/s^{O(1/\eps)},s^{O(1/\eps)}\right)\]

This proves the correctness of \Cref{lem:approxsparsecut} (remember that both $A_{(h,s)}(C)$ and $DLSC$ increases as $h$ increases or $s$ decreases by \Cref{lem:monotonicity}).

\section{Expander Decomposition from Approximately DLSC}\label{sec:expanderfromDLSC}

In this section, we will show how to get length-constraint expander decomposition by repeatedly calling and cutting approximately demand-size-largest sparse cut. We will prove the following theorem.

\vertexLCED*

Notice that if $\eps$ is too small, say less than $o(1/\log\log n)$, then the length slackness becomes an arbitrary large polynomial on $n$. In this case, \Cref{thm:vertexLC-ED} is trivial since every cut cannot separate any demand as $h\cdot s$ is too large, so $G$ is an expander anyway. Thus, in what follows we assume $\eps=\Omega(1/\log \log n)$. 

Recall that from \Cref{lem:approxsparsecut} and the definition \Cref{dfn:apxLargeCut}, we can use $\textsc{ApxDLSC}(G,A,h,s,\phi,\eps)$ to get a $(h,s)$-length $\phi$-sparse cut $C$ such that the demand size is at least $\frac{1}{\alpha}$ fraction of the largest demand size among all $(h/\alpha_s,s\cdot \alpha_s)$-length $\phi/\alpha_\phi$-sparse cut. Here the parameters $\alpha,\alpha_s,\phi_\phi$ depend on $n,\eps,s$ can be found in \Cref{lem:approxsparsecut}, we state them here for convenience. 

\ApxSparsecut*
To recall the definition of approximate sparse cut, we restate the definition as follows.

\apxsparsecut*

\paragraph{Intuition.} The algorithm repeatedly finds and cuts the approximately demand-size-largest $(h,s)$-length $\phi$-sparse cut until the demand size decreases by a factor of $1/n^{O(\eps)}$. We can prove that there will be not too many such cuts because of \Cref{lem:LDSCSAtMostLDSC}. For convenience, we restate the lemma as below.

\LDSCSAtMostLDSC*

The demand-size of a cut sequence is not too large than the demand-size of one cut. We call the above procedure an `epoch'. One can see that after this epoch, the demand-size-largest $(h/\alpha_s,s\cdot \alpha_s)$-length $\phi/\alpha_\phi$-sparse cut get reduced by a factor of $1/n^{O(\eps)}$. We update $h,s,\phi$ to the corresponding new parameters and continue to the next epoch to further reduce the demand-size. After $O(1/\eps)$ epochs, the graph will have no sparse cut and it becomes a good expander. The detailed algorithm is stated as follows.

\Cref{lem:approxsparsecut} did not specify the constant hidden inside $O$ for $\alpha,\alpha_\phi,\alpha_s$. We use $c$ in the algorithm to denote a sufficiently large constant that is larger than every constant hidden inside these $O$. The same holds for the constants hidden inside \Cref{lem:LDSCSAtMostLDSC}. This is to prevent potential confusion.

\newcommand{\li}{\l}
\begin{algorithm}
\caption{$C \leftarrow \textsc{VertexLC-ED}(G,A,h,\phi)$}
\label{alg:EDsfromCuts}
\begin{algorithmic}[1]
    \State $z \gets n^{2c\sqrt{\eps}},\quad \li_{0,z} \gets 0$
    \State $h_0 \gets h\cdot (1/\eps)^{(3c/\eps)^{1/\sqrt{\eps}}},\quad s_0 \gets 1/\eps,\quad \phi_0 \gets \phi \cdot n^{5c\sqrt{\eps}}$
    \For{$\epo = 1, \ldots, 1/\sqrt{\eps}$}
        \State $\li_{\epo,0} \gets \li_{\epo-1,z}$
        \For{$j = 1, \ldots, z$}
            \State $C_{\epo,j} \gets \textsc{ApxDLSC}(G_{\epo,j-1},A,h_\epo,s_\epo,\phi_\epo,\eps)$
            \State $\li_{\epo,j} \gets \li_{\epo,j-1}+h_{\epo}\cdot C_{\epo,j}$
            \State $G_{\epo,j} \gets G+\li_{\epo,j}$
        \EndFor
        \State $h_{\epo+1} \gets h_{\epo}/s_{\epo}^{2c/\eps}$
        \State $s_{\epo+1} \gets s_{\epo}\cdot s_{\epo}^{2c/\eps}$
        \State $\phi_{\epo+1} \gets \phi_{\epo}/n^{5c\eps}$
    \EndFor
    \State \textbf{return} $\li_{1/\sqrt{\eps},z}/h$
\end{algorithmic}
\end{algorithm}

\begin{proof}[Proof of \Cref{thm:vertexLC-ED}]
    We use the algorithm as described in \Cref{alg:EDsfromCuts}. We first show the correctness. The correctness follows from the following lemma, describing the invariance throughout the algorithm. 
    
    \begin{lemma}\label{lem:invariance}
        At the end of the $\epo$ iteration, we have 
        \[DLSC_{G_{\epo,z}}\left(\phi_{\epo+1}\cdot n^{2c\eps},2h_{\epo+1},\frac{s_{\epo+1}-2}{2}\right)\le |A|\cdot \left(\frac{1}{n^{c\sqrt{\eps}}}\right)^\epo\]
    \end{lemma}
    \begin{proof}
        We prove it by induction on $\epo$. When $\epo=0$, we view `the end of the $0$ iteration' as `the start of the $1$ iteration'. The induction hypothesis holds because the demand-size of every cut is at most $|A|$ according to the definition of demand-size of a cut \Cref{def:demandSize}. 

        Suppose $\epo\ge 1$. Suppose to the contrary, at the end of iteration $\epo$, we still have $DLSC_{G_{\epo,z}}(\phi_{\epo+1}\cdot n^{2c\eps},2h_{\epo+1},(s_{\epo+1}-2)/2)>|A|\cdot \left(1/n^{c\sqrt{\eps}}\right)^\epo$. Consider the $\epo$ iteration. Notice that according to the monotonicity \Cref{lem:monotonicity}, we have $DLSC_{G_{\epo,j}}(2\phi_{\epo+1}\cdot n^{2c\eps},2h_{\epo+1},(s_{\epo+1}-2)/2)>|A|\cdot \left(1/n^{c\sqrt{\eps}}\right)^\epo$ throughout the whole $\epo$ iteration since $G_{\epo,j}$ only increases its length as $j$ increases. 
        According to the output guarantee of \textsc{ApxDLSC} \Cref{lem:approxsparsecut}, the demand size of $C_{\epo,j}$ for every $j\in[z]$ is at least (notice that $s_{\epo}>1/\eps$ and $s_{\epo}<(1/\eps)^{(3c/\eps)^{1/\sqrt{\eps}}}<n^{\eps}$ for sufficiently large $n$ because of $\eps>\Omega(1/\log\log n)$)
        \begin{align*}
        &A_{G_{\epo,j-1}}(h_{\epo},s_{\epo},C_{\epo,j})\\
            \ge&\frac{1}{s^3_{\epo}\cdot n^{c\left(\eps+\frac{1}{s_{\epo}}\right)}}\cdot DLSC_{G_{\epo,j-1}}\left(\frac{\phi_{\epo}}{s^3_{\epo}\cdot n^{c\left(\eps+\frac{1}{s_{\epo}}\right)}},\frac{h_\epo}{s_{\epo}^{c/\eps}},s_\epo\cdot s_{\epo}^{c/\eps}\right)\\
            \ge&\frac{1}{n^{3c\eps}}\cdot DLSC_{G_{\epo,j-1}}\left(\frac{\phi_{\epo}}{n^{3c\eps}},\frac{h_\epo}{s_{\epo}^{c/\eps}},s_\epo\cdot s_{\epo}^{c/\eps}\right)\\
            \ge&\frac{1}{n^{3c\eps}}\cdot DLSC_{G_{\epo,j-1}}\left(2\phi_{\epo+1}\cdot n^{2c\eps},2h_{\epo+1},\frac{s_{\epo+1}-2}{2}\right)\\
            \ge&\frac{|A|}{n^{3c\eps}}\cdot \left(\frac{1}{n^{c\sqrt{\eps}}}\right)^{\epo}
        \end{align*}
        
        We have that $(C_{\epo,0},C_{\epo,1},...)$ is a sequence of $(h_\epo,s_\epo)$-length $\phi_\epo$-sparse moving cut of graph $G_{\epo-1,z}$. Thus, according to \Cref{lem:LDSCSAtMostLDSC}, we get that (remember that $s_{\epo}>1/\eps$ and $s_{\epo}<(1/\eps)^{(3c/\eps)^{1/\sqrt{\eps}}}<n^{\eps}$ for sufficiently large $n$ because of $\eps>\Omega(1/\log\log n)$)
        
        \begin{align*}
            &\sum_{j\in[z]}A_{G_{\epo,j-1}}(h_{\epo},s_{\epo},C_{\epo,j})\\
            \le & s_{\epo}^3\cdot \log ^3n\cdot n^{c/s_\epo}\cdot DLSC_{G_{\epo-1,z}}\left(s^3\cdot \log ^3n\cdot n^{c/s_\epo}\cdot n^{O(\eps)},2h_{\epo},\frac{s_{\epo}-2}{2}\right)\\
            \le & n^{2c\eps}\cdot DLSC_{G_{\epo-1,z}}\left(\phi_{\epo}\cdot n^{2c\eps},2h_{\epo},\frac{s_{\epo}-2}{2}\right)\\
            \le &n^{2c\eps}\cdot |A|\cdot \left(\frac{1}{n^{c\sqrt{\eps}}}\right)^{\epo-1}
        \end{align*}

        The last inequality is from the induction hypothesis. Combining the two inequalities, we get

        \[z\cdot \frac{|A|}{n^{3c\eps}}\cdot \left(\frac{1}{n^{c\sqrt{\eps}}}\right)^{\epo}\le n^{2c\eps}\cdot |A|\cdot \left(\frac{1}{n^{c\sqrt{\eps}}}\right)^{\epo-1}\]

        This is a contradiction as $z$ is as large as $n^{2c\sqrt{\eps}}$. 
    \end{proof}

    \Cref{lem:invariance} shows that at the end of the algorithm, when $\epo=1/\sqrt{\eps}$, we get (remember that size of $A$ scaling up to a integer is polynomial on $n$ according to our assumption)
    \[DLSC_{G_{\epo,z}}\left(\phi_{\epo+1}\cdot n^{2c\eps},2h_{\epo+1},\frac{s_{\epo+1}-2}{2}\right)<\min_v A(v)\]

    Also, notice that according to the update rule of $h_\epo,s_\epo,\phi_{\epo}$, we have
    \[s_{\epo}\le (s_0)^{(2c/\eps+1)^{\epo}}\le (1/\eps)^{(2c/\eps+1)^{1/\sqrt{\eps}}}=O(\exp(\exp(O(\log (1/\eps)/\sqrt{\eps})))\]
    \[h_{\epo}\ge h_0/\prod_{\epo\le 1/\eps}(1/\eps)^{(2c/\eps+1)^{\epo}}\ge h_0/ (1/\eps)^{2\cdot (2c/\eps+1)^{1/\sqrt{\eps}}}\ge h\]
    \[\phi_{\epo}\ge \phi_0/\left(n^{5c\eps}\right)^{1/\sqrt{\eps}}\ge \phi\]

    Thus, according to the monotonicity \Cref{lem:monotonicity}, by setting $s=\exp(\exp(O(\log (1/\eps)/\sqrt{\eps}))$, we get
    \[DLSC_{G_{1/\sqrt{\eps},z}}(\phi,h,s)<\min_v A(v)\]

    In other words, on the graph $G_{1/\sqrt{\eps},z}$, there is no $(h,s)$-length $\phi$-sparse cut (since demand size of a sparse cut should be at least $\min_v A(v)$). Thus, $G_{1/\sqrt{\eps},z}$ is a $(h,s)$-length $\phi$-expander for $A$. Moreover, $G_{1/\sqrt{\eps},z}=G+\ell_{1/\sqrt{\eps},z}$. If we interpreted the returned cut $C:=\ell_{1/\sqrt{\eps},z}/h$ as a $h\cdot s$-length cut, then $G_{1/\sqrt{\eps},z}=G-C$ is $(h,s)$-length $\phi$-expanding for $A$ (here we again use the monotonicity \Cref{lem:monotonicity}). Thus, the returned cut $C$ is a $(h,s)$-length $\phi$-expander decomposition. 

    Now analyze the size of the returned cut. Notice that $\ell_{1/\sqrt{\eps,z}}$ is the summation of $h_{\epo}\cdot C_{\epo,j}$ for $\epo\in[1/\sqrt{\eps}],j\in[z]$, where $C_{\epo,j}$ is a $(h_{\epo},s_{\epo})$-length $\phi_{\epo}$-sparse cut according to \Cref{lem:approxsparsecut}. Thus, we have that $|C_{\epo,j}|\le |A|\cdot \phi_{\epo}$ according to \Cref{def:sparsity}. Now we can give an upper bound for the returned cut $C$ through the following calculations.
    \begin{align*}
        |C|&\le \frac{1}{h}\cdot \sum_{\epo\in[1/\sqrt{\eps}],j\in[z]}|C_{\epo,j}|\cdot h_{\epo}\\
        &\le \frac{1}{h}\cdot \sum_{\epo\in[1/\sqrt{\eps}],j\in[z]}|A|\cdot \phi_{\epo}\cdot h_{\epo}\\
        &\le \frac{|A|}{h}\cdot \sum_{\epo\in[1/\sqrt{\eps}],j\in[z]}\phi\cdot n^{5c\sqrt{\eps}}\cdot h\cdot (1/\eps)^{(3c/\eps)^{1/\sqrt{\eps}}}\\
        &\le |A|\cdot \phi\cdot \left(\frac{1}{\sqrt{\eps}}\cdot n^{2c\sqrt{\eps}} \cdot n^{5c\sqrt{\eps}}\cdot (1/\eps)^{(3c/\eps)^{1/\sqrt{\eps}}}\right)\\
        &\le |A|\cdot \phi\cdot n^{O(\sqrt{\eps})}\\
    \end{align*}

    So the cut slack $\kappa=n^{O(\sqrt{\eps})}$.

    \paragraph{Witness.} We just argued that at the end of the algorithm, the graph $G_{1/\sqrt{\eps},z}$ contains no $(h,s)$-length $\phi$-sparse cut. According to \Cref{lem:approxsparsecut}, the algorithm must return a witness with length

    \[h_{1/\sqrt{\eps}}/s^{O(1/\eps)}\ge h\]

    length slackness 
    
    \[s_{1/\sqrt{\eps}}/\eps\le s\]

    and expansion

    \[\phi_{1/\sqrt{\eps}}/n^{O(\eps)}\ge \phi\]

    So it is a $(h,s)$-length $\phi$-expansion witness for $G_{1/\sqrt{\eps}}=G-C$. %
    
    \paragraph{Complexity.} It remains to argue the complexity of \Cref{alg:EDsfromCuts}. The algorithm consists of $1/\sqrt{\eps}=O(\log \log n)$ interactions, each iteration consists of $z=n^{2c\sqrt{\eps}}$ iterations where each of them contains a call to \textsc{ApxDLSC}. According to \Cref{lem:approxsparsecut}, the work for one call is $\poly(h_{\epo})\cdot m\cdot n^{f(\eps)}$ and depth is $\poly(h_{\epo})\cdot n^{f(\eps)}$ where $f$ is a function that goes to $0$ as $\eps$ goes to $0$. Notice that the number of iterations is subsumed by $n^{f(\eps)}$. Moreover, since we have $h_{\epo}\le h\cdot (1/\eps)^{(3c/\eps)^{1/\sqrt{\eps}}}$ and $\eps>1/\sqrt{\log \log n}$, we get that $\poly(h_{\epo})=\poly(h)\cdot n^{O(\eps)}$.

    The theorem follows by setting $\eps=\sqrt{\eps}$. 
\end{proof}
\end{document}